\documentclass{article}

\usepackage[affil-it]{authblk}
\usepackage[usenames,dvipsnames]{xcolor}
\usepackage{amsfonts}
\usepackage{amsmath,amsthm,amssymb,dsfont}
\usepackage{enumerate}
\usepackage{graphicx}	
\usepackage{subcaption}
\usepackage[margin=3cm]{geometry}
\usepackage{url}
\usepackage{todonotes}
\usepackage{bbm}

\usepackage{tikz}

\usepackage{pifont}
\usepackage{multirow}
\usepackage{makecell}

\usepackage{epsfig}
\usetikzlibrary{shapes.symbols,patterns} % for source symbols
\usetikzlibrary{decorations.pathreplacing,angles,quotes}
\usepackage{pgfplots}
\pgfplotsset{compat=1.10}
\usepgfplotslibrary{fillbetween}

\definecolor{linkblue}{HTML}{001487}
\usepackage{hyperref}
\hypersetup{colorlinks=true,citecolor=linkblue,linkcolor=linkblue,filecolor=linkblue,urlcolor=linkblue,breaklinks=true}

\usepackage{nicefrac}
\usepackage{mathtools}
\usepackage[nameinlink,capitalize,noabbrev]{cleveref}

\usepackage{algorithm}
\usepackage{algorithmic}

\usepackage{mdframed}
\usepackage{aligned-overset}
\theoremstyle{plain}
\newtheorem{theorem}{Theorem}[section]
\newtheorem{lemma}[theorem]{Lemma}

\newtheorem{corollary}[theorem]{Corollary}
\newtheorem{proposition}[theorem]{Proposition}

\theoremstyle{definition}
\newtheorem{definition}[theorem]{Definition}
\newtheorem{definition-2}[theorem]{Definition}
\newtheorem{remark}[theorem]{Remark}
\newtheorem{example}[theorem]{Example}
\newtheorem{question}[theorem]{Question}
\newtheorem{notation}[theorem]{Notation}
\newtheorem{assumption}[theorem]{Assumption}

\usepackage{tcolorbox}
\tcbuselibrary{breakable}
\tcolorboxenvironment{definition}{breakable, pad at break*=2mm}
\tcolorboxenvironment{definition-2}{colback=white,colframe=black!30!white}
\tcolorboxenvironment{theorem}{colback=pink!25!white,colframe=pink!100!black}
\tcolorboxenvironment{lemma}{breakable, pad at break*=2mm, colback=yellow!5!white,colframe=yellow!75!black}
\tcolorboxenvironment{corollary}{colback=Dandelion!5!white,colframe=Dandelion!75!black}
\tcolorboxenvironment{proposition}{colback=Emerald!5!white,colframe=Emerald!75!black}
\tcolorboxenvironment{claim}{colback=RoyalBlue!5!white,colframe=RoyalBlue!75!black}
\tcolorboxenvironment{conjecture}{colback=red!5!white,colframe=red!75!black}
\tcolorboxenvironment{example}{colback=white,colframe=lightgray}
\tcolorboxenvironment{remark}{colback=white,colframe=lightgray}
\tcolorboxenvironment{question}{colback=SkyBlue!5!white,colframe=lightgray!85!SkyBlue}
\tcolorboxenvironment{fact}{colback=Purple!5!white,colframe=Purple!75!black}
\tcolorboxenvironment{notation}{colback=white,colframe=black!30!white}
\tcolorboxenvironment{assumption}{colback=CornflowerBlue!20!white,colframe=CornflowerBlue!75!black}
\crefname{assumption}{Assumption}{Assumptions}
\crefname{notation}{Notation}{Notations}

\DeclareRobustCommand{\abbrevcrefs}{%
\Crefname{theorem}{Thm.}{Thms.}%
\Crefname{corollary}{Cor.}{Cors.}%
\Crefname{lemma}{Lem.}{Lems.}%
\Crefname{proposition}{Prop.}{Props.}%
\Crefname{equation}{Eq.}{Eqs.}%
\Crefname{example}{Ex.}{Exs.}%
\Crefname{remark}{Rmk.}{Rmks.}%
}

\DeclareRobustCommand{\Cshref}[1]{{\abbrevcrefs\Cref{#1}}}

\newcommand*{\E}{\mathbb{E}}

\newcommand*{\cA}{\mathcal{A}}
\newcommand*{\cB}{\mathcal{B}}

\newcommand*{\cD}{\mathcal{D}}
\newcommand*{\cE}{\mathcal{E}}
\newcommand*{\cF}{\mathcal{F}}

\newcommand*{\cH}{\mathcal{H}}

\newcommand*{\cK}{\mathcal{K}}
\newcommand*{\cN}{\mathcal{N}}
\newcommand*{\cM}{\mathcal{M}}
\newcommand*{\cP}{\mathcal{P}}
\newcommand*{\cQ}{\mathrm{Q}}
\newcommand*{\cR}{\mathcal{R}}
\newcommand*{\cS}{\mathcal{S}}
\newcommand*{\cT}{\mathcal{T}}
\newcommand*{\cU}{\mathcal{U}}
\newcommand*{\cV}{\mathcal{V}}
\newcommand*{\cW}{\mathcal{W}}
\newcommand*{\cX}{\mathcal{X}}

\newcommand*{\N}{\mathbb{N}}

\newcommand*{\R}{\mathbb{R}}

\newcommand*{\C}{\mathbb{C}}

\newcommand*{\LOCC}{\mathsf{LOCC}}

\newcommand*{\St}{\mathrm{S}}

\newcommand*{\reg}{\mathrm{reg}}

\newcommand*{\eps}{\varepsilon}

\newcommand*{\id}{\mathds{1}}
\newcommand*{\poly}{\mathrm{poly}}

\newcommand*{\supp}{\mathrm{supp}}
\newcommand*{\tr}{\mathrm{tr}}
\newcommand*{\ket}[1]{| #1 \rangle}
\newcommand*{\bra}[1]{\langle #1 |}

\newcommand{\proj}[1]{|#1\rangle\!\langle #1|}
\newcommand*{\braket}[1]{\langle #1 \rangle}

\newcommand*{\Pos}{\mathrm{P}}

\newcommand*{\di}{\mathrm{d}} % integration d

\newcommand{\norm}[1]{\left\lVert#1\right\rVert}

\newcommand*{\ox}{\otimes}

\DeclareMathOperator*{\argmax}{arg\,max}

\usepackage{float}

 \allowdisplaybreaks

\title{Operational transformation rates of quantum states}

  \author{\normalsize Giulia Mazzola and Renato Renner}
   \affil{\small Institute for Theoretical Physics, ETH Zurich\\
  }
 \date{}

\newcommand*{\OP}{\mathrm{OP}}
\usepackage{mathrsfs}
\newcommand*{\F}{\mathscr{F}}
\newcommand{\subsetsim}{\mathrel{\underset{\raisebox{0.2ex}{$\sim$}}{\subset}}}
\newcommand{\insim}{\mathrel{\underset{\raisebox{0.2ex}{$\sim$}}{\in}}}

\usepackage{stmaryrd}

\begin{document}

\maketitle

 \begin{abstract}
Typically, multiple copies of a quantum state are described by a tensor power of that state. However, from an operational point of view, an experimentalist cannot, in general, rule out imperfections at unknown positions or correlations between copies, so assuming independent and identically distributed (iid) copies is not justified a priori. In the following, we relax the iid assumption by considering a whole set of resource states that are fundamentally indistinguishable from a tensor power state, and in that sense look like many copies of the quantum state to the experimentalist. Using a de Finetti-type argument, we show that this operational set of states can asymptotically be represented by almost-iid states. As an application, we introduce operational analogues of asymptotic transformation rates, such as entanglement cost and distillation. This raises the question of whether these rates coincide with their iid counterparts.

%As an application, we introduce a modified notion of asymptotic state transformation rates, such as entanglement cost and distillation, that accounts for this more general structure. We investigate how these operational rates relate to their iid counterparts.

%As an application, we study a modified notion of asymptotic state transformation rates, such as entanglement cost and distillation, that accounts for this more general structure.
 %  However, from an operational point of view, where an experimentalist \red{has access to each individual copy separately}, this iid assumption is a priori not justified. 
 \end{abstract}
 
\subsection*{Introduction}
Many tasks in quantum information theory are formulated in terms of multiple copies of a quantum state. For instance, hypothesis testing, data compression, communication, quantum cryptography, and resource manipulation all process long sequences of quantum systems, which are typically represented by a tensor power $\rho^{\otimes n}$ of a fixed state $\rho$. This independent and identically distributed (iid) model is mathematically convenient and often physically well motivated. From a fundamental standpoint, however, an exact tensor power structure cannot in general be certified: an experimentalist performing measurements on the systems emitted by a source cannot rule out correlations among them, nor exclude a small number of defective systems at unknown positions. Thus, even when a source consistently appears to produce the state $\rho$, its joint $n$-system state $\rho^{(n)}$ need not equal $\rho^{\otimes n}$ (see~\cref{fig:intro-iid-source}). Already in classical probability theory, de Finetti cautioned that the ``most common and misleading error" is to regard the iid assumption as fundamental~\cite{dF70}.

\begin{figure}[htb!]
    \centering
    \includegraphics[width=0.44\columnwidth]{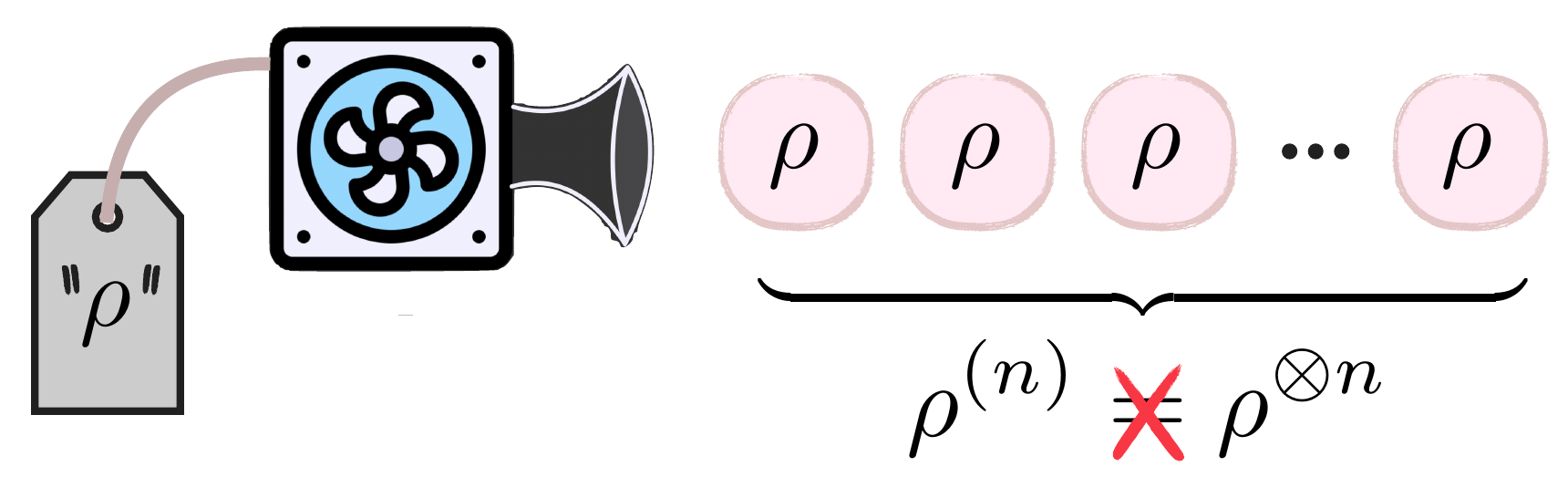}
    \caption{A $\rho$-source is expected to emit many copies of a state $\rho$. While tomography may consistently identify the individual outputs with $\rho$, it cannot in general certify that their joint state $\rho^{(n)}$ has the exact tensor power form $\rho^{\otimes n}$.}
    \label{fig:intro-iid-source}
\end{figure}
\vspace{-0.25cm}
To make the motivation concrete, consider the following examples. The experimentalist receives a single $n$-system block in a state $\rho^{(n)}$, with no additional copies, and wishes to use it as a resource for further information processing.\footnote{Appealing to repeated preparations of the entire block, e.g., for standard tomography, would introduce an independence assumption at a higher level – one would then need to justify that the blocks themselves are iid, and the same verification question arises again.} Verification should therefore preserve most of that resource and ideally use only free operations or a negligible amount of additional resources. The underlying limitation for such verification already arises in classical settings, such as coin tossing.
%Sacrificing a small sample is compatible with this aim; consuming or replacing the entire resource is not. 
%To make the motivation concrete, consider the following examples. We regard $\rho^{(n)}$ as a resource for further information-processing tasks. Verification should therefore preserve most of that resource and ideally use only free operations or a negligible amount of additional resources. Sacrificing a small sample is compatible with this aim; consuming or replacing the entire resource is not. 
\begin{enumerate}[(i)]
\item \textit{Fair coins with a parity constraint.} Consider an $n$-bit string drawn uniformly from all even-parity strings. Every proper subset consists of independent fair (iid) bits, although the full string is correlated. Even reading the entire string does not determine its generating law with certainty, since one observed even-parity string is compatible with iid fair coins as well as the even-parity distribution. In an even more extreme case, a deterministic source that always emits one fixed string is indistinguishable from iid fair coins given a single observation.
%In a more extreme case, a single observed string is even compatible with a deterministic source that always emits that particular string. 
Reading the string again, or making copies of it, does not provide fresh samples from the generating distribution and thus, is not helpful.

These distinctions have direct operational consequences. If the string is used as a one-time-pad key in cryptography, an adversary who knows that the source produces only even-parity strings can infer the parity of the original message with certainty from a single encrypted message. If the source instead always emits a deterministic fixed string known to the adversary, the entire original message can be recovered.

\item \textit{Fair coins with unknown defects.} Consider the task of distinguishing $n$ iid fair coins from a family of sources that produce independent fair bits except at $k = o(\sqrt{n})$ positions, where they always output heads. If the defective positions are known, one can test the iid hypothesis by inspecting those $k$ bits and accepting if at least one shows tails. This test accepts the iid source with probability $1-2^{-k}$ and rejects the defective source with certainty.

If the defective positions are unknown, the test must instead work uniformly over all possible placements. No single test can achieve reliable discrimination in this setting, even when the entire string is available. Indeed, for iid fair coins, the number of heads is centered around $\frac{n}{2}$, with fluctuations of order $\sqrt n$. The $k$ defected coins shift this mean by $\frac{k}{2}$, which is smaller than the natural fluctuations when $k=o(\sqrt n)$. More precisely, the total variation distance between the iid fair-coin distribution and the uniform mixture over all defect placements tends to zero. Consequently, any test that accepts iid fair coins with probability at least $1-\delta$ must also accept some defective source with probability at least $1-\delta-o(1)$.

%the uniform mixture over all placements approaches the iid fair-coin distribution in total variation whenever $k=o(\sqrt n)$. Consequently, any test that accepts iid fair coins with probability at least $1-\delta$ must also accept some defective source with probability at least $1-\delta-o(1)$.

%\item \textit{Fair vs.~deterministic coins.}

\item \textit{Copies of a mixed state.} Let $\rho=\frac{1}{2}\id_2$ be the maximally mixed qubit state. Consider an accept/reject measurement on an $n$-system block, with acceptance POVM element $0\le E\le\id$. If it accepts the
iid target with certainty, then $\tr[E\rho^{\otimes n}]=1$. Since $\rho^{\otimes n}$ has full rank, this forces $E=\id$. The test therefore accepts every input state with certainty, including arbitrarily correlated states. In contrast to a projection onto a pure target, a test that always accepts this mixed target cannot reject any alternative. Even if a small rejection probability $\delta$ was allowed on the iid target, every such test would necessarily accept some highly correlated states with probability at least $1-\delta$. This is the quantum analogue of the limitation illustrated by the classical coin examples. 
%In both cases, a test that accepts the iid target with certainty cannot rule out alternative preparations.
%\item \text{Copies of a mixed state.} Let $\rho = \frac{1}{2}\id$ denote the maximally mixed state, and consider any accept/reject protocol on an n-block system that is represented by a POVM element $0\leq E \leq \id$ that accepts the iid state $\rho^{\otimes n}$ with certainty. Because $\rho^{\otimes n}$ is a full-rank state, it immediately follows that $E = \id$ and thus, the protocol accepts any input multi-copy state with certainty. Therefore, contrary to the pure state case, no test can verify the iid state with certainty.
\end{enumerate}
In the preceding examples, the obstruction to certifying an iid structure stems from statistical uncertainty. For pure quantum states, on the other hand, the challenge lies in the resource cost of verification, as the following examples illustrate.

 \begin{enumerate}[(i)]
 \setcounter{enumi}{3}
\item \textit{Copies of a pure Bell state.} Consider the resource state $\ket{\Phi^+}^{\otimes n} $ where \smash{$\ket{\Phi^{+}}:= \frac{1}{\sqrt{2}}(\ket{00}+\ket{11})$} denotes a maximally entangled Bell state. An ideal nondestructive implementation of the projection $P_n =  \proj{\Phi^+}^{\otimes n} $ yields exactly $n$ iid Bell pairs upon success. When the pairs are shared between two distant parties, however, this projection cannot be implemented using local operations and classical communication alone. In fact, any exact nondestructive implementation must consume at least $n$ shared Bell pairs. On the other hand, local measurements can test Bell correlations without additional shared entanglement, but destroy the entanglement in the measured pairs.

Nevertheless, approximate verification with a small error may require fewer resources. In particular, the techniques developed in this work yield such a procedure for iid Bell pairs without consuming additional shared entanglement. This is achieved by combining local tomography with a decoupling procedure (see~\cref{sec:LOCC_rates}).
 %\GM{Add the general resource argument in words or with formulas.}
%Allowing a small verification error can reduce this resource cost. The techniques developed in this work provide an approximate verification procedure for iid Bell pairs that requires no additional shared entanglement.

\item \textit{Tomography on a few systems.} Replace one of the $n$ intended $\proj{\Phi^+}$ Bell pairs from the previous example by $\proj{\Phi^-}$ at a uniformly random position, where \smash{$\ket{\Phi^{-}}:= \frac{1}{\sqrt{2}}(\ket{00}-\ket{11})$}. The resulting joint state is orthogonal to the ideal tensor power, yet the reduced state of any $m$ pairs has trace distance only $\frac{m}{n}$ from the ideal sample.
%while a uniformly chosen sample of $m$ pairs contains the defect with probability only $\frac{m}{n}$ and therefore, its reduced state has trace distance exactly $\frac{m}{n}$ from the ideal sample. 
Thus, even an optimal measurement confined to a sublinear sample, $m=o(n)$, has a vanishing distinguishing advantage. %and is thus unable to make the complete resource perfectly distinguishable from the ideal tensor power. 
This limits what can be concluded from tomography performed on a few sampled pairs. A procedure that returns an estimate close to $\proj{\Phi^+}$ with high probability on the ideal input will do so with nearly the same probability on the defective input. A tomographic estimate close to the ideal Bell state can therefore occur with high probability even when the full input has maximal trace distance from the ideal tensor power.

\end{enumerate}
These examples illustrate why successful verification using few resources need not establish an exact tensor power structure. In particular, an accurate single-system estimate from tomography does not by itself certify that the full resource is close to an ideal tensor power state. %They motivate asking whether subsequent information-processing protocols remain reliable for all inputs consistent with the chosen verification procedure.

Considerable effort has been devoted to studying information processing and resource manipulation beyond the iid setting. For instance, the information-spectrum approach was developed as a general framework for this purpose. It describes arbitrary sequences of states or channels through spectral rates, which generalize the entropic quantities familiar from the iid setting, and yields operational characterizations of tasks such as hypothesis testing, source and channel coding, and entanglement manipulation~\cite{VerduHan1994,Han2003,NagaokaHayashi2007,BowenDatta2006,Hayashi2006,BowenDatta2011,BuscemiDatta2010}. The framework was later connected to smooth entropies and one-shot information theory, which also provide finite-blocklength bounds and refined asymptotic expansions~\cite{RenWol04a,tomamichel09,DattaRenner2009,datta09,TH13,DL15}.

%Crucially, these approaches typically consider a given sequence of resources and allow the protocol to depend on that sequence. 
Crucially, these approaches generally characterize what can be achieved for each fixed resource sequence, allowing the optimizing protocol to depend on that sequence. Our setting is different. We fix a reference state $\rho$ and consider a source that is expected to produce many copies of $\rho$. Moreover, we fix an operational criterion that determines when a multi-copy state $\rho^{(n)}$ is accepted as being consistent with $\rho$. This criterion may arise, for example, from tomography, property estimation, or hypothesis testing. Importantly, the joint multi-copy state $\rho^{(n)}$ is otherwise unknown. We then consider the entire class of multi-copy states that satisfy this criterion, and we ask the following: 
% \GM{soruce $\equiv$ class of states (characterized by $\rho$). Add in question that the protocol may depend on $\rho$.}
%Their joint states need not be close to $\rho^{\otimes n}$ and may contain correlations or defects that the test does not detect. Therefore, we ask the following.

%Crucially, these approaches generally characterize what can be achieved for each fixed resource sequence, allowing the optimizing protocol to depend on that sequence. We instead ask for a single protocol that works uniformly over an entire class of sources. To define such a class, we fix a reference state $\rho$ and an operational criterion that determines when a multi-copy source is accepted as being consistent with $\rho$. This criterion may arise, for example, from tomography, property estimation, or hypothesis testing. The actual joint state $\omega^{(n)}$ is otherwise unknown: it need not be close to $\rho^{\otimes n}$ and may contain correlations or defects that the criterion does not detect. The distinction is therefore one of quantifiers: instead of optimizing separately for each fixed sequence, we require one protocol to work for every sequence satisfying the operational criterion.

\medskip
\noindent \textbf{Question.} \textit{Consider a $\rho$-source for a fixed state $\rho$. Can a single protocol, chosen without knowing the actual resource state $\rho^{(n)}$, process every resource state in this operationally defined class at the usual iid rate?}

\medskip
\noindent In other words, we require one protocol, chosen without knowledge of the correlations or defects, that works uniformly for every sequence compatible with the operational criterion. In terms of quantifiers, the information-spectrum approach generally asks what can be achieved for each fixed sequence, whereas we ask what can be achieved by a single protocol over an entire operationally defined class.

\begin{figure}[htb!]
    \centering
    \includegraphics[width=0.7\columnwidth]{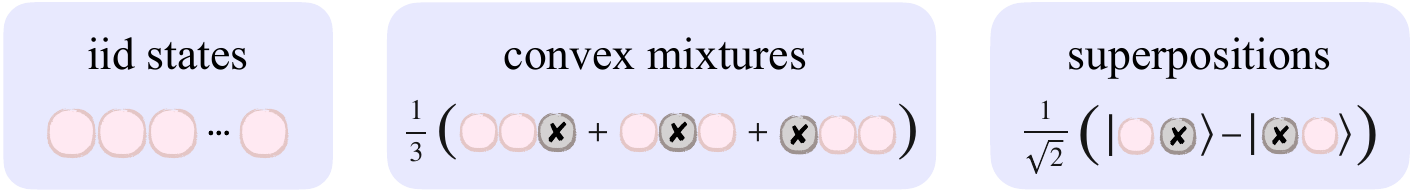}
    \caption{Illustration of almost-iid states. Each circle represents a single-copy system. The pink circles correspond to the iid components, represented by $\rho$, while the crossed circles indicate the defective and possibly correlated subsystems. A precise formulation is given in~\cref{def_almost_product_state_mixed}.}
     \label{fig:intro-iid-states}
\end{figure}
\vspace{-0.25cm}
A natural candidate class compatible with this operational viewpoint is provided by \emph{almost-iid states}~\cite{renner_phd,MSR2026}. Informally, these states have a tensor power structure up to a small number of defects, whose positions may be unknown and whose contribution may be coherent (see~\cref{fig:intro-iid-states}). Their operational relevance is supported by the fact that product measurements yield asymptotically the same empirical statistics as for the corresponding iid source~\cite[Proposition~2.7]{MSR2026}. 
% Product measurements on such states nevertheless yield empirical statistics that asymptotically approach those of the corresponding iid source~\cite[Proposition~2.7]{MSR2026}. 
%up to a sublinear number of possibly coherent defects. 
Motivated in part by that observation, the central question raised in~\cite{MSR2026} was whether almost-iid resources are as effective as perfect iid resources for information-processing tasks. Subsequent works introduced broader notions of almost-iid behavior~\cite{GPL26} and continued this investigation for several information-processing tasks~\cite{GDPL26,ND26,ND26_2,MSR2026-2}. Importantly, these approaches postulate the admissible class through a structural condition on the joint state, rather than determine it through an operational criterion.

In this work, we make the operational viewpoint more concrete: rather than postulating the form of the joint state, we ask which state sequences an experimentalist would actually recognize as many copies of $\rho$. To this end, we fix a tomography protocol \texttt{tom($\cdot$)} and consider all sequences $\pmb{\omega}=(\omega^{(n)})_{n\in\N}$ on which the protocol returns an estimate close to $\rho$ with high confidence and with a controlled sample complexity. Importantly, this does not mean that $\omega^{(n)}$ is close to $\rho^{\otimes n}$ in trace distance, or that no global measurement could distinguish the two states. It only means that the sequence is consistent with the operational test available to the experimentalist.

%Our main example of such an operational criterion comes from quantum state tomography. A tomography procedure examines part of a multi-copy source and accepts it whenever the resulting estimate is sufficiently close to $\rho$. The corresponding operational class consists of all joint states that pass this procedure with the prescribed confidence, including states with correlations that are not captured by the estimate.

\begin{tcolorbox}[colback=white,colframe=lightgray]
\textbf{Operational set of multi-copy states} \textit{(informal, see~\cref{def_OP-set})}\textbf{.}
For a tomography protocol \texttt{tom($\cdot$)} and a single-copy state $\rho$, we define
\begin{align} \label{intro:op_set_def}
 \OP(\texttt{tom($\cdot$)};\rho)\,:=\,\bigg\{ \,\pmb{\omega}=\big(\omega^{(n)}\big)_{n\in\N}\  \Big| \   \Pr\big[\texttt{tom}\big(\omega^{(n)}\big) \approx \rho \big] \overset{n\to\infty}{\longrightarrow} 1 \,\bigg\}\,.
%\OP(\texttt{tom($\cdot$)};\rho) :=\Bigg\{\pmb{\omega}=(\omega^{(n)})_{n\in\N}\ \Bigg|\  \begin{array}{c} \texttt{tom}(\omega^{(n)})\text{ estimates }\rho\text{ with}\\[-0.2em] \text{efficiently improving accuracy and confidence} \end{array}\Bigg\}.
\end{align}
Thus, every element of $\OP(\texttt{tom($\cdot$)};\rho)$ is interpreted operationally as a sequence of multiple copies of $\rho$, whose members may contain correlations or defects. %A precise formulation is given in~\cref{def_OP-set}.
\end{tcolorbox}

\noindent We impose only a few natural requirements on the tomography protocol: It should be correct on genuine iid inputs, be consistent under random symmetrization and sampling, and remain reliable when a sufficiently small number of arbitrary systems is corrupted (see~\cref{def_tom12,def_tom34}). Then, we show that standard tomography schemes based on informationally complete single-copy measurements and empirical frequencies satisfy these requirements (see~\cref{app_standard-tom-protocols}). Consistently with the measurement-statistics bound in~\cite[Proposition~2.7]{MSR2026}, almost-iid state sequences with a sublinear amount of defects belong to the corresponding operational set. The operational set is nevertheless larger since its definition places no direct structural restriction on the joint state (for instance, see~\cref{example-dF-thrm}).

This operational construction is also motivated by a complementary line of work on quantum state learning beyond the iid assumption. 
% A complementary line of work asks what can be learned about a quantum source when its outputs are not assumed to be iid. 
Reliable tomography can provide statistically valid predictions about unmeasured systems without assuming that the initial joint state is a tensor power~\cite{Christandl-Renner-2012}. Similarly, more recent results show that learning procedures designed for iid inputs can be adapted to arbitrary, possibly correlated inputs~\cite{Fawzi-Kueng-2024}, while related hypothesis-testing methods can efficiently identify the average state of non-identically distributed resource states~\cite{depalma2025}. These results concern the information that can be inferred from measurement data, rather than the subsequent use of the source in an information-processing task. In particular, learning that a source is consistent with a state $\rho$ does not imply that its joint state is close to $\rho^{\otimes n}$. This motivates our operational perspective: we consider all multi-copy states that a learning or tomography procedure accepts as copies of $\rho$, and ask whether a single protocol can process every state in this class at the usual iid rate.

\medskip
\noindent \textbf{Asymptotic transformation rates.} Having specified what it means operationally to possess many copies of a state, we use the resulting operational sets to define asymptotic transformation rates. A transformation from $\rho$ to $\sigma$ must map every sequence accepted as copies of $\rho$ to a sequence accepted as copies of $\sigma$, using a single sequence of allowed operations that is chosen independently of the actual resource state. Optimizing over such protocols gives operational counterparts of the usual iid conversion rates, including operational versions of entanglement distillation and entanglement cost. This develops the robustness question raised in~\cite{MSR2026}, where the stability of entanglement measures under almost-iid perturbations was investigated. The main difficulty is already visible when a resource state contains a small number of arbitrary defects at unknown positions. A protocol designed for the ideal tensor power input need not work simultaneously for all possible defect patterns (see~\cref{fig:intro-iid-transformations}).

\begin{figure}[htb!]
    \centering
    \includegraphics[width=0.71\columnwidth]{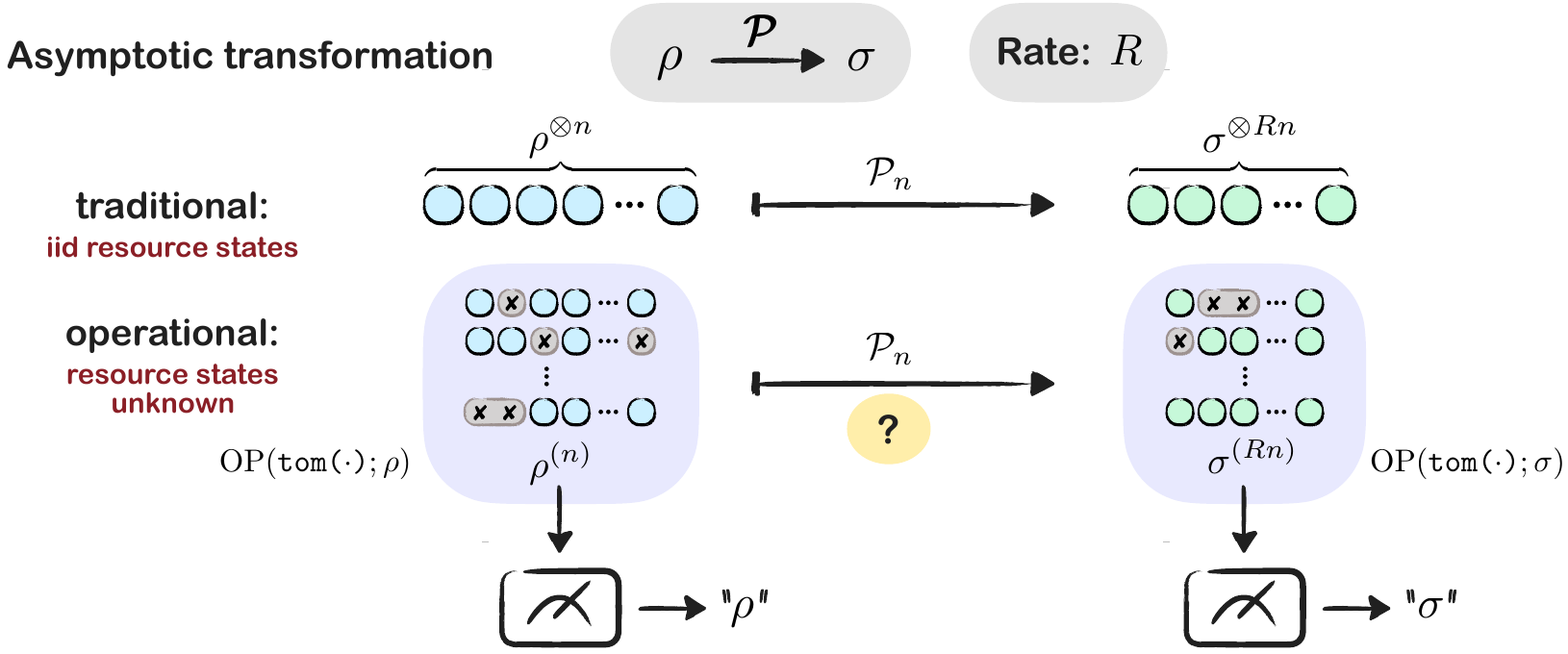}
    \caption{Asymptotic state transformation rates characterize the conversion of $n$ copies of $\rho$ into approximately $Rn$ copies of $\sigma$ via a protocol $\cP_n$ consisting of free operations. Traditionally, such protocols are designed for iid input resources $\rho^{\otimes n}$ and produce outputs approximating the iid target state $\sigma^{\otimes  Rn}$. In practice, however, multi-copy resources may contain correlated defects that are practically undetectable. States with such defects therefore belong to the operational set of multi-copy states. The question is whether $\cP_n$ remains effective uniformly over this larger class of resource states.}
     \label{fig:intro-iid-transformations}
\end{figure}
\vspace{-0.25cm}
Several robustness questions raised in~\cite{MSR2026} have since received partial answers~\cite{MSR2026-2,ND26_2}. These results provide evidence that iid conclusions can remain valid under almost-iid perturbations, but they do not establish the full stability of entanglement distillation and entanglement cost. As outlined above, we take a different route: we define admissible resource states through an operational criterion and study asymptotic transformations in this new setting. This change of perspective does not by itself resolve the stability question, but we hope that it provides a useful framework in which to investigate it. An affirmative answer would show that the usual rates and protocols are robust under deviations that cannot be excluded by the operational criterion. A negative answer would expose a gap between the resources that can be certified experimentally and the resources that can actually be used at the rates predicted by iid theory.

%One main contribution of this project: Precise formulation of the question and operational set of sequences as a motivation for almost-iid notions. Goal of this project:  Are the operational definition compatible with the more idealized iid definitions? We hope yes. But if no, it has drastic consequences! Study those new definitions, and try to relate them to the traditional iid counterpart. The benefit is hopefully leading to a robustness of QIT techniques and protocols.

\subsection*{Overview and main results}
\let\savedthetheorem\thetheorem
\renewcommand{\thetheorem}{\arabic{theorem}}
In the following, we provide a concise overview of the main definitions and results of this work. Unless stated otherwise, all systems considered are finite-dimensional.

\medskip
\noindent
\textbf{Notation.} To simplify the presentation of our main results, we introduce the following notation.
\begin{itemize}
\item \emph{State sequences:} We denote a sequence of states $\big(\rho^{(n)}\big)_{n\in \N} $ with $\rho^{(n)} \in \St(\cH^{\otimes n})$ by ${\pmb{\rho}} := \big(\rho^{(n)}\big)_{n\in \N} \,$. 
\item \emph{iid state sequences:} We denote an iid state sequence by $\pmb{\rho}^{\mathrm{iid}} := \big(\rho^{\otimes n}\big)_{n\in \N} \,$. 
\item \emph{Free operations with rate $r$:} For a sequence of free operations $\big( \cP^{(r)}_n\big)_{n\in\N}$ with  $\cP^{(r)}_n \in \F \big(C^n \to D^{ m_n}\big)$ and $\lim_{n\to\infty} \frac{m_n}{n} = r$, we write \smash{${\pmb{\mathcal{P}}}^{(r)} := \big( \cP^{(r)}_n\big)_{n\in\N} \in \F\,$}.
\item \emph{Trace operations:} For any $\alpha \in (0,1)$ and $k_n := \lceil  n^\alpha \rceil$, we denote the sequence $\big(\tr_{k_n}\big)_{n\in \N} $ of trace operations, where $k_n$ systems are traced out, by  $\pmb{\tr}_\alpha :=   \big(\tr_{k_n}\big)_{n\in \N} \,.$
\item \emph{Subexponentially decreasing functions:} A function $g : \mathbb{N} \to \mathbb{R} $ is called \emph{subexponentially decreasing} if there exist constants $K, L, \alpha >0 $ such that $g(n) \leq K\, 2^{-Ln^\alpha}$ for all $n\in \N$.
\item \emph{Equivalence relation for state sequences:} Two state sequences $\pmb{\rho}$ and $\pmb{\tau}$ are said to be \emph{equivalent} if the function $g(n) := \tfrac{1}{2} \| \rho^{(n)} - \tau^{(n)} \|_1 $ is subexponentially decreasing (here, $\| \cdot \|_1$ denotes the trace norm). In this case, we write $\pmb{\rho} \simeq \pmb{\tau}$. This defines an equivalence relation. 
\item \emph{Negligible ball:} For any set $\cA$ of state sequences, we define the \emph{negligible ball} as 
\[ B( \mathcal{A})\,:=\, \Big\{\,  \pmb{\tau} \ \big| \ \exists\ \pmb{\rho} \in \mathcal{A} \ \textrm{ such that }\,   \pmb{\tau} \simeq  \pmb{\rho}\,\Big\}\,. \]
\item \emph{Set of power-law sublinear almost-iid state sequences:} For a (single-copy) state $\rho$, we define the \emph{set of power-law sublinear almost-iid state sequences along $\rho$} as
\[
\mathcal{A}(\rho) \,:=\, \Big\{\pmb{\rho} \ \Big| \  \exists \, L\geq0,\, \kappa \in [0,1)\ \forall \, n\in \N\,:\  \rho^{(n)} \textrm{ is $\binom{n}{r_n}$-almost-iid along } \rho \textrm{ with } r_n:=\ \lceil L\,n^\kappa \rceil \,\Big\} \,.
\]
\end{itemize}

\medskip
\noindent
\textbf{Rate definitions.} Using the above notation, we may express traditional iid asymptotic transformation rates as follows.
\begin{tcolorbox}[colback=white,colframe=lightgray]
\textbf{iid transformation rates.} 
Let $\rho \in \St(\cH_C)$ and $\sigma\in \St(\cH_D)$, and $\F$ be any fixed set of free operations. The iid asymptotic state transformation rate is defined by
\begin{align} \label{intro_eq:iid_rate}
R_\F(\rho \to \sigma) \,=\, \sup\Big\{\, r>0 \ \big| \ \exists\ {\pmb{\mathcal{P}}}^{(r)} \in \F \,:\ {\pmb{\mathcal{P}}}^{(r)}\big( \pmb{\rho}^{\mathrm{iid}}\big) \, \simeq \,\pmb{\sigma}^{\mathrm{iid}} \Big\} \,,
\end{align}
%where \smash{$\pmb{\rho}^{\mathrm{iid}} :=\big(\rho^{\otimes n}\big)_{n\in \N} $ and  $\pmb{\sigma}^{\mathrm{iid}} :=\big(\sigma^{\otimes n}\big)_{n\in \N} \,$}.
If the set in~\cref{intro_eq:iid_rate} is empty, we set $ R_\F(\rho \to \sigma) = 0$.
\end{tcolorbox}
\noindent Now we can use the notion of an operational set of multi-copy states introduced in~\cref{intro:op_set_def} to define an operational asymptotic transformation rate as follows.
\begin{tcolorbox}[colback=white,colframe=lightgray]
\textbf{Operational transformation rates} \textit{(informal, see~\cref{def_OP-rate})}\textbf{.}
Let \texttt{tom($\cdot$)} denote a sensible tomography protocol and  $\F$ a set of free operations. For any density operators $\rho \in \St(\cH_C)$ and $\sigma\in \St(\cH_D)$ , we define the asymptotic transformation rate by
\begin{align} \label{intro_eq:op_rate}
  R^{ \texttt{tom}}_\F(\rho \to \sigma) \,=\,\sup\Big\{\, r>0\ \big| \ \exists\ {\pmb{\mathcal{P}}}^{(r)} \in \F \,:\ {\pmb{\mathcal{P}}}^{(r)}\big(\OP(\texttt{tom($\cdot$)};\rho) \big) \,\subseteq \,\OP(\texttt{tom($\cdot$)};\sigma) \Big\}\,.
\end{align}
If the set in~\cref{intro_eq:op_rate} is empty, we set $ R^{ \texttt{tom}}_\F(\rho \to \sigma) = 0$.
\end{tcolorbox}
\noindent Although this definition is natural from an operational point of view, its structure may appear complicated at first sight. Furthermore, the operational transformation rates seem to rely on the choice of the tomography protocol \texttt{tom($\cdot$)}. As part of our main results, we provide a simpler characterization of both the operational set $\OP(\texttt{tom($\cdot$)};\rho)$ and rate $R^{ \texttt{tom}}_\F(\rho \to \sigma)$, and further show that the latter is actually independent of the chosen tomography protocol.

%%%%%%%%%%%%%%%%%%%%%%%%%%%%%%%%%%%%

\medskip
\noindent
\textbf{Main results.}
%Although~\cref{def_OP-set,def_OP-rate} are natural from an operational point of view, they may appear complicated at first sight. Furthermore, the rate in~\cref{def_OP-rate} seems to rely on the choice of the tomography protocol \texttt{tom($\cdot$)}. 
The first main result characterizes the operational set $\OP(\texttt{tom($\cdot$)};\rho)$ when restricted to state sequences \smash{$\pmb{\rho}=\big(\rho^{(n)}\big)_{n\in \N}$} whose elements are permutation invariant. The permutation symmetry is natural from a physical point of view, and can usually be enforced by means of an additional random permutation. In the case of asymptotic transformation rates, for instance, the permutation symmetry of the state sequences is justified whenever the corresponding protocols admit a random permutation as a free operation. Then, we show that by tracing out a sublinear amount of copies, a permutation invariant state sequence from the operational set is asymptotically equal to an almost-iid state with a power-law sublinear number of defects.

\begin{theorem}[{Characterizing the operational set; \textit{informal, see~\cref{thm_almost-iid-dF}}}] \label{intro_thrm-dF}
 Let \texttt{tom($\cdot$)} be any sensible tomography protocol. Let $\pmb{\rho} \in  \OP(\texttt{tom($\cdot$)};\rho)$ and assume $\pmb{\rho}$ to be permutation invariant. Then, for any $\alpha \in (0,1)$, we find
 \begin{align}
 	\pmb{\tr}_\alpha\big( \pmb{\rho}\big)\, \simeq\, \pmb{\tau} \,\in\, \mathcal{A}(\rho)\,.
 \end{align}
\end{theorem}
\noindent The proof for characterizing the operational set is based on two main ideas: \textit{(i)} We use the permutation symmetry to represent the states $\rho^{(n)}$ as de Finetti states~\cite{Renner_2007}, and \textit{(ii)} make use of the tomography property to restrict the convex sum of the de Finetti states to (single-copy) states that are close to $\rho$. Crucially, this de Finetti-type approach yields a description of operational state sequences in terms of almost-iid states. Conceptually,
\begin{align*}
\text{operational \& symmetric resource $\pmb{\rho}$} \quad\xrightarrow{\ \text{discard $k_n = o(n)$ systems}\ }\quad\text{almost-iid resource $\pmb{\tau}$}.
\end{align*}

As a consequence,~\cref{intro_thrm-dF} allows us to simplify the structure of operational transformation rates. Under certain (standard) assumptions on the set of free operations $\F$ (see~\cref{def_QRTs}), we can show that the operational rate can be expressed in terms of almost-iid state sequences.
\begin{corollary}[Characterizing the operational transformation rate; \textit{informal, see~\cref{thrm_rate-tom-indep}}]  \label{thrm_RATEE}
For a given sensible tomography protocol  \texttt{tom($\cdot$)}, set of free operations $\F$, and any two density operators $\rho$ and $\sigma$, we find
\begin{align} \label{thrm_rate}
R^{ \texttt{tom}}_\F(\rho \to \sigma) \,=\,  \sup\Big\{\, r>0\ \big| \ \exists\ {\pmb{\mathcal{P}}}^{(r)} \in \F \,:\ {\pmb{\mathcal{P}}}^{(r)}\big(\mathcal{A}(\rho)\big) \,\subseteq \, B\big(\mathcal{A}(\sigma)\big) \Big\} \,=:\, \tilde{R}_\F(\rho \to \sigma)\,.
\end{align}
\end{corollary}
%\begin{proposition} \label{prop_trans}
%For any density operators $\rho \in \St(\cH_C)$, $\sigma\in \St(\cH_D)$ and $\rho \in \St(\cH_E)$, we find
%\begin{align}
% \tilde{R}_\F(\rho \to \sigma) \,\geq\, \tilde{R}_\F(\rho \to \tau)\tilde{R}_\F(\tau \to \sigma)\,.
%\end{align}
%\end{proposition}
\noindent In particular,~\cref{thrm_RATEE} shows that the operational transformation rate is independent of the actual choice of the tomography protocol \texttt{tom($\cdot$)}. In addition, it can be shown that
\begin{align} \label{intro:composability}
 \tilde{R}_\F(\rho \to \sigma) \,\geq\, \tilde{R}_\F(\rho \to \tau)\tilde{R}_\F(\tau \to \sigma)
\end{align}
for any density operators $\rho$, $\sigma$, and $\tau$ (see~\cref{prop_rates-transitivity}). Hence, the operational transformation rate satisfies the necessary composability property for a reasonable transformation rate. Furthermore,~\cref{thrm_rate} substantially simplifies the analysis of the operational transformation rate and allows us to compare it to other asymptotic transformation rates that are widely studied in the literature~\cite{chit19,Gour_2025}. Indeed, in a second main result, we can relate the operational rates to the traditional iid transformation rates as follows.
 \begin{theorem} \label{intro:thm_almost-iid-iid-rates-relation}
 Let $\rho \in \St(\cH_C)$ and $\sigma\in \St(\cH_D)$, and $\F$ be a fixed set of free operations as above. Then, the (operational) almost-iid asymptotic transformation rate $ \tilde{R}_\F(\rho \to \sigma)$ is lower bounded by the iid asymptotic rate, i.e., 
 \begin{align}
R_\F(\rho \to \sigma) \,\leq\,  \tilde{R}_\F(\rho \to \sigma)\,.
 \end{align}
 \end{theorem}
\noindent Here, the key idea is to apply a conventional iid protocol to suitably chosen blocks of increasing size, thereby ensuring that any power-law sublinear number of input defects gives rise to only a power-law sublinear number of output defects. This result is conceptually noteworthy: almost-iid resources are not only as effective as perfect iid resources for interconversion tasks, but may even offer improved performance because the structural requirements imposed on the output states are less restrictive.

On the other hand, under the stated assumptions on the set of free states, almost-iid transformation rates satisfy a converse bound in terms of the regularized relative entropy of resource,
\begin{align} \label{intro:converse-rel-entropy}
\tilde{R}_\F(\rho\to\sigma) \,\leq\, \frac{D^\reg(\rho\| \F)}{D^\reg(\sigma\| \F)}\,,
\end{align}
whenever the denominator is nonzero (see~\cref{prop_rates-upper-bound}). Ultimately, the question is whether the operational transformation rate $\tilde{R}_\F$ and its iid counterpart $R_\F$ may actually coincide.
\begin{question} \label{intro:quesion-rates-equal}
 Is it true that $R_\F(\rho\to\sigma) \overset{?}{=} \tilde{R}_\F(\rho\to\sigma)$ for all $\rho_{C} \in \St(\cH_{C})$ and $\sigma_{D} \in \St(\cH_{D})$  ?
\end{question}

\medskip
\noindent
\textbf{Entanglement theory.}
We can partially answer this question in the particular case of entanglement manipulation where the set of free operations is given by local operations and classical communication, $\F = \LOCC$. Let  $\Phi_2 =  \proj{\Phi_2}$ with \smash{$\ket{\Phi_2} := \frac{1}{\sqrt{2}}(\ket{00}+\ket{11})$} denote maximally entangled Bell state. For a bipartite state $\rho \in \St(\cH_{AB})$, the (iid) entanglement distillation and entanglement cost are defined as
\begin{align} \label{eq_E_D_and_E_C}
E_D(A:B)_{\rho} \,:=\, R_{\LOCC}(\rho \to \Phi_2) \qquad \textnormal{and} \qquad E_C(A:B)_{\rho} \,:=\, R_{\LOCC}(\Phi_2 \to \rho)^{-1} \, .
\end{align}
Analogously, we introduce the notion of \emph{operational entanglement distillation} and \emph{cost} through operational asymptotic state transformation rates.
\begin{tcolorbox}[colback=white,colframe=lightgray]
\textbf{Operational entanglement cost and distillation.} 
Let $\rho \in \St(\cH_{AB})$ and $\F=\LOCC$. The operational entanglement distillation and entanglement cost are defined as
\begin{align}
\tilde{E}_D(A:B)_{\rho} := \tilde{R}_{\LOCC}(\rho \to \Phi_2) \qquad \textnormal{and} \qquad \tilde{E}_C(A:B)_{\rho} :=  \tilde{R}_{\LOCC}(\Phi_2 \to \rho)^{-1} \, ,
\end{align}
respectively, where $\tilde{R}_{\LOCC}$ denotes the (operational) almost-iid asymptotic transformation rate.
\end{tcolorbox}
\noindent As a consequence of~\cref{intro:thm_almost-iid-iid-rates-relation}, it holds that $E_D(A:B)_{\rho} \leq \tilde{E}_D(A:B)_{\rho}$ and $\tilde{E}_C(A:B)_{\rho} \leq E_C(A:B)_{\rho}$ for any bipartite state $\rho$. Notably, we further show that the operational and standard asymptotic transformation rates for entanglement distillation coincide.
\begin{corollary} \label{intro:ED-is-operational}
For any density operator $\rho \in \St(\cH_{AB})$, we find
\begin{align} 
E_D(A:B)_{\rho} \,=\, \tilde{E}_D(A:B)_{\rho}\,.
\end{align}
\end{corollary}
\noindent The proof employs a decoupling protocol~\cite{DBWR14} to convert defective almost-iid Bell pairs into perfect iid Bell pairs at unit rate. This technique relies on the specific structure of the output Bell states and therefore cannot be directly applied to the entanglement cost setting. In summary, for a general mixed state $\rho \in \St(\cH_{AB})$, our results give the following hierarchy:
%\begin{tcolorbox}[colback=Cyan!3!white,colframe=Cyan!75!black]
\begin{align}
E_D(A:B)_{\rho} \overset{\textnormal{\Cshref{intro:ED-is-operational}}}{=} \tilde{E}_D(A:B)_{\rho}  \overset{\textnormal{\Cshref{intro:converse-rel-entropy}}}{\leq} D^\reg(\rho \| \mathrm{SEP})  \overset{\textnormal{\Cshref{intro:converse-rel-entropy}}}{\leq} \tilde{E}_C(A:B)_{\rho} \overset{\textnormal{\Cshref{intro:thm_almost-iid-iid-rates-relation}}}{\leq} E_C(A:B)_{\rho}\,, 
\end{align}
%\end{tcolorbox}
where $ \mathrm{SEP}$ denotes the set of separable states. If $\rho = \psi_{AB}$ is pure, then all of the above notions coincide and equal the entanglement entropy $H(A)_{\psi}$. In the more general case of asymptotic transformations with pure endpoints,~\cref{intro:composability} with the relations above imply the following formulas. Let $\rho \in\St(\cH_{AB})$, and let $\ket{\psi}_{AB},\,\ket{\phi}_{AB} \in \cH_{AB}$ be entangled pure states. Then,
\begin{align} \label{eq:moving-pure-endpoints}
\tilde{R}_{\LOCC}(\rho\to\phi) \,=\, \frac{E_D(A:B)_\rho}{H(A)_\phi}\,=\, R_{\LOCC}(\rho\to\phi)\,, \qquad \tilde{R}_{\LOCC}(\psi\to \rho) \,=\, \frac{H(A)_\psi}{\tilde{E}_C(A:B)_\rho}\,.
\end{align}
In particular, it holds
\begin{align}
\tilde{R}_{\LOCC}(\psi\to\phi) \,=\, \frac{H(A)_\psi}{H(A)_\phi} \,=\, R_{\LOCC}(\psi\to\phi) \,.
\end{align}

\medskip
\noindent
\textbf{Outlook.}
In this work, we studied the setting in which an experimentalist has access to multiple copies of a quantum state of interest. To capture the operational limitations of such a scenario, we introduced a class of states, referred to as the operational set, that are generically indistinguishable from tensor power states, as they appear to the experimentalist like many copies of the same state. While the operational set provides a notion of multi-copy resource states that is verifiable in principle, we showed that such verification can be achieved through a practical tomography procedure based on symmetrized product measurements. Based on this framework, we defined an asymptotic transformation rate that incorporates this broader operational perspective.

For instance, we found that the corresponding operational distillable entanglement coincides with its iid counterpart as shown in~\cref{intro:ED-is-operational}. However, it remains open whether the equality between operational asymptotic transformation rates $\tilde{R}_\F$ and the traditional iid rates $R_\F$ holds in general. In particular, for the case of entanglement cost,~\cref{intro:thm_almost-iid-iid-rates-relation} implies that the operational version is smaller than the standard one, $\tilde{E}_C(A:B)_{\rho} \leq E_C(A:B)_{\rho}$, while the converse inequality remains unclear. In that case,~\cref{intro:quesion-rates-equal} reads:

\medskip
%\begin{tcolorbox}[colback=SkyBlue!5!white,colframe=lightgray!85!SkyBlue]
\noindent \textbf{Question.} Is it true that $\tilde{E}_C(A:B)_{\rho} \overset{?}{=} E_C(A:B)_{\rho}$ for all $\rho_{AB} \in \St(\cH_{AB})$ ?
%\end{tcolorbox}

\medskip
\noindent Importantly, operational asymptotic transformation rates are, by construction, robust to small perturbations and experimental imperfections. Should equality between $\tilde{R}_\F$ and $R_\F$ hold, then this robustness would extend to standard iid transformation rates as well. Conversely, if the two rates ultimately differ, we argue that the operational definition constitutes the more natural notion of an asymptotic transformation rate from a fundamental and experimental standpoint, even if the standard iid definition possesses a higher degree of mathematical elegance.
\let\thetheorem\savedthetheorem

\paragraph{Acknowledgements.} We thank David Sutter and Nilanjana Datta for helpful discussions on almost-iid information theory. This project began several years ago, and its central definitions, ideas, and insights were developed before capable generative AI models became widely available. We acknowledge the subsequent use of OpenAI’s GPT models Sol and Astra for language polishing, assistance with correctness checks, and minor technical adjustments. GM and RR acknowledge support from the NCCR SwissMAP, the ETH Zurich Quantum Center, the SNSF project No.~20QU-1 225171, and the CHIST-ERA project MoDIC.

\newpage 
%\section*{Technical manuscript}
\tableofcontents
\newpage
%%%%%%%%%%%%%%%%%%%%%%%%%%%%%%%%%%%%%%%%%%%%%%%%%%%%%%%%%%%%%%%%%%%%%%%%%%%%%%%%%%%%
\section{Preliminaries and notation}
Throughout this work (unless otherwise stated), we work with finite-dimensional systems. Moreover, we use the logarithms to base 2 and we denote $\log(\cdot) \equiv \log_2(\cdot)$. We write $\poly(n)$ for terms of order at most polynomial in $n$. For two Hermitian matrices $X$ and $Y$, we say that $X \leq Y$ if $Y - X$ is positive semidefinite. We write $X>0$ and $X\geq 0$ for $X$ being a positive definite or positive semidefinite matrix, respectively. Let $\St(\cH_A):=\{X \geq 0: \tr[X]=1\}$ denote the set of density matrices on a Hilbert space $\cH_A$ with dimension $d_A$, and we may also write $\cH_A\equiv A$. For any state $\rho \in \St(\cH)$ and $\epsilon\in(0,1)$ define the $\epsilon$-ball around $\rho$ with respect to the trace distance by $B_{\epsilon}(\rho):=\{\rho' \in \St(\cH): T(\rho,\rho') \leq \epsilon \}$. For any two density matrices $\rho,\,\sigma \in \St(A)$, the trace distance $T$ is defined by $T(\rho,\sigma) := \frac{1}{2}\|\rho -\sigma \|_1$ where \smash{$\| X\|_1 = \tr[\sqrt{X^\dagger X} ]$} for any (bounded) linear operator $X$. Accordingly, the fidelity $F$ is defined as $F(\rho,\sigma) := \|\sqrt{\rho}\sqrt{\sigma}\|_1^2$.

% \GM{Convert this into a list?}

\subsection{Subexponentially decreasing functions and sample complexities}

\begin{lemma}[Polynomial bounds] \label{lem_poly-1}
   Let $P(X)= c_k X^{\beta_k} + c_{k-1} X^{\beta_{k-1}} + \cdots + c_1 X^{\beta_1}$  be a polynomial with real exponents $\beta_i \in \R$ such that $\beta_k > \dots > \beta_1$, and real coefficients $c_i \in \R$ for $i \in [k]$. %with $c_k \neq 0$. 
   Then there exists a constant $C \geq 0$ such that 
   \begin{align}
       |P(X)| \,\leq \, C X^{\beta_k} \quad \forall\, X \in  \R_{\geq 1}\,.
   \end{align}
\end{lemma}
\begin{proof}
    For any $X \in \R$ with $X \geq 1$, we have that $X^\beta = \exp(\beta\,\ln(X))$ is monotonically increasing in $\beta \in \R$. Therefore, $0 \leq X^{\beta_1} \leq \dots \leq X^{\beta_k}$. Hence, 
    \begin{align}
    |P(X)| \leq  |c_k X^{\beta_k}| + |c_{k-1} X^{\beta_{k-1}}| + \cdots + |c_1 X^{\beta_1}|  \leq (|c_k| + \dots + |c_1|)X^{\beta_k} \equiv CX^{\beta_k}  
    \end{align}
    with $C := |c_k| + \dots + |c_1| \geq 0$.
\end{proof}

\begin{notation}[Subexponentially decreasing functions] \label{def_subexp}
% A function $g : \mathbb{N} \to \mathbb{R} $ is called \emph{subexponentially decreasing} if there exist constants $K, L, \alpha >0 $ such that $|g(n)| \leq K\, 2^{-Ln^\alpha}$ for all $n\in \N$.
A function $g : \mathbb{N} \to \mathbb{R} $ is called \emph{subexponentially decreasing} if there exist constants $N_0 \in \N$ and $\alpha \in (0,1) $ such that $|g(n)| \leq 2^{-n^\alpha}$ for all $n \geq N_0$. 
 %\red{If $\sup\{\alpha \} =1$, we say that the function is \emph{exponentially decreasing}}.
If there exist constants $N_0 \in \N$ and $c >0$ such that $|g(n)| \leq 2^{-cn}$ for all $n \geq N_0$, we say that the function is \emph{exponentially decreasing}.
\end{notation}

\begin{remark}[Equivalent notions for subexponentially decreasing functions] \label{def_subexp-equivalence}
Let a function $g : \mathbb{N} \to \mathbb{R} $ and constants $L, \alpha, \tilde{L}, \tilde{\alpha} >0 $ be given such that $\tilde{L} < L$ and $\tilde{\alpha} < \alpha$. Consider the following statements:
%Then, the following are equivalent:
\begin{enumerate}[(i)]
%\item $\exists \, K, L, \alpha >0 \ \ \forall\, n\in \N: \ \  |g(n)| \leq K\, 2^{-Ln^\alpha}$.
%\item $\exists \, N_0, \tilde{L}, \alpha >0 \ \ \forall\, n \geq N_0: \ \  |g(n)| \leq 2^{-\tilde{L}n^\alpha}$.
%\item $\exists \, N_0, \tilde{\alpha} >0 \ \ \forall\, n \geq N_0: \ \  |g(n)| \leq 2^{-n^{\tilde\alpha}}$.
\item $\exists \,K \geq 0 \ \ \forall\, n\in \N: \   |g(n)| \leq K\, 2^{-Ln^\alpha}$ \label{rem_subex-it1}
\item $\exists \, N_0 \in \N \ \ \forall\, n \geq N_0: \  |g(n)| \leq 2^{-\tilde{L}n^\alpha}$  \label{rem_subex-it2}
% \item $\exists \,\bar{K} \geq 0 \ \ \forall\, n\in \N: \   |g(n)| \leq \bar{K}\, 2^{-\tilde{L}n^\alpha}$. \label{rem_subex-it4}
\item $\exists \, \bar{N}_0 \in \N \ \ \forall\, n \geq \bar{N}_0:\   |g(n)| \leq 2^{-n^{\tilde\alpha}}$  \label{rem_subex-it3}
\item $\exists \,\bar{K} \geq 0 \ \ \forall\, n\in \N: \   |g(n)| \leq \bar{K}\, 2^{- n^{\tilde{\alpha}}}$ \label{rem_subex-it4}
\end{enumerate}

Then the following chain of implications holds: $\eqref{rem_subex-it1} \implies \eqref{rem_subex-it2} \implies \eqref{rem_subex-it3} \implies \eqref{rem_subex-it4} $.
%\\ \\
%In addition, the choice of the constants $K,\,\bar{K},\,N_0,\,\bar{N}_0$ is constructive as can be seen from the proof.
\end{remark} \noindent
This remark shows that the properties $\eqref{rem_subex-it1}$-$\eqref{rem_subex-it4}$ above are essentially equivalent. In particular, a function $g : \mathbb{N} \to \mathbb{R} $ that satisfies any of the properties $\eqref{rem_subex-it1}$-$\eqref{rem_subex-it4}$ is a subexponentially decreasing function. Furthermore, a function $g$ that satisfies any of the properties $\eqref{rem_subex-it1}$-$\eqref{rem_subex-it2}$ with exponent $\alpha =1$ is exponentially decreasing.
\begin{proof}
\underline{$\eqref{rem_subex-it1} \implies \eqref{rem_subex-it2}$}: If $K \leq 1$, then $K\, 2^{-Ln^\alpha} \leq 2^{-\tilde{L}n^\alpha}$ for all $n \in \N$ and the claim directly follows for all $n \in \N$. Now let $K > 1$. Then we have
\begin{align}
 K\, 2^{-Ln^\alpha} \,\leq\, 2^{-\tilde{L}n^\alpha} \,&\iff\, -L n^\alpha + \log(K) \,\leq\, -\tilde{L}  n^\alpha
 \,\iff\, \log(K) \,\leq\, n^\alpha(L-\tilde{L}) \\
 \,&\iff\, \left(\frac{\log(K)}{L-\tilde{L}}\right)^{\frac{1}{\alpha}} \,\leq\, n\,. \label{proof_subexp-equivalence-eq1}
\end{align}
By setting 
\begin{align}
N_0 := \Big\lceil  \left(\frac{\log(K)}{L-\tilde{L}}\right)^{\frac{1}{\alpha}} \Big\rceil \in \N\,, \label{proof_subexp-equivalence-eq2}
\end{align} 
the claim follows.
 
\medskip
\noindent
\underline{$\eqref{rem_subex-it2} \implies \eqref{rem_subex-it3}$}: We have
\begin{align}
 2^{-\tilde{L}n^\alpha} \,\leq\, 2^{-n^{\tilde{\alpha}}} \,\iff\,  n^{\tilde{\alpha}} \,\leq\,  \tilde{L}n^\alpha   \,\iff\, \left(\frac{1}{\tilde{L}} \right)^{\frac{1}{\alpha - \tilde{\alpha}}} \,\leq\, n\,.
\end{align}
By setting $\bar{N}_0 := \max\Big\{N_0, \,\big\lceil \tilde{L}^{-\frac{1}{\alpha - \tilde{\alpha}}} \big\rceil \Big\} \in \N$, the claim follows.

\medskip
\noindent
\underline{$\eqref{rem_subex-it3} \implies \eqref{rem_subex-it4}$}: Define $ \bar{K}_0 := \max_{n < \bar{N}_0}\Big\{ |g(n)|\, 2^{n^{\tilde\alpha}} \Big\} < \infty$, and $ \bar{K} := \max \big\{ 1, \bar{K}_0 \big\} $. For $n < \bar{N}_0$, we then obtain
\begin{align}
|g(n)| \,\leq \, \bar{K}_0\,2^{- n^{\tilde{\alpha}}} \,\leq\, \bar{K} \,2^{- n^{\tilde{\alpha}}}\,,
\end{align}
while for $n \geq \bar{N}_0$, we have
\begin{align}
 |g(n)| \leq 2^{-n^{\tilde\alpha}} \, \leq \, \bar{K} \, 2^{-n^{\tilde\alpha}} \,,
\end{align}
which proves the claim.
\end{proof}

\begin{lemma}[Combinations of subexponentially decreasing functions] \label{lem_sums-of-subexp}
Let $g : \mathbb{N} \to \mathbb{R} $ and $ h : \mathbb{N} \to \mathbb{R} $ be two subexponentially decreasing functions. Then the sum
\begin{align}
f : \mathbb{N} &\to \mathbb{R}  \nonumber \\
n \ &\mapsto \ f(n) := g(n) + h(n)
\end{align}
is a subexponentially decreasing function. \\
\\
Furthermore, for any polynomial $P(X)= c_k X^{\beta_k} + c_{k-1} X^{\beta_{k-1}} + \cdots + c_1 X^{\beta_1}$ with real exponents $\beta_i \in \R$, the function
\begin{align}
\tilde{f} : \mathbb{N} &\to \mathbb{R}  \nonumber \\
n \ &\mapsto \ \tilde{f}(n) := P(n)g(n)
\end{align}
is also a subexponentially decreasing function. 

\medskip
\noindent Finally, the above assertions remain valid if ``subexponentially decreasing'' is replaced by ``exponentially decreasing''.
\end{lemma}
\begin{proof}
By~\cref{def_subexp,def_subexp-equivalence}, there exist constants $K, L, \alpha >0 $ and $K', L', \alpha' >0 $ such that $|g(n)| \leq K\, 2^{-Ln^\alpha}$ and $|h(n)| \leq K'\, 2^{-L'n^{\alpha'}}$ for all $n\in \N$. Set $\tilde{K} := \max\{K, K'\}$, $\tilde{L} := \min\{L, L'\}$, and $\tilde{\alpha} := \min\{\alpha, \alpha'\}$. Then, $|f(n)| \leq  |g(n)| + |h(n)| \leq 2\tilde{K}\,2^{-\tilde{L}n^{\tilde{\alpha}}}$. In addition, $\tilde{\alpha} = 1$ for $\alpha= \alpha' =1$, which proves the first claim. \\
For the second claim, consider first the expression $\tilde{g}(n) := n^\beta 2^{-Ln^\alpha}$ for any $\beta \in \R$. We want to show that $\tilde{g}$ is subexponentially decreasing. To this end, let $\tilde{L} = L - \nu$ for any $\nu \in (0,L)$. Then, we aim to find a constant $\tilde{K}$ such that
\begin{align} \label{eq_exp-bound-1}
	\tilde{g}(n) = n^\beta 2^{-Ln^\alpha} \,\overset{!}{\leq}\, \tilde{K}2^{-\tilde{L}n^\alpha} \quad \iff \quad n^\beta 2^{-\nu n^\alpha} \leq \tilde{K}\,.
\end{align}
For $\beta \leq 0$ one has $n^\beta \leq 1$ and one can directly take $\tilde{K} =1$ (and $\nu = 0$). For $\beta > 0$,~\cref{eq_exp-bound-1} is achievable as the function $n^\beta 2^{-\nu n^\alpha}$ attains a global maximum at $n_0 = (\frac{\beta}{\alpha \nu \ln 2})^{\frac{1}{\alpha}}$, and therefore we can choose $\tilde{K}:= \max_{n\in\N} n^\beta 2^{-\nu n^\alpha} < \infty$. Now set $\beta := \max_{i\in [k]} \beta_i$. Together with~\cref{lem_poly-1}, we find
\begin{align}
|\tilde{f}(n)| = |P(n)||g(n)| \overset{\textnormal{\Cshref{lem_poly-1}}}{\leq} C n^\beta \, K\, 2^{-Ln^\alpha}  \overset{\textnormal{\Cshref{eq_exp-bound-1}}}{\leq} CK\tilde{K} 2^{-\tilde{L}n^\alpha}\,.
\end{align}
In particular, since the exponent of $\tilde{f}$ is given by $\alpha$, $\tilde{f}$ is exponentially decreasing in case $\alpha =1$, which proves the second claim.
\end{proof}

\begin{lemma}[\cite{Bellare2002}] \label{lem_bounded-subexp}
A function $g : \mathbb{N} \to \mathbb{R} $ is subexponentially decreasing if and only if there exists a subexponentially decreasing function $h : \mathbb{N} \to \mathbb{R} $ and $N_0 \in \N$ such that $|g(n)| \leq |h(n)|$ for all $n \geq N_0$.
\end{lemma}
\begin{proof}
The proof is straightforward and analogous to~\cite[Proposition 2.4]{Bellare2002}.
\end{proof}

\begin{lemma} \label{lem_max-subexp}
Let $\{f_j\}_j$ denote a finite set of subexponentially decreasing functions $f_j : \mathbb{N} \to \mathbb{R} $. Then, the function $g : \mathbb{N} \to \mathbb{R} $, defined by
\begin{align}
g(n) \,:=\, \max_j f_j(n) \quad \forall\, n\in \N\,,
\end{align}
is also subexponentially decreasing.
\end{lemma}
\begin{proof}
By definition, for any $j$ there exists a constant $N^j_0 \in \N$ and $\alpha_j \in (0,1)$ such that for all $n \geq N^j_0$ it holds that $|f_j(n)| \leq 2^{-n^{\alpha_j}}$. By choosing $\bar{N}_0 := \max_j N^j_0$ and $\bar{\alpha} := \min_j \alpha_j$, we find
\begin{align}
|g(n)| \,=\,  |\max_j f_j(n)| \,\leq\, \max_j |f_j(n)| \,\leq \, 2^{-n^{\bar{\alpha}}}
\end{align}
for all $n \geq \bar{N}_0$, which proves the claim.
\end{proof}

\begin{example} \label{ex_subexp-function}
Let $L>0 $, $\alpha \in (0,1]$ and $N_0 \in \N$. Consider a function $g : \mathbb{N} \to \mathbb{R} $ with $|g(n)| \leq 2^{-n^\alpha(L+\nu(n))}$ for all $n\geq N_0$ with some function $\nu : \mathbb{N} \to \mathbb{R} $ such that $\lim_{n\to\infty}|\nu(n)| = 0$. Then $g$ is subexponentially decreasing (with exponent $\alpha$).
\end{example}
\begin{proof}
Since $\lim_{n\to\infty}|\nu(n)| = 0$, there exists ${N'_0}\in \N $ such that $|\nu(n)| \leq \frac{L}{2}$ for all $n \geq {N'_0}$. By further noting that $\nu(n) \geq - |\nu(n)|$ for any $n \in \N$, and setting $\bar{N}_0 := \max\{N_0, N'_0\}$, we obtain
\begin{align}
|g(n)|\,\overset{n\geq N_0}{\leq}\, 2^{-n^\alpha(L+\nu(n))} \,\leq \,  2^{-n^\alpha(L-|\nu(n)|)} \,\overset{n\geq N'_0}{\leq}\, 2^{-n^\alpha \frac{L}{2}} \label{eq_lem-bounded-subexp}
\end{align}
for all $n\geq \bar{N}_0$. Since the RHS of~\cref{eq_lem-bounded-subexp} is subexponentially decreasing with exponent $\alpha$, the function $g$ is also subexponentially decreasing with exponent $\alpha$ due to~\cref{lem_bounded-subexp}.
\end{proof}

\begin{notation}[Sample complexity] \label{def_sample-complexity}
A \emph{sample complexity} $N_0$ is a function
\begin{align*}
    N_0 \,\colon \  \  (0,1)\times (0,\tfrac{1}{2}] \ &\longrightarrow \ \N \\
    (\epsilon,\delta)\  &\longmapsto \  \bar{N}_0(\epsilon,\delta) \,.
\end{align*}
A sample complexity $N_0$ is called \emph{poly-logarithmic} if
\begin{align}
    N_0(\epsilon,\delta) \,\leq \, \poly\big(\tfrac{1}{\epsilon}, \,\log(\tfrac{1}{\delta})\big)\,.
\end{align}
A sample complexity $\bar{N}_0$ is called \emph{efficient} if
\begin{align}
    \bar{N}_0(\epsilon,\delta) \,\leq \, \poly\big(\tfrac{1}{\epsilon}\big) \cdot\log(\tfrac{1}{\delta})\,.
\end{align}
\end{notation}

\noindent
As in statistical learning, the parameters $\epsilon$ and $\delta$ correspond to the accuracy and confidence of a protocol, respectively. For given parameters $(\epsilon,\delta)$, the sample complexity $N_0(\epsilon,\delta)$ then gives the number of copies that are needed for the protocol to succeed within the accuracy $\epsilon$ with confidence $1-\delta$. The efficient sample complexity appears for instance in the context of classical or quantum state learning when assuming an iid setting~\cite{Hoeffding1963, Nielsen_Chuang_2010-chap8, ODonnell2016}. Note that an efficient sample complexity corresponds to a poly-logarithmic sample complexity with a fixed scaling with respect to the confidence parameter $\delta$.

%----------------------------------------------------------------

\begin{lemma}[Sample complexities and subexponentially decreasing functions] \label{lem_sample-subexp}
Let $f(\epsilon, n): (0,1)  \times \N \to \R$ be a function. The following statements are equivalent:
\begin{enumerate}[(i)]
\item $f(\epsilon,n) \leq K \,2^{-n^\alpha \epsilon^\mu L} $ for some constants $K,L >0$, $\mu \geq 0$ and $\alpha \in (0,1]$. \label{lem_subex-it1}
\item  \label{lem_subex-it2} There exists a constant $K\geq 0$ such that $f(\epsilon, n )\leq K$, and a poly-logarithmic sample complexity $N_0$ such that
\begin{align}
 \forall\,(\epsilon,\delta) \in (0,1)\times (0,\tfrac{1}{2}] \ \ \forall\, n\geq N_0(\epsilon,\delta)\,:\   f(\epsilon, n) \leq \delta\,.
\end{align}
\end{enumerate}
% such that $f(\epsilon,n) \leq K \,2^{-n^\alpha \epsilon^\mu L} $ for some constants $K,L,\mu >0$ and $\alpha \in (0,1]$. Then there exists a poly-logarithmic sample complexity $N_0$ such that
%\begin{align}
% \forall\,(\epsilon,\delta) \in (0,1)\times (0,\tfrac{1}{2}] \ \ \forall\, n\geq N_0(\epsilon,\delta)\,:\   f(\epsilon, n) \leq \delta\,.
%\end{align}
Furthermore, the sample complexity can be chosen to be efficient iff $\alpha = 1$.
\end{lemma}
\begin{proof}
\underline{$\eqref{lem_subex-it1} \implies \eqref{lem_subex-it2}$}:
First, note that it immediately follows that $f(\epsilon, n) \leq K$ globally. Now, let $\epsilon \in (0,1)$ be fixed, and define $\tilde{L}:=\frac{ \epsilon^\mu L }{2}$. By~\cref{def_subexp-equivalence}~$\eqref{rem_subex-it2}$, there exists $N_0(\epsilon) \in \N$ such that $f(\epsilon,n) \leq 2^{-n^\alpha \tilde{L}} $ for all $n\geq N_0(\epsilon) $. In addition, because of~\cref{proof_subexp-equivalence-eq2} and~\cref{lem_poly-1}, we can bound
\begin{align}
N_0(\epsilon) \,\leq\, C \left(\frac{1}{\epsilon} \right)^q
\end{align}
for some constant $C\geq 0$ and $q := {\lceil\frac{\mu}{\alpha} \rceil} \in \N$. Now for any $\delta \in (0,\frac{1}{2}]\,$, set 
\begin{align}
N_0(\epsilon,\delta) \,:= \, \left\lceil \bar{C} \left(\frac{1}{\epsilon} \right)^q \log\left(\frac{1}{\delta}\right)^s \,\right\rceil \,\geq\, N_0(\epsilon) \label{eq_sample-subexp-1}
\end{align}
with $\bar{C} := \max\big\{C, \big(\frac{2}{{L}}\big)^{1/\alpha} \big\}$ and $s := {\lceil\frac{1}{\alpha} \rceil} \in \N$. Since
\begin{align}
N_0(\epsilon,\delta) \,\geq\,  \left(\frac{2}{{L}}\right)^{\frac{1}{\alpha}}  \left(\frac{1}{\epsilon} \right)^{\frac{\mu}{\alpha}} \log\left(\frac{1}{\delta} \right)^{\frac{1}{\alpha}} \label{eq_sample-subexp-2}
\end{align}
by definition, we find for all $n\geq N_0(\epsilon,\delta) $ that
% it can be seen from~\cref{proof_subexp-equivalence-eq1} that $N_0(\epsilon)$ can be chosen to scale in $\epsilon$ as $\frac{1}{}$
\begin{align}
f(\epsilon,n) \, \overset{\textnormal{\Cshref{eq_sample-subexp-1}}}{\leq} \, 2^{-n^\alpha \tilde{L}} \, = \,  2^{-n^\alpha \epsilon^\mu \frac{L}{2}} \, \leq \, 2^{-N_0(\epsilon, \delta)^\alpha \epsilon^\mu \frac{L}{2}} \,\overset{\textnormal{\Cshref{eq_sample-subexp-2}}}{\leq}\,\delta\,.
\end{align}
Finally, note that $s = 1$ for $\alpha = 1$ and hence, the sample complexity $N_0(\epsilon, \delta)$ defined in~\cref{eq_sample-subexp-1} is efficient in that case.

\medskip
\noindent
\underline{$\eqref{lem_subex-it2} \implies \eqref{lem_subex-it1}$}: 
By definition of a poly-logarithmic sample complexity, there exist constants $C\geq 0$ and $q ,\,s \in \N$ such that ${N}_0(\epsilon,\delta) \leq C\left(\tfrac{1}{\epsilon}\right)^q \log\left(\tfrac{1}{\delta}\right)^s$. Now, define $\delta_{n,\epsilon} := 2^{-(\frac{1}{C}\epsilon^{q} n)^{\frac{1}{s}}}$. For any $\epsilon \in (0,1)$ and $n\in \N $ with $n\geq C\left(\tfrac{1}{\epsilon}\right)^q  $, we have $\delta_{n,\epsilon} \leq \frac{1}{2} $. In that case, it holds that ${N}_0(\epsilon, {\delta}_{n,\epsilon}) \leq n$ and hence,
\begin{align} 
f(\epsilon, n)\, \leq\, \delta_{n,\epsilon} \,.
\end{align}
For the case where $n < C\left(\tfrac{1}{\epsilon}\right)^q  $, we exploit the fact that $f(\epsilon,n)$ is bounded from above. Formally, we have $f(\epsilon, n) \leq K \leq  K 2 \delta_{n,\epsilon}$, where the last step follows since $\delta_{n,\epsilon} \geq \frac{1}{2}$. Hence, for $K':= 2\max\{K,1\}$, we overall find
\begin{align}
f(\epsilon, n)\, \leq\, K'\delta_{n,\epsilon}  \,=\, K' 2^{-(\frac{1}{C}\epsilon^{q} n)^{\frac{1}{s}}}
\end{align}
for any $\epsilon \in (0,1)$ and $n\in \N$. By further setting $L = C^{-\frac{1}{s}},\, \mu = \frac{q}{s},\,\alpha = \frac{1}{s}$, the claim follows. Finally, note that $\alpha =1$ if $s=1$.
\end{proof}

\begin{example} \label{ex_sums-of-subexp-eps}
Let $g,\,h : (0,1)\times \mathbb{N} \to \mathbb{R} $ be two functions such that
\begin{align}
g(\epsilon,n) \,\leq\,K2^{-L\epsilon^\mu n^\alpha}\,,\quad h(\epsilon,n) \,\leq\,\bar{K}2^{-\bar{L}\epsilon^{\bar{\mu}} n^{\bar{\alpha}}}
\end{align}
for some constants $ K, \bar{K}, L, \bar{L} >0$, $\mu,\bar{\mu} \geq 0$ and $\alpha,\bar{\alpha} \in (0,1]$. Then, their sum satisfies
\begin{align}
f(\epsilon,n) \,:=\, g(\epsilon,n)+h(\epsilon,n) \,\leq\, \tilde{K}2^{-\tilde{L}\epsilon^{\tilde{\mu}} n^{\tilde{\alpha}}}\,
\end{align}
with $\tilde{K} := 2\max\{K, \bar{K}\}$, $\tilde{L} := \min\{L, \bar{L}\}$, $\tilde{\mu} := \max\{\mu,\bar{\mu} \}$ and $\tilde{\alpha} := \min\{\alpha, \bar{\alpha}\}$. Hence, one may apply~\cref{lem_sample-subexp} to the function $f(\epsilon,n)$.
\end{example}

\begin{lemma} \label{lem_sample-subexp-function-eps}
Consider a function  $g : (0,1)\times \mathbb{N} \to \mathbb{R} $ defined by
\begin{align}
g(\epsilon, n) \,=\, 2^{-n^\alpha L\epsilon^\mu +r_n }
\end{align}
for some constants $L>0$, $\mu \geq 0$, $\alpha \in (0,1]$, and $r_n = o(n^\alpha)$ any sublinear function in $n$ (i.e., such that $\lim_{n\to \infty}\frac{r_n}{n^\alpha} = 0$). Then, there exists a function ${N}_0 : (0,1) \to \mathbb{N} $ such that 
\begin{align}
g(\epsilon, n)\,\leq\, 2^{-n^\alpha \frac{L}{2}\epsilon^\mu} \quad \textrm{for all} \quad n \geq N_0(\epsilon)\,.
\end{align}
Furthermore, the scaling of $N_0(\epsilon)$ with respect to $\epsilon$ depends on the precise form of $r_n$ and is given by solving the inequality $\frac{r_n}{n^\alpha} \leq \frac{L}{2}\epsilon^\mu$ with respect to $n$.
\end{lemma}
\begin{proof}
Let $\epsilon \in (0,1)$. Since $r_n = o(n^\alpha)$ by assumption, there exists ${N_0}(\epsilon)\in \N $ such that $\frac{r_n}{n^\alpha} \leq \frac{L}{2}\epsilon^\mu$ for all $n \geq {N_0}(\epsilon)$. Note that solving this inequality with respect to $n$ would exhibit the scaling of the function $N_0(\epsilon)$ with respect to $\epsilon$. Then, for all $n \geq {N_0}(\epsilon)$, we find
\begin{align}
g(\epsilon, n) \,=\, 2^{-n^\alpha L\epsilon^\mu +r_n } \,\leq\,  2^{-n^\alpha L\epsilon^\mu + n^\alpha\frac{L}{2}\epsilon^\mu } \,=\,  2^{-n^\alpha \frac{L}{2}\epsilon^\mu }\,.
\end{align}
\end{proof}

\begin{example} \label{ex_subexp-function-eps}
Let $L,\alpha>0$, $c, \mu \geq 0 $, and $N_0 \in \N$. Consider a function $g : (0,1)\times \mathbb{N} \to \mathbb{R} $ with $|g(\epsilon,n)| \leq 2^{-n^\alpha\big(L\epsilon^\mu-c\frac{\log(n+1)}{n^\alpha}\big)}$ for all $n\geq N_0$. Then, there exists a function $\bar{N}_0 : (0,1) \to \mathbb{N} $ such that $\bar{N}_0(\epsilon) \leq C \big( \frac{1}{\epsilon}\big)^{q}$ with $C>0$, $q=\frac{1}{\alpha}(\lfloor \mu \rfloor  +1)$, and such that $|g(\epsilon,n)| \leq 2^{-n^\alpha\frac{L}{2}\epsilon^\mu}$ for all $n\geq \bar{N}_0(\epsilon)$.
\end{example}
\begin{proof}
Define the function \smash{$h(\epsilon, n):= 2^{-n^\alpha\big(L\epsilon^\mu-c\frac{\log(n+1)}{n^\alpha}\big)} $}. Since \smash{$\lim_{n\to\infty} c\,\frac{\log(n)}{n^\alpha} = 0$}, it follows directly from~\cref{lem_sample-subexp-function-eps} that there exists a function $N'_0(\epsilon)$ such that $h(\epsilon, n) \leq 2^{-n^\alpha\frac{L}{2}\epsilon^\mu}$ for all $n\geq N'_0(\epsilon)$. Now define the function $\bar{N_0}  : (0,1) \to \mathbb{R}  $ by setting $\bar{N_0}(\epsilon) := \max\{N_0, N'_0(\epsilon)\}$. Then, for any $\epsilon \in (0,1)$, we directly obtain
\begin{align}
|g(\epsilon,n)| \,\overset{n\geq N_0}{\leq}\, 2^{-n^\alpha\big(L\epsilon^\mu-c\frac{\log(n+1)}{n^\alpha}\big)} \,\overset{n\geq N'_0(\epsilon)}{\leq} \,2^{-n^\alpha\frac{L}{2}\epsilon^\mu}
\end{align}
for all $n\geq \bar{N}_0(\epsilon)$. To determine the precise form of the function $\bar{N}_0 (\epsilon)$, we need to solve the inequality
\begin{align} \label{eq_ex_subexp-function-eps-0}
c\,\frac{\log(n+1)}{n^\alpha} \,\overset{!}{\leq} \,\frac{L}{2}\epsilon^\mu \quad \implies \quad \frac{2c}{L}\left(\frac{1}{\epsilon} \right)^\mu \,\leq\, \frac{n^\alpha}{\log(n+1)}\,.
\end{align}
For $\mu =0$, the $\epsilon$-dependence disappears, and since $\lim_{n\to\infty} \frac{\log(n+1)}{n^\alpha} = 0$, there exists a constant $\bar{N}_0 \in \N$ such that~\cref{eq_ex_subexp-function-eps-0} is satisfied for all $n\geq \bar{N}_0$. Now let $\mu >0$. First, note that since $\log(x) \leq \frac{x^p}{p}$ for any $x\geq 1$ and $p\in(0,1)$, one has
\begin{align} \label{eq_ex_subexp-function-eps-1}
\frac{n^\alpha}{\log(n+1)} \,\geq\, \frac{n^\alpha p}{(n+1)^p} \,\geq\, \frac{n^\alpha p}{(2n)^p} \,=\,  \frac{p}{2^p}\,n^{\alpha-p}\,.
\end{align}
Therefore, it is sufficient to set $p := \alpha-  \frac{\mu}{q} \in (0,1)$ for $q =\frac{1}{\alpha}(\lfloor \mu \rfloor  +1) >\frac{\mu}{\alpha}$, $K:= \left(\frac{2^{p+1}c}{pL}\right)^{\frac{1}{\alpha-p}}$, and $N'_0(\epsilon) :=\big\lceil K\left(\frac{1}{\epsilon} \right)^q \big\rceil$. Indeed, for all $n \geq N'_0(\epsilon) $ we then find 
\begin{align}
 \frac{n^\alpha}{\log(n+1)} \,\overset{\textnormal{\Cshref{eq_ex_subexp-function-eps-1}}}{\geq}\,  \frac{p}{2^p}\,n^{\alpha-p} \,\geq\,  \frac{p}{2^p}\,{N'_0(\epsilon)}^{\alpha-p} \,\geq\, \frac{p}{2^p} \frac{2^{p+1}c}{pL}\left(\frac{1}{\epsilon} \right)^\mu \,=\,  \frac{2c}{L}\left(\frac{1}{\epsilon} \right)^\mu\,,
\end{align}
as desired. Furthermore, due to~\cref{lem_poly-1}, there exists a constant $\tilde{C}>0$ such that $N'_0(\epsilon) \leq \tilde{C} \big( \frac{1}{\epsilon}\big)^{q}$. Finally, since $N_0 \leq N_0 \big( \frac{1}{\epsilon}\big)^{q}$, we directly find that $\bar{N}_0(\epsilon) \leq C \big( \frac{1}{\epsilon}\big)^{q}$ with $C:=\max\{N_0, \tilde{C} \}$, which proves the claim.
\end{proof}

%---------------------------------------------------------------- State sequences
\subsection{Sequences of states and channels}
\begin{notation}[State sequences]
We denote a sequence of states $\big(\rho^{(n)}\big)_{n\in \N} $ with $\rho^{(n)} \in \St(\cH^{\ox n})$ for all $n \in \N$  by 
\begin{align}
{\pmb{\rho}} := \big(\rho^{(n)}\big)_{n\in \N} \,.
\end{align} 

If $\rho^{(n)}$ is permutation invariant (over the $n$ subsystems) for all $n \in \N$, then we say that $\pmb{\rho}$ is permutation invariant.
\end{notation}

\begin{example}[iid state sequences]
Let $\rho \in \St(\cH)$. The iid state sequence along $\rho$ is defined by 
\begin{align}
\pmb{\rho}^{\mathrm{iid}} :=\big(\rho^{\otimes n}\big)_{n\in \N} \,.
\end{align}
\end{example}

\begin{notation}[Equivalence relation for state sequences] \label{def_eq-rel}
Two state sequences $\pmb{\rho}$ and $\pmb{\tau}$ are said to be \emph{equivalent} if the function $g(n) := \tfrac{1}{2} \| \rho^{(n)} - \tau^{(n)} \|_1 $ is subexponentially decreasing (here, $\| \cdot \|_1$ denotes the trace norm). In this case, we write 
\begin{align}
\pmb{\rho} \simeq \pmb{\tau}\,.
\end{align}
\end{notation}

\begin{proposition} \label{prop_eq-rel}
\cref{def_eq-rel} defines an equivalence relation.
\end{proposition}
\begin{proof}
For any two state sequences $\pmb{\rho}$ and $\pmb{\tau}$, the properties $\pmb{\rho} \simeq \pmb{\rho}$ (reflexivity) and $\pmb{\rho} \simeq  \pmb{\tau} \iff   \pmb{\tau} \simeq \pmb{\rho}$ (symmetry) follow immediately from the properties of the trace norm. For the transitivity property, let $\pmb{\sigma}$ be a third state sequence such that $\pmb{\rho} \simeq \pmb{\sigma}$ and $\pmb{\sigma} \simeq \pmb{\tau}$. Then, due to the triangle inequality of the trace norm, we find that the function
\begin{align}
f(n) \,:=\, \tfrac{1}{2} \| \rho^{(n)} - \tau^{(n)} \|_1 \,\leq \, \tfrac{1}{2} \| \rho^{(n)} - \sigma^{(n)} \|_1 +  \tfrac{1}{2} \| \sigma^{(n)} - \tau^{(n)} \|_1
\end{align}
is bounded by a sum of two subexponentially decreasing functions, and by~\cref{lem_sums-of-subexp,lem_bounded-subexp}, it is itself subexponentially decreasing. Hence, $\pmb{\rho} \simeq  \pmb{\tau}$.
\end{proof} 

\begin{notation}[Channel sequences] \label{not_channel-seq}
We denote a sequence of CPTP maps \smash{$(\mathcal{P}_{n})_{n\in \N}$}, with $\mathcal{P}_{n}:  \St(\cH^{\ox n}_C) \to  \St(\cH^{\ox m_n}_D)$ for any sequence $(m_n)_{n\in\N}$ of integers $m_n \in \N $ such that $\lim_{n\to\infty} m_n = \infty$, by
\begin{align}
{\pmb{\mathcal{P}}} := \big( \cP_n\big)_{n\in\N}\,.
\end{align}
The action of a channel sequence $\pmb{\mathcal{P}}$ on an input state sequence \smash{$\pmb{\rho}= \big(\rho^{(n)}\big)_{n \in \N} $} produces a new output state sequence \smash{$\pmb{\sigma} =  \big(\sigma^{(m)}\big)_{m \in \N} $} containing the elements \smash{$ \sigma^{(m_n)}:= \cP_n(\rho^{(n)}) $}. More precisely, we fix the following convention: Write $I:=\{m_n:n\in\N\}$ and define the following input index representative,
\begin{align}\label{eq_channel-selector}
j(m) \,:=\, \begin{cases} \min\{n\in\N:m_n=m\},&m\in I,\\ \min\{n\in\N:m_n\ge m\},&m\notin I. \end{cases}
\end{align}
Both minima are well-defined because $m_n\to\infty$. For any $\pmb{\rho}=\big(\rho^{(n)}\big)_{n\in\N}$, we then set
\begin{align}\label{eq_channel-completed-action}
\sigma^{(m)} \,:=\, \tr_{m_{j(m)}-m} \Big[\cP_{j(m)}\big(\rho^{(j(m))}\big)\Big]\,,
\end{align}
for all $m\in \N$, where $\tr_0$ is the identity. Thus attained output sizes use their earliest exact representative, while missing sizes use the next larger attained size followed by a partial trace. %All selections are independent of the input state sequence.

%if there exists $m = m_n = m_{n'}$ for $n\neq n'$, we take the smallest number $\bar{n}$ such that $m = m_{\bar{n}}$ and define \smash{$\sigma^{(m)} := \sigma^{(m_{\bar{n}})} $}. On the other hand, for any $m \in \N$ such that $m_n > m > m_{n-1}$ for some $n\in \N$, we may write $m = m_n - k$ with $k \in \N$ and define $\sigma^{(m)} := \tr_k[\sigma^{(m_n)}]$ (where  the last $k$ subsystems are traced out).

Accordingly, we denote the action of a channel sequence $\pmb{\mathcal{P}}$ on an input state sequence $\pmb{\rho}$ with output sequence $\pmb{\sigma}$ by 
\begin{align}
\pmb{\mathcal{P}}(\pmb{\rho}) \,=\,\pmb{\sigma}\,.
\end{align}
For any set \smash{$\mathcal{A}$} of state sequences, we further define
\begin{align}
\pmb{\cP}(\cA) \,:=\, \Big\{  \pmb{\cP}(\pmb{\tau}) \ \big| \  \pmb{\tau} \in \cA \Big\} \,.
\end{align}
\end{notation}

\begin{example}[Trace operations] \label{not_trace-ops}
For any $\alpha \in (0,1)$ and $k_n := \lceil  n^\alpha \rceil$, we denote the sequence $\big(\tr_{k_n}\big)_{n\in \N} $ of trace operations with $\tr_{k_n}:  \St(\cH^{\ox n}_C) \to  \St(\cH^{\ox n-k_n}_D)$ and where $k_n$ systems are traced out, by  
\begin{align}
\pmb{\tr}_\alpha :=   \big(\tr_{k_n}\big)_{n\in \N} \,.
\end{align}
\end{example}

\begin{notation}[Negligible ball] \label{not_negligible-ball}
Consider any set \smash{$\mathcal{A}$} of state sequences. We define the \emph{negligible ball} with respect to \smash{$\mathcal{A}$} as
\begin{align}
B( \mathcal{A})\,:=\, \Big\{  \pmb{\tau} \equiv \big(\tau^{(n)}\big)_{n\in \N} \ \Big| \ \exists\ \pmb{\rho} \in \mathcal{A} \ \textrm{ such that }\   \pmb{\tau} \simeq  \pmb{\rho}\Big\} \,.
\end{align}
For any state sequence $\pmb{\tau} $ such that \smash{$\pmb{\tau} \in B( \mathcal{A})$}, we may write \smash{$\pmb{\tau}  \insim  \mathcal{A}\,$}. Similarly, for any set \smash{$\mathcal{C}$} of state sequences such that \smash{$\mathcal{C} \subseteq B( \mathcal{A})$}, we may write \smash{$\mathcal{C} \subsetsim  \mathcal{A}\,$}.
\end{notation}

\begin{lemma}[Data processing] \label{lem_data-processing-neg-ball}
Consider any channel sequence $\pmb{\cP}$ with output system size $m_n = \poly(n)$, and any set $\mathcal{A}$ of state sequences. Then
\begin{align}
\pmb{\cP}\big(B( \mathcal{A})\big) \, \subseteq \,  B\big(\pmb{\cP}(\mathcal{A})\big)\,.
\end{align}
\end{lemma}
\begin{proof}
Let $\pmb{\tau}\in B(\cA)$. By the definition of the negligible ball, there exists a state sequence $\pmb{\rho} \in \mathcal{A}$ such that $\pmb{\tau} \simeq \pmb{\rho}$, meaning that 
\begin{align}
\tfrac{1}{2}\big\| \tau^{(n)} - \rho^{(n)}\big\|_1 \,\leq\, f(n)
\end{align}
for some subexponentially decreasing function $f(n)$. Then it follows by the contractivity of the trace distance that 
\begin{align}
\tfrac{1}{2}\norm{\cP_n\big(\tau^{(n)}\big) - \cP_n\big(\rho^{(n)}\big)}_1 \,\leq\, f(n)\,.
\end{align}
Since $m_n =\poly(n)$ by assumption, there exists constants $C,\,q >0$ such that $m_n \leq Cn^q$ for all $n\in \N$ due to~\cref{lem_poly-1}. Consequently, the function $g(m):= f(\bar{n}_m)$, with $\bar{n}_m := \min\{n\in \N \,:\, m_n = m\}$, is also subexponentially decreasing in $m$. Hence, $\pmb{\cP}(\pmb{\tau}) \simeq \pmb{\cP}(\pmb{\rho}) \in \pmb{\cP}(\mathcal{A})$. Therefore, $\pmb{\cP}(\pmb{\tau}) \in  B\big(\pmb{\cP}(\mathcal{A})\big)$. Since $\pmb{\tau}\in B(\cA)$ was chosen arbitrarily, the claim follows.
\end{proof}

\subsection{Entropic quantities}
Recall that $\St(\cH):=\{X \geq 0: \tr[X]=1\}$ denotes the set of density matrices on a Hilbert space $\cH$ with dimension $d$, and that the fidelity is defined as $F(\rho,\sigma) := \|\sqrt{\rho}\sqrt{\sigma}\|_1^2$ for any $\rho,\,\sigma \in \St(\cH)$. For $\alpha \in [1/2,1) \cup (1,\infty)$ the \emph{sandwiched R\'enyi divergence}~\cite{MLDSFT13,WWY14} is given by
\begin{align} \label{eq_sandwhiched}
D_{\alpha}(\rho \| \sigma) :=\frac{1}{\alpha -1 } \log \tr\big[(\sigma^{\frac{1-\alpha}{2\alpha}} \rho \, \sigma^{\frac{1-\alpha}{2\alpha}})^{\alpha}\big] \, .
\end{align}
For $\alpha = \frac{1}{2}$ we have \smash{$D_{\frac{1}{2}}(\rho \| \sigma)=- \log F(\rho,\sigma)$}.
In the limits $\alpha \to 1$ and $\alpha \to \infty$, the sandwiched R\'enyi divergence converges to the relative entropy
\begin{align}
D(\rho \| \sigma) := \tr\big[ \rho (\log \rho -  \log \sigma )\big]\,,
\end{align}
and the max-relative entropy~\cite{renner_phd,datta09}
\begin{align}
D_{\max}(\rho \| \sigma) :=\inf\big\{\lambda \in \R: \rho \leq 2^\lambda \sigma \big\} \, ,
\end{align}
respectively.
For $\rho_{AB} \in \St(\cH_{AB})$, the R\'enyi divergence can be used to define a conditional R\'enyi entropy~\cite{marco_book}
\begin{align} \label{eq_def_cond_renyi_entropy}
H_{\alpha}(A|B)_\rho := - \min_{\sigma_B \in \St(\cH_B)} D_{\alpha}(\rho_{AB} \| \id_A \otimes \sigma_B) \, , 
\end{align}
which converges to the conditional entropy, 
\begin{align}
H(A|B)_{\rho} := H(AB)_\rho - H(B)_\rho \,,\qquad H(A)_\rho := -\tr\big[\rho_{A} \log \rho_{A} \big]\,,
\end{align}
for $\alpha \to 1$, and to the conditional min-entropy,
\begin{align} 
H_{\min}(A|B)_{\rho} := \sup_{\sigma_B \in \St(\cH_B)} \sup \big\{\lambda \in \R: \rho_{AB} \leq 2^{-\lambda}(\id_A \otimes \sigma_B) \big\}\,,
\end{align}
for $\alpha \to \infty$. For $\alpha = \frac{1}{2}$, we obtain the conditional max-entropy 
\begin{align}
H_{\max}(A|B)_{\rho} := \max_{\sigma_B \in \St(\cH_B)}  \log F(\rho_{AB},\id_A \otimes \sigma_B)   \, .
\end{align}
For any two states $\rho,\, \sigma \in \St(\cH)$, the purified distance~\cite{marco_book} is defined as \smash{$P(\rho,\sigma):=\sqrt{1-F(\rho,\sigma)}$}. The Fuchs-van de Graaf inequality~\cite{fuchs99} implies $P(\rho,\sigma) \geq T(\rho,\sigma)$, where $T$ denotes the trace distance.
Let $\St_{\leq}(\cH):=\{X \geq 0: \tr[X]\leq1\}$ denote the set of subnormalized density matrices on a Hilbert space $\cH$. For $\rho \in \St(\cH)$ and $\eps\in(0,1)$ define the $\eps$-ball around $\rho$ with respect to the purified distance by $\cB^P_{\eps}(\rho):=\{\rho' \in \St_{\leq}(\cH): P(\rho,\rho') \leq \eps \}$.\footnote{For subnormalized states $\rho,\,\sigma \in \St_{\leq}(\cH)$, the purified distance is defined via the generalized fidelity, $P(\rho,\sigma) := \sqrt{1-F_{*}(\rho,\sigma)}$, where $F_{*}(\rho,\sigma) := \left(\| \sqrt{\rho} \sqrt{\sigma} \|_1 + \sqrt{(1-\tr\,\rho)(1-\tr\,\sigma)}\right)^2 $ (for instance, see~\cite{marco_book}).}
We then define a smooth variant of the min- and max-entropy by
\begin{align}
H^\eps_{\min}(A|B)_{\rho} :=\max_{\rho' \in \cB^P_{\eps}(\rho)} H_{\min}(A|B)_{\rho'} \qquad \textnormal{and} \qquad  H^\eps_{\max}(A|B)_{\rho} :=\min_{\rho' \in \cB^P_{\eps}(\rho)} H_{\max}(A|B)_{\rho'} \, .
\end{align}
% Similarly, we define the smooth max-relative entropy
% \begin{align}
%  D_{\max}^\eps(\rho \| \sigma ) := \min_{\rho' \in \cB_{\eps}(\rho)}  D_{\max}(\rho' \| \sigma ) \, .
% \end{align}

\subsection{Almost-iid states}
In this section we shortly recap the notion almost-iid states~\cite{renner_phd,MSR2026} and mention some of their properties that will be useful later. Let $\St(\cH)$ denote the set of density matrices on a Hilbert space $\cH$ whose dimension is denoted by $d$, and let $S_n$ be the symmetric group and $U_\pi$ denote the unitary representation of $\pi \in S_n$ that acts on the tensor product space $\cH^{\ox n}$ by permuting the factors. The symmetric subspace is given by $\mathrm{Sym}^n(\cH):= \mathrm{span} \{ \ket{\theta}^{\otimes n}: \ket{\theta} \in \cH \}$. Let $\cV^n_r(\cH, \ket{\theta}):=\{U_\pi(\ket{\theta}^{\otimes n-r} \otimes \ket{\Omega^{(r)}}): \pi \in S_n,\, \ket{\Omega^{(r)}} \in \cH^{\otimes r} \}$ for any $n,\,r\in \N_0$ with $r\leq n$.

\begin{definition}[Almost-iid states~\cite{MSR2026}] \label{def_almost_product_state_mixed}
Let $\cH_C$ be a Hilbert space, $\tau_C \in \St(\cH_C)$, and $n,r \in \N_0$ such that $r \leq n$. Then, $\rho_{C_1^n} \in \St(\cH_C^{\otimes n})$ is called a $\binom{n}{r}$-\emph{almost-iid state along} $\tau_C $ if there exists a purification $\ket{\theta}_{CE}$ of $\tau_C$ and an extension $\rho_{C_1^n E_1^n}$ of $\rho_{C_1^n}$ such that
\begin{enumerate}[(i)]
\item $\rho_{C_1^n E_1^n}$ is invariant under simultaneous permutations $(C_i,E_i) \leftrightarrow (C_j,E_j)$;\label{it_first_def}
\item $\supp(\rho_{C_1^n E_1^n}) \subseteq \,  \mathrm{span}\,\cV^n_r(\cH_{CE},\ket{\theta}_{CE})$. \label{it_second_def}
\end{enumerate}
The set of $\binom{n}{r}$-almost-iid states along $\tau_C$ is denoted by $\St^n_r(\cH_C, \tau_C)$.
\end{definition}
\noindent If a state $\rho_{C_1^n}$ satisfies~\eqref{it_second_def} but not necessarily~\eqref{it_first_def}, we call it a $\binom{n}{r}$-\emph{generalized-almost-iid-state along} $\tau_C$. We then write  $\rho_{C_1^n} \in \bar \St^n_r(\cH_C, \tau_C)$. Trivially, we have $\St^n_r(\cH_C, \tau_C) \subseteq \bar \St^n_r(\cH_C, \tau_C)$.

\begin{remark}[Properties of almost-iid states~\cite{MSR2026}] \label{rmk_properties_almost_iid_states}
An almost-iid state $\rho_{C_1^n} \in \St^n_r(\cH_C, \tau_C)$, according to~\cref{def_almost_product_state_mixed}, satisfies the following properties:
\begin{enumerate}[(a)]
\item For $\emph{any}$ purification $\ket{\theta}_{CE}$ of $\tau_C$ there exists an extension $\rho_{C_1^n E_1^n}$ of $\rho_{C_1^n}$ that satisfies~\eqref{it_first_def} and~\eqref{it_second_def}.
%\item In case that $\sigma$ and $\rho$ are pure,~\cref{def_almost_product_state_mixed} reduces to~\cref{def_almost_product_state} since the extension can be chosen to be trivial.
\item There exists an orthonormal basis $\{\ket{\Psi_t}\}_{t \in \cT}$ of $\mathrm{span}\,\cV^n_r(\cH_{CE},\ket{\theta}_{CE})$ with vectors $\ket{\Psi_t} \in \cV^n_r(\cH_{CE},\ket{\theta}_{CE})$ and with
\begin{align} \label{eq_sizeT}
    |\cT| \leq \binom{n}{r}\,d_{CE}^{\,r} \leq 2^{n h(r/n)}\, d_{CE}^{\,r},
\end{align}
where $h(x):=-x\log x - (1-x)\log(1-x)$ for $x \in [0,1]$ is the binary entropy function. With respect to this basis, condition~\eqref{it_second_def} can then be rewritten  as
\begin{align} \label{eq_almost-iid-form}
\rho_{C_1^n E_1^n} = \sum_{i,j \in \cT} \beta_{i,j} \ket{\Psi_i} \bra{\Psi_j} \, ,
\end{align}
for coefficients $\beta_{i,j} \in \C$ with $\beta_{i,i} \in [0,1]$ and $\sum_{i \in \cT} \beta_{i,i}=1$.
\item \smash{$\rho_{C_1^n}$} is permutation-invariant.
\item  $\rho_{C_1^{n+m}} \in \St^{n+m}_r(\cH_C, \tau_C)$  implies $\rho_{C_1^n} \in \St^n_r(\cH_C, \tau_C)$. 
\item For any \smash{$\rho_{C_1^n} \in \bar \St^n_r(\cH_C, \tau_C)$} and \smash{$\tilde\rho_{C_1^m} \in \bar \St^m_{\tilde{r}}(\cH_C, \tau_C)$},  we have $\rho_{C_1^n} \otimes \tilde\rho_{C_1^m}  \in \bar \St^{n+m}_{r+\tilde r}(\cH_C, \tau_C)$. 
\item  The set of almost-iid states is convex and compact.
\end{enumerate}
\end{remark}
In the analysis of almost-iid states, the representation in~\cref{eq_almost-iid-form} is particularly useful because it exploits the product structure of the space $\cV^n_r(\cH_{CE},\ket{\theta}_{CE})$. Nonetheless, it is sometimes preferable to pass to a purified description, i.e., to work with a purification of~\cref{eq_almost-iid-form}. The following lemma establishes that the underlying structure of $\cV^n_r(\cH_{CE},\ket{\theta}_{CE})$ persists even in this purified formulation.

\begin{lemma}[Purification in a preferred basis~\cite{MSR2026}] \label{lem_purif_almost_iid}
%Let $\rho_{C_1^n} \in \St(\cH_C^{\otimes n}, \tau_C^{\otimes n-r})$ be a $\binom{n}{r}$-\emph{almost-iid state} in $\tau_C $. 
Let $\ket{\Theta}_{AH} \in \cH_{AH}$ for some composite Hilbert space $\cH_{AH}$, and let $\cV_A \subseteq \cH_A$ be any subset of vectors such that $\,\mathrm{span}\,\cV_A$ admits an orthonormal basis $\big\{ \ket{\Psi_j}_A\big\}_{j \in \cT}$ with vectors $\ket{\Psi_j}_A \in \cV_A \ \  \forall\,j \in \cT$. 
Then, the following are equivalent:
\begin{enumerate}[(i)]
\item $\supp\big(\tr_H\big[\proj{\Theta}_{AH}\big]\big) \,\subseteq\, \mathrm{span}\, \cV_A$ \label{it_1}
\item $\ket{\Theta}_{AH} =\sum\limits_{j\in \cT} \alpha_j \ket{\Psi_j}_A \ket{\tilde{h}_j}_H$ for coefficients $ \alpha_j  \in \C$ and normalized vectors $\ket{\tilde{h}_j}_H \in \cH_H$ for all $j\in \cT$. \label{it_2}
\end{enumerate}
In addition, if $\ket{\Theta}_{AH}$ is normalized and satisfies $(ii)$, we have $1 = \braket{\Theta|\Theta}_{AH} = \sum_{j\in\cT} |\alpha_j|^2 $.
\end{lemma}

\begin{notation}[Set of power-law sublinear almost-iid state sequences] \label{not_subpoly-almost-iid}
For a (single-copy) state $\rho \in \St(\cH)$, we define the set of \emph{power-law sublinear almost-iid state sequences along $\rho$} as
\begin{align}
\mathcal{A}(\rho) \,:=\, \Big\{\pmb{\rho} \ \big| \  \exists \, L\geq 0,\, \kappa \in [0,1)\ \forall \, n\in \N\,:\  \rho^{(n)} \textrm{ is an $\binom{n}{\lceil Ln^\kappa \rceil}$-almost-iid state along } \rho \,\Big\}\,.
\end{align}
By convention, for (the finitely many) $n\in \N$ such that $\lceil Ln^\kappa \rceil \geq n$, the state $\rho^{(n)} \in \St(\cH^{\otimes n})$ may be chosen arbitrarily.
\end{notation}

\begin{lemma}[Stability of the set of power-law sublinear almost-iid state sequences]  \label{lem_subpoly-almost-iid}
Let $\rho \in \St(\cH)$ and $\pmb{\rho} \in \mathcal{A}(\rho)$. Then, for any sequence $(k_n)_{n\in \N}$ such that $k_n \leq cn$ eventually for some $c\in (0,1)$, there exist constants $L\geq 0$ and $\kappa \in [0,1)$ such that for all $n\in \N$ (with $n>k_n$),
\begin{align}
 \tr_{k_n}\big[\rho^{(n)}\big] \textrm{  is $\binom{n-k_n}{\big\lceil L(n-k_n)^\kappa \big\rceil}$-almost-iid along $\rho$}\,.
\end{align}
In particular, the above holds for sequences $(k_n)_{n\in \N}$ that are sublinear in $n$, thereby implying that $\pmb{\tr}_\alpha\big( \pmb{\rho}\big)  \in \mathcal{A}(\rho)$ and $\bar{\pmb{\rho}}:=\big( \tr_{k_n}\big[\rho^{(n+k_n)}\big] \big)_{n\in\N} \in \mathcal{A}(\rho)$.
\end{lemma}
\begin{proof}
By~\cref{not_subpoly-almost-iid}, there exist constants $ L\geq 0,\, \kappa \in [0,1)$ such that $ \rho^{(n)}$ is an $\binom{n}{\lceil Ln^\kappa \rceil}$-almost-iid state along $\rho$ for all $n\in \N$. Therefore, $ \tr_{k_n}\big[\rho^{(n)}\big] $ is $\binom{n-k_n}{\lceil Ln^\kappa \rceil}$-almost-iid along $\rho$ due to~\cref{rmk_properties_almost_iid_states}. It remains to be shown that $\lceil Ln^\kappa \rceil \leq \lceil \bar{L} (n-k_n)^\kappa \rceil$ for some other constant $\bar{L}$. Indeed, we can choose
\begin{align}
\bar{L} \,:= \, \max_{n \in \N \,:\, n> k_n} \,L \left( \frac{n}{n-k_n}\right)^\kappa \,< \,\infty\,,
\end{align}
and the value $\bar{L}$ is bounded because $k_n \leq cn$ eventually for some $c\in (0,1)$ and thus, $ \frac{n}{n-k_n} \leq  \frac{1}{1-c} <\infty$ as $n\to \infty$. Furthermore, note that the state $ \tr_{k_n}[\rho^{(n)}] $ is only non-trivial when $n>k_n$. Then, we find
\begin{align}
\lceil Ln^\kappa \rceil \,=\, \left\lceil L  \left( \frac{n}{n-k_n}\right)^\kappa (n-k_n)^\kappa \right\rceil \,\leq\,  \left\lceil \bar{L}  (n-k_n)^\kappa \right\rceil\,,
\end{align}
for all $n\in \N$ such that $n>k_n$, which also directly shows that $\pmb{\tr}_\alpha\big( \pmb{\rho}\big)  \in \mathcal{A}(\rho)$. In order to treat the sequence $\bar{\pmb{\rho}}$, for each $n\in \N$ we set $m_n:= n+k_n$. Since $k_n \leq cn \leq cm_n$ for $c\in (0,1)$ eventually, it follows that $\bar{\pmb{\rho}}  \in \mathcal{A}(\rho)$ by the same argument as above.
\end{proof}

%--------------------------- STATISTICS lemma
%\begin{proposition}[Statistics of almost-iid states~\cite{MSR2026}] \label{prop_statistics}
%Let $n,r \in \N$ such that $r\leq \tfrac{1}{2}n$, $\sigma \in \St(\cH)$, and $\rho^{(n)} \in \St(\cH^{\otimes n}, \sigma^{\otimes n-r})$. Let $\cM = \{M_x\}_{x\in\cX}$ be a POVM on $\cH$, and $P_X(x) = \tr[M_x \sigma]$ for all $x\in \cX$. Then, for any $\eps > 0$, we have
%\begin{align}
%\underset{{\bf{x}}}{\PP}\big[ \norm{\lambda_{\bf{x}} - P_X }_1 >  f(\eps,r,n)\big] \leq \eps \,,
%\end{align}
%where the probability is taken over the outcomes ${\bf{x}} = (x_1,\dots,x_n)$ of the product measurement $\cM^{\otimes n}$ applied to $\rho^{(n)}$, and where 
%\begin{align}
%    f(\eps,r,n) = 2\sqrt{\frac{\log{\big(\frac{1}{\eps}\big)}}{n} + \frac{|\cX|}{n}\log\Big( \frac{n}{2}+1\Big) + h\Big(\frac{r}{n}\Big) + \frac{2r}{n}\log(d) } + \frac{2r}{n}\,,
%\end{align}
%with $h(x)=-x\log x - (1-x)\log(1-x)$, and $ d = \dim{\cH}$.
%Furthermore, if $r=o(n)$, then $\lim_{n\to \infty} f(\eps,r,n) = 0$.
%\end{proposition}

%--------------------------------------------------------------------------------------------------------------------------------------------
\newpage
\section{Operational state sequences}
Consider a source that generates $n$ copies of some quantum state $\rho \in \St(\cH)$.
%Let $\rho \in \St(\cH_C)$ and $\sigma\in \St(\cH_D)$ be density operators. For a given set of (free) operations $\F$, such as the set of $\LOCC$ operations, one may consider the conversion of $n$ copies of $\rho$ into $m$ copies of $\sigma$, i.e.,
%\begin{align}
%    \rho^{(n)} \overset{\F}{\longrightarrow} \sigma^{(m)}\,,
%\end{align}
%such that the asymptotic rate $R \equiv m/n$ is maximal. 
%Here, $\rho^{(n)}\in \St(\cH^{\ox n}_C)$ and $\sigma^{(m)} \in \St(\cH^{\ox m}_D)$ represent the $n$ and $m$ copies of $\rho$ and $\sigma$, respectively. 
Typically, multiple copies of any state are assumed to be described by a tensor product, i.e., $\rho^{(n)} = \rho^{\ox n}$. However, from an operational point of view, where an experimentalist has access to each individual copy separately, this assumption is a priori not justified, as was motivated in the introduction. In the following, we want to relax this assumption by considering a whole \emph{set} of states $\rho^{(n)}$ that are operationally indistinguishable from a tensor product state $\rho^{\ox n}$, and in that sense look like many copies of $\rho$ to the experimentalist (see~\cref{fig:intro-iid-source}). This is relevant to many tasks in quantum information theory that involve processing entire sequences of quantum systems.
% \GM{Add other examples in quantum information where many copies of a state are considered.}
%\begin{figure}[h]
%    \centering
%    \includegraphics[width=0.35\columnwidth]{fig1f.pdf} % Or .eps, .png
%    \caption{Consider a source that generates $n$ copies of some quantum state $\sigma$. A common assumption is that the global state \smash{$\rho_{A^n_1}$}, describing these copies jointly, is given by a tensor power (iid) state $\sigma^{\otimes n}$. Fundamentally, however, such a perfect iid structure cannot be certified in practice.}
%    \label{fig:set-up}
%\end{figure}

Formally, we capture this set of states as follows.
%by the following definition. 
Let $\rho \in \St(\cH)$ be a density operator on a finite-dimensional Hilbert space $\cH$. We consider a setup where an experimentalist is given a source that is supposed to produce \emph{multiple copies} of $\rho$. In order for the experimentalist to verify that the source is actually producing the state $\rho$, they can perform a tomography procedure in order to estimate the single-copy state. 

\subsection{Tomography protocols \& operational sets}
Here, we allow any general tomography protocol \texttt{tom($\cdot$)} satisfying the following natural property:
\begin{definition}[Tomography protocols] \label{def_tom12}
%\phantom{abc}
For any fixed (single-copy) Hilbert space $\cH$, a tomography protocol \texttt{tom($\cdot$)} should be represented by a $\mathrm{POVM}$ (including eventual classical post-processing steps) with input space $\St(\cH^{\ox n})$ for any $n \in \N$,
%\GM{Should we additionally require a product structure for the map to make sure that it only acts "locally" (i.e., on the individual copies separately)? Maybe this is important in Definition 1.1 below, maybe not.}
and should output a (probabilistic) classical description, say $``\hat{\omega}"$, of the guessed (single-copy) state $\hat{\omega} \in \St(\cH)$. We require that:
\begin{enumerate}
% \item \label{def_tom12-1} 
\item %For any state $\omega \in \St(\cH_C)$, 
\label{def_tom12-2}
There exists an efficient sample complexity $\bar{N}_0$ (see~\cref{def_sample-complexity}) such that
\begin{align}
\forall\,(\epsilon,\delta) \in (0,1)\times (0,\tfrac{1}{2}]\ \ \forall\, n\geq \bar{N}_0(\epsilon,\delta) \ \ \forall\,\omega \in \St(\cH)\,:\ \Pr\big[\texttt{tom}\big(\omega^{\ox n}\big) \notin B_\epsilon(\omega)\big] \leq \delta\,. \label{eq_tom-def1} 
% \\
% &\red{\mathrm{where} }\quad N_0(\epsilon,\delta) \leq \poly\left(\tfrac{1}{\epsilon},\,\log\tfrac{1}{\delta}\right)\,. \label{eq_tom-def2}
\end{align} 
%\item Consider a sequence of states $\big(\rho^{(n)}\big)_{n\in \N}$ with $\rho^{(n)} \in \St(\cH^{\ox n}_C)$ for all $n\in \N$. Let $(k_n)_{n\in\N}$ be any (strictly monotonically increasing) sequence of integers with $k_n \leq n$, and define a reduced state sequence $\big(\bar{\rho}^{(k_n)}\big)_{k_n\in \N}$ by tracing out some of the systems such that $\bar{\rho}^{(k_n)}:=  \tr_{n-k_n}\big[\rho^{(n)}\big]$. Then, we require that
% \begin{align}
%      \Pr\big[\texttt{tom}\big(\bar{\rho}^{(k_n)}\big) \notin B_\epsilon(\rho)\big] \,=\, \Pr\big[\texttt{tom}\big(\rho^{(k_n)}\big) \notin B_\epsilon(\rho)\big]\,.    
% \end{align}
\end{enumerate}
\end{definition}
Intuitively, the relation in~\cref{def_tom12-2} is the defining property of any sensible tomography protocol in that it guesses the correct single-copy state in the case where the (multi-copy) input state has a tensor power structure (iid), while the requirement of an efficient sample complexity ensures that the number of copies needed for a given guessing accuracy (as specified by $\epsilon$ and $\delta$) does not grow too fast with increasing accuracy. For our work, we need to make a few additional assumptions about the tomography protocol to control its behavior when applied to more general input states $\rho^{(n)}$, as explained below.

For a given tomography protocol \texttt{tom($\cdot$)}, and for any $\rho \in \St(\cH)$, we define the following operational set of (multi-copy) state sequences \smash{$\pmb{\rho} =  (\rho^{(n)})_{n\in \N}$}, where each state \smash{$\rho^{(n)} \in \St(\cH^{\ox n})$} represents $n$ copies of $\rho$. Intuitively, we consider all possible state sequences $\pmb{\rho}$ such that by performing any (sensible) tomography procedure on the elements \smash{$\rho^{(n)}$}, for large $n$, the experimentalist will obtain a classical description $``\hat{\rho}"$ such that $\hat{\rho} \approx_\epsilon \rho$ with high probability. Recall~\cref{def_sample-complexity} for the definition of sample complexities.

\begin{definition}[Operational set of multi-copy states] \label{def_OP-set} Let \texttt{tom($\cdot$)} denote a tomography protocol and $\rho \in \St(\cH)$. For any sample complexity $N_0$, we define the set of state sequences
\begin{align*}
    \OP^{(N_0)}(\texttt{tom($\cdot$)};\rho)\,:=\,\bigg\{ \pmb{\rho} \equiv \big(\rho^{(n)}\big)_{n\in \N} \ \Big| \    \forall\,\epsilon,\delta  \in (0,1)\times (0,\tfrac{1}{2}] \ \ \forall\, n\geq N_0(\epsilon,\delta)\,&:\ \\
    \Pr\big[\texttt{tom}\big(\rho^{(n)}\big) \notin B_\epsilon(\rho)\big] &\leq \delta \,\bigg\}\,.
    % \OP(\texttt{tom($\cdot$)};\rho)\,:=\,\bigg\{\big(\rho^{(n)}\big)_{n\in \N} \ \Big| \  \forall\,\epsilon,\delta>0\ \ \exists\, N_0 \in \N \ \ \forall\,n\geq &N_0\,:\ \Pr\big[\texttt{tom}\big(\rho^{(n)}\big) \notin B_\epsilon(\rho)\big] \leq \delta\,, \quad \   \\
    % &\,\mathrm{where} \quad N_0(\epsilon,\delta) \leq \poly\left(\tfrac{1}{\epsilon},\,\log\tfrac{1}{\delta}\right) \bigg\}\,, 
\end{align*}
The \emph{operational set} is then defined as the union over all poly-logarithmic sample complexities,
\begin{align*}
    \OP(\texttt{tom($\cdot$)};\rho)\,:=& \bigcup_{N_0 \textrm{ is poly-log.}}   \OP^{(N_0)}(\texttt{tom($\cdot$)};\rho)
    \\ =&\,\bigg\{\ \pmb{\rho} \ \, \Big| \ \exists\ \, \textrm{poly-log.~sample complexity } N_0 \ \  \textrm{s.t.}\quad\\ 
    &\phantom{\,\bigg\{\ } \ \,   \forall\,(\epsilon,\delta) \in (0,1)\times (0,\tfrac{1}{2}]
    \ \ \forall\, n\geq N_0(\epsilon,\delta)\,:\ 
    \Pr\big[\texttt{tom}\big(\rho^{(n)}\big) \notin B_\epsilon(\rho)\big] \leq \delta \,\bigg\} \,.
    % \OP(\texttt{tom($\cdot$)};\rho)\,:=\,\bigg\{\big(\rho^{(n)}\big)_{n\in \N} \ \Big| \  \forall\,\epsilon,\delta>0\ \ \exists\, N_0 \in \N \ \ \forall\,n\geq &N_0\,:\ \Pr\big[\texttt{tom}\big(\rho^{(n)}\big) \notin B_\epsilon(\rho)\big] \leq \delta\,, \quad \   \\
    % &\,\mathrm{where} \quad N_0(\epsilon,\delta) \leq \poly\left(\tfrac{1}{\epsilon},\,\log\tfrac{1}{\delta}\right) \bigg\}\,, 
\end{align*}
Analogously, we may define an \emph{efficient} version of the operational set as
\begin{align*}
    \OP^{\textrm{(eff.)}}(\texttt{tom($\cdot$)};\rho)\,:=& \bigcup_{\bar{N}_0 \textrm{ is efficient}}   \OP^{(\bar{N}_0)}(\texttt{tom($\cdot$)};\rho)
    \\ =&\,\bigg\{\ \pmb{\rho} \ \, \Big| \ \exists\ \, \textrm{efficient sample complexity } \bar{N}_0 \ \  \textrm{s.t.}\quad\\ 
    &\phantom{\,\bigg\{\   } \ \,   \forall\,(\epsilon,\delta) \in (0,1)\times (0,\tfrac{1}{2}]
    \ \ \forall\, n\geq \bar{N}_0(\epsilon,\delta)\,:\ 
    \Pr\big[\texttt{tom}\big(\rho^{(n)}\big) \notin B_\epsilon(\rho)\big] \leq \delta \,\bigg\} \,.
    % \OP(\texttt{tom($\cdot$)};\rho)\,:=\,\bigg\{\big(\rho^{(n)}\big)_{n\in \N} \ \Big| \  \forall\,\epsilon,\delta>0\ \ \exists\, N_0 \in \N \ \ \forall\,n\geq &N_0\,:\ \Pr\big[\texttt{tom}\big(\rho^{(n)}\big) \notin B_\epsilon(\rho)\big] \leq \delta\,, \quad \   \\
    % &\,\mathrm{where} \quad N_0(\epsilon,\delta) \leq \poly\left(\tfrac{1}{\epsilon},\,\log\tfrac{1}{\delta}\right) \bigg\}\,, 
\end{align*}
\end{definition}

Note that if, for instance, the protocol \texttt{tom($\cdot$)} is described by a tensor product measurement \smash{$\big(\cM^{\otimes n}\big)_{n\in\N}$}, then the iid state sequence \smash{$\pmb{\rho}^{\mathrm{iid}} :=\big(\rho^{\otimes n}\big)_{n\in \N} $} is contained in the efficient operational set~\cite{Hoeffding1963, Nielsen_Chuang_2010-chap8, ODonnell2016}, see also~\cref{app_item2}.

For any given tomography protocol $\texttt{tom($\cdot$)}$ and state $\rho \in \St(\cH)$, one may wish to characterize the multi-copy state sequences contained in the set $\OP(\texttt{tom($\cdot$)};\rho)$. In general, this can be a difficult task. However, a more precise characterization becomes possible (as shown in~\cref{thm_almost-iid-dF}) if we require the tomography protocol to satisfy two additional properties as stated in~\cref{def_tom34}. In the following, we define the symmetrization map by 
\begin{align} \label{eq_symmetrization-map}
    \mathcal{S}_n :X\mapsto \frac{1}{n!} \sum_{\pi\in S_n}U_\pi X U^\dagger_\pi\,,
\end{align}
for any $X \in \St(\cH^{\ox n})$, where $S_n$ denotes the symmetric group and $U_\pi$ is the unitary representation of $\pi \in S_n$ on the space $\cH^{\ox n}$.
\begin{definition}[robust sampling tomography protocols] \label{def_tom34}
Let $\texttt{tom($\cdot$)}$ denote a tomography protocol that satisfies~\cref{def_tom12}, and let % $\big(\rho^{(n)}\big)_{n\in \N} \in \OP(\texttt{tom($\cdot$)};\rho)$ for some 
$\rho \in \St(\cH)$. We say that $\texttt{tom($\cdot$)}$ is \emph{robust \& sampling} if it satisfies the following additional properties:
\begin{enumerate}
\setcounter{enumi}{1}
    %\item[3.] (\emph{permutation invariance}): $\big(\mathcal{S}_n(\rho^{(n)})\big)_{n\in \N} \in \OP(\texttt{tom($\cdot$)};\rho)\,$.
   \item   \label{def_tom34-3} (\emph{consistency under random sampling}): Let $\pmb{\rho} \in \OP(\texttt{tom($\cdot$)};\rho)$. Then, there exists a poly-logarithmic sample complexity $N_0$ such that for any sequence $(m_n)_{n\in\N}$ of non-negative integers $m_n \in \N_0 $, the reduced state sequence $\bar{\pmb{\rho}}$, defined by 
\begin{align}
    % \bar{\rho}^{(m_n)}:=  \tr_{n}\big[\mathcal{S}_{n+m_n}(\rho^{(n+m_n)})\big]\,,
     \bar{\rho}^{(n)}:=  \tr_{m_n}\big[\mathcal{S}_{n+m_n}(\rho^{(n+m_n)})\big]\,,
\end{align}
where $\tr_{m_n}$ denotes the partial trace over (any) $m_n$ subsystems, is contained in the operational set with sample complexity $N_0$, i.e.,
\begin{align}
    %\pmb{\rho} \in \OP(\texttt{tom($\cdot$)};\rho) \quad \implies \quad \bar{\pmb{\rho}} \in \OP(\texttt{tom($\cdot$)};\rho)\,.
     \bar{\pmb{\rho}} \in \OP^{(N_0)}(\texttt{tom($\cdot$)};\rho) \subseteq  \OP(\texttt{tom($\cdot$)};\rho)  \,.
\end{align}

% \red{monotonically increasing} sequence
    \item \label{def_tom34-4} (\emph{robustness under noise}): Let $(r_n)_{n\in\N}$ be any sequence of non-negative integers $r_n \in \N_0 $ such that $r_n \leq n$ and $r_n = o(n^\beta)$ for some $\beta \in (0,1)$. Then, there exists an efficient sample complexity $\bar{N}_0$ such that for any \smash{$\Omega^{(r_n)} \in \St(\cH^{\otimes r_n})$} and any permutation $\pi\in S_n$, the state sequence $\tilde{\pmb{\rho}}$, defined by 
\begin{align}
    % \bar{\rho}^{(m_n)}:=  \tr_{n}\big[\mathcal{S}_{n+m_n}(\rho^{(n+m_n)})\big]\,,
     \tilde{\rho}^{(n)}:= U_\pi\big(\rho^{\otimes n-r_n} \otimes  \Omega^{(r_n)} \big) U_\pi^\dagger\,,
\end{align}
is contained in the efficient operational set with sample complexity $\bar{N}_0$, i.e.,
\begin{align}
	  \tilde{\pmb{\rho}} \in   \OP^{(\bar{N}_0)}(\texttt{tom($\cdot$)};\rho) \subseteq \OP^{\textrm{(eff.)}}(\texttt{tom($\cdot$)};\rho)\,.
\end{align}
\end{enumerate}
\end{definition}
The \emph{consistency under random sampling} property captures a weak form of the requirement that an experimentalist has access to each individual copy of the quantum state separately, and that the tomography protocol does not rely on global properties of a multi-copy state. A stronger requirement, for instance, would be to impose some product structure on the protocol. 
In addition, note that by choosing $m_n = 0 \ \forall\,n\in \N$, the above condition further captures an \emph{asymptotic permutation invariance} of the protocol, formally
\begin{align} \label{def_tom-perm-inv}
	  \big({\rho}^{(n)}\big)_{n \in \N} \in \OP(\texttt{tom($\cdot$)};\rho) \quad \implies \quad \big(\mathcal{S}_n(\rho^{(n)})\big)_{n\in \N} \in \OP(\texttt{tom($\cdot$)};\rho)\,,
\end{align}
which is a property that arises naturally for tomography protocols that are based on random sampling.
%Note that a stronger requirement, for instance, would be to impose some product structure on the protocol. 
% The \emph{permutation invariance} is a property that arises naturally for tomography protocols that are based on random sampling. It is  necessary in order to define a notion of robustness under partial trace, and it will further be used for characterizing the operational transformation rates as defined below. The \emph{robustness under partial trace} property captures a weak form of the requirement that an experimentalist has access to each individual copy of the quantum state separately. (A stronger requirement, for instance, would be to impose some product structure on the protocol.) \GM{We should somehow show/argue that known tomography arguments satisfy this definitions.}
On the other hand, the $\emph{robustness under noise}$ assumption ensures that the tomography protocol keeps functioning even after the insertion of a few ``defects" (represented by $\Omega^{(r_n)}$) into a tensor power (iid) state. Furthermore, note that~\cref{def_tom34-4} represents a generalization of~\cref{def_tom12-2} for the case where $r_n=0$.\footnote{Actually, we remark that the sample complexity in~\cref{def_tom34-4} a priori may depend on the single-copy state $\rho$, while the sample complexity from~\cref{def_tom12-2} is universal for all single-copy states. For standard tomography protocols, the sample complexity in~\cref{def_tom34-4} turns out to be independent of $\rho$ as shown in~\cref{app_item4}.} Finally, we argue in~\cref{app_standard-tom-protocols} that~\cref{def_tom12,def_tom34} are satisfied by most standard quantum state tomography procedures. 
% \GM{TODO. Further include counterexamples, showing that all assumptions are needed.} \GM{Comment about other ``operational criteria", such as hypothesis testing.}

\subsection{Basic properties of the operational set}
The following lemma shows that the equivalence of two state sequences is compatible with the membership in the operational set. In particular, its proof uses an operational characterization of the trace distance, thereby illustrating why the operational set is defined with respect to the trace distance (see also~\cite{ferradini-sandfuchs2025}, for instance, for a discussion of why the trace distance is central to defining security in quantum key distribution). %\GM{The robustness would also hold for other equivalent metrics (even if they scale (polynomially) in $n$).}
\begin{lemma} (Robustness of the operational set) \label{lem_OP-set-robust}
Let $\rho \in \St(\cH)$. Consider two states sequences $\pmb{\rho}$ and $\pmb{\sigma}$ such that $\pmb{\rho} \in \OP(\texttt{tom($\cdot$)};\rho)$ and $\pmb{\sigma} \simeq \pmb{\rho}$. Then, it follows that $\pmb{\sigma} \in \OP(\texttt{tom($\cdot$)};\rho)$. In particular, this implies that
\begin{align}
B(\pmb{\rho}) \,\subseteq \, \OP(\texttt{tom($\cdot$)};\rho)
\end{align}
where $B(\pmb{\rho}) $ denotes the negligible ball around $\pmb{\rho}$.
\end{lemma}
\begin{proof}
For $\pmb{\rho} \in \OP(\texttt{tom($\cdot$)};\rho)$, there exists a poly-log.~sample complexity $N_0(\epsilon,\delta)$ such that
\begin{align}
 \forall\,(\epsilon,\delta) \in (0,1)\times (0,\tfrac{1}{2}] \ \ \forall\, n\geq N_0(\epsilon,\delta)\,:\ \Pr\big[\texttt{tom}\big(\rho^{(n)}\big) \notin B_\epsilon(\rho)\big] \leq \delta 
 \end{align}
 by definition. By~\cref{lem_sample-subexp}, there exist constants $K,L >0$, $\mu \geq 0$ and $\alpha \in (0,1]$ such that
 \begin{align} \label{eq_robust-op-set-it1}
 \Pr\big[\texttt{tom}\big(\rho^{(n)}\big) \notin B_\epsilon(\rho)\big] \, \leq \,K \,2^{-n^\alpha \epsilon^\mu L} \,=:\, \delta_{n,\epsilon}
 \end{align}
 for all $n\in \N$ and $\epsilon \in (0,1)$. On the other hand, because $\pmb{\sigma} \simeq \pmb{\rho}$, the trace distance \smash{$ \frac{1}{2}\norm{\rho^{(n)}- \sigma^{(n)}}_1$} is a subexponentially decreasing function in $n$, meaning that
% let $\xi_n := \frac{1}{2}\norm{\rho^{(n)}- \sigma^{(n)}}_1$ for any $n \in \N$. Because $\pmb{\sigma} \simeq \pmb{\rho}$, the quantity $\xi_n$ is a subexponentially decreasing function in $n$, meaning that
\begin{align}  \label{eq_robust-op-set-it2}
 \frac{1}{2}\norm{\rho^{(n)}- \sigma^{(n)}}_1 \,\leq\, \bar{K} 2^{-n^{\bar{\alpha}} \bar{L}} \,=:\, \xi_n
\end{align}
for some constants $\bar{K}, \bar{L} > 0$ and $\bar{\alpha} \in (0,1]$ due to~\cref{def_subexp-equivalence}. Now consider the event $\Omega^n_\epsilon := \texttt{tom}\big(X^{(n)}\big) \notin B_\epsilon(\rho)$ with corresponding POVM element given by $\Lambda^n_\epsilon$ such that
\begin{align} \label{eq_POVM-event}
\Pr\big[\texttt{tom}\big(X^{(n)}\big) \notin B_\epsilon(\rho)\big] \,=\, \tr\big[\Lambda^n_\epsilon X^{(n)} \big]
\end{align}
for any $X^{(n)}\in \St(\cH^{\otimes n})$. Then, by using the variational characterization of the trace distance~\cite[Eq. (9.22)]{nielsenChuang_book},
\begin{align} \label{eq_var-trace-distance}
\tfrac{1}{2}\norm{\rho - \sigma }_1 \,=\, \max_{0\leq\Lambda \leq \id} \tr\big[\Lambda (\rho - \sigma) \big]
\end{align}
for any two states $\rho$ and $\sigma$, we can bound
\begin{align}
\Pr\big[\texttt{tom}\big(\sigma^{(n)}\big) \notin B_\epsilon(\rho)\big] \,\overset{\textnormal{\Cshref{eq_POVM-event}}}&{=}\,\tr\big[\Lambda^n_\epsilon \sigma^{(n)} \big] \,=\, \tr\big[\Lambda^n_\epsilon\big( \sigma^{(n)}  - \rho^{(n)}\big)\big] + \tr\big[\Lambda^n_\epsilon \rho^{(n)} \big] \\
\,\overset{\textnormal{\Cshref{eq_var-trace-distance}}}&{\leq}\, \frac{1}{2}\norm{\rho^{(n)}- \sigma^{(n)}}_1 + \tr\big[\Lambda^n_\epsilon \rho^{(n)} \big] \\
\,\overset{\textnormal{\Cshref{eq_POVM-event,eq_robust-op-set-it1,eq_robust-op-set-it2}}}&{\leq}\, \xi_n + \delta_{n,\epsilon}\,. \label{eq_robust-op-set-it3}
\end{align}
Combining~\cref{eq_robust-op-set-it3} with~\cref{ex_sums-of-subexp-eps,lem_sample-subexp}, we find that there exists a poly-log.~sample complexity $N'_0(\epsilon,\delta)$ such that\footnote{Proceeding in this way yields a scaling of the sample complexity in \smash{$\frac{1}{\epsilon}$} that is polynomial, but not optimal. In fact, one can show that for the sum in~\cref{eq_robust-op-set-it3}, one can find a sample complexity that has the same scaling in \smash{$\frac{1}{\epsilon}$} as one would find for $\delta_{n,\epsilon}$.}
\begin{align}
 \forall\,(\epsilon,\delta) \in (0,1)\times (0,\tfrac{1}{2}] \ \ \forall\, n\geq N'_0(\epsilon,\delta)\,:\ \Pr\big[\texttt{tom}\big(\sigma^{(n)}\big) \notin B_\epsilon(\rho)\big] \leq \delta \,,
 \end{align}
 implying that $\pmb{\sigma} \in \OP(\texttt{tom($\cdot$)};\rho)$.
\end{proof}

In order to better understand the operational set, it is instructive to find examples of state sequences that are contained in that set. The iid state sequence \smash{$\pmb{\rho}^{\mathrm{iid}} =\big(\rho^{\otimes n}\big)_{n\in \N} $} is one such example. In the following proposition, we show that, more generally, almost-iid state sequences with a limited amount of defects are also contained in the operational set.
\begin{proposition}[Almost-iid states are contained in the efficient operational set] \label{prop_almost-iid-in-OP}
Let \texttt{tom($\cdot$)} denote a tomography protocol satisfying~\cref{def_tom12,def_tom34}. Let $\rho \in \St(\cH)$ and $\pmb{\rho}$ denote a sequence of \smash{$\binom{n}{r_n}$}-almost iid states along $\rho$ with $(r_n)_{n\in\N}$ being any sequence of non-negative integers $r_n \in \N_0 $ such that $r_n \leq n$ and $r_n = o(n^\beta)$ for some $\beta \in (0,1)$. Then, $\pmb{\rho} \in  \OP^{\textrm{(eff.)}}(\texttt{tom($\cdot$)};\rho)$.
\end{proposition}
\begin{proof}
Let $n\in \N$, $r \equiv r_n$, and $\rho^{(n)}$ denote the $n$-th element in the sequence $\pmb{\rho}$. Since $\rho^{(n)}$ is an \smash{$\binom{n}{r_n}$}-almost iid state along $\rho$, by~\cref{rmk_properties_almost_iid_states}, for any purification $\ket{\theta}_{CE}$ of $\rho \equiv \rho_C$ there exists an extension \smash{$\rho_{C_1^nE_1^n}$} of \smash{$\rho^{(n)} \equiv \rho_{C_1^n}$} such that
\begin{align} \label{lem_almost-iid-OP-it1}
\rho_{C_1^n E_1^n} \,=\, \sum_{i,j \in \cT_n} \beta_{i,j} \ket{\Psi_i} \bra{\Psi_j} \, ,\qquad   |\cT_n| \,\leq\, 2^{n h(r/n)}\, d_{CE}^{\,r}\,,
\end{align}
for a family $\{\ket{\Psi_j}_{C_1^n E_1^n}\}_{j \in \cT_n}$ of orthonormal vectors from  \smash{$\cV^n_r\big(\cH_{CE}, \ket{\theta}_{CE} \big)$}. Furthermore, due to the Pinching inequality, the extension $\rho_{C_1^n E_1^n}$ satisfies the operator inequality
\begin{align} \label{lem_almost-iid-OP-it2}
\rho_{C_1^n E_1^n} \, \overset{\textnormal{\cite[Lemma~4.3]{MSR2026}}}{\leq} \,|\cT_n| \sum_{j \in \cT_n} \beta_{j,j} \proj{\Psi_j}_{C_1^n E_1^n} \,=:\, |\cT_n|\, \tilde \rho_{C_1^n E_1^n}   \, .
\end{align}
Now let $\epsilon \in (0,1)$. We aim to bound the probability of the event $\Omega^n_\epsilon := \texttt{tom}\big(\rho^{(n)}\big) \notin B_\epsilon(\rho)$, with corresponding POVM element given by $\Lambda^n_\epsilon$ (see~\cref{eq_POVM-event}). One finds
\begin{align}
 \Pr\big[\texttt{tom}\big(\rho^{(n)}\big) \notin B_\epsilon(\rho)\big]  \,&=\, \tr\big[\Lambda^n_\epsilon \rho^{(n)}\big] \,=\,  \tr\big[\Lambda^n_\epsilon \otimes \id_{E_1^n} \rho_{C_1^n E_1^n} \big] \\
 \,\overset{\textnormal{\Cshref{lem_almost-iid-OP-it2}}}&{\leq} \, |\cT_n| \sum_{j \in \cT_n} \beta_{j,j}\, \tr\big[\Lambda^n_\epsilon \otimes \id_{E_1^n} \proj{\Psi_j}_{C_1^n E_1^n} \big] \\
 \,&=\,  |\cT_n| \sum_{j \in \cT_n} \beta_{j,j}\, \tr\big[\Lambda^n_\epsilon \underbrace{U_j  \big(\rho^{\otimes n-r} \otimes \omega^{(r)}_j \big) U_j^\dagger}_{=:\, \tilde{\rho}^{(n)}_j} \big] \\
 \,&\leq\,  |\cT_n| \, \max_{j \in \cT_n} \, \Pr\big[\texttt{tom}\big(\tilde{\rho}^{(n)}_j\big) \notin B_\epsilon(\rho)\big]\,,  \label{lem_almost-iid-OP-it3}
\end{align}
where the penultimate step holds since \smash{$\ket{\Psi_j} \in \cV^n_r\big(\cH_{CE}, \ket{\theta}_{CE} \big)$} with $U_j$ denoting some fixed permutation and \smash{$\omega_j^{(r)} \in \St(\cH^{\otimes r})$} for all $j \in \cT_n$, and the last step uses that $\beta_{j,j} \geq 0$ and $\sum_{j\in \cT_n} \beta_{j,j} = 1$.
Next, define 
\begin{align} \label{lem_almost-iid-OP-max}
F(\epsilon,n) :=  \max_{j \in \cT_n}  \, \Pr\big[\texttt{tom}\big(\tilde{\rho}^{(n)}_j\big) \notin B_\epsilon(\rho)\big] \,,
\end{align} 
which we wish to bound. Due to~\cref{def_tom34-4} in~\cref{def_tom34}, there exists an efficient sample complexity $N_0$, depending only on $\rho$ and the defect-count sequence $(r_n)_{n\in \N}$, such that for all $\delta \in (0,\frac{1}{2}]$ and $n\geq N_0(\epsilon,\delta)$, we have
\begin{align}
 \Pr\big[\texttt{tom}\big(\tilde{\rho}^{(n)}_j\big) \notin B_\epsilon(\rho)\big] \,\leq\, \delta
\end{align}
for every $j \in \cT_n$, therefore implying that $F(\epsilon,n) \,\leq\,\delta$ for all $n\geq N_0(\epsilon,\delta)$.
%\begin{align} \label{lem_almost-iid-OP-argmax}
%j^*_n :=  \argmax_{j \in \cT_n}  \, \Pr\big[\texttt{tom}\big(\tilde{\rho}^{(n)}_j\big) \notin B_\epsilon(\rho)\big] \,.
%\end{align} 
%Then, the sequence \smash{$\tilde{\pmb{\rho}} := \big(\tilde{\rho}^{(n)}_{j^*_n} \big)_{n\in\N}$} fulfils the assumptions of~\cref{def_tom34-4} in~\cref{def_tom34}, meaning that \smash{$\tilde{\pmb{\rho}} \in  \OP^{(N_0)}(\texttt{tom($\cdot$)};\rho)$} for some efficient sample complexity $N_0$. Note that even though the sequence $\tilde{\pmb{\rho}}$ may implicitly depend on $\epsilon$ via~\cref{lem_almost-iid-OP-argmax},~\cref{def_tom34} assures that the sample complexity $N_0$ is independent of the choice of $j^*_n$ and hence also independent of $\epsilon$. 
Then, due to~\cref{lem_sample-subexp}, there exist constants $K,L >0$ and $\mu \geq 0$ such that
 \begin{align}  \label{lem_almost-iid-OP-it4}
 %\Pr\big[\texttt{tom}\big(\tilde{\rho}^{(n)}_{j^*_n}\big) \notin B_\epsilon(\rho)\big] \, \leq \,K \,2^{-n \epsilon^\mu L} \,=:\, \delta_{n,\epsilon} \,,
F(\epsilon,n)\, \leq \,K \,2^{-n \epsilon^\mu L} \,=:\, \delta_{n,\epsilon} \,,
 \end{align}
and hence, we have
\begin{align}  \label{lem_almost-iid-OP-it5}
 \Pr\big[\texttt{tom}\big(\rho^{(n)}\big) \notin B_\epsilon(\rho)\big] \,\overset{\textnormal{\Cshref{lem_almost-iid-OP-it3,lem_almost-iid-OP-it4}}}{\leq} \,  |\cT_n| \delta_{n,\epsilon} \,.
\end{align}
Since the binary entropy function satisfies the inequality $h(x) \leq 2\sqrt{x}$ for all $x \in [0,1]$, and $r_n \leq \tilde{L} n^\beta$ for some $\tilde{L} \geq 0$ and $\beta \in (0,1)$ by assumption, one can bound
\begin{align}  \label{lem_almost-iid-OP-it6}
nh\left(\frac{r}{n}\right) + r\log(d_{CE}) \,\leq\, 2n\sqrt{\frac{r}{n}} +  r\log(d_{CE}) \, \leq \, 2\sqrt{\tilde{L}} n^{\frac{1+\beta}{2}} + \tilde{L}n^\beta \log(d_{CE}) \,\overset{\frac{1+\beta}{2} > \beta}{\leq}\,\bar{L}\,n^{\frac{1+\beta}{2}}
% |\cT_n|  \, \overset{\textnormal{\Cshref{lem_almost-iid-OP-it1}}}{\leq} \,  2^{-n \epsilon^\mu L + 2n\sqrt{\frac{r}{n}} + r\log(d_{CE})}\,.
\end{align}
for some constant $\bar{L} \geq 0$. Hence, we find
\begin{align} \label{lem_almost-iid-OP-it7}
 \Pr\big[\texttt{tom}\big(\rho^{(n)}\big) \notin B_\epsilon(\rho)\big] \,\overset{\textnormal{\Cshref{lem_almost-iid-OP-it5}}}{\leq} \,  |\cT_n| \delta_{n,\epsilon}  \,\overset{\textnormal{\Cshref{lem_almost-iid-OP-it1,lem_almost-iid-OP-it6,lem_almost-iid-OP-it4}}}{\leq} \,  K\,2^{-n \epsilon^\mu L +  \bar{L}\,n^{\frac{1+\beta}{2}} } \,=:\, g(n,\epsilon)\,. 
\end{align}
By~\cref{lem_sample-subexp-function-eps} (and since \smash{$\frac{1+\beta}{2} <1$}), it follows that there exists a function $N'_0(\epsilon)$ such that 
\begin{align} \label{lem_almost-iid-OP-it8}
g(n,\epsilon)\, \leq \, K\,2^{-n \epsilon^\mu \frac{L}{2}} 
\end{align}
for all $n \geq N'_0(\epsilon)$. The scaling of $N'_0(\epsilon)$ in $\epsilon$ is determined by solving the inequality 
\begin{align}
\frac{\bar{L}n^{\frac{1+\beta}{2}}}{n} \,\overset{!}{\leq}  \, \frac{L}{2}\epsilon^\mu \quad \iff \quad  \frac{2\bar{L}}{L}\left(\frac{1}{\epsilon}\right)^\mu \,\leq\, n^{\frac{1-\beta}{2}}  \quad \iff \quad \bar{K} \left(\frac{1}{\epsilon}\right)^{\frac{2\mu}{1-\beta}} \,\leq\, n\,,
\end{align}
for $\bar{K} \geq 0$, showing a polynomial scaling of $N'_0(\epsilon)$ in $\frac{1}{\epsilon}$. Therefore, it follows from~\cref{lem_almost-iid-OP-it7,lem_almost-iid-OP-it8,lem_sample-subexp} that there exists an efficient sample complexity $\bar{N}'_0(\epsilon,\delta)$ such that
\begin{align}
 \forall\,(\epsilon,\delta) \in (0,1)\times (0,\tfrac{1}{2}] \ \ \forall\, n\geq \bar{N}_0(\epsilon,\delta)\,:\ \Pr\big[\texttt{tom}\big(\rho^{(n)}\big) \notin B_\epsilon(\rho)\big] \leq \delta \,,
 \end{align}
 where we set $\bar{N}_0(\epsilon,\delta) :=  \max\{N'_0(\epsilon)\log(\frac{1}{\delta}), N_0(\epsilon, \delta), \bar{N}'_0(\epsilon,\delta)  \}$. Since $\bar{N}_0(\epsilon, \delta)$ is again an efficient sample complexity, this proves that $\pmb{\rho} \in  \OP^{\textrm{(eff.)}}(\texttt{tom($\cdot$)};\rho)$.
\end{proof}

\newpage
\section{Characterization of the operational set} \label{sec:characterization-OP-set}
In the previous section, we have introduced the notion of operational state sequences and sets (see~\cref{def_OP-set}). A priori, these definitions appear to depend on a specific tomography protocol \texttt{tom($\cdot$)}. 
%These sequences admit a characterization in terms of ...
In this section, we will show that, as long as \texttt{tom($\cdot$)} satisfies~\cref{def_tom12,def_tom34}, these sequences admit a characterization in terms of almost-iid states independently of \texttt{tom($\cdot$)}.
%In~\cref{def_OP-rate}, we have proposed an operational definition for asymptotic transformation rates. A priori, this definition seems to rely on a specific tomography protocol \texttt{tom($\cdot$)}. 
%In this section, we will show that these rates turn out to be independent of the chosen protocol \texttt{tom($\cdot$)} as long as it satisfies~\cref{def_tom12,def_tom34}.
The main idea is to use a de Finetti-type argument to characterize the operational set $\OP(\texttt{tom($\cdot$)};\rho)$ as shown in~\cref{thm_almost-iid-dF}. The proof of the latter is built on two central parts which are presented in~\cref{prop_modified-deFinetti,prop_close-almost-iid}.

\begin{proposition} \label{prop_modified-deFinetti}
Let \texttt{tom($\cdot$)} denote a tomography protocol satisfying~\cref{def_tom12,def_tom34}. Let $\rho \in \St(\cH_C)$, and $\pmb{\rho} \in \OP(\texttt{tom($\cdot$)};\rho)$ be permutation invariant. Furthermore, let \smash{$k_n:= \lceil n^\alpha \rceil,\ r_n:=\lfloor n^\beta\rfloor$} with $\alpha,\,\beta \in (0,1)$ such that $\alpha + \beta >1$. 
Then, there exist $\gamma>0$, a measure $\nu_n$ on $\St(\cH_C)$, and, for each $\omega \in \St(\cH_C)$, a sequence of states \smash{$\sigma_n^{(\omega)}$} that are $\binom{n}{r_n}$-almost-iid along $\omega$, such that, for $\epsilon_n = \frac{1}{n^\gamma}$ (or any fixed value $\epsilon_n \equiv \epsilon>0$), we find that
\begin{align}\label{eq_peaked-de-Finetti}
    \left\|\,\tr_{k_n}\big[\rho^{(n+k_n)}\big]\,-\, \int_{\omega \in B_{\epsilon_n}(\rho)} \sigma_n^{(\omega)} \,\nu_n(\omega)\, \right\|_1 \quad \longrightarrow\ \textrm{ is subexponentially decreasing in $n$} \,.
\end{align}
Crucially, the above integration is restricted to states $\omega$ that are (asymptotically) close to $\rho$.
\end{proposition}
\begin{proof} 
This proposition is a modified version of~\cite[Theorem 4.3.2]{renner_phd}, and in what follows, we will make use of various techniques introduced in the proof therein.

\medskip
\noindent %First, we need to state some preparatory relations. 
\textbf{Step 1:} \textit{(Tomography part)} %First, for the integer sequence $\big(k_n = \lceil n^\alpha \rceil\big)_{n\in \N}$,   $\big(\bar{\rho}^{(k_n)}\big)_{k_n}$ with
For $n \in \N$ and $k_n = \lceil n^\alpha \rceil$, we consider reduced states of the form
\begin{align} \label{proof_eq_reduced-seq}
	\bar{\rho}^{(k_n)} := \tr_{n}\big[\rho^{(n+k_n)}\big]  \,.
	% \equiv  \tr_{\ell_k}\big[\rho^{(\ell_k +k )}\big]  \,.
\end{align}
The states in~\cref{proof_eq_reduced-seq} do not yet form a well-defined state sequence $\bar{\pmb{\rho}} = \big(\bar{\rho}^{(k)}\big)_{k \in \N}$, since multiple values $n$ result in the same value $k_n$ due to rounding. However, a valid sequence can always be defined by picking any representative value $n^*_k \in J_k := \{ n\in \N \,: \, k_n = k\}$ for any $k \in \N$ and taking $\tilde{\rho}^{(k)} = \bar{\rho}^{(k_{n^*})}$. Then, by rewriting $\tilde{\rho}^{(k)} =  \tr_{\ell_k}\big[\rho^{(\ell_k +k)}\big]  $ with $\ell_k := n^*_k$, it follows directly from~\cref{def_tom34-3} in~\cref{def_tom34} that \smash{$\bar{\pmb{\rho}} = \big(\tilde{\rho}^{(k)}\big)_{k \in \N}  \in \OP(\texttt{tom($\cdot$)};\rho)$}. More precisely, there exists a poly-logarithmic sample complexity $N_0$ such that for any choice of $\ell_k$ we have \smash{$\bar{\pmb{\rho}} \in \OP^{(N_0)}(\texttt{tom($\cdot$)};\rho)$}.

%In order to make use of~\cref{def_tom34-3} in~\cref{def_tom34}, we need to rewrite the above reduced states. 
%Thus, for any $k \in \N$, consider the set $J_k := \{ j\in \N \,: \, k_j = k\}$. Let $\ell_k$ be any sequence of integers such that 
Second, recall that by~\cref{def_OP-set,lem_sample-subexp}, there exist constants $K,L >0$, $q \geq 0$ and $\nu \in (0,1]$ such that
 \begin{align}  \label{proof_eq_tomo1}
  %\forall\,(\epsilon,\delta) \in (0,1)\times (0,\tfrac{1}{2}] \ \ \forall\, n\geq N_0(\epsilon,\delta)\,:\  
  \Pr\big[\texttt{tom}\big(\tilde{\rho}^{(k)} \big) \notin B_\epsilon(\rho)\big] \, \leq \,K \,2^{-k^\nu \epsilon^{q\nu} L} 
 \end{align}
 for any $k\in \N$ and $\epsilon \in (0,1)$.
%Second, recall that by~\cref{def_OP-set}, for any $\epsilon,\,\delta >0$ there exists a $N_0(\epsilon,\delta) \in \N$ and a tomography protocol \texttt{tom($\cdot$)} such that for all $k_n \geq N_0(\epsilon,\delta)$ we have
%\begin{align} \label{proof_eq_tomo1}
%    \Pr\big[\texttt{tom}\big(\bar{\rho}^{(k_n)}\big) \notin B_\epsilon(\rho)\big] \leq \delta\,, \quad \textrm{and} \quad N_0(\epsilon,\delta) \leq \poly\left(\tfrac{1}{\epsilon},\,\log\tfrac{1}{\delta}\right)\,.
%\end{align}
%W.l.o.g., we can assume that $\epsilon,\,\delta \in (0,1)$, as otherwise one may just choose $N_0(\epsilon,\delta) =1 $. The latter condition in~\cref{proof_eq_tomo1} then implies that there exists some constant $K >0$ and positive integers $q,s \in \N$ such that 
%\begin{align} \label{proof_eq_thresh}
%    N_0(\epsilon,\delta) \leq K\left(\tfrac{1}{\epsilon}\right)^{q}\left(\log\tfrac{1}{\delta}\right)^s \quad \implies \quad \delta \,\leq\,
%    \exp\Big(-\tilde{K}N_0^{\frac{1}{s}}\epsilon^{\frac{q}{s}}\Big)\,,
%\end{align}
%with $\tilde{K} = K^{-\frac{1}{s}}$. 
Importantly, the values $K,\,L,\,q,\,\nu$ are independent of $\epsilon$ and the integer sequence $(\ell_k)_{k\in\N}$, and depend only on the explicit tomography protocol \texttt{tom($\cdot$)} and the original state sequence $\pmb{\rho}$.

Now, let $n\in \N$ be fixed. Let $\zeta_n :=\frac{c}{n^\gamma}$ for any constant $c\in(0,1)$ and any fixed $\gamma$ satisfying $0\leq\gamma< \frac{\alpha}{q}\,$ (in the case $q=0$, the parameter $\gamma \geq 0$ may be chosen arbitrarily). Furthermore, for any $k\in \N$, define 
\begin{align} \label{proof_eq-argmax}
n^*_k :=  \argmax_{j \in J_{k}}  \,  \Pr\big[\texttt{tom}\big(\bar{\rho}^{(k_j)}\big) \notin B_{\zeta_j}(\rho)\big]  \,,
\end{align}
resulting in a valid state sequence $\bar{\pmb{\rho}}$ as described above.
Then, it follows from~\cref{proof_eq_tomo1,proof_eq-argmax} that
\begin{align}
    \Pr\big[\texttt{tom}\big(\bar{\rho}^{(k_n)}\big) \notin B_{\zeta_n}(\rho)\big] \,\overset{\textnormal{\Cshref{proof_eq-argmax}}}&{\leq}\,  \Pr\big[\texttt{tom}\big(\tilde{\rho}^{(k_n)}\big) \notin B_{\zeta_{n^*}}(\rho)\big] \\
     \,\overset{\textnormal{\Cshref{proof_eq_tomo1}}}&{\leq}\,K \,2^{-L k_n^\nu \zeta_{n^*}^{q\nu} } \\ 
     \,\overset{n^*\, \leq\, k_n^{1/\alpha}}&{\leq}\, K \,2^{- L c^{q\nu} k_n^{\nu(1-\gamma q / \alpha  ) }} \\
     \,\overset{k_n \, \geq \,n^\alpha}&{\leq}\, K \,2^{- L c^{q\nu} n^{\nu(\alpha-\gamma q   ) }}  \quad \longrightarrow\ \textrm{ is subexponentially decreasing in $n$} \,.  \label{proof_eq_tomo2}
\end{align}
%Then, by choosing $\delta_n := \exp\big(-\tilde{K}c^{\frac{q}{s}}n^\frac{\alpha-\gamma q}{s}\big)$, it follows due to~\cref{proof_eq_thresh} that
%\begin{align}
%    N_0(\zeta_n,\delta_n) \leq K\left(\tfrac{1}{\zeta_n}\right)^{q}\left(\log\tfrac{1}{\delta_n}\right)^s \leq k_n\,, 
%\end{align}
%and therefore, because of~\cref{proof_eq_tomo1},
%\begin{align} \label{proof_eq_tomo2}
%    \Pr\big[\texttt{tom}\big(\bar{\rho}^{(k_n)}\big) \notin B_{\zeta_n}(\rho)\big] \,\leq\, \exp\Big(-\tilde{K}c^{\frac{q}{s}}\,n^{\frac{\alpha-\gamma q}{s}} \Big) \quad \longrightarrow\quad \textrm{is negligible in}\ \  n\,.
%\end{align}

As a next step, let $\zeta_n = \frac{1}{2n^\gamma}$ and consider a POVM $\big\{\Pi_{\zeta_n}^m,\overline{\Pi}_{\zeta_n}^m \big\}$, which is determined by the tomography protocol \texttt{tom($\cdot$)}, such that
\begin{align}  \label{proof_eq_tomo3}
\Pr\big[\texttt{tom}\big(X^{(m)}\big) \in B_{\zeta_n}(\rho)\big] \,=\, \tr\big[\Pi_{\zeta_n}^m X^{(m)} \big]\,, \quad  \overline{\Pi}_{\zeta_n}^m \,:=\, \id_C^{\ox m} - \Pi_{\zeta_n}^m\,,
\end{align}
for any $X^{(m)}\in \St(\cH_C^{\otimes m})$.
%As a next step, we may describe the tomography protocol \texttt{tom($\cdot$)} by a POVM, denoted by $\{T^m_{\hat{\rho}} \}_{\hat{\rho}}\,$, which acts on $m$ copies, and where $\hat{\rho}$ labels any possible output single-copy state in $\St(\cH_C)$.\GM{To be precise, one should replace the sums below by integrals or sth. similar.} With that, and choosing $\zeta_n = \frac{1}{2n^\gamma}$, we define a coarse grained POVM $\big\{\Pi_{\zeta_n}^m,\overline{\Pi}_{\zeta_n}^m \big\}$ by
%\begin{align} \label{proof_eq_tomo3}
%    \Pi_{\zeta_n}^m \,:=\, \red{\sum_{\hat{\rho}\in B_{\zeta_n}(\rho)}} T^m_{\hat{\rho}}\,, \quad \textrm{and} \quad \overline{\Pi}_{\zeta_n}^m \,:=\, \id_C^{\ox m} - \Pi_{\zeta_n}^m \,= \red{\sum_{\hat{\rho}\notin B_{\zeta_n}(\rho)}} T^m_{\hat{\rho}} \,.
%\end{align}
Then, for $m:=k_n$,~\cref{proof_eq_tomo2,proof_eq_tomo3} imply that
\begin{align}
    p_{\zeta_n}^{k_n} \,:=\,\tr\Big[\big(\id_C^{\ox n} \ox \overline{\Pi}_{\zeta_n}^{k_n}  \big) \,\rho^{(n+k_n)}\Big] \,\overset{\textnormal{\Cshref{proof_eq_tomo3}}}&{=} \, \Pr\big[\texttt{tom}\big(\bar{\rho}^{(k_n)}\big) \notin B_{\zeta_n}(\rho)\big] \label{proof_eq_tomo4} \\
    \quad \overset{\textnormal{\Cshref{proof_eq_tomo2}}}&{\longrightarrow}\ \textrm{ is subexponentially decreasing in $n$}\,. \label{proof_eq_lowerbound1}
\end{align}

\medskip
\noindent \textbf{Step 2:} \textit{(Restricting the integration, part 1)} After the above preparation, consider now the first term in~\cref{eq_peaked-de-Finetti}. Because $\rho^{(n+k_n)}$ is permutation invariant by assumption, there exists a symmetric purification $\ket{\Psi^{(n+k_n)}} \in \mathrm{Sym}^{n+k_n}\big(\cH_C \ox \cH_{\bar{C}}\big)$ of $\rho^{(n+k_n)}$~\cite[Lemma 4.2.2]{renner_phd}. Then, we can rewrite the first term in~\cref{eq_peaked-de-Finetti} as
\begin{align}
    \tr_{k_n}\big[\rho^{(n+k_n)}\big] &=\tr_{k_n}\tr_{\bar{C}}\Big[\proj{\Psi}^{(n+k_n)}_{C\bar{C}} \Big]   \label{proof_eq_LHS4-1} \\
    &= \tr_{k_n}\tr_{\bar{C}}\Big[\id^{\ox n}_{C\bar{C}} \ox \id_{\mathrm{Sym}^{(k_n)}(C\bar{C})} \cdot \proj{\Psi}^{(n+k_n)}_{C\bar{C}} \Big]  \label{proof_eq_LHS4-2}\\
    \overset{\textnormal{\cite[Lem. 4.3.1]{renner_phd}}}&{=} \tr_{\bar{C}}\left( \binom{k_n+d_C^2-1}{k_n} \int_{\cS_1(C\bar{C})}\tr_{k_n}\Big[\id^{\ox n}_{C\bar{C}} \ox \proj{\theta}^{\ox k_n}_{C\bar{C}}  \cdot \proj{\Psi}^{(n+k_n)}_{C\bar{C}}  \Big] \di \theta \right)  \label{proof_eq_LHS4-3} \\
    &= \tr_{\bar{C}}\left( \int_{\cS_1(C\bar{C})}  \rho^{\ket{\theta}}_n \,\di \theta \right) \label{proof_eq_LHS4}
\end{align}
where $d_C := \dim(\cH_C)$, $\cS_1(C\bar{C})$ denotes the set of all normalized vectors $\ket{\theta} \in \cH_{C\bar{C}}$ and $\di \theta$ is the uniform probability (Haar) measure on $\cS_1(C\bar{C})$. The second step above holds since $\mathrm{Sym}^{n+k_n}(\cH) \subseteq \mathrm{Sym}^{n}(\cH)\ox \mathrm{Sym}^{k_n}(\cH) $ for any Hilbert space $\cH$, and in the last step we have defined
\begin{align} \label{proof_eq_def-rho-theta}
    \rho^{\ket{\theta}}_n \,:=\, \binom{k_n+d_C^2-1}{k_n}\, \tr_{k_n}\Big[\id^{\ox n}_{C\bar{C}} \ox \proj{\theta}^{\ox k_n}_{C\bar{C}}  \cdot \proj{\Psi}^{(n+k_n)}_{C\bar{C}}  \Big]\,.
\end{align}
For any $\gamma$ satisfying $0\leq\gamma< \frac{\alpha}{q}\,$, define $\epsilon_n := \frac{1}{n^\gamma}$. Now consider (normalized) vectors $\ket{\theta}_{C\bar{C}}$ with marginal $\theta_C := \tr_{\bar{C}}(\proj{\theta}_{C\bar{C}})$. 
%A natural measure on the reduced space of density operators in $\St(\cH_C)$ is then given by the one induced by the Hilbert-Schmidt metric~\cite{Zyczkowski_2001, Christandl_2009}, denoted by $\mu(\cdot)$. 
Then, we want to bound the following term,
\begin{align}
    \frac{1}{2}\,\bigg\| \tr_{k_n}\big[\rho^{(n+k_n)}\big] \,-\,\int_{\theta_C\in B_{\epsilon_n}(\rho)} &\tr_{\bar{C}}\big( \rho^{\ket{\theta}}_n\big)\,\di \theta  \bigg\|_1  \label{proof_eq_tprob0}\\
    \,\overset{\textnormal{DPI}}&{\leq} \, \frac{1}{2}\,\bigg\|\tr_{k_n}\Big[ \proj{\Psi}^{(n+k_n)}_{C\bar{C}}\Big] \,-\, \int_{\theta_C\in B_{\epsilon_n}(\rho)} \rho^{\ket{\theta}}_n \,\di \theta \bigg\|_1 \\
    \overset{\textnormal{\Cshref{proof_eq_LHS4-1,proof_eq_LHS4-2,proof_eq_LHS4-3,proof_eq_LHS4}}}&{=}\frac{1}{2}\,\bigg\| \int_{\theta_C\notin B_{\epsilon_n}(\rho)}\rho^{\ket{\theta}}_n \,\di \theta \bigg\|_1 \\ 
    &= \frac{1}{2}\,\tr\bigg[\int_{\theta_C\notin B_{\epsilon_n}(\rho)}\rho^{\ket{\theta}}_n \,\di \theta   \bigg] \label{proof_eq_tprob3} \\
    \overset{\textnormal{\Cshref{proof_eq_def-rho-theta}}}&{=}\frac{1}{2}\binom{k_n+d_C^2-1}{k_n} \int_{\theta_C\notin B_{\epsilon_n}(\rho)}  \tr\Big[\id^{\ox n}_{C\bar{C}} \ox \proj{\theta}^{\ox k_n}_{C\bar{C}}  \cdot \proj{\Psi}^{(n+k_n)}_{C\bar{C}}\Big]\,\di \theta   \\
    &= \frac{1}{2}\binom{k_n+d_C^2-1}{k_n} \int_{\theta_C\notin B_{\epsilon_n}(\rho)} p(\theta) \,\di \theta \label{proof_eq_tprob5}
\end{align}
where~\cref{proof_eq_tprob3} follows since $\rho^{\ket{\theta}}_n $ is positive semidefinite for any $\ket{\theta}$, and in the last step we have defined
\begin{align}
    p(\theta) \,:=\, \tr\Big[\id^{\ox n}_{C\bar{C}} \ox \proj{\theta}^{\ox k_n}_{C\bar{C}}  \cdot \proj{\Psi}^{(n+k_n)}_{C\bar{C}}\Big]
\end{align}
for any $\ket{\theta} \in \cS_1(C\bar{C})$.

\medskip
\noindent \textbf{Step 3:} \textit{(Restricting the integration, part 2)} In order to bound the term in~\cref{proof_eq_tprob5}, we will next derive a bound for $p(\theta)$ for any $\ket{\theta} \in \cS_1(C\bar{C})$ such that $\theta_C = \tr_{\bar{C}}(\proj{\theta}_{C\bar{C}}) \notin B_{\epsilon_n}(\rho)$ (for $\epsilon_n = \frac{1}{n^\gamma}$ with $0\leq\gamma< \frac{\alpha}{q}\,$). Define the positive semidefinite operator
\begin{align} \label{proof_eq_def-POVM}
    P^{k_n} := \id^{\ox n}_{C\bar{C}} \ox \Pi^{k_n}_{\zeta_n} \ox \id^{\ox k_n}_{\bar{C}}\,.
\end{align}
It follows that
\begin{align}
    \tr\Big[P^{k_n} \proj{\Psi}^{(n+k_n)}_{C\bar{C}} \Big] = \tr\Big[\big(\id_C^{\ox n} \ox \Pi_{\zeta_n}^{k_n}  \big) \,\rho^{(n+k_n)}\Big]\, \overset{\textnormal{\Cshref{proof_eq_tomo3}, }\eqref{proof_eq_tomo4}}{=}\,  1- p_{\zeta_n}^{k_n} 
\end{align}
and therefore, due to the Fuchs-van de Graaf inequality and the gentle measurement lemma (see~\cref{lem_gentle-measurement}),
\begin{align}
    \frac{1}{2}\,\Bigg|\, p(\theta)\, -\, \tr\Bigg[ \id^{\ox n}_{C\bar{C}} \ox \proj{\theta}^{\ox k_n}_{C\bar{C}}  &\cdot \frac{\sqrt{P^{k_n}} \proj{\Psi}^{(n+k_n)}_{C\bar{C}}\sqrt{P^{k_n}}}{ 1- p_{\zeta_n}^{k_n}}\Bigg] \Bigg| \\
    \overset{\textnormal{DPI}}&{\leq}    
    \frac{1}{2}\,\left\| \proj{\Psi}^{(n+k_n)}_{C\bar{C}}  -\frac{\sqrt{P^{k_n}} \proj{\Psi}^{(n+k_n)}_{C\bar{C}}\sqrt{P^{k_n}}}{ 1- p_{\zeta_n}^{k_n}} \right\|_1 \\ 
    \overset{\textnormal{F.v.d.G}}&{\leq} \sqrt{1-F\left(\proj{\Psi}^{(n+k_n)}_{C\bar{C}}, \frac{\sqrt{P^{k_n}} \proj{\Psi}^{(n+k_n)}_{C\bar{C}}\sqrt{P^{k_n}}}{ 1- p_{\zeta_n}^{k_n}}  \right)} \\
    \overset{\textnormal{\Cshref{lem_gentle-measurement}}}&{\leq} \sqrt{ p_{\zeta_n}^{k_n}} \,.
\end{align}
With the above, we can now bound
\begin{align} \label{proof_eq_lowerbound2}
    p(\theta) \,\leq\, 2\sqrt{ p_{\zeta_n}^{k_n}}\, +\, \underbrace{\tr\left[ \id^{\ox n}_{C\bar{C}} \ox \proj{\theta}^{\ox k_n}_{C\bar{C}}  \cdot \frac{\sqrt{P^{k_n}} \proj{\Psi}^{(n+k_n)}_{C\bar{C}}\sqrt{P^{k_n}}}{ 1- p_{\zeta_n}^{k_n}}\right]}_{=: (\epsilon)}\,.
\end{align}
The first term in~\cref{proof_eq_lowerbound2} is negligible in $n$ as argued in~\cref{proof_eq_lowerbound1}, while the second term $(\epsilon)$ can be bounded as follows. First, note that for any $\bar{\rho}\in B_{\zeta_n}(\rho)$ and $\theta_C\notin B_{\epsilon_n}(\rho)$ we have
\begin{align}
    \tfrac{1}{2}\big\|\bar{\rho}-\theta_C \big\|_1 \,\geq\,  \tfrac{1}{2}\big\|\rho-\theta_C \big\|_1 -\,  \tfrac{1}{2}\big\|\bar{\rho}-\rho \big\|_1 \,>\, \epsilon_n - \zeta_n \,=\, \zeta_n\,, \\
    \implies \quad \bar{\rho} \notin  B_{\zeta_n}(\theta_C) \label{proof_eq_triangle-ineq}
\end{align}
due to the triangle inequality. Furthermore, recall that the trace is cyclic and satisfies $\tr(AB) \leq \tr(A)\tr(B)$ for any two positive semidefinite operators $A,\,B$. Hence, as $\omega^{(k_n)}:= \tr_n\big[\proj{\Psi^{(n+k_n)}}\big]$ is normalized, we obtain
\begin{align}
    (\epsilon) \,&=\, \frac{1}{\big(1-p_{\zeta_n}^{k_n}\big)}\, \tr\left[\sqrt{P^{k_n}}\, \id^{\ox n}_{C\bar{C}} \ox \proj{\theta}^{\ox k_n}_{C\bar{C}}\,  \sqrt{P^{k_n}} \cdot\proj{\Psi}^{(n+k_n)}_{C\bar{C}}\right] \\
     \overset{\textnormal{\Cshref{proof_eq_def-POVM}}}&{=}\,  \frac{1}{\big(1-p_{\zeta_n}^{k_n}\big)}\, \tr\left[\sqrt{\Pi^{k_n}_{\zeta_n}} \ox \id^{\ox k_n}_{\bar{C}} \, \proj{\theta}^{\ox k_n}_{C\bar{C}}\, \sqrt{\Pi^{k_n}_{\zeta_n}} \ox \id^{\ox k_n}_{\bar{C}} \cdot\omega^{(k_n)}_{C\bar{C}}\right] \\
    &\leq \, \frac{1}{\big(1-p_{\zeta_n}^{k_n}\big)}\, \tr\left[\Pi^{k_n}_{\zeta_n} \ox \id^{\ox k_n}_{\bar{C}} \, \proj{\theta}^{\ox k_n}_{C\bar{C}}  \right] \\
    %\overset{\textnormal{\Cshref{proof_eq_tomo3}}}&{=}  \,\frac{1}{\big(1-p_{\zeta_n}^{k_n}\big)}\,\sum_{\bar{\rho}\in B_{\zeta_n}(\rho)}  \tr\left[ T^{k_n}_{\bar{\rho}} \theta_C^{\ox k_n} \right] \\
    \overset{\textnormal{\Cshref{proof_eq_tomo3}}}&{=}  \,\frac{1}{\big(1-p_{\zeta_n}^{k_n}\big)}\, \Pr\big[\texttt{tom}\big(\theta_C^{\ox k_n}\big) \in B_{\zeta_n}(\rho)\big] \\
    % \overset{\textnormal{\Cshref{proof_eq_triangle-ineq}}}&{\leq} \,\frac{1}{\big(1-p_{\zeta_n}^{k_n}\big)}\,\sum_{\bar{\rho}\notin B_{\zeta_n}(\theta_C)}  \tr\left[ T^{k_n}_{\bar{\rho}} \theta_C^{\ox k_n} \right] \\
      \overset{\textnormal{\Cshref{proof_eq_triangle-ineq}}}&{\leq} \,\frac{1}{\big(1-p_{\zeta_n}^{k_n}\big)}\, \Pr\big[\texttt{tom}\big(\theta_C^{\ox k_n}\big) \notin B_{\zeta_n}(\theta_C)\big] \\
    % &{=} \,\frac{1}{\big(1-p_{\zeta_n}^{k_n}\big)}\, \Pr\big[\texttt{tom}\big(\theta_C^{\ox k_n}\big) \notin B_{\zeta_n}(\theta_C)\big]  \\
     \overset{\textnormal{\cref{def_tom12-2},  \Cshref{lem_sample-subexp} \& \Cshref{proof_eq_tomo2},\eqref{proof_eq_lowerbound1}}}&{\longrightarrow}\ \textrm{ is subexponentially decreasing in $n$} \,.
\end{align}
Note that for the last step above, one needs to further restrict the parameter $\gamma$, which enters in $\zeta_n = \frac{1}{2n^\gamma}$, to satisfy $\gamma <\frac{\alpha}{\bar{q}}$ where $\bar{q}$ is the parameter analogous to $q$ that arises from the efficient sample complexity from~\cref{def_tom12-2}. 

In summary, we have shown that for any $\ket{\theta}_{C\bar{C}}$ such that $\theta_C \notin B_{\epsilon_n}(\rho)$, it follows that $p(\theta)$ is subexponentially decreasing in $n$.

\medskip
\noindent \textbf{Step 4:} \textit{(de Finetti part)} With Step 3 it is now clear that the bound in~\cref{proof_eq_tprob5} is subexponentially decreasing in $n$. Hence, we can now focus on the second term in~\cref{proof_eq_tprob0}. For that, we can repeat the arguments of~\cite[Theorem 4.3.2]{renner_phd} as follows. Let $P^{\ket{\theta}}$ denote the projector onto the space of pure states that are $\binom{n}{r_n}$-almost-iid along $\ket{\theta}$, denoted by $\mathrm{Sym}\big((C\bar{C})^{\ox n},\, \ket{\theta}^{\ox n - r_n} \big)$. Define the projected normalized density operator
\begin{align}
    \bar{\rho}^{\ket{\theta}}_n \,:=\, \frac{1}{q(\ket{\theta})}\,P^{\ket{\theta}}\rho^{\ket{\theta}}_n P^{\ket{\theta}}\,,
\end{align}
where $q(\ket{\theta}) := \tr\big(P^{\ket{\theta}}\rho^{\ket{\theta}}_n P^{\ket{\theta}}\big)$. In the case $q(\ket{\theta}) = 0$, the state $ \bar{\rho}^{\ket{\theta}}_n$ may be chosen arbitrarily. Then, it holds that
\begin{align}
    \frac{1}{2}\,\bigg\|\int_{\theta_C\in B_{\epsilon_n}(\rho)} \tr_{\bar{C}}\big( \rho^{\ket{\theta}}_n\big)\,\di \theta  \,&-\,\int_{\theta_C\in B_{\epsilon_n}(\rho)} \tr_{\bar{C}}\big( \bar{\rho}^{\ket{\theta}}_n\big)\,q(\ket{\theta})\,\di \theta  \bigg\|_1 \\
    \overset{\textnormal{DPI}}&{\leq} \frac{1}{2}\,\bigg\|\int_{\theta_C\in B_{\epsilon_n}(\rho)}  \rho^{\ket{\theta}}_n\,\di \theta \,-\,\int_{\theta_C\in B_{\epsilon_n}(\rho)} \bar{\rho}^{\ket{\theta}}_n\,q(\ket{\theta})\,\di \theta \bigg\|_1 \\
    &\leq  \frac{1}{2}\,\int_{\theta_C\in B_{\epsilon_n}(\rho)} \Big\| \rho^{\ket{\theta}}_n \,-\, P^{\ket{\theta}}\rho^{\ket{\theta}}_n P^{\ket{\theta}}  \Big\|_1 \di\theta \\
    &\leq  \frac{1}{2}\,\int_{\cS_1(C\bar{C})} \Big\| \rho^{\ket{\theta}}_n \,-\, P^{\ket{\theta}}\rho^{\ket{\theta}}_n P^{\ket{\theta}}  \Big\|_1 \di\theta\,, \label{proof_eq_old-deFinetti}
\end{align}
and at this point, we can just repeat exactly the same steps as in~\cite[Theorem 4.3.2]{renner_phd} to show that the term in~\cref{proof_eq_old-deFinetti} satisfies
\begin{align}
    \frac{1}{2}\,\int_{\cS_1(C\bar{C})} \Big\| \rho^{\ket{\theta}}_n \,-\, P^{\ket{\theta}}\rho^{\ket{\theta}}_n P^{\ket{\theta}}  \Big\|_1 \di\theta \,&\leq\, k_n^{d^2_C/2}\exp\left(- \frac{k_n(r_n+1)}{2(n+ k_n)} \right) \\
    &\leq\,\big(n^\alpha +1\big)^{d^2_C/2}\exp\left(- \frac{n^{\alpha+\beta}}{2(n+ n^\alpha+1)} \right) \\
    \quad &\longrightarrow\ \textrm{ is subexponentially decreasing in $n$} \,.
\end{align}

\medskip
\noindent \textbf{Step 5:} 
%A natural measure on the reduced space of density operators in $\St(\cH_C)$ is then given by the one induced by the Hilbert-Schmidt metric~\cite{Zyczkowski_2001, Christandl_2009}, denoted by $\mu(\cdot)$. 
Now we can put all the pieces together. First, for any vector $\ket{\theta}_{C\bar{C}}$ we relabel $\omega := \theta_C$ and we fix any purification $\ket{\bar{\theta}_\omega}_{C\bar{C}}$ of $\omega$ (e.g., the canonical one). Then, we define
\begin{align}
q(\omega)\,:=\, \int_{\cU(\bar{C})} q(\id_C \otimes U_{\bar{C}} \ket{\bar{\theta}_\omega}) \,\di U \,\geq\,0\,,
\end{align}
where ${\cU(\bar{C})}$ denotes the set of all unitaries on $\bar{C}$ and $\di U$ is the Haar measure defined on this set. With that, we rewrite the following integral, by a change of variables, as
\begin{align}
\int_{\theta_C\in B_{\epsilon_n}(\rho)} \tr_{\bar{C}}\big( \bar{\rho}^{\ket{\theta}}_n\big)\,q(\ket{\theta})\,\di \theta \,&=\,\int_{\omega \in B_{\epsilon_n}(\rho)} \di\mu(\omega)  \int_{\cU(\bar{C})}  \tr_{\bar{C}}\Big( \bar{\rho}^{(\id_C \otimes U_{\bar{C}})\ket{\bar{\theta}_\omega}}_n\Big)\,q(\id_C \otimes U_{\bar{C}}\ket{\bar{\theta}_\omega})\,\di U\  \\
\,&=\, \int_{\omega \in B_{\epsilon_n}(\rho)} q(\omega) \di\mu(\omega)  \int_{\cU(\bar{C})}  \tr_{\bar{C}}\Big( \bar{\rho}^{\id_C \otimes U_{\bar{C}}\ket{\bar{\theta}_\omega}}_n\Big)\,\frac{q(\id_C \otimes U_{\bar{C}}\ket{\bar{\theta}_\omega})}{q(\omega)} \,\di U \\
\,&=\, \int_{\omega \in B_{\epsilon_n}(\rho)} \sigma_n^{(\omega)}  \,q(\omega) \di\mu(\omega) \,,
\end{align}
where $\mu(\cdot)$ is the measure on the reduced space of density operators in $\St(\cH_C)$ induced by the Hilbert-Schmidt metric~\cite{Zyczkowski_2001, Christandl_2009}, and we have defined
\begin{align}
 \sigma_n^{(\omega)} \,:=\, \begin{cases} \int_{\cU(\bar{C})}  \tr_{\bar{C}}\big( \bar{\rho}^{\id_C \otimes U_{\bar{C}}\ket{\bar{\theta}_\omega}}_n\big)\,\frac{q(\id_C \otimes U_{\bar{C}}\ket{\bar{\theta}_\omega})}{q(\omega)} \,\di U  & q(\omega) > 0 \\ \omega^{\otimes n}   & q(\omega) = 0\end{cases} \,.
\end{align}
Note that since $\bar{\rho}^{\ket{\theta}}_n$ is $\binom{n}{r_n}$-almost-iid along $\ket{\theta}$ for any vector $\ket{\theta} \in \cS_1(C\bar{C})$, this implies that $ \tr_{\bar{C}}\big( \bar{\rho}^{(\id_C \otimes U_{\bar{C}})\ket{\bar{\theta}_\omega}}_n\big)$ is $\binom{n}{r_n}$-almost-iid along $\omega$. On the other hand, $\frac{q(\id_C \otimes U_{\bar{C}}\ket{\bar{\theta}_\omega})}{q(\omega)} \,\di U$ defines a probability measure on $\cU(\bar{C})$ by construction. Thus, by convexity of the set of almost-iid states, it follows that $\sigma^{(\omega)}_n$ is $\binom{n}{r_n}$-almost-iid along $\omega$ for all $\omega \in \St(\cH_C)$.\footnote{More precisely, it holds that $\sigma^{(\omega)}_n \in \overline{\mathrm{conv}( \St^n(\cH_C, \omega^{\otimes n-r_n}))} =  \St^n(\cH_C, \omega^{\otimes n-r_n}) $ due to~\cite[Proposition~1.2.12]{analysis_banach} and since $ \St^n(\cH_C, \omega^{\otimes n-r_n})$ is convex and compact (see~\cref{rmk_properties_almost_iid_states}).}

Finally, by setting $\di\nu_n(\omega):= q(\omega) \di\mu(\omega)$, which is a (not necessarily normalized) measure on $ \St(\cH_C)$, we overall find
\begin{align}
     \frac{1}{2}\,\bigg\|\,\tr_{k_n}\big[\rho^{(n+k_n)}\big]\,&-\, \int_{\omega \in B_{\epsilon_n}(\rho)} \sigma_n^{(\omega)} \,\di\nu_n(\omega)\, \bigg\|_1 \\
     &\leq \,\frac{1}{2}\,\bigg\| \tr_{k_n}\big[\rho^{(n+k_n)}\big] \,-\,\int_{\theta_C\in B_{\epsilon_n}(\rho)} \tr_{\bar{C}}\big( \rho^{\ket{\theta}}_n\big)\,\di \theta \bigg\|_1 \nonumber \\
     &\phantom{aaaa}\,+\, \frac{1}{2}\,\bigg\|\int_{\theta_C\in B_{\epsilon_n}(\rho)} \tr_{\bar{C}}\big( \rho^{\ket{\theta}}_n\big)\,\di \theta \,-\,\int_{\theta_C\in B_{\epsilon_n}(\rho)} \tr_{\bar{C}}\big( \bar{\rho}^{\ket{\theta}}_n\big)\,q(\ket{\theta})\,\di \theta \bigg\|_1 \\
     \nonumber \\
     \quad &\longrightarrow\ \textrm{ is subexponentially decreasing in $n$} \,,
\end{align}
which completes the proof.
\end{proof}

\begin{proposition} \label{prop_close-almost-iid}
   For $n\in \N$, let $r =\lfloor n^\beta\rfloor$, and $\epsilon = \frac{1}{n^\gamma}$ for some $\beta,\,\gamma \in (0,1)$. Let $\rho \in \St(\cH_C)$, $\omega \in B_{\epsilon}(\rho)$, and let $\sigma_n^{(\omega)}\in \St(\cH_C^{\ox n})$ be $\binom{n}{r}$-almost-iid along $\omega$.

   Then, for any $\kappa \in (0,1)$ with $\kappa < \min\{\gamma,1-\beta\}$, there exists a sequence of states $\sigma'_n \in\St(\cH_C^{\ox n})$ that are $\binom{n}{r'}$-almost-iid along $\rho$ with $r' = r+{\lceil (n-r)2\epsilon' \rceil}$ and $\epsilon' = \frac{1}{n^\kappa}$, such that 
   \begin{align}
       \tfrac{1}{2}\Big\|\sigma_n^{(\omega)} - \sigma'_n \Big\|_1 \,\leq\,g(n) \ \ \longrightarrow\ \textrm{ is subexponentially decreasing in $n$} \,.
   \end{align}
Moreover, the function $g(n)$ is independent of $\omega$ and only depends on the parameters $\beta,\,\gamma,\,\kappa$.
\end{proposition}
\begin{proof}
Let $\ket{\rho}_{CE} \in \cH_{CE}$ be any purification of $\rho$ with $E\simeq C$. Then, by Uhlmann's theorem (combined with the Fuchs-van de Graaf inequalities), there exists a purification $\ket{\theta}_{CE} \in \cH_{CE}$ of $\omega$ such that
\begin{align} \label{prop_proof_trace-dist}
    \tfrac{1}{2}\Big\|\proj{\rho}_{CE} - \proj{\theta}_{CE} \Big\|_1 \,\leq\, \sqrt{2\epsilon}\,.
\end{align}
In the following, we may write
\begin{align} \label{prop_proof_decomp}
    \ket{\theta}_{CE} \,=\, a \ket{\rho}_{CE} + b \ket{\rho^\perp}_{CE}\,,\quad a,\,b \in \R^+ \quad \textrm{and} \quad a^2 + b^2 = 1\,,
\end{align}
where $\ket{\rho^\perp}_{CE}\in \cH_{CE} $ is chosen to satisfy $\braket{\rho^\perp|\rho} = 0$, and $a$ and $b$ can be assumed to be real and non-negative since the phases of $\ket{\theta}_{CE}$ and $\ket{\rho^\perp}_{CE}$ are irrelevant. Note that
\begin{align}
    b \,=\,\tfrac{1}{2}\Big\|\proj{\rho}_{CE} - \proj{\theta}_{CE} \Big\|_1 \overset{\textnormal{\Cshref{prop_proof_trace-dist}}}{\leq} \sqrt{2\epsilon} \\ \overset{\textnormal{\Cshref{prop_proof_decomp}}}{\implies} \quad a \geq \sqrt{1-2\epsilon}\,.
\end{align}
%By the definition of almost-iid states, there exists an orthonormal family of states $\big\{\ket{\Psi_i} \big\}_{i\in\cT}$ from $\cV(\cH^{\ox n}_{CE},\,\ket{\theta}_{CR}^{\ox n -r} )$, and an extension $\bar{\sigma}^{(\omega)}_{C^nE^n}$ of $\sigma_n^{(\omega)}$ that is of the form
%\begin{align} \label{prop_proof_def-almost-iid}
%    \bar{\sigma}^{(\omega)}_{C^nE^n} \,=\, \sum_{i,j} c_{i,j} \ket{\Psi_i}\bra{\Psi_j}\,,
%\end{align}
%with coefficients $c_{i,j} \in \C $. 
To simplify the notation, we will write $\ket{0} \equiv \ket{\rho}_{CE} $ and $\ket 1 \equiv \ket{\rho^\perp}_{CE} $ in what follows. In addition, for the (Bernoulli) probability distribution $Q^{(k,m)}$ specified by $Q^{(k,m)}(1) := \frac{k}{m} $ and $Q^{(k,m)}(0) := \frac{m-k}{m} $, we denote by \smash{$\Lambda^{Q^{(k,m)}}_m$} the corresponding type class, i.e., the set of all $m$-tuples ${\bf{x}} = (x_1,\dots,x_m)$ with $x_i \in \{0,1\}$ and with frequency distribution $\lambda_{\bf{x}} = Q^{(k,m)}$.

For $m:=n-r$, consider the tensor power state
\begin{align} \label{prop_proof_binomial1}
    \ket{\theta}_{CE}^{\ox m}\,=\,\big(a \ket{0} + b\ket{1} \big)^{\ox m} \,=\,\sum_{k=0}^m a^{m-k} b^k\left(\sum_{(x_1,\dots,x_m)\in \Lambda^{Q^{(k,m)}}_m} \ket{x_1}\ox \cdots \ox \ket{x_m}\right)\,.
\end{align}
The idea is now to truncate the sum in~\cref{prop_proof_binomial1} so as to limit the number of \emph{faulty} states $\ket 1$ in such a way that the resulting state is still close to the original state $ \ket{\theta}_{CE}^{\ox m}$. To do so, note that the index $k$ in the sum corresponds just to the number of states $\ket 1$ contained in the summand. Hence, the probability of measuring $k$ times the outcome ``$1$" in $\ket{\theta}_{CE}^{\ox m}$ is given by $\Pr[k] = \binom{m}{k}\big|a^{m-k} b^k\big|^2$, thus corresponding to a binomial distribution with mean $\mu = mb^2 \leq m2\epsilon$. This suggests truncating~\cref{prop_proof_binomial1} at a value $k$ that is strictly greater than the mean, say $\lceil m2\epsilon' \rceil$ for some $\epsilon' > \epsilon$. More concretely, let $\kappa \in (0,\gamma)$ and $\epsilon' := \frac{1}{n^\kappa}$. In the following, we consider $n\in \N$ large enough such that $2\epsilon' \leq 1$. Then, we define
\begin{align}
    \ket{\theta'^{(m)}} \,:=\, \frac{1}{\sqrt{\cN}}\,\sum_{k=0}^{\lceil m2\epsilon' \rceil} a^{m-k} b^k\left(\sum_{(x_1,\dots,x_m)\in \Lambda^{Q^{(k,m)}}_m} \ket{x_1}\ox \cdots \ox \ket{x_m}\right)\,.
\end{align}
with normalization
\begin{align}
    \cN \,=\, \sum_{k=0}^{\lceil m2\epsilon' \rceil} {(a^2)}^{m-k} (b^2)^k \binom{m}{k}\,.
\end{align}

Then, it holds that
\begin{align}
    \Big| \braket{\theta'^{(m)}|\theta^{\ox m}} \Big|^2 \,=\,\left|\frac{1}{\sqrt{\cN}}\,\sum_{k=0}^{\lceil m2\epsilon' \rceil} {(a^2)}^{m-k} (b^2)^k \binom{m}{k}  \right|^2 \,=\,\cN\,,
    %\quad \overset{m\to \infty}{\longrightarrow} \quad 1\,,
\end{align}
and therefore
\begin{align}
    \tfrac{1}{2}\Big\|\proj{\theta'^{(m)}} - \proj{\theta}^{\ox m} \Big\|_1 \,=\,\sqrt{1-\Big| \braket{\theta'^{(m)}|\theta^{\ox m}} \Big|^2} \,=\,\sqrt{1-\cN}\,.
    %\quad \overset{m\to \infty}{\longrightarrow} \quad 0\,.
\end{align}
Next, we need to bound the value of $\cN$. Fortunately, the quantity $\cN$ also appears when studying the sum of independent Bernoulli random variables. In fact, let $X_1,\dots,X_m$ be independent random variables that are distributed according to the Bernoulli distribution with probability $p= b^2$, and let $S_m := X_1+\cdots +X_m$. Then we have
\begin{align} \label{proof_prop_prob-Ber}
    \Pr\big[S_m\leq \lceil m2\epsilon' \rceil \big] \,=\, \cN \quad \implies \quad 1-\cN \,=\,\Pr\big[S_m > \lceil m2\epsilon' \rceil \big] \,\leq\, \Pr\big[S_m\geq \ m2\epsilon' \big] \,.
\end{align}
Now we can use Hoeffding's inequality~\cite[Theorem 1]{Hoeffding1963} to bound the last term in~\cref{proof_prop_prob-Ber}  as follows,
\begin{align} \label{proof_prop_Hoeffding0}
    \Pr\big[S_m\geq \ m2\epsilon' \big] \,\leq\, 2^{-m\,\Lambda(2\epsilon',b^2)}\,,
\end{align}
where
\begin{align} 
    \Lambda\big(q,p\big)\,:=\, q \log\left(\frac{q}{p} \right) + (1-q)\log\left(\frac{1-q}{1-p} \right)
\end{align}
corresponds to the Kullback–Leibler divergence between two Bernoulli distributions with probability $q$ and $p$, respectively. Furthermore, since $\Lambda\big(q,x\big)$ is monotonically decreasing in $x$ for $x\in(0,q)$, we find $\Lambda\big(2\epsilon',b^2\big) \geq \Lambda\big(2\epsilon',2\epsilon\big)$, and therefore
\begin{align} \label{proof_prop_Hoeffding}
     1-\cN \,\leq\, 2^{-m\,\Lambda(2\epsilon',2\epsilon)}\,.
\end{align}
The bound in~\cref{proof_prop_Hoeffding} will converge to zero (with increasing $n$) quickly enough for our purposes, as we shall see shortly.

Recall that by the definition of almost-iid states, there exists an orthonormal family of states $\big\{\ket{\Psi_i} \big\}_{i\in\cT}$ from $\cV^n_r(\cH_{CE},\,\ket{\theta}_{CE} )$, and an extension $\bar{\sigma}^{(\omega)}_{C^nE^n}$ of $\sigma_n^{(\omega)}$ that is of the form
\begin{align} \label{prop_proof_def-almost-iid}
    \bar{\sigma}^{(\omega)}_{C^nE^n} \,=\, \sum_{i,j \in\cT} c_{i,j} \ket{\Psi_i}\!\bra{\Psi_j}\,,
\end{align}
with coefficients $c_{i,j} \in \C $. Next, we employ~\cref{lem_purif_almost_iid} to construct a (normalized) purification $\ket{\Sigma^{(\omega)}}_{C^nE^n H}$ of $\bar{\sigma}^{(\omega)}_{C^nE^n} $ specified by
\begin{align}  \label{prop_proof_def-almost-iid-pure}
\ket{\Sigma^{(\omega)}}_{C^nE^n H} \,:=\, \sum_{i \in \cT} \alpha_i \ket{\Psi_i}_{C^nE^n} \ket{h_i}_H
\end{align}
where $\alpha_i \in \C$ such that $\sum_{i \in \cT} |\alpha_i|^2 = 1$ and $ \{\ket{h_i}_H\}_{i\in \cT}$ are normalized (but not necessarily orthogonal) states.

%Finally, we construct a new state $\sigma'_{C^nE^n} \in\St(\cH_{CE}^{\ox n})$ by taking~\cref{prop_proof_def-almost-iid} and rotating the subsystems in $\ket{\Psi_i}$ that are of the form $\ket{\theta}^{\ox m}$ into $\ket{\theta'^{(m)}}$ for all $i\in \cT$. Due to~\cref{lem_basis-unitary}, there exists a unitary map $U$ on $\cH_{CE}$ such that $U\ket{\theta} = \ket{\xi}$ and $\norm{\id - U}_{\infty} = \norm{\ket{\theta}-\ket{\xi}}_2$, where $\norm{\cdot}_2$ and \smash{$\norm{A}_\infty := \sqrt{\lambda_{\max}(A^\dagger A)}$} denote the Euclidean norm and the corresponding operator norm, respectively.

Next, we construct a new (unnormalized) vector $\ket{\Sigma'}_{C^nE^n H} \in \cH^{\otimes n}_{CE} \otimes \cH_H$ by taking~\cref{prop_proof_def-almost-iid-pure} and replacing all subsystems in $\ket{\Psi_i}$ that are of the form $\ket{\theta}^{\ox m}$ by $\ket{\theta'^{(m)}}$ for all $i\in \cT$. The resulting vectors are denoted by $\ket{\Psi'_i}$. Note that this is possible since $\ket{\theta}^{\ox m}$ is pure and, therefore, occurs as a tensor product with the remaining subsystems in $\ket{\Psi_i}$. Formally,
\begin{align}
%    \sigma'_{C^nE^n} \,:=\,  \sum_{i,j} c_{i,j} \ket{\Psi'_i}\bra{\Psi'_j}\,,\quad \ket{\Psi'_i}\in \Big\{\pi\big(\ket{\theta'^{(m)}} \otimes \ket{\Omega^{(n -m)}}\big)\ :\  \pi \in \cS_n,\, \ket{\Omega^{(n -m)}} \in \cH_{CE}^{\otimes n -m} \Big\} \quad \forall\,i\in \cT\,.
\ket{\Sigma'}_{C^nE^n H} \,:=\, \sum_{i \in \cT} \alpha_i \ket{\Psi'_i}_{C^nE^n} \ket{h_i}_H    \,.  % \quad \ket{\Psi'_i}\in \cV(\cH^{\ox n}_{CE},\,\ket{\theta'^{(m)}}_{CE} ) \quad \forall\,i\in \cT\,.
\end{align}
Note that the states $ \ket{\Psi'_i}$ are normalized and satisfy $\braket{\Psi'_i|\Psi_i} = \braket{\theta'^{(m)}|\theta^{\ox m}}= \sqrt{\cN}$ for all $i\in \cT$. By construction, $\ket{\theta'^{(m)}}$ is symmetric and $\binom{m}{{\lceil m2\epsilon' \rceil}}$-almost-iid along $\ket{0} \equiv \ket{\rho}_{CE}$. Hence, we have that 
\begin{align} \label{proof_prop_supp1}
 \ket{\Psi'_i} \in \mathrm{span}\,  \cV^n_{r + \lceil m2\epsilon' \rceil }\big(\cH_{CE},\,\ket{\rho}_{CE} \big) \quad \forall\, i\in \cT\,.
 \end{align}
%$\binom{n}{r+{\lceil m2\epsilon' \rceil}}$-almost-iid in $\ket{\rho}_{CE}$ for all $i\in \cT$. 
Furthermore, for $\norm{\cdot}_2$ denoting the Euclidean norm, we evaluate
\begin{align}
\big\| \ket{\Sigma'}_{C^nE^n H}- \ket{\Sigma^{(\omega)}}_{C^nE^n H} \big\|_2 \,&=\, \norm{\sum_{i\in\cT} \alpha_i \ket{h_i}_H \Big(\ket{\Psi'_i}_{C^nE^n} - \ket{\Psi_i}_{C^nE^n} \Big)}_2 \\
\overset{\textnormal{triangle}}&{\leq} \, \sum_{i\in\cT} | \alpha_i | \,  \big\| \ket{\Psi'_i}_{C^nE^n} - \ket{\Psi_i}_{C^nE^n} \big\|_2 \\
\overset{\textnormal{\Cshref{lem_norm-identities}}}&{=} \,  \sum_{i\in\cT} | \alpha_i |  \sqrt{2\big(1-\mathrm{Re}(\braket{\Psi_i|\Psi'_i})\big)} \\
%&\leq\, \sum_{i\in\cT} | \alpha_i |  \sqrt{2\big(1-|\braket{\Psi_i|\Psi'_i}|^2\big)} \\
&\leq \, |\cT|  \sqrt{2(1-\cN)} \\
 \overset{\textnormal{\Cshref{proof_prop_Hoeffding}}}&{\leq}\, \sqrt{2} |\cT| \sqrt{ 2^{-m\,\Lambda(2\epsilon',2\epsilon)}}\,, \label{proof_prop_bound1}
\end{align}
where we have also used that $\ket{\Psi_i}$, $\ket{\Psi'_i}$ and  $\ket{h_i}_H$ are all normalized and that $0\leq\braket{\Psi'_i|\Psi_i} \leq 1 $ for all $i\in \cT$.

Finally, we define the state
\begin{align}
 \sigma'_{C^nE^n}  \,:=\, \cS_n \left( \tr_{H}\left[\tfrac{1}{|\braket{\Sigma'|\Sigma'}|} \proj{\Sigma'}_{C^nE^n H} \right]  \right) \,,
\end{align}
where $\cS_n$ denotes the symmetrization map introduced in~\cref{eq_symmetrization-map}. Notice that~\cref{proof_prop_bound1} combined with the estimate in~\cref{proof_prop_bound3} given below justify that $\braket{\Sigma'|\Sigma'} >0$ for large enough values of $n$. By construction, $ \sigma'_{C^nE^n} $ is a valid normalized state and permutation invariant. In addition, it follows from~\cref{lem_purif_almost_iid} that the support of $ \sigma'_{C^nE^n} $ is contained in the space spanned by (permutations of) the states $ \{\ket{\Psi'_i}\}_{i\in\cT}$. As a consequence of that and~\cref{proof_prop_supp1}, we have
\begin{align}
\supp(\sigma'_{C^nE^n}) \subseteq \,  \mathrm{span}\,  \cV^n_{r + \lceil m2\epsilon' \rceil }\big(\cH_{CE},\,\ket{\rho}_{CE} \big) \,.
\end{align}
% As the support of $ \sigma'_{C^nE^n} $ is contained in the space spanned by the states $ \ket{\Psi'_i}$,
Therefore, $ \sigma'_{C^nE^n} $ is $\binom{n}{r+{\lceil m2\epsilon' \rceil}}$-almost-iid along $\ket{\rho}_{CE}$. Hence, we find that $ \sigma'_n := \tr_{E^n}[ \sigma'_{C^nE^n} ]$ is $\binom{n}{r'}$-almost-iid along $\rho$, with
\begin{align}
 r' \,:=\,  r+{\lceil (n-r)2\epsilon' \rceil}\,.
    % r' \,:=\, \lceil 2n\epsilon' \rceil +\lceil r(1-2\epsilon') \rceil + 1 \,\geq\, r+{\lceil m2\epsilon' \rceil}\,.
\end{align}
Furthermore, we can bound
\begin{align}
 \tfrac{1}{2}\Big\| \sigma'_n - \sigma_n^{(\omega)} \Big\|_1 \,\overset{\textnormal{DPI}}&{\leq}  \, \tfrac{1}{2}\Big\| \tfrac{1}{|\braket{\Sigma'|\Sigma'}|} \proj{\Sigma'}_{C^nE^n H} - \proj{\Sigma^{(\omega)}}_{C^nE^n H} \Big\|_1 \\
 \overset{\textnormal{\Cshref{lem_norm-identities} \eqref{lem_norm-identities-it3}}}&{\leq}  \,  \Big\| \tfrac{1}{\sqrt{|\braket{\Sigma'|\Sigma'}|}} \ket{\Sigma'}_{C^nE^n H} - \ket{\Sigma^{(\omega)}}_{C^nE^n H} \Big\|_2 \\
  \overset{\textnormal{triangle}}&{\leq}  \,   \Big\| \Big(\tfrac{1}{\sqrt{|\braket{\Sigma'|\Sigma'}|}} -1 \Big) \ket{\Sigma'}_{C^nE^n H}   \Big\|_2 +   \Big\| \ket{\Sigma'}_{C^nE^n H} - \ket{\Sigma^{(\omega)}}_{C^nE^n H} \Big\|_2 \\
  &= \,  \Big|1 -  \sqrt{|\braket{\Sigma'|\Sigma'}|} \Big|  +   \Big\| \ket{\Sigma'}_{C^nE^n H} - \ket{\Sigma^{(\omega)}}_{C^nE^n H} \Big\|_2 \\
   \,\overset{\textnormal{triangle}}&{\leq}  \, 2  \Big\| \ket{\Sigma'}_{C^nE^n H} - \ket{\Sigma^{(\omega)}}_{C^nE^n H} \Big\|_2 \\
   \,\overset{\textnormal{\Cshref{proof_prop_bound1}}}&{\leq} \, 2 \sqrt{2}   |\cT| \sqrt{ 2^{-m\,\Lambda(2\epsilon',2\epsilon)}}\,, \label{proof_prop_bound2}
\end{align}
where we also used that $\bar{\sigma}^{(\omega)}_{C^nE^n} $ is permutation invariant by assumption in the first step, and that $ \ket{\Sigma^{(\omega)}}_{C^nE^n H}$ is normalized in the penultimate step.

It remains to be shown that the bound in~\cref{proof_prop_bound2} vanishes subexponentially fast in $n$. Recalling that $m=n-r$ and
\begin{align}
    |\cT| \,\leq\, 2^{nh(r/n)}d_{CE}^{r} \,=\,\exp_2\Big(-r\log(\tfrac{r}{n})- n(1-\tfrac{r}{n})\log(1-\tfrac{r}{n}) +r\log(d_{CE})\Big)\,,
\end{align}
the bound in~\cref{proof_prop_bound2} becomes
\begin{align}
    \tfrac{1}{2}\Big\| \sigma'_n - \sigma_n^{(\omega)} \Big\|_1 \,&\leq  \,2\sqrt{2}\,\exp_2\Big(-\big[n\epsilon'\log(\tfrac{\epsilon'}{\epsilon})
    + \tfrac{1}{2}n(1-2\epsilon')\log(\tfrac{1-2\epsilon'}{1-2\epsilon}) \nonumber \\ 
    & \hphantom{\leq\,2\sqrt{2}\,\exp\Big(-\big[\ }
    - r\epsilon'\log(\tfrac{\epsilon'}{\epsilon}) - \tfrac{1}{2}r(1-2\epsilon')\log(\tfrac{1-2\epsilon'}{1-2\epsilon})  \nonumber \\
    &\hphantom{\leq\,2\sqrt{2}\,\exp\Big(-\big[\ } +r\log(\tfrac{r}{n})+ n(1-\tfrac{r}{n})\log(1-\tfrac{r}{n}) -r\log(d_{CE}) 
    \big] \Big) \\
    &\sim \,2\sqrt{2}\,\exp_2\Big(-\big[n^{1-\kappa} \log(n^{\gamma - \kappa})
    +  \tfrac{1}{2}n(1-\tfrac{2}{n^\kappa})\log\big(\tfrac{1-2/n^\kappa}{1-2/n^\gamma}\big) \nonumber \\ 
    & \hphantom{\leq\,2\sqrt{2}\,\exp\Big(-\big[\ }
    - n^{\beta-\kappa}\log(n^{\gamma - \kappa}) -  \tfrac{1}{2}n^\beta(1-\tfrac{2}{n^\kappa})\log\big(\tfrac{1-2/n^\kappa}{1-2/n^\gamma}\big)  \nonumber \\
    &\hphantom{\leq\,2\sqrt{2}\,\exp\Big(-\big[\ } -n^\beta \log\big({n^{1-\beta}})+ n\big(1-\tfrac{1}{n^{1-\beta}}\big) \log\big(1-\tfrac{1}{n^{1-\beta}}\big) -n^\beta \log(d_{CE}) \big] \Big) \\
    \overset{n \textrm{ large}}&{\sim} \,2\sqrt{2}\,\exp_2\Big(-\big[n^{1-\kappa} \log(n^{\gamma - \kappa})
    - n^{1-\kappa} + n^{1-\gamma} \nonumber \\ 
    & \hphantom{\leq\,2\sqrt{2}\,\exp\Big(-\big[\ ++ }
    - n^{\beta-\kappa}\log(n^{\gamma - \kappa}) + n^{\beta-\kappa} - n^{\beta - \gamma}  \nonumber \\
    &\hphantom{\leq\,2\sqrt{2}\,\exp\Big(-\big[\ ++} -n^\beta\log\big({n^{1-\beta}})- n^\beta -n^\beta \log(d_{CE}) \big] \Big) \\
    &=\,2\sqrt{2}\,\exp_2\Big(- n^{1-\kappa} \log(n^{\gamma - \kappa})\big[ 1+o(1) \big] \Big) \quad \overset{n\to \infty}{\longrightarrow} \quad 0 \,, \label{proof_prop_bound3}
\end{align}
where we used that $\epsilon' = \frac{1}{n^\kappa}$, $\epsilon = \frac{1}{n^\gamma}$, and that $\kappa < \min\{\gamma,1-\beta\}$. Since the bound above is subexponentially decreasing in $n$ and independent of the state $\omega$, this completes the proof.
\end{proof}

We now turn to our main theorem, which is obtained by combining~\cref{prop_modified-deFinetti,prop_close-almost-iid}. Informally, the theorem states that every sequence of states that belongs to the operational set, after tracing out a small number of subsystems, can be approximated by a sequence of almost-iid states with only a few defects. To state the theorem precisely and concisely, we recall the notation introduced in~\cref{not_channel-seq,not_trace-ops,not_subpoly-almost-iid,def_eq-rel}.
\begin{theorem}  \label{thm_almost-iid-dF}
Let $\rho \in \St(\cH)$ and \texttt{tom($\cdot$)} denote a tomography protocol satisfying~\cref{def_tom12,def_tom34}. Let $\pmb{\rho} \in \OP(\texttt{tom($\cdot$)};\rho)$ be permutation invariant. Then, for any $\alpha \in (0,1)$, we find
 \begin{align} 
 	\pmb{\tr}_\alpha\big( \pmb{\rho}\big)\, \simeq\, \pmb{\tau} \,\in\, \mathcal{A}(\rho)\,.
 \end{align}
Moreover, this relation is independent of the convention adopted in~\cref{not_channel-seq} for selecting the output state sequence.
\end{theorem}

\begin{proof}
First, recall that $k_n = \lceil n^\alpha \rceil$ for any $n\in \N$, and consider the (slightly modified) sequence of states $\big(\tr_{k_n}[\rho^{(n+k_n)}]\big)_n$. Due to~\cref{prop_modified-deFinetti}, there exist $\beta,\gamma>0$, a measure $\nu_n$ on $\St(\cH)$ and, for each $\omega \in \St(\cH)$, a sequence of states \smash{$\sigma_n^{(\omega)}$} that are $\binom{n}{r_n}$-almost-iid along $\omega$ with $r_n = \lfloor n^\beta \rfloor$, such that, for $\epsilon_n = \frac{1}{n^\gamma}$, we have
\begin{align} \label{thm_proof-it0}
    \tr_{k_n}\big[\rho^{(n+k_n)}\big] \, \simeq\, \int_{\omega \in B_{\epsilon_n}(\rho)} \sigma_n^{(\omega)} \,\di\nu_n(\omega) \,.
\end{align}
Note that the measure $\nu_n$ can be rescaled to a probability measure $\tilde{\nu}_n$ on $ B_{\epsilon_n}(\rho)$ since, for  $c_n :=  \int_{\omega \in B_{\epsilon_n}(\rho)} \di\nu_n(\omega)$, we evaluate
\begin{align}
\Big\|  \int_{\omega \in B_{\epsilon_n}(\rho)} \sigma_n^{(\omega)} \,\di\nu_n(\omega) -  \int_{\omega \in B_{\epsilon_n}(\rho)} \sigma_n^{(\omega)} \,\tfrac{\di\nu_n(\omega)}{c_n} \Big\|_1 \,&=\, \big| 1-\tfrac{1}{c_n} \big|\,\tr\Big[  \int_{\omega \in B_{\epsilon_n}(\rho)} \sigma_n^{(\omega)} \,\di\nu_n(\omega)\Big] \\
\,&=\,  \big| c_n -1 \big|  \\
\,\overset{\textnormal{DPI}}&{\leq} \,  \Big\|  \tr_{k_n}\big[\rho^{(n+k_n)}\big] -  \int_{\omega \in B_{\epsilon_n}(\rho)} \sigma_n^{(\omega)} \,\di\nu_n(\omega)\Big\|_1 \label{thm_proof-it00}
\end{align}
%\begin{align}
%\Big\|  \tr_{k_n}\big[\rho^{(n+k_n)}\big] -  \int_{\omega \in B_{\epsilon_n}(\rho)} \sigma_n^{(\omega)} \,\frac{\nu_n(\omega)}{c_n}\Big\|_1 \,\overset{\textnormal{triangle}}&{\leq} \, \Big\|  \tr_{k_n}\big[\rho^{(n+k_n)}\big] -  \int_{\omega \in B_{\epsilon_n}(\rho)} \sigma_n^{(\omega)} \,\nu_n(\omega)\Big\|_1 \nonumber  \\ &\qquad \qquad \qquad \,+\, \Big\| \left(1-\frac{1}{c_n}\right) \int_{\omega \in B_{\epsilon_n}(\rho)} \sigma_n^{(\omega)} \,\nu_n(\omega\Big\|_1 \\
%\,&=\, \Big\|  \tr_{k_n}\big[\rho^{(n+k_n)}\big] -  \int_{\omega \in B_{\epsilon_n}(\rho)} \sigma_n^{(\omega)} \,\nu_n(\omega)\Big\|_1 + \big| c_n -1 \big| \\
%\,\overset{\textnormal{DPI}}&{\leq} \, 2 \Big\|  \tr_{k_n}\big[\rho^{(n+k_n)}\big] -  \int_{\omega \in B_{\epsilon_n}(\rho)} \sigma_n^{(\omega)} \,\nu_n(\omega)\Big\|_1\,,
%\end{align}
Notice that since $\lim_{n\to\infty} c_n = 1$ due to~\cref{thm_proof-it0,thm_proof-it00}, we have $c_n >0$ eventually. Hence, for $\di\tilde{\nu}_n(\omega) := \frac{\di\nu_n(\omega)}{c_n}$, we find because of~\cref{thm_proof-it0} that
\begin{align} \label{thm_proof-it1}
    \tr_{k_n}\big[\rho^{(n+k_n)}\big]  \,\simeq\, \int_{\omega \in B_{\epsilon_n}(\rho)} \sigma_n^{(\omega)} \,\di\nu_n(\omega) \, \simeq\, \int_{\omega \in B_{\epsilon_n}(\rho)} \sigma_n^{(\omega)} \,\di\tilde{\nu}_n(\omega) \,.
\end{align}

Next, fix some $\kappa \in (0,1)$ with $\kappa < \min\{\gamma,1-\beta\}$. Then, by~\cref{prop_close-almost-iid}, for any state $ \sigma_n^{(\omega)}$ there exists a state $\tilde{\sigma}^{(\omega)}_n \in\St(\cH^{\ox n})$ that is $\binom{n}{r'_n}$-almost-iid along $\rho$ with $r'_n = r_n+{\lceil (n-r_n)2\epsilon'_n \rceil}$ and $\epsilon'_n = \frac{1}{n^\kappa}$, such that 
\begin{align}  \label{thm_proof-it2}
   \tfrac{1}{2}\Big\| \sigma_n^{(\omega)} - \tilde{\sigma}^{(\omega)}_n \Big\|_1 \,\leq\,g(n) \,,
\end{align}
where $g(n)$ is a subexponentially decreasing function and independent of $\omega$. Accordingly, the state $\tau^{(n)}$ defined by
\begin{align}  \label{thm_proof-it3}
\tau^{(n)} \,:=\,  \int_{\omega \in B_{\epsilon_n}(\rho)} \tilde{\sigma}^{(\omega)}_n \,\di\tilde{\nu}_n(\omega)\,
\end{align}
is also $\binom{n}{r'_n}$-almost-iid along $\rho$ because of convexity of the set of almost-iid states.\footnote{More precisely, it holds that $\tau^{(n)} \in \overline{\mathrm{conv}( \St^n(\cH, \rho^{\otimes n-r'}))} =  \St^n(\cH, \rho^{\otimes n-r'}) $ due to~\cite[Proposition~1.2.12]{analysis_banach} and since $ \St^n(\cH, \rho^{\otimes n-r'})$ is convex and compact (see~\cref{rmk_properties_almost_iid_states}).} Hence, since
\begin{align} 
 \left\|\,\int_{\omega \in B_{\epsilon_n}(\rho)} \sigma_n^{(\omega)} \,\di\tilde{\nu}_n(\omega)\, \,-\, \tau^{(n)}\,  \right\|_1 \,\overset{\textnormal{\Cshref{thm_proof-it3}}}&{\leq}\, \int_{\omega \in B_{\epsilon_n}(\rho)}  \left\|\, \sigma_n^{(\omega)} -    \tilde{\sigma}^{(\omega)}_n \right\|_1 \,\di\tilde{\nu}_n(\omega) \\
 &\leq\, \sup_{\omega \in B_{\epsilon_n}(\rho)}  \left\|\, \sigma_n^{(\omega)} -    \tilde{\sigma}^{(\omega)}_n \right\|_1 \\
 \overset{\textnormal{\Cshref{thm_proof-it2}}}&{\longrightarrow}\ \textrm{ is subexponentially decreasing in $n$} \,,  \label{thm_proof-it4}
\end{align}
we find
\begin{align} \label{thm_proof-it5}
 \tr_{k_n}\big[\rho^{(n+k_n)}\big] \, \overset{\textnormal{\Cshref{thm_proof-it1},\eqref{thm_proof-it4}}}{\simeq}\, \tau^{(n)} \,.
\end{align}
Furthermore, since
\begin{align}
r'_n =  r_n+{\lceil (n-r_n)2\epsilon'_n \rceil} = \lfloor n^\beta \rfloor + \lceil (n-\lfloor n^\beta \rfloor)\tfrac{2}{n^\kappa}\rceil \leq n^\beta + \frac{2(n-n^\beta+1)}{n^\kappa} + 1 \overset{\textnormal{\Cshref{lem_poly-1}}}{\leq} Ln^{1-\kappa}
\end{align}
for some constant $L\geq 0$, we have $\pmb{\tau} \in \cA(\rho)$.

\medskip
\noindent Now, we consider the states $ \tr_{k_n}\big[ \rho^{(n)}\big]$ for $n \in \N$. In order to apply~\cref{thm_proof-it5}, we need to match the corresponding indices. To do so, consider the following sets of integers,
\begin{gather}
S_\alpha := \big\{n \in \N \,:\,  n = \ell + k_\ell\,, \  \ell \in \N  \big\}\,, \\
 S_{\frac{1}{s}} := \big\{n \in \N \,:\,  n = \ell + \lceil \ell^{\frac{1}{s}} \rceil \,, \  \ell \in \N  \big\}\,, \quad  S_{\frac{1}{2s}} := \big\{n \in \N \,:\,  n = \ell + \lceil \ell^{\frac{1}{2s}} \rceil \,, \  \ell \in \N  \big\}\,,
\end{gather}
where $s\in \N$ is chosen such that $\frac{1}{s} \leq \alpha$. It turns out that there exist infinitely many integer values that are \emph{not} contained in $S_\alpha$. However, it is shown in~\cref{lem_number-theory} that $S_{\frac{1}{s}} \cup S_{\frac{1}{2s}} = \N \setminus \{1,3 \} $. Next, recall that by convention (see~\cref{not_channel-seq}), we would only need to consider integers $n\equiv n(m)$ that satisfy $n(m)=\min\{n\in \N \,:\,n-k_n = m\}$ for any $m \in \N$. However, the remaining arguments remain valid regardless of the choice of representative $n^*_m \in \{ n\in \N \,:\,n-k_n = m\}$.

For $n \in S_\alpha$, we may write $n = \ell + k_\ell$ for some (unique) $\ell \in \N$ and  $\Delta := k_n - k_\ell \geq 0$. Due to~\cref{thm_proof-it5}, we find
\begin{align}  \label{thm_proof-it6}
\norm{\tr_{k_n}\big[ \rho^{(n)}\big] - \tr_{\Delta}\big[ \tau^{(\ell)}\big]  }_1 \,\overset{\textnormal{DPI}}{\leq} \, \norm{\tr_{k_\ell}\big[ \rho^{(n)}\big] -  \tau^{(\ell)}  }_1 \,\overset{\textnormal{\Cshref{thm_proof-it5}}}{\leq}\, f(\ell)\,,
\end{align}
where $f:\N \to \R$ denotes a subexponentially decreasing function. Furthermore, as $\pmb{\tau} \in \cA(\rho)$, it follows from~\cref{lem_subpoly-almost-iid} that 
\begin{align}  \label{thm_proof-it7}
\bar{\tau}^{(n-k_n)} := \tr_{\Delta}\big[ \tau^{(\ell)}\big] \ \textrm{  is $\binom{n-k_n}{\big\lceil \bar{L}(n-k_n)^{1-\kappa} \big\rceil}$-almost-iid along $\rho$}
\end{align}
for some constant (in $n$) $\bar{L}\geq 0$.
% $\tr_{\Delta}\big[ \tau^{(\ell)}\big] $ is $\binom{n-k_n}{r'_n}$-almost-iid along $\rho$ as mentioned in~\cref{rmk_properties_almost_iid_states}. 

%the state sequence $\bar{\pmb{\tau}}$, defined by
%\begin{align}
%\bar{\pmb{\tau}} \,:=\,\big(\tr_{\Delta}\big[ \tau^{(\ell)}\big]  \big)_{m \in \N}
%\end{align}
%also satisfies $\bar{\pmb{\tau}} \in \cA(\rho)$.

For $\bar{n} \notin S_\alpha$ and $\bar{n} \geq 4$, we still have $\bar{n} \in S_{\frac{1}{s}}$ or $\bar{n} \in S_{\frac{1}{2s}}$. Hence, by repeating the whole argument above for $k^j_\ell := \lceil \ell^j\rceil$ with $j\in \{\tfrac{1}{s},\tfrac{1}{2s}\}$, there exist sequences of states $\pmb{\tau}_j$ such that
\begin{align} \label{thm_proof-it8}
\norm{\tr_{k_{\bar{n}}} \big[ \rho^{(\bar{n})}\big] - \tr_{k_{\bar{n}}-k^j_{\bar{n}}} \big[\tr_{\Delta_j}\big[ \tau_j^{(\bar{\ell})}\big]\big]  }_1 \,\overset{\textnormal{DPI}}{\leq} \, \norm{\tr_{k^j_{\bar{n}}} \big[ \rho^{(\bar{n})}\big] - \tr_{\Delta_j}\big[ \tau_j^{(\bar{\ell})}\big] }_1  
 \,\overset{\textnormal{\Cshref{thm_proof-it6}}}{\leq}\, f_j(\bar{\ell})\,,
\end{align}
where $\bar{n} = \bar{\ell} + k^j_{\bar{\ell}}$, $\Delta_j  := k^j_{\bar{n}} - k^j_{\bar{\ell}} \geq 0$, and $f_j :\N \to \R$ denotes a subexponentially decreasing function. Again, it follows from~\cref{lem_subpoly-almost-iid} that
\begin{align}  \label{thm_proof-it9}
\bar{\tau}_j^{(\bar{n}-k_{\bar{n}})} :=\tr_{k_{\bar{n}}-k^j_{\bar{n}}} \big[\tr_{\Delta_j}\big[ \tau_j^{(\bar{\ell})}\big]\big]  \ \textrm{  is $\binom{\bar{n}-k_{\bar{n}}}{\big\lceil \bar{L}_j(\bar{n}-k_{\bar{n}})^{1-\kappa_j} \big\rceil}$-almost-iid along $\rho$}
\end{align}
with constants $\bar{L}_j \geq 0$ and $\kappa_j \in (0,1)$.

Finally, we can combine all of the above by defining $\tilde{L} := \max\big\{\bar{L}, \bar{L}_{\frac{1}{s}},  \bar{L}_{\frac{1}{2s}} \big\}$, $\tilde{\kappa} := \min\big\{\kappa, \kappa_{\frac{1}{s}},  \kappa_{\frac{1}{2s}} \big\}$, and $\tilde{f}(n) := \max\big\{f(n), f_{\frac{1}{s}}(n), f_{\frac{1}{2s}}(n) \big\}$ for all $n\in \N$. Furthermore, we define the sequence $\bar{\pmb{\tau}}$ as follows: Given an output blocklength $m\in \N$, pick $n \equiv n(m)$, $j\equiv j(n) \in \{\alpha,\tfrac{1}{s},\tfrac{1}{2s}\}$ such that $n\in S_j$, and $\ell \equiv \ell(n)$ such that $n=\ell+k_\ell^{j}$, and set $ \Delta_n \,:=\, k_n-k_\ell^{j}$. Then, we define
\begin{align}
\bar{\tau}^{(m)} \,:=\, \tr_{\Delta_{n}}\big[ \tau_j^{(\ell)}\big]
\end{align}
for all $m\in \N$ analogously to~\cref{thm_proof-it7,thm_proof-it9}, implying that $\bar{\tau}^{(m)}$ is $\binom{m}{\lceil \tilde{L}m^{1-\tilde{\kappa}}\rceil}$-almost-iid along $\rho$, and hence, $\bar{\pmb{\tau}} \in \cA(\rho)$. Furthermore, because of~\cref{thm_proof-it6,thm_proof-it8}, we have
\begin{align}  \label{thm_proof-it10}
\norm{\tr_{k_n}\big[ \rho^{(n)}\big] -  \bar{\tau}^{(n-k_n)} }_1  \,\leq\, \tilde{f}(\ell(n)) \,\overset{\textnormal{\Cshref{lem_max-subexp} and \Cshref{def_subexp-equivalence}}}{\leq} \, K\,2^{-\ell(n)^\nu} \,\overset{k_n \geq k_\ell}{\leq} K\,2^{-(n-k_n)^\nu}
\end{align}
for some constants $K\geq 0$, $\nu \in (0,1)$. The bound in~\cref{thm_proof-it10} is a subexponentially decreasing function due to~\cref{ex_subexp-function}, which completes the proof.
\end{proof}

\begin{example} \label{example-dF-thrm}
Let $\tau := \frac{1}{2}\id_2$. Consider the state sequence $\pmb{\xi} = \big(\xi^{(n)} \big)_{n\in \N}$ defined by
\begin{align}
\xi^{(n)} \,:=\, \begin{dcases} \frac{1}{\binom{n}{n/2}}\sum_{\vec{z}\in \{0,1\}^n\,:\, |\vec{z}| = \frac{n}{2}} \proj{\vec{z}}, & n \textrm{ even}, \\ \tr_1\big[\xi^{(n+1)}\big], & n \textrm{ odd}, \end{dcases}
\end{align}
for all $n\in \N$, where $\tr_1$ denotes tracing out the last subsystem. For even $n$, this corresponds to a uniform (classical) mixture of bit strings with an equal number of zeros and ones. In~\cite[Proposition 25]{GPL26}, it was shown that the states $\xi^{(n)}$ cannot be almost-iid along the maximally mixed state $\tau$ with a sublinear number of defects $r_n = o(n)$, i.e., $\pmb{\xi} \notin \cA(\tau)$. However,~\cref{thm_almost-iid-dF} implies for any $\alpha \in (0,1)$ that
 \begin{align} 
 \pmb{\tr}_\alpha\big( \pmb{\xi}\big)\, \in\, B\big(\mathcal{A}(\tau)\big)\,.
 \end{align}
\end{example}
\begin{proof}
We show that the state sequence $\pmb{\xi}$ belongs to the operational set with respect to the maximally mixed state $\tau$. The intuition is simple. Under standard local tomography, as described in~\cref{app_standard-tom-protocols}, measurements in the $x$- and $y$-bases yield independent, unbiased outcomes. For even block lengths, measurements in the $z$-basis sample without replacement from a bit string containing equally many zeros and ones, so their outcome frequencies also concentrate around $\frac{1}{2}$ (even more tightly than for the iid state $\tau^{\ox n}$). Thus, the observed frequencies in all three bases lead to an estimate close to the maximally mixed state $\tau$. More quantitatively, 
one has the operator inequality
\begin{align} \label{example-dF-thrm-it1}
\xi^{(n)} \,\leq\, C_n \tau^{\ox n}\,, \qquad C_n \,:=\, \frac{2^n}{\binom{n}{n/2}} \,=\,\sqrt{\frac{\pi n}{2}} + o(1)\,,
\end{align}
for even $n$, and consequently, after taking the partial trace, an analogous inequality holds also for odd $n$. Crucially, the quantity $C_n$ scales only polynomially in $n$. Therefore, for any $n\in \N$ and $\epsilon \in (0,1)$, we can bound
\begin{align}
 \Pr\big[\texttt{tom}\big(\xi^{(n)}\big) \notin B_\epsilon(\tau)\big] \,=\,  \tr\big[\Pi_{\eps}^n \xi^{(n)} \big] \,\overset{\textnormal{\Cshref{example-dF-thrm-it1}}}{\leq} \, C_n\, \tr\big[\Pi_{\eps}^n \tau^{\ox n} \big] \,=\, C_n\Pr\big[\texttt{tom}\big(\tau^{\ox n}\big) \notin B_\epsilon(\tau)\big]\,,
\end{align}
where  \texttt{tom($\cdot$)} denotes the tomography protocol and $\Pi_{\eps}^n $ the corresponding POVM element for the above failure event. Due to~\cref{def_tom12-2} in~\cref{def_tom12},~\cref{lem_sample-subexp} and~\cref{ex_subexp-function-eps}, it follows that 
\begin{align}
\pmb{\xi}\in \OP(\texttt{tom($\cdot$)};\tau)\,.
\end{align} 
Since $\pmb{\xi}$ is further permutation invariant,~\cref{thm_almost-iid-dF} then finally proves the desired claim.
\end{proof}

Finally, we remark that in the derivation of~\cref{prop_modified-deFinetti,prop_close-almost-iid,thm_almost-iid-dF}, we never used the assumption in~\cref{def_tom34-4} of~\cref{def_tom34}. In fact,~\cref{thm_almost-iid-dF,lem_OP-set-robust} together would imply that any almost-iid state sequence $\pmb{\tau} \in \cA(\rho)$ that results from~\cref{thm_almost-iid-dF} must satisfy $\pmb{\tau} \in \OP(\texttt{tom($\cdot$)};\rho)$. However, it is a priori \emph{not} clear whether all state sequences contained in $\cA(\rho)$ can arise in that way and thus, one cannot conclude $\cA(\rho) \subseteq \OP(\texttt{tom($\cdot$)};\rho)$. Assuming~\cref{def_tom34-4} in~\cref{def_tom34} then closes this gap, and is further necessary in the study of asymptotic interconversion rates discussed in~\cref{sec:OP-rates}.
%------------------------------%------------------------------
%%%%%%%%%%%%%%%%%%%%%%%%%%%%%%%%%%%%%%%%%%%%%%%%%%%%%%%%%%%%%%%%%%%%%%%

%%%%%%%%%%%%%%%%%%%%%%%%%%%%%%%%%%%%%%%%%%%%%%%%%%%%%%%%%%%%%%%%%%%%
\newpage
\section{Standard tomography protocols fulfil the basic assumptions} \label{app_standard-tom-protocols}
There is a vast literature on the field of quantum state estimation (e.g., see~\cite{Paris2004}) and accordingly, many different tomography protocols or equivalent hypothesis testing approaches have been developed and studied to date. These include (but are not restricted to) standard approaches relying on an iid structure with a fixed~\cite{Nielsen_Chuang_2010-chap8, Vogel-Risken-1989, Aaronson-2007, Blume-Kohout_2010, Cramer2010, Adamson-Steinberg-2010, Toth-2010, Flammia2011, Flammia2012, Kalev2015, Wang-Scholz-Renner-2018} or adaptive~\cite{Huszar-Houlsby-2012, Ferrie-2014} choice of the measurement settings involving maximum likelihood estimation (MLE) or Bayesian estimation techniques, approaches based on neural network methods~\cite{Quek2021, Lange2023}, or schemes relying on joint measurements across multiple copies~\cite{Keyl2006, ODonnell2016, Haah-Harrow-2017}, and techniques generalizing quantum state estimation and learning beyond the iid setting~\cite{Christandl-Renner-2012, Fawzi-Kueng-2024, depalma2025}.

In this section, we assume a particular structure for a tomography protocol \texttt{tom($\cdot$)} that is employed by many standard approaches, and argue that, with this structure, our requirements for a sensible tomography protocol (\cref{def_tom12,def_tom34}) are indeed satisfied. % \red{We focus in particular on the assumptions in~\cref{def_tom34}, as these rely on the operational set introduced in~\cref{def_OP-set,def_eff-OP-set}, which is novel}. 
In addition, this provides a simpler handle to check for a generic protocol whether it satisfies~\cref{def_tom12,def_tom34} since the assumptions defining this particular structure (described below) are usually easier to check than the requirements listed in~\cref{def_tom12,def_tom34}. For instance,~\cref{def_tom34-3} in~\cref{def_tom34} relies on the operational set introduced in~\cref{def_OP-set} which is novel, while~\cref{def_tom12-2} in~\cref{def_tom12} %depends strongly on the precise structure of the protocol and on the classical data post-processing. Note that as this requirement 
is formulated for an iid input state sequence and therefore, it is typically analysed (and proven to hold) already in the existing literature.

In the following, the structure of the tomography protocol \texttt{tom($\cdot$)} is fixed by additional assumptions on the POVM elements and the classical data post-processing, as described in~\cref{app_assumption-POVM,app_assumption-PP} below. Then, it will be shown that~\cref{app_assumption-POVM,app_assumption-PP} imply~\cref{def_tom12,def_tom34}.

\medskip
\noindent \textbf{POVM elements:}
In general, a protocol \texttt{tom($\cdot$)} (according to~\cref{def_tom12}) is described by a POVM, denoted by $\{T^n_{\hat{\rho}} \}_{\hat{\rho}}$, which acts on $n$ copies, and where $\hat{\rho}$ labels any possible output single-copy state in $\St(\cH)$. In most of the above mentioned approaches, the POVM elements $T^n_{\hat{\rho}}$ have a tensor product structure over the $n$ copies (one important exception is given by EYD (empirical Young diagram) based protocols~\cite{Keyl2006, ODonnell2016, Haah-Harrow-2017}). More precisely, a (possibly adaptive) protocol is often described by a collection of measurement configurations, labelled by $\alpha \in \cA$, where each measurement configuration $\alpha$ specifies itself a POVM on a single-copy system, \smash{$\big\{ M^{\alpha}_x \big\}_{x \in \cX_\alpha}$}, and $x \in \cX_\alpha$ denote possible outcomes of a single-copy measurement. Accordingly, a multi-copy POVM element with tensor product structure then takes the form 
\begin{align} \label{eq_product-povm-1}
T^n_{(x_1,\dots,x_n)} =M^{\alpha_1}_{x_1}\otimes M^{\alpha_2(x_1)}_{x_2} \otimes \cdots \otimes   M^{\alpha_n(x_1,\dots,x_{n-1})}_{x_n}\,,
\end{align}
where the choice of the measurement configurations $\alpha_i \in \cA$ for $i \in [n]$ may depend on the previously collected data, $\alpha_i = \alpha_i(x_1,\dots,x_{i-1})$. The single-copy state $\hat{\rho}$ is then determined from the data $\vec{x} = (x_1,\dots,x_n)$ by some classical post-processing or estimation technique (as further explained below). % Here, we assume the sets $\cA$ and $\cX_\alpha$ to be fixed and of finite size.

In the remainder of this section, we will treat only the simpler case of a non-adaptive protocol, i.e., when the measurement setting is fixed from the beginning such that $|\cA|=1$. In fact, the whole point of an adaptive protocol is to improve protocols with a fixed choice of POVMs. It is thus reasonable to assume that the adaptivity can only make the tomography procedure more efficient, as otherwise one could just use a fixed measurement setting. Hence, if our requirements for a sensible tomography protocol (\cref{def_tom12,def_tom34}) are satisfied in the non-adaptive case (which we are going to show), then it is reasonable to demand that these properties should continue to hold for adaptive protocols as well.

In addition, we impose a \textbf{permutation invariance} on the POVM elements 
%without loss of generality. 
More precisely, consider the symmetrization map $\cS_n$ defined by \(    \mathcal{S}_n :X\mapsto \frac{1}{n!} \sum_{\pi\in S_n}U_\pi X U^\dagger_\pi\, \) for any linear operator $X \in \mathrm{L}(\cH^{\ox n})$, where $S_n$ denotes the symmetric group and $U_\pi$ is the standard unitary representation of $\pi \in S_n$ on the space $\cH^{\ox n}$. Then, the probability of observing the data $\vec{x}$ is given by
\begin{align}
 \Pr_{\rho^{\otimes n}}\big[\vec{x} \big] \equiv \tr\Big[T^n_{(x_1,\dots,x_n)} \rho^{\otimes n} \Big] = \tr\Big[T^n_{(x_1,\dots,x_n)} \cS_n\big(\rho^{\otimes n}\big) \Big] = \tr \Big[ \cS_n\big(T^n_{(x_1,\dots,x_n)} \big) \rho^{\otimes n} \Big]\,,
\end{align}
where the second equality holds as \smash{$\rho^{\otimes n}$} is permutation invariant, and the last step follows due to the cyclicity of the trace and since $\sum_{\pi\in S_n}U_\pi X U^\dagger_\pi =  \sum_{\pi\in S_n}U^\dagger_\pi X U_\pi $ for any linear operator $X$. That is, the symmetrized POVM with elements $\big\{ \cS_n\big(T^n_{\vec{x}}\big)\big\}_{\vec{x}}$ yield the same probability of observing the data $\vec{x}$ as the original POVM if the input state is iid. Since tomography procedures are usually designed for iid states, it is natural to assume the permutation invariant structure
\begin{align} \label{eq_product-povm-2}
T^n_{(x_1,\dots,x_n)} = \frac{1}{n!} \sum_{\pi\in S_n}U_\pi \big( M_{x_1}\otimes M_{x_2} \otimes \cdots \otimes   M_{x_n} \big)   U^\dagger_\pi\,.
\end{align}
Intuitively,~\cref{eq_product-povm-2} corresponds to first picking one system of the multi-copy state at random before performing the measurements. Note that~\cref{eq_product-povm-2} indeed defines a POVM since \smash{$T^n_{\vec{x}} \geq 0 $} for all $\vec{x}\in \cX^n $ and \smash{$\sum_{\vec{x}} T^n_{\vec{x}} = \id^{\otimes n}$}.

Finally, we require the (single-copy) POVM defining the tomography protocol, \smash{$ \cM = \{ M_x \}_{x \in \cX}$}, to be \textbf{informationally complete}, and we assume the set $\cX$ to be finite, $|\cX| < \infty$.\footnote{This is justified for finite-dimensional Hilbert spaces $\cH$, see for instance~\cite{Flammia2005}.} In summary, we have:
\begin{assumption}[POVM elements] \label{app_assumption-POVM}
Let $\cM = \big\{ M_x \big\}_{x \in \cX}$ denote a (single-copy) POVM. We define a corresponding (multi-copy) POVM, denoted by $\{T^n_{\vec{x}} \}_{\vec{x}\in \cX^n}$ where $\vec{x}$ denotes the observed data, with the following assumptions:
\begin{enumerate}
\item The set of outcomes is finite, $|\cX| < \infty$, and $\cM$ is informationally complete.
\item The POVM elements $T^n_{\vec{x}}$ with $\vec{x} = (x_1,\dots,x_n)$ have the tensor-product structure
\begin{align} \label{eq_product-povm-3}
T^n_{(x_1,\dots,x_n)} = \frac{1}{n!} \sum_{\pi\in S_n}U_\pi \big( M_{x_1}\otimes M_{x_2} \otimes \cdots \otimes   M_{x_n} \big)   U^\dagger_\pi\,.
\end{align}
\end{enumerate}
\end{assumption}
\phantom{abc}
\\ 
\textbf{Classical data post-processing:}
On the other hand, the structure of the function $\vec{x} \mapsto \hat{\rho}(\vec{x})$ that assigns the guessed single-copy state $\hat{\rho}$ from the given data $\vec{x} = (x_1,\dots,x_n)$ is equally relevant for the analysis of the protocol's performance. In many approaches, this mapping only depends on the empirical frequencies that are constructed from the data $\vec{x}$. More precisely, let $\vec{x} = (x_1,\dots,x_n) $ be the observed data of length $n\in \N$. The empirical frequency distribution $\lambda_{\vec{x}}$ of $\vec{x}$ is defined by 
\begin{align}
\lambda_{\vec{x}} (y) = \frac{1}{n} \sum_{i=1}^n \id[x_i = y]
\end{align}
for $y \in \cX$, and where $\id[\cdot]$ denotes the indicator function. Whenever the classical post-processing only depends on the empirical frequency distributions of the data, the mapping $\vec{x} \mapsto \hat{\rho}(\vec{x})$ shows the following schematic processing structure:
\begin{align} \label{eq_data-processing-scheme}
\vec{x}   \ \mapsto \   \lambda_{\vec{x}}  \ \mapsto \ \hat{\rho}\big(\lambda_{\vec{x}} \big) \,\in\,\St(\cH) \,.
\end{align}

On the other hand, it needs to be assured that the mapping \smash{$\lambda_{\vec{x}}  \ \mapsto \ \hat{\rho}\big(\lambda_{\vec{x}} \big)$} outputs a state $\hat{\rho}$ that is consistent with the empirical frequency distributions. Said differently, if $\hat{\rho}$ were measured w.r.t.~the POVM \smash{$\big\{ M_x \big\}_{x \in \cX}$}, we should approximately 
%\footnote{One needs to allow for approximations as not all frequencies $ \lambda^\alpha_{\vec{x}}$ (for finite $n$) are compatible with a valid (positive semi-definite) density matrix.}
retrieve the empirical frequency distribution by means of the Born rule. Formally, for some given data $\vec{x}$ (of length $n$) and guessed state $\hat{\rho}(\vec{x})$, define the distribution $\mu_{\hat{\rho}}$ such that 
% for any set of distributions \smash{$\big\{ \lambda^\alpha \big\}_{\alpha \in \cA}$} and guessed state \smash{$\hat{\rho} = \hat{\rho}\big(\big\{ \lambda^\alpha \big\} \big)$}
\begin{align}  \label{eq_inverse-freq-map}
\tr \big[ M_y \hat{\rho}(\vec{x}) \big] \,=\, \mu_{\hat{\rho}}(y) \qquad  \forall\, y\in \cX\,.
\end{align}
Then, one may wish to obtain $ \lambda_{\vec{x}}= \mu_{\hat{\rho}}$. However, such an equality may not hold in practice as there exist empirical frequency distributions $ \lambda_{\vec{x}}$ that cannot arise from any valid (positive semi-definite) density matrix $\hat{\rho}$ by means of~\cref{eq_inverse-freq-map}. Therefore, we need to allow for an approximate relation. More precisely, we require that there exists an efficient sample complexity $\bar{N}_0$ such that
\begin{align}  \label{eq_inverted-born-rule}
 \forall\, n\geq \bar{N}_0(\epsilon,\delta)\,:\quad  \Pr_{\sigma^{\otimes n}}\Big[ \big\|  \lambda_{\vec{x}} - \mu_{\hat{\rho}}\big\|_1 >  \epsilon \Big] \,\leq\, \delta
\end{align}
for all $(\epsilon,\delta) \in (0,1)\times (0,\tfrac{1}{2}]$, and where the probability is taken over the data $\vec{x}$. As indicated by the notation above, $\vec{x}$ is assumed to be distributed according to a distribution resulting from measuring any iid state $\sigma^{\otimes n}$ of some underlying (``true") state $\sigma \in \St(\cH)$. Said differently, $\vec{x}$ is sampled iid from the classical distribution \smash{$P_X$} where \smash{$P_X(x) := \tr[M_x \sigma]$ for $x\in \cX$}. The above assumption is justified since for many tomography procedures, the distributions \smash{$ \lambda_{\vec{x}}$} and \smash{$ \mu_{\hat{\rho}}$} will actually coincide for a large enough number of samples $n$.~\cref{eq_inverted-born-rule} then constitutes even a weaker requirement. % Furthermore, in the case of a non-adaptive protocol, $|\ca| = 1$, implies~\cref{eq_lipschitz-continuity}
%\red{Also, this is exactly what~\cref{eq_tom-def1} is capturing in more general terms. Indeed, if the tomographic model is informationally complete,~\cref{eq_inverted-born-rule,eq_tom-def1} turn out to be equivalent, see~\cref{app_item2}. However, the formulation in~\cref{eq_inverted-born-rule} is based on the above structure while~\cref{eq_tom-def1} can be expressed for more general tomography procedures. } 
In summary, we have:
\begin{assumption}[Classical data post-processing] \label{app_assumption-PP}
Let \smash{$\cM = \big\{ M_x \big\}_{x \in \cX}$} denote a (single-copy) POVM as in~\cref{app_assumption-POVM}. The classical data post-processing defines a map $\vec{x}\mapsto \hat{\rho}(\vec{x})$ that assigns an observed data $\vec{x}\in \cX^n$ to some guessed (single-copy) state $\hat{\rho}(\vec{x}) \in \St(\cH)$. We assume the following structure:
\begin{enumerate}
\item The map only depends on the empirical frequency of the data: $\vec{x}   \ \mapsto \   \lambda_{\vec{x}}  \ \mapsto \ \hat{\rho}\big(\lambda_{\vec{x}} \big)$.
\item The map is consistent with the Born rule: For any $ \sigma \in \St(\cH)$, let $\mu_{\sigma}$ be a probability distribution defined by
\begin{align}  \label{eq_inverse-freq-map--}
 \mu_{\sigma}(y) \,=\, \tr \big[ M_y \sigma \big] \qquad  \forall\, y\in \cX \,.
\end{align}
Then there exists an efficient sample complexity $\bar{N}_0$ such that
\begin{align} \label{eq_inverted-born-rule--}
\forall\,(\epsilon,\delta) \in (0,1)\times (0,\tfrac{1}{2}]\  \forall\, n\geq \bar{N}_0(\epsilon,\delta)  \ \forall\,\sigma \in \St(\cH)  \,:\  \Pr_{\sigma^{\otimes n}}\Big[ \big\|  \lambda_{\vec{x}} - \mu_{\hat{\rho}(\vec{x})}\big\|_1 >  \epsilon \Big] \,\leq\, \delta
\end{align}
where the notation $\Pr_{\sigma^{\otimes n}}$ indicates that $\vec{x}$ is sampled according to the iid distribution $(P_X)^n$ defined by $P_X(x) := \tr[M_x \sigma]$ for $\sigma\in \St(\cH) $.
\end{enumerate}
\end{assumption}
\phantom{abc}
\\
\textbf{Lipschitz continuity:} Additionally, we point out that~\cref{app_assumption-POVM,app_assumption-PP} imply a type of Lipschitz continuity of the mapping \smash{$\lambda_{\vec{x}}  \ \mapsto \ \hat{\rho}\big(\lambda_{\vec{x}} \big)$}. To see this, consider first the following related statement.
\begin{lemma}[Lipschitz continuity] \label{lem_lipschitz-continuity}
Let \smash{$ \cM = \{ M_x \}_{x \in \cX}$} denote a POVM that is informationally complete and such that $|\cX| < \infty$. For any $ \sigma \in \St(\cH)$, let $\mu_{\sigma}$ be a probability distribution defined by $\mu_{\sigma}(x) \,=\, \tr \big[ M_x \sigma \big]$ for all $ x\in \cX$. Then, for any two states $\sigma,\tau \in \St(\cH)$, it holds
\begin{align}
 \big\| \mu_{\sigma} - \mu_{\tau}  \big\|_1  \,\leq\,   \big\| \sigma - \tau \big\|_1\, \leq\,C \big\| \mu_{\sigma} - \mu_{\tau}  \big\|_1  \,.
\end{align}
for some constant $C> 0$.
\end{lemma}
\begin{proof}
Define the function $f: \mathrm{Herm(\cH)} \to \R^{|\cX|}$ by \smash{$f(Y) =  \big(\tr[M_x Y]\big)_{x\in \cX} $}. Then $f$ is a linear map between finite dimensional vector spaces, hence also bounded. On the one hand, this implies that there exists a constant $\tilde{C}\geq0$ such that
\begin{align}
\big\| f(Y) \big\|_1 \,\leq\,  \tilde{C}  \| Y \|_1  \,.
\end{align}
for any $Y \in \mathrm{Herm(\cH)} $. More precisely, one can choose $\tilde{C} = 1$. This follows since any $Y \in \mathrm{Herm(\cH)} $ admits a decomposition $Y = Y_+ - Y_-$ with $Y_+,\,Y_- \geq 0$ and $Y_+Y_- = 0$. Consequently, $\|Y\|_1 =\tr(Y_+) + \tr(Y_-) $ and hence, one obtains
\begin{align}
\big\| f(Y) \big\|_1\, =\, \sum_{x\in\cX}\big|\tr[M_x Y]\big|  \,\overset{\substack{\mathrm{triangle}\\ \mathrm{ineq.}}}{\leq}\, \sum_{x\in\cX}\tr[M_x Y_+] + \tr[M_x Y_-]\, \overset{\mathrm{POVM}}{=}\, \tr[ Y_+] + \tr[ Y_-] \,=\, \|Y\|_1 \,.
\end{align} 
 Setting $Y = \sigma - \tau$ for any two states $\sigma,\tau \in \St(\cH)$ then yields
 \begin{align} \label{eq_lipschitz-2}
  \big\| \mu_{\sigma} - \mu_{\tau}  \big\|_1 \,=\,  \big\| f(\sigma-  \tau) \big\|_1 \,\leq\,   \big\| \sigma - \tau \big\|_1\,.
\end{align} 
On the other hand, since the POVM \smash{$ \{ M_x \}_{x \in \cX}$} is assumed to be informationally complete, the map $f$ is injective. Hence, this map has an inverse defined on its image, $f^{-1}: \mathrm{Im}(f) \to \mathrm{Herm(\cH)} $, that is linear and bounded. In particular, there exists a positive constant $C>0$ such that for any $Y \in \mathrm{Herm(\cH)} $ we obtain
\begin{align}
\| Y \|_1 \, = \, \big\|f^{-1}(f(Y)) \big\|_1 \,\leq\, C \big\| f(Y) \big\|_1 \,.
\end{align}
% Setting $Y = \hat{\rho} - \hat{\rho}'$ yields
Again setting $Y = \sigma - \tau$ for any two states $\sigma,\tau \in \St(\cH)$ yields
\begin{align} \label{eq_lipschitz-1}
  \big\| \sigma - \tau \big\|_1 \,\leq\,  C \big\| f(\sigma-  \tau) \big\|_1 = C \big\| \mu_{\sigma} - \mu_{\tau}  \big\|_1 
\end{align} 
for \smash{$\mu_{\sigma} := f(\sigma)$} and \smash{$\mu_{\tau} := f(\tau)$}. 
\end{proof}

\begin{lemma}[Lipschitz continuity with high probability] \label{lem_lipschitz-continuity-whp}
Let \smash{$ \cM = \{ M_x \}_{x \in \cX}$} be a (single-copy) POVM, $\sigma \in \St(\cH)$, and $P_X(x) := \tr[M_x \sigma]$ denote the classical distribution of the outcomes of $\cM$ applied to $\sigma$, respectively.  Let $\vec{x}\mapsto \hat{\rho}(\vec{x})$ be a classical data post-processing map. If~\cref{app_assumption-POVM,app_assumption-PP} are satisfied, then there exists a constant $C> 0$ such that for any $n\in \N$ and $\epsilon \in (0,1)$, it holds
\begin{align}
\Pr  \Big[  \big\| \hat{\rho}(\vec{x}) - \hat{\rho}(\vec{x}') \big\|_1 > C &\big\| \lambda_{\vec{x}} - \lambda_{\vec{x}'} \big\|_1 + \epsilon \Big] \,\leq\, g_{\epsilon}(n)\,,
\end{align}
where $\vec{x}$ and $\vec{x}'$ are sampled independently from the iid distribution $(P_X)^n$ and $ g_{\epsilon}(n)$ is an exponentially decreasing function.
\end{lemma}
\begin{proof}
%Now let \smash{$ \cM = \{ M_x \}_{x \in \cX}$} be the POVM from~\cref{app_assumption-POVM}, $\sigma \in \St(\cH)$, and $P_X(x) := \tr[M_x \sigma]$ denote the classical distribution of the outcomes of $\cM$ applied to $\sigma$, respectively. 
Let $\vec{x}$ and $\vec{x}'$ be two data sets of length $n\in\N$ that are distributed according to the iid distribution $(P_X)^n$. Let $\hat{\rho} = \hat{\rho}(\lambda_{\vec{x}})$ and $\hat{\rho}' = \hat{\rho}'(\lambda_{\vec{x}'})$ denote the guessed single-copy states w.r.t.~to the data sets $\vec{x}$ and $\vec{x}'$, respectively. By~\cref{app_assumption-PP,lem_sample-subexp}, there exist constants $K,L >0$ and $q \geq 0$ such that
\begin{align}
\Pr_{\vec{x}} \Big[ \big\|  \lambda_{\vec{x}} -\mu_{\hat{\rho}(\vec{x})}  \big\|_1 >  \epsilon \Big] \, \leq \, K\,2^{-{L}\epsilon^q n} =: {\delta}_{n,\epsilon} 
\end{align}
%Recall that the efficient sample complexity in~\cref{app_assumption-PP}, by definition, satisfies $\bar{N}_0(\epsilon,\delta) \leq K\left(\tfrac{1}{\epsilon}\right)^q \log\left(\tfrac{1}{\delta}\right)$ for some constants $K,\, q >0$. In particular, for \smash{${\delta}_{n,\epsilon} := 2^{-\frac{1}{K}\epsilon^q n}$}, \red{TODO: choose $n$ large enough such that $\delta \leq 1/2$.} it follows that \smash{$\bar{N}_0(\epsilon, {\delta}_{n,\epsilon}) \leq n$}. 
for any $\epsilon \in (0,1)$. Consequently, by combining the above with~\cref{lem_lipschitz-continuity} for $ \sigma = \hat{\rho}, \  \tau = \hat{\rho}' $, and by noting that $\|  \mu_{\hat{\rho}} - \mu_{\hat{\rho}'}  \|_1 \leq  \| \mu_{\hat{\rho}}- \lambda_{\vec{x}}  \|_1 + \| \lambda_{\vec{x}} - \lambda_{\vec{x}'}\|_1 + \| \lambda_{\vec{x}'} -  \mu_{\hat{\rho}'} \|_1 $ due to the triangle inequality, we find
\begin{align}
\Pr  \Big[  \big\| \hat{\rho} - \hat{\rho}'  \big\|_1 > C &\big\| \lambda_{\vec{x}} - \lambda_{\vec{x}'} \big\|_1 + \epsilon \Big] \, \nonumber \\
& \leq \, \Pr \Big[  \big\| \hat{\rho} - \hat{\rho}'  \big\|_1 > C \big( \big\| \mu_{\hat{\rho}} - \mu_{\hat{\rho}'}  \big\|_1 -  \big\| \mu_{\hat{\rho}}- \lambda_{\vec{x}}  \big\|_1 -  \big\|\lambda_{\vec{x}'} -  \mu_{\hat{\rho}'} \big\|_1 \big) + \epsilon \Big] \\
\overset{\textnormal{\Cshref{lem_lipschitz-continuity}}}&{\leq} \, \Pr \Big[   \big\| \lambda_{\vec{x}} - \mu_{\hat{\rho}}  \big\|_1 +  \big\|\lambda_{\vec{x}'} -  \mu_{\hat{\rho}'} \big\|_1 > \tfrac{\epsilon}{C} \Big] \\
& \leq \,\Pr \Big[  \big\| \lambda_{\vec{x}} - \mu_{\hat{\rho}}  \big\|_1 >  \tfrac{\epsilon}{2C} \quad \textrm{or}\quad  \big\|\lambda_{\vec{x}'} -  \mu_{\hat{\rho}'} \big\|_1 >   \tfrac{\epsilon}{2C} \Big] \\ 
 \overset{{\mathrm{union}\, \mathrm{bound}}}&{\leq} \, \Pr \Big[  \big\| \lambda_{\vec{x}} - \mu_{\hat{\rho}}  \big\|_1 >  \tfrac{\epsilon}{2C}  \Big] \,+\,  \Pr \Big[  \big\|\lambda_{\vec{x}'} -  \mu_{\hat{\rho}'} \big\|_1 >  \tfrac{\epsilon}{2C}  \Big] \\
\overset{\textnormal{\Cshref{eq_inverted-born-rule--}}}&{\leq}\,  2\,{\delta}_{n,\frac{\epsilon}{2C}}
\end{align}
for any $\epsilon \in (0,1)$. Finally, it follows from~\cref{def_subexp-equivalence} that $g_\epsilon(n):=   2\,{\delta}_{n,\frac{\epsilon}{2C}}$ is an exponentially decreasing function.
\end{proof}
In other words, this shows that $ \| \hat{\rho} - \hat{\rho}'  \|_1 \leq C \| \lambda_{\vec{x}} - \lambda_{\vec{x}'} \|_1 + \epsilon$ for any $\epsilon \in (0,1)$ with high probability. Intuitively, this means that if the two data sets ${\vec{x}}$ and ${\vec{x}'}$ differ in $k$ samples, then the trace distance between the guessed states $\hat{\rho}$ and $\hat{\rho}'$ (w.r.t.~the trace distance) is roughly upper bounded by \smash{$C\frac{k}{n}$}. Indeed, the above argument is the basis for the stability of many tomography protocols, and the constant $C$ is essentially determined by the smallest singular value of the linear map $f$.

%On the other hand, since $f$ is itself a linear bounded function, one also has the relation

\subsection{Proof of~\cref{def_tom12-2}} \label{app_item2}
This requirement is the defining property of any sensible tomography protocol if the (multi-copy) input state has a tensor power structure (iid). 
% Certifying this property, especially determining the sample complexity, strongly depends on the details of the protocol. Here, we show that~\cref{app_assumption-POVM,app_assumption-PP} imply~\cref{def_tom12-2}.
%Therefore, we do not aim to give a general proof that holds for all of the protocols, but we argue that~\cref{eq_product-povm-2,eq_data-processing-scheme,eq_inverted-born-rule} are sufficient for~\cref{def_tom12-2} to be satisfied.
Let $\rho \in \St(\cH)$ and let the input state sequence be given by \smash{$( \rho^{\otimes n} )_{n\in \N}$}. Since the classical post-processing scheme in~\cref{app_assumption-PP} relies only on the empirical frequency of the data and how fast it converges to the expected value (in the sense of~\cref{eq_inverted-born-rule--}), the sample complexity is determined by this convergence, and the latter is controlled by Hoeffding's inequality~\cite{Hoeffding1963,cover}. Formally, let \smash{$ \cM = \{ M_x \}_{x \in \cX}$} and $P_X(x) := \tr[M_x \rho]$ denote the POVM and the classical distribution of the outcomes of $\cM$ applied to $\rho$, respectively. Let $n\in \N$ and $\vec{x}$ be the outcome of the multi-copy POVM described in~\cref{app_assumption-POVM} applied to $ \rho^{\otimes n}$. By~\cref{prop_typical-sequences}~\eqref{lem_typical_itm2}, we have for any $\epsilon >0$ that
\begin{align}
\Pr_{\vec{x}} \Big[ \big\|  \lambda_{\vec{x}} - P_X \big\|_1 >  \epsilon \Big] \, \leq \, 2^{-n\big(\frac{\epsilon^2}{2 \ln{2}} - |\cX|\frac{\log{(n+1)}}{n} \big)} \,=: \delta_{n,\epsilon}\,.
\end{align}
On the other hand, let \smash{$\hat{\rho} = \hat{\rho}(\lambda_{\vec{x}})$} be the guessed state with associated distribution $\mu_{\hat{\rho}}$ according to~\cref{eq_inverse-freq-map--}. Because the sample complexity in~\cref{eq_inverted-born-rule--} is assumed to be efficient, it follows from~\cref{lem_sample-subexp} that there exist constants $K,L >0,\, q \geq 0$ such that
\begin{align} \label{eq_it2-assumption-rewritten}
\Pr_{\vec{x}} \Big[ \big\|  \lambda_{\vec{x}} -\mu_{\hat{\rho}(\vec{x})}  \big\|_1 >  \epsilon \Big] \, \leq \, K\,2^{-{L}\epsilon^q n} =: \tilde{\delta}_{n,\epsilon}
\end{align}
for the data $\vec{x}$ of length $n$.
%The efficient sample complexity in~\cref{eq_inverted-born-rule}, by definition, satisfies $\bar{N}_0(\epsilon,\delta) \leq K\left(\tfrac{1}{\epsilon}\right)^q \log\left(\tfrac{1}{\delta}\right)$ for some constants $K,\, q >0$. Therefore, for $\tilde{\delta}_{n,\epsilon} := 2^{-\frac{1}{K}\epsilon^q n}$, it follows that $\bar{N}_0(\epsilon, \tilde{\delta}_{n,\epsilon}) \leq n$ and hence,
%\begin{align} \label{eq_it2-assumption-rewritten}
%\Pr_{\vec{x}} \Big[ \big\|  \lambda_{\vec{x}} -\mu_{\hat{\rho}}  \big\|_1 >  \epsilon \Big] \, \leq \, \tilde{\delta}_{n,\epsilon}
%\end{align}
%for the data $\vec{x}$ of length $n$. 
By the triangle inequality \smash{$\| \mu_{\hat{\rho}} - P_X \|_1 \leq \| \mu_{\hat{\rho}} - \lambda_{\vec{x}}  \|_1 + \|  \lambda_{\vec{x}} - P_X \|_1 $}, we then obtain
\begin{align}
\Pr_{\vec{x}} \Big[ \big\| \mu_{\hat{\rho}(\vec{x})} - P_X \big\|_1 >  \epsilon \Big] \, &\leq \,\Pr_{\vec{x}} \Big[ \big\|\mu_{\hat{\rho}(\vec{x})} - \lambda_{\vec{x}}  \big\|_1 +  \big\| \lambda_{\vec{x}} - P_X  \big\|_1 >  \epsilon \Big]  \\
& \leq \,\Pr_{\vec{x}} \Big[ \big\|\mu_{\hat{\rho}(\vec{x})} - \lambda_{\vec{x}}  \big\|_1 >  \tfrac{\epsilon}{2} \quad \textrm{or}\quad  \big\| \lambda_{\vec{x}} - P_X  \big\|_1 >   \tfrac{\epsilon}{2} \Big] \\ 
 \overset{{\mathrm{union}\, \mathrm{bound}}}&{\leq} \, \Pr_{\vec{x}} \Big[ \big\|\mu_{\hat{\rho}(\vec{x})} - \lambda_{\vec{x}}  \big\|_1 >  \tfrac{\epsilon}{2}  \Big] \,+\,  \Pr_{\vec{x}} \Big[  \big\| \lambda_{\vec{x}} - P_X  \big\|_1 >  \tfrac{\epsilon}{2}  \Big] \\
& \leq \, \tilde{\delta}_{n,\frac{\epsilon}{2}} + {\delta}_{n,\frac{\epsilon}{2}}\,. \label{eq_it2-ineq3}
\end{align}
Now we may use the above to show that the guessed state $\hat{\rho}$ is close to the true state $\rho$ with high probability. To that end, first note that $\mu_{\hat{\rho}}(x) = \tr[M_x \hat{\rho}]$ and $P_X(x) = \tr[M_x \rho]$ for all $x \in \cX$ by definition. Due to~\cref{lem_lipschitz-continuity}, there exists a constant $C\geq0$ such that
\begin{align} \label{eq_it2-lower-bound}
  \big\| \hat{\rho}(\vec{x}) - \rho \big\|_1 \,\leq\, C \big\| \mu_{\hat{\rho}(\vec{x})} - P_X \big\|_1 \,.
\end{align} 
Combining~\cref{eq_it2-ineq3,eq_it2-lower-bound}, we finally obtain
\begin{align}
 \Pr_{\vec{x}} \Big[  \big\| \hat{\rho}(\vec{x}) - \rho  \big\|_1 >  \epsilon \Big] \,\leq \, \Pr_{\vec{x}} \Big[ \big\| \mu_{\hat{\rho}(\vec{x})} - P_X \big\|_1 >  \frac{\epsilon}{C} \Big]\, \leq \, \tilde{\delta}_{n,\frac{\epsilon}{2C}} + {\delta}_{n,\frac{\epsilon}{2C}}\,. \label{eq_it2-ineq4}
\end{align}
By~\cref{ex_subexp-function-eps}, there exists a function ${N}_0 : (0,1) \to \mathbb{N} $ such that \smash{${N}_0(\epsilon) \leq C \big( \frac{1}{\epsilon}\big)^{3}$} with $C>0$, and such that for all $n\geq {N}_0(\epsilon)$, we can bound \smash{$ {\delta}_{n,\frac{\epsilon}{2C}} \leq 2^{-nL'\epsilon^2}$} for some constant $L'>0$. Hence, the bound in~\cref{eq_it2-ineq4} is an exponentially decreasing function for any fixed $\epsilon \in (0,1)$ as a consequence of~\cref{lem_sums-of-subexp}. More precisely, it is visible from the proof of~\cref{lem_sums-of-subexp} (see also~\cref{ex_sums-of-subexp-eps}) that, when keeping the $\epsilon$-dependence explicit,~\cref{eq_it2-ineq4} is bounded for all $n\geq {N}_0(\epsilon)$ by 
%there exists $N_0 \in \N$ such that for all $n\geq N_0$,
%Note that $ \tilde{\delta}_{n,\frac{\epsilon}{2C}}$ and $ {\delta}_{n,\frac{\epsilon}{2C}}$ are exponentially decreasing functions for any $\epsilon \in (0,1)$ due to~\cref{ex_subexp-function}. \GM{ATTENTION: Relationsship between $n$ and $\epsilon$. Still polynomial, but to be made precise!} 
\begin{align}
 \Pr_{\vec{x}} \Big[  \big\| \hat{\rho}(\vec{x}) - \rho  \big\|_1 >  \epsilon \Big] \,\leq \, \tilde{K} 2^{- n \epsilon^{\bar{q}} \tilde{L}}
\end{align}
for some constants $\tilde{K}, \tilde{L}>0$ and where $\bar{q} := \max\{2,q\}$. Hence, by~\cref{lem_sample-subexp} there exists an efficient sample complexity $N'_0(\epsilon, \delta)$ such that 
%\smash{$ \tilde{K} 2^{- n \epsilon^{\bar{q}} \tilde{L}} \leq \delta $} for all $(\epsilon,\delta) \in (0,1)\times (0,\tfrac{1}{2}]$ and $n \geq N'_0(\epsilon, \delta) $. Finally, setting $\bar{N}_0(\epsilon, \delta) := \max\{N_0(\epsilon)\log(\frac{1}{\delta}), N'_0(\epsilon, \delta) \} \geq N_0(\epsilon)$, which is an efficient sample complexity, we find
\begin{align}
 \forall\,(\epsilon,\delta) \in (0,1)\times (0,\tfrac{1}{2}] \ \ \forall\, n\geq \bar{N}_0(\epsilon,\delta)\,:\   \Pr_{\vec{x}} \Big[  \big\| \hat{\rho}(\vec{x}) - \rho  \big\|_1 >  \epsilon \Big]\, \leq \, \delta\,, \label{eq_it2-ineq5}
\end{align}
where we set $\bar{N}_0(\epsilon, \delta) := \max\{N_0(\epsilon)\log(\frac{1}{\delta}), N'_0(\epsilon, \delta) \} \geq N_0(\epsilon)$. Since $\bar{N}_0(\epsilon, \delta)$ is again an efficient sample complexity, and its form is independent of $\rho$ (since $\delta_{n,\epsilon}$ and $\tilde{\delta}_{n,\epsilon} $ are both independent of $\rho$), this proves~\cref{def_tom12-2}.

\subsection{Proof of~\cref{def_tom34-3}} \label{app_item3}
Let $\rho \in \St(\cH)$, $\pmb{\rho} \in \OP(\texttt{tom($\cdot$)};\rho)$ and $(m_n)_{n\in\N}$ be any sequence of non-negative integers $m_n \in \N_0 $. Define the reduced state sequence $\bar{\pmb{\rho}}$ by 
\begin{align} \label{eq_it3-def-seq}
     \bar{\rho}^{(n)}:=  \tr_{m_n}\big[\mathcal{S}_{n+m_n}(\rho^{(n+m_n)})\big]\,,
\end{align}
where $\cS$ is the symmetrization map from~\cref{eq_symmetrization-map} and $\tr_{m_n}$ denotes the partial trace over (any) $m_n$ subsystems. Here, we want to show that $\pmb{\rho} \in \OP(\texttt{tom($\cdot$)};\rho)$ implies $\bar{\pmb{\rho}} \in \OP(\texttt{tom($\cdot$)};\rho)$. The corresponding proof will be more involved than in~\cref{def_tom12-2} since we do not know anything about the structure of $\pmb{\rho}$ a priori. The proof idea is to relate the probabilities \smash{$ \Pr\big[\texttt{tom}\big(\rho^{(N)}\big) \notin B_\epsilon(\rho)\big]$} and \smash{$  \Pr\big[\texttt{tom}\big(\bar{\rho}^{(n)}\big) \notin B_\epsilon(\rho)\big]$} by converting the expressions into a classical statement about the empirical frequency distributions of the observed data, and this conversion is done via~\cref{lem_lipschitz-continuity} together with~\cref{app_assumption-PP}.

As a warm up, observe that due to the permutation invariance of the (multi-copy) POVM (see~\cref{app_assumption-POVM}) we have for the data outcome probabilities
\begin{align} \label{eq_it3-PI}
\Pr_{ \cS_n( \rho^{(n)})}\big[\vec{x}_n \big] \,\equiv \,\tr\Big[ T^n_{(x_1,\dots,x_n)}\, \cS_n\big( \rho^{(n)}\big) \Big]\,\overset{\mathrm{cycl.}}{=}\,\tr\big[ \cS_n\Big(  T^n_{(x_1,\dots,x_n)}\big)\,\rho^{(n)} \Big] \, \overset{\mathrm{PI}}{=}\, \tr\Big[ T^n_{(x_1,\dots,x_n)}\,\rho^{(n)} \Big] \,\equiv\, \Pr_{  \rho^{(n)}}\big[\vec{x}_n \big] 
%\tr\Big[ T^n_{(x_1,\dots,x_n)}\, \cS_n\big( \rho^{(n)}\big) \Big] \,&=\, \frac{1}{n!} \sum_{\pi \in S_n} \tr\Big[ T^n_{(x_1,\dots,x_n)}\, U_\pi \rho^{(n)} U^\dagger_\pi \Big] \\
% \,\overset{\mathrm{cyclicity}}&{=}\,\frac{1}{n!} \sum_{\pi \in S_n} \tr\Big[U^\dagger_\pi  T^n_{(x_1,\dots,x_n)} U_\pi \, \rho^{(n)} \Big] \\
% \,&=\, \tr\big[ \cS_n\Big(  T^n_{(x_1,\dots,x_n)}\big)\,\rho^{(n)} \Big]  \\
%  \,&=\, \tr\Big[ T^n_{(x_1,\dots,x_n)}\,\rho^{(n)} \Big]
\end{align}
for any $n \in \N$ and $\vec{x}_n =(x_1,\dots,x_n)  \in \cX^n$. This means that the probability of any event is the same whether the data $\vec{x}_n $ is sampled w.r.t.~\smash{$ \cS_n( \rho^{(n)})$ or $ \rho^{(n)}$}. In particular, this directly implies that $\cS(\pmb{\rho}) \in \OP(\texttt{tom($\cdot$)};\rho) $ if $\pmb{\rho}  \in \OP(\texttt{tom($\cdot$)};\rho) $, which proves~\cref{def_tom34-3} for the special case where $m_n = 0$ for all $n \in \N$. 

Now let $n \in \N$ be fixed. We can relate the probabilities of observing some data $\vec{x}$ with respect to the two state sequences $\pmb{\rho}$ and $\bar{\pmb{\rho}}$ due to~\cref{app_assumption-POVM} as follows:
\begin{align}
 \Pr_{\bar{\rho}^{(n)}}\big[\vec{x}_n \big] \,&= \, \tr\Big[ T^n_{(x_1,\dots,x_n)}\,\bar{\rho}^{(n)} \Big] \\
 \,\overset{\textnormal{\Cshref{eq_it3-def-seq}}}&{=}\, \tr_n \tr_{m_n} \Big[T^n_{(x_1,\dots,x_n)} \otimes \id_{m_n}\,\cS_{n+m_n}\big(\rho^{(n+m_n)} \big) \Big] \\
 \,\overset{\mathrm{cyclicity}}&{=}\,  \tr_n \tr_{m_n} \Big[\cS_{n+m_n}\big(T^n_{(x_1,\dots,x_n)} \otimes \id_{m_n} \big)\,\rho^{(n+m_n)} \Big] \\
 \,&=\, \sum_{x_{n+1},\dots,x_{n+m_n}} \tr_n \tr_{m_n} \Big[\cS_{n+m_n}\Big(T^n_{(x_1,\dots,x_n)} \otimes T^{m_n}_{(x_{n+1},\dots,x_{n+m_n})} \Big)\,\rho^{(n+m_n)} \Big]  \\
 \,\overset{\textnormal{\Cshref{eq_product-povm-3}}}&{=}\,  \sum_{x_{n+1},\dots,x_{n+m_n}} \tr_n \tr_{m_n} \Big[T^{n+m_n}_{(x_{1},\dots,x_{n+m_n})} \,\rho^{(n+m_n)} \Big]  \\
 \,&=\,  \sum_{x_{n+1},\dots,x_{N}}  \Pr_{ {\rho}^{(N)}}\big[\vec{x}_N  \big] \,, \label{eq_it3-prob1}
\end{align}
where we have set $N := n + m_n$ and $\vec{x}_N := (x_{1},\dots,x_{n}, x_{n+1},\dots,x_{N})$ in the last step. This shows that \smash{$ \Pr_{\bar{\rho}^{(n)}}[\vec{x}_n ]$} corresponds to the marginal distribution of \smash{$ \Pr_{ {\rho}^{(N)}}[\vec{x}_N ]$}.

Next, we study the statistics induced by $\bar{\pmb{\rho}}$ at the level of the empirical frequency distributions, and show that it is determined by the statistics of $\pmb{\rho}$ with high probability. Let $P_X(x) := \tr[M_x \rho]$ denote the classical distribution of the outcomes of (the single-copy POVM) $\cM$ applied to $\rho$, and consider the set \smash{$F^{n}_{\epsilon} := \big\{ \vec{x}_n \,: \,\| \lambda_{\vec{x}_n} - P_X \|_1 > \epsilon \big\}$} for any $\epsilon \in (0,1)$, where $\lambda_{\vec{x}_n}$ denotes the empirical frequency distribution w.r.t.~$\vec{x}_n$. Then, the probability that \smash{$\vec{X}_n \in F^{n}_{\epsilon}  $} can be written as
\begin{align} \label{eq_it3-freq-steps-1}
% = \Pr\big[\vec{x}_n \in F^{n}_{\epsilon}  \, |\, \bar{\rho}^{(n)} \big]
 \Pr_{\bar{\rho}^{(n)} }\big[\vec{X}_n \in F^{n}_{\epsilon}  \big]  \,=\, \sum_{x_1, \dots,x_n} \id\big[\vec{x}_n \in F^{n}_{\epsilon}  \big] \,  \Pr_{\bar{\rho}^{(n)}}\big[\vec{x}_n \big]  \,\overset{\textnormal{\Cshref{eq_it3-prob1}}}{=}\,  \sum_{x_1, \dots,x_N}  \id\big[\vec{x}_n \in F^{n}_{\epsilon}  \big] \,  \Pr_{ {\rho}^{(N)} }\big[\vec{x}_N \big]\,.
 \end{align}
Note that the probability distribution \smash{$\Pr_{ {\rho}^{(N)}}[\vec{x}_N  ] $} is permutation invariant since the operator \smash{$T^{n+m_n}_{(x_{1},\dots,x_{N})}$} is permutation invariant. Thus, for any index subset $I = \{i_1,\dots,i_n \} \subset \{1,\dots,N\}$ of length $|I|=n \leq N$ (with ordering $i_1 <\dots < i_n$), we may rewrite the above as
\begin{align}
\sum_{x_1, \dots,x_N}  \id\big[\vec{x}_n \in F^{n}_{\epsilon}  \big] \,  \Pr_{ {\rho}^{(N)} }\big[\vec{x}_N \big] \,&=\,\sum_{x_1, \dots,x_N} \frac{1}{\binom{N}{n}} \sum_{I \,:\,| I | = n}  \id\big[\vec{x}_I \in F^{n}_{\epsilon} \,\big|\, \vec{x}_N  \big] \,  \Pr_{ {\rho}^{(N)} }\big[\vec{x}_N \big] \\
 \,&=\, \sum_{x_1, \dots,x_N} \Pr_{I \small \sim \mathrm{unif.}\,:\, |I| = n  }\Big[\vec{x}_I \in F^{n}_{\epsilon} \,\big|\, \vec{x}_N \Big]    \,  \Pr_{ {\rho}^{(N)} }\big[\vec{x}_N \big] \,,
\end{align}
where we have introduced the notation $\vec{x}_I := (x_{i_1},\dots,x_{i_n})$ and $\Pr_{I \small \sim \mathrm{unif.}\,:\, |I| = n  }$ indicating that the index subset $I \subset  \{1,\dots,N\}$ of length $|I| = n$ is sampled uniformly at random. Now observe that since the set $F^n_\epsilon$ is agnostic of the sequence ordering, sampling index subsets $I$ is equivalent to sampling elements $\vec{y}_n = (y_1,\dots, y_n)$ from a fixed set $\vec{x}_N = (x_1,\dots, x_N)$ at random without replacement. To see this, we evaluate for any fixed sequence $\vec{x}_N$ 
\begin{align}
  \Pr_{I \small \sim \mathrm{unif.}\,:\, |I| = n  }\Big[\vec{x}_I \in & F^{n}_{\epsilon} \,\big|\, \vec{x}_N \Big] \,=\,\frac{1}{\binom{N}{n}} \sum_{I \,:\,| I | = n}  \id\big[\vec{x}_I \in F^{n}_{\epsilon}\, \big|\, \vec{x}_N  \big] 
  \,=\, \frac{1}{\binom{N}{n}} \sum_{y_1<\dots< y_n}  \id\big[\vec{y}_n \in F^{n}_{\epsilon}\, \big|\, \vec{x}_N  \big] \\
  \,&=\, \tfrac{(N-n)!}{N!}\,n!  \sum_{y_1<\dots< y_n}  \id\big[\vec{y}_n \in F^{n}_{\epsilon}\, \big|\, \vec{x}_N  \big] 
  \,\overset{(\star)}{=}\,  \tfrac{(N-n)!}{N!}  \sum_{y_1,\dots, y_n}  \id\big[\vec{y}_n \in F^{n}_{\epsilon}\, \big|\, \vec{x}_N  \big]  \\
  \,&=\, \sum_{y_1,\dots, y_n}  \id\big[\vec{y}_n \in F^{n}_{\epsilon}\, \big|\, \vec{x}_N  \big] \,\Pr_{\substack{\vec{Y}_n \small\sim \mathrm{random} \\ \mathrm{w.o. \ rep.}}}\big[\vec{y}_n\, \big|\, \vec{x}_N \big] \,=\,\Pr_{\substack{\vec{Y}_n \small\sim \mathrm{random} \\ \mathrm{w.o.\ rep.}}}\Big[\vec{Y}_n \in F^{n}_{\epsilon}\, \big|\, \vec{x}_N  \Big],
\end{align}
where we used in step $(\star)$ that $F^n_\epsilon$ is ordering agnostic. In the following, we will write $ \Pr_{\vec{Y}_n } \equiv \Pr_{\substack{\vec{Y}_n \small\sim \mathrm{random} \\ \mathrm{w.o. \ rep.}}} $ to simplify notation. 

Consider the event $B_{\bar{\epsilon}} := \big\{ \vec{y}_n \,: \,\| \lambda_{\vec{y}_n} -  \lambda_{\vec{x}_N} \|_1 \leq \bar{\epsilon} \big\}$ for some $\bar{\epsilon} \in (0,\epsilon)$. Then, we use the law of total probability to bound
\begin{align}
 \Pr_{\vec{Y}_n } \Big[\vec{Y}_n \in F^{n}_{\epsilon}\, \big|\, \vec{x}_N  \Big] \,&=\,  \Pr_{\vec{Y}_n } \Big[\vec{Y}_n \in F^{n}_{\epsilon}\, \big|\, B_{\bar{\epsilon}},\, \vec{x}_N  \Big] \Pr_{\vec{Y}_n } \Big[ B_{\bar{\epsilon}}\,\big|\, \vec{x}_N  \Big] + \Pr_{\vec{Y}_n } \Big[\vec{Y}_n \in F^{n}_{\epsilon}\, \big|\, B^c_{\bar{\epsilon}},\, \vec{x}_N  \Big] \Pr_{\vec{Y}_n } \Big[ B^c_{\bar{\epsilon}}\,\big|\, \vec{x}_N  \Big] \\
 \,&\leq\,  \Pr_{\vec{Y}_n } \Big[\vec{Y}_n \in F^{n}_{\epsilon}\, \big|\, B_{\bar{\epsilon}},\, \vec{x}_N  \Big]  + \Pr_{\vec{Y}_n } \Big[ B^c_{\bar{\epsilon}}\,\big|\, \vec{x}_N  \Big] \,. \label{eq_it3-lawtotalprob-1}
\end{align}
For the first term in~\cref{eq_it3-lawtotalprob-1}, we use that for all $\vec{y}_n \in  F^{n}_{\epsilon} \cap B_{\bar{\epsilon}}$, we get $\|\lambda_{\vec{x}_N} - P_X \|_1 \geq \|\lambda_{\vec{y}_n} - P_X \|_1 - \|\lambda_{\vec{x}_N} - \lambda_{\vec{y}_n} \|_1 > \epsilon - \bar{\epsilon}$ due to the triangle inequality, therefore yielding
\begin{align}
 \Pr_{\vec{Y}_n } \Big[\vec{Y}_n \in F^{n}_{\epsilon}\, \big|\, B_{\bar{\epsilon}},\, \vec{x}_N  \Big] \,\leq\,  \Pr_{\vec{Y}_n } \Big[\|\lambda_{\vec{x}_N} - P_X \|_1 > \epsilon - \bar{\epsilon} \, \big|\,  \vec{x}_N  \Big] \,=\, \id\Big[\|\lambda_{\vec{x}_N} - P_X \|_1 > \epsilon - \bar{\epsilon} \Big]\,,
\end{align}
while for the second term in~\cref{eq_it3-lawtotalprob-1}, we can directly use~\cref{prop_typical-sequences}~\eqref{lem_typical_itm3} to find that
\begin{align} \label{eq_it3-freq-steps--1}
\Pr_{\vec{Y}_n } \Big[ B^c_{\bar{\epsilon}}\,|\, \vec{x}_N  \Big]  =\Pr_{\vec{Y}_n } \Big[ \| \lambda_{\vec{y}_n} -  \lambda_{\vec{x}_N} \|_1 > \bar{\epsilon}\,\big|\, \vec{x}_N  \Big] \, \leq \, 2^{-n\big(\frac{\bar{\epsilon}^2}{2 \ln{2}} - |\cX|\frac{\log{(n+1)}}{n} \big)} \,=:\, \delta_{n,\bar{\epsilon}}\,.
\end{align}
Note that the quantity $\delta_{n,\epsilon}$ is independent of $m_n$. By combining all the steps~\eqref{eq_it3-freq-steps-1}-\eqref{eq_it3-freq-steps--1} above, we obtain
\begin{align} \label{eq_it3-intermediate-1}
 \Pr_{\bar{\rho}^{(n)} }\big[\vec{X}_n \in F^{n}_{\epsilon}  \big] \,=\,\sum_{x_1, \dots,x_N} \Pr_{\vec{Y}_n }\Big[\vec{Y}_n \in F^{n}_{\epsilon}\, \big|\, \vec{x}_N  \Big]   \,  \Pr_{ {\rho}^{(N)} }\big[\vec{x}_N \big] \,\leq\,  \Pr_{ {\rho}^{(N)} } \big[\vec{X}_N \in F^{N}_{\epsilon - \bar{\epsilon}}  \big]+  \delta_{n,\bar{\epsilon}} \,.
\end{align}

Now we need to take into account how the empirical frequencies are mapped to the guessed (single-copy) states. To do so, we will make use of a post-selection de Finetti technique~\cite{Christandl_2009} together with~\cref{app_assumption-PP}. The post-selection de Finetti theorem states that there exists a measure $\mu$ on the space of single-copy density operators on $\St(\cH)$ such that for any permutation invariant state $\omega^{(n)} \in \St(\cH^{\otimes n})$, we have
\begin{align} \label{eq_it3-postselection-dF}
\omega^{(n)} \,\leq\, c_{n,d} \int_{\St(\cH)} \sigma^{\otimes n} \di\mu(\sigma)
\end{align}
with \smash{$c_{n,d} = \binom{n+d^2 -1}{n} \leq (n+1)^{d^2-1}$} for $d = \dim{\cH}$. As a consequence, we can bound the probability of any event $A$ with respect to $\omega^{(n)}$ by the probability of $A$ with respect to some iid state $ \sigma^{\otimes n}$. Formally, consider the operator $M_A := \sum_{\vec{x}_n \in A} T^n_{(x_1,\dots,x_n)}$ which satisfies $0\leq M_A \leq \id$. Then, one finds
\begin{align}
\Pr_{\omega^{(n)}}\big[\vec{X}_n \in A\big] \,&= \, \sum_{\vec{x}_n}\id\big[\vec{x}_n \in A\big] \Pr_{\omega^{(n)}}\big[ \vec{x}_n \big] = \sum_{\vec{x}_n \in A} \tr\big[ T^n_{(x_1,\dots,x_n)}\,\omega^{(n)} \big] \,=\, \tr\big[M_A\,\omega^{(n)} \big] \\
\,\overset{\textnormal{\Cshref{eq_it3-postselection-dF}}}&{\leq} c_{n,d} \int_{\St(\cH)} \tr\big[M_A\,\sigma^{\otimes n} \big]  \di\mu(\sigma) \,\leq\,  (n+1)^{d^2-1} \sup_{\sigma \in \St(\cH)}  \tr\big[M_A\,\sigma^{\otimes n} \big] \\
\,&=\, (n+1)^{d^2-1} \sup_{\sigma \in \St(\cH)} \Pr_{\sigma^{\otimes n}}\big[\vec{X}_n \in A\big]\,. \label{eq_it3-postselection-dF-p}
\end{align}

Let $\epsilon \in (0,1)$ and  ${\vec{x}}  \,\mapsto \,\lambda_{\vec{x}}  \, \mapsto \, \hat{\rho}(\lambda_{\vec{x}} )$ be the map that assigns a guessed state $\hat{\rho}$ to an observed data set $\vec{x}$ via its empirical frequency distribution $\lambda_{\vec{x}}$. Then, we can use~\cref{lem_lipschitz-continuity} to bound
\begin{align}
 \Pr_{\bar{\rho}^{(n)} }\Big[\tfrac{1}{2}\big\| \hat{\rho}(\vec{x}_n) - \rho \big\|_1 > \epsilon \Big] \,\overset{\textnormal{\Cshref{lem_lipschitz-continuity}}}&{\leq}\, \Pr_{\bar{\rho}^{(n)} }\Big[\tfrac{1}{2}\big\| \mu_{\hat{\rho}} - P_X \big\|_1 > \tfrac{\epsilon}{C} \Big] \,\\
 \overset{{\mathrm{triangle}\, \mathrm{ineq.}}}&{\leq}\, \Pr_{\bar{\rho}^{(n)} }\Big[\tfrac{1}{2}\big\| \mu_{\hat{\rho}} - \lambda_{\vec{x}_n} \big\|_1 + \tfrac{1}{2}\big\| \lambda_{\vec{x}_n} - P_X \big\|_1 > \tfrac{\epsilon}{C} \Big]  \\
 \, \overset{{\mathrm{union}\, \mathrm{bound}}}&{\leq} \, \Pr_{\bar{\rho}^{(n)} }\Big[\tfrac{1}{2}\big\| \mu_{\hat{\rho}} - \lambda_{\vec{x}_n} \big\|_1 > \tfrac{\epsilon}{2C} \Big] + \Pr_{\bar{\rho}^{(n)} }\Big[\tfrac{1}{2}\big\| \lambda_{\vec{x}_n} - P_X \big\|_1 > \tfrac{\epsilon}{2C} \Big] \\
 \,&=\,  \Pr_{\bar{\rho}^{(n)} }\Big[\big\| \mu_{\hat{\rho}} - \lambda_{\vec{x}_n} \big\|_1 > \tfrac{\epsilon}{C} \Big] +  \Pr_{\bar{\rho}^{(n)} }\Big[\vec{X}_n \in F^{n}_{\frac{\epsilon}{C}}  \Big]\,,  \label{eq_it3-intermediate-2}
 \end{align}
for $C\geq0$ and $\mu_{\hat{\rho}}$ defined as in~\cref{eq_inverse-freq-map--}. Since $\bar{\rho}^{(n)}$ is permutation invariant by definition, one finds for the first term in~\cref{eq_it3-intermediate-2} that
\begin{align}
 \Pr_{\bar{\rho}^{(n)} }\Big[\big\| \mu_{\hat{\rho}} - \lambda_{\vec{x}_n} \big\|_1 > \tfrac{\epsilon}{C} \Big] \overset{\textnormal{\Cshref{eq_it3-postselection-dF-p}}}&{\leq} (n+1)^{d^2-1} \sup_{\sigma \in \St(\cH)} \Pr_{\sigma^{\otimes n}} \Big[\big\| \mu_{\hat{\rho}} - \lambda_{\vec{x}_n} \big\|_1 > \tfrac{\epsilon}{C} \Big]  \\
 \overset{\textnormal{\Cshref{eq_inverted-born-rule--,lem_sample-subexp}}}&{\leq} (n+1)^{d^2-1} \tilde{K} \,2^{-\tilde{L} (\frac{\epsilon}{C})^{\tilde{q}} n} \,=: \,\tilde{\delta}_{n,\frac{\epsilon}{C}} \label{eq_it3-postselection-delta}
\end{align}
for some constants $\tilde{K},\tilde{L}>0,\,q\geq 0$. Note that the quantity $\tilde{\delta}_{n,\epsilon}$ is independent of $\bar{\rho}$ and $m_n$. The second term in~\cref{eq_it3-intermediate-2} was already bounded in~\cref{eq_it3-intermediate-1}, and we obtain
\begin{align}
 \Pr_{\bar{\rho}^{(n)} }\Big[\vec{X}_n \in F^{n}_{\frac{\epsilon}{C}}  \Big] \,\overset{\textnormal{\Cshref{eq_it3-intermediate-1}}}&{\leq}\,\Pr_{ {\rho}^{(N)} } \big[\vec{X}_N \in F^{N}_{\frac{\epsilon}{C} - \bar{\epsilon}}  \big]+  \delta_{n,\bar{\epsilon}} \\
 \overset{{\mathrm{triangle}\, \mathrm{ineq.}}}&{\leq}\,  \Pr_{{\rho}^{(N)} }\Big[\tfrac{1}{2}\big\|\lambda_{\vec{x}_N} -  \mu_{\hat{\rho}(\vec{x}_N)} \big\|_1 + \tfrac{1}{2}\big\|  \mu_{\hat{\rho}(\vec{x}_N)} - P_X \big\|_1 > \tfrac{\epsilon}{2C}-\tfrac{\bar{\epsilon}}{2} \Big] + \delta_{n,\bar{\epsilon}}  \\
  \overset{{\mathrm{union}\, \mathrm{bound}}}&{\leq}\,  \Pr_{{\rho}^{(N)} }\Big[\big\|\lambda_{\vec{x}_N} -  \mu_{\hat{\rho}(\vec{x}_N)} \big\|_1 > \tfrac{\epsilon}{2C}-\tfrac{\bar{\epsilon}}{2} \Big] +  \Pr_{{\rho}^{(N)} }\Big[\big\|  \mu_{\hat{\rho}(\vec{x}_N)} - P_X \big\|_1 > \tfrac{\epsilon}{2C}-\tfrac{\bar{\epsilon}}{2} \Big]+ \delta_{n,\bar{\epsilon}}\, \label{eq_it3-intermediate-3}
%    \overset{\textnormal{\Cshref{eq_it3-postselection-dF-p,eq_lipschitz-2}}}&{\leq} (N+1)^{d^2-1} \sup_{\sigma \in \St(\cH)} \Pr_{\sigma^{\otimes N}} \Big[\big\|\lambda_{\vec{x}_N} -  \mu_{\hat{\rho}(\vec{x}_N)} \big\|_1 > \tfrac{\epsilon}{C}-\bar{\epsilon}  \Big]  \nonumber \\
%&\hspace{5cm}  + \Pr_{{\rho}^{(N)} }\Big[\big\|  \hat{\rho}(\vec{x}_N) - \rho \big\|_1 > \tfrac{\epsilon}{C}-\bar{\epsilon} \Big] + \delta_{n,\bar{\epsilon}} \\
% \overset{\textnormal{\Cshref{eq_it3-postselection-delta}}}&{\leq} \,\tilde{\delta}_{N,\frac{\epsilon}{C}-\bar{\epsilon}} \,+\,  \Pr_{{\rho}^{(N)} }\Big[\big\|  \hat{\rho}(\vec{x}_N) - \rho \big\|_1 > \tfrac{\epsilon}{C}-\bar{\epsilon} \Big]
%\,+\, \delta_{n,\bar{\epsilon}} \,.  \label{eq_it3-intermediate-3}
%\end{align}
\end{align}
for any $\bar{\epsilon} \in (0,\frac{\epsilon}{C})$. Again due to the post-selection de Finetti result, the first term in~\cref{eq_it3-intermediate-3} can be bounded as
 \begin{align}
 \Pr_{{\rho}^{(N)} }\Big[\big\|\lambda_{\vec{x}_N} -  \mu_{\hat{\rho}(\vec{x}_N)} \big\|_1 > \tfrac{\epsilon}{2C}-\tfrac{\bar{\epsilon}}{2} \Big] \,\overset{\textnormal{\Cshref{eq_it3-PI}}}&{=}  \Pr_{\cS_N({\rho}^{(N)}) }\Big[\big\|\lambda_{\vec{x}_N} -  \mu_{\hat{\rho}(\vec{x}_N)} \big\|_1 >\tfrac{\epsilon}{2C}-\tfrac{\bar{\epsilon}}{2}  \Big]  \\
 \overset{\textnormal{\Cshref{eq_it3-postselection-dF-p}}}&{\leq}\, (N+1)^{d^2-1} \sup_{\sigma \in \St(\cH)} \Pr_{\sigma^{\otimes N}} \Big[\big\|\lambda_{\vec{x}_N} -  \mu_{\hat{\rho}(\vec{x}_N)} \big\|_1 > \tfrac{\epsilon}{2C}-\tfrac{\bar{\epsilon}}{2}  \Big]  \\
 \, \overset{\textnormal{\Cshref{eq_it3-postselection-delta}}}&{\leq} \,\tilde{\delta}_{N,\frac{\epsilon}{2C}-\frac{\bar{\epsilon}}{2}}\,,
 \end{align}
 while for the second term in~\cref{eq_it3-intermediate-3} we find
 \begin{align}  \label{eq_it3-intermediate-4}
 \Pr_{{\rho}^{(N)} }\Big[\big\|  \mu_{\hat{\rho}(\vec{x}_N)} - P_X \big\|_1 > \tfrac{\epsilon}{2C}-\tfrac{\bar{\epsilon}}{2} \Big] \,  \overset{\textnormal{\Cshref{lem_lipschitz-continuity}}}&{\leq} \,\Pr_{{\rho}^{(N)} }\Big[\big\|  \hat{\rho}(\vec{x}_N) - \rho \big\|_1 > \tfrac{\epsilon}{2C}-\tfrac{\bar{\epsilon}}{2} \Big] \,.
 \end{align}
By combining the steps~\eqref{eq_it3-intermediate-2}-\eqref{eq_it3-intermediate-4}, we finally find
\begin{align} \label{eq_it3-final}
\Pr_{\bar{\rho}^{(n)} }\Big[\tfrac{1}{2}\big\| \hat{\rho}(\vec{x}_n) - \rho \big\|_1 > \epsilon \Big]  \,\leq\,  \Pr_{{\rho}^{(N)} }\Big[\big\|  \hat{\rho}(\vec{x}_N) - \rho \big\|_1 > \tfrac{\epsilon}{2C}-\tfrac{\bar{\epsilon}}{2} \Big] \,+\, \tilde{\delta}_{n,\frac{\epsilon}{C}} \,+\,\tilde{\delta}_{N,\frac{\epsilon}{2C}-\frac{\bar{\epsilon}}{2}} \,+\,\delta_{n,\bar{\epsilon}} \,.
\end{align}
for any $\bar{\epsilon} \in (0,\frac{\epsilon}{C})$.

As a final step, we can argue with~\cref{eq_it3-final} that $\pmb{\rho} \in \OP(\texttt{tom($\cdot$)};\rho)$ implies $\bar{\pmb{\rho}} \in \OP(\texttt{tom($\cdot$)};\rho)$. By assumption, there exists a poly-log.~sample complexity $N_0(\epsilon, \delta)$ such that
\begin{align}
\forall\, N\geq N_0(\epsilon,\delta)\,:\  \Pr\big[\texttt{tom}\big(\rho^{(N)}\big) \notin B_\epsilon(\rho)\big] \leq \delta \quad \iff \quad  \Pr_{{\rho}^{(N)} }\Big[\tfrac{1}{2}\big\|  \hat{\rho}(\vec{x}_N) - \rho \big\|_1 > {\epsilon} \Big]  \leq \delta
\end{align}
for all $(\epsilon,\delta) \in (0,1)\times (0,\tfrac{1}{2}]$. Because of~\cref{lem_sample-subexp}, this implies that there exist constants $K,L>0,\,\mu \geq0$ and $\alpha \in (0,1]$ such that  
\begin{align} \label{eq_it3-assumption}
 \Pr_{{\rho}^{(N)} }\Big[\tfrac{1}{2}\big\|  \hat{\rho}(\vec{x}_N) - \rho \big\|_1 > {\epsilon} \Big]  \leq  K \,2^{-N^\alpha \epsilon^\mu L} =: \delta'_{N,\epsilon}
\end{align}
for all $\epsilon \in (0,1)$ and $N\in \N$. Hence, by choosing $\bar{\epsilon}:=\frac{\epsilon}{2C}$, we find
\begin{align}
\Pr_{\bar{\rho}^{(n)} }\Big[\tfrac{1}{2}\big\| \hat{\rho}(\vec{x}_n) - \rho \big\|_1 > \epsilon \Big] \,\overset{\textnormal{\Cshref{eq_it3-final,eq_it3-assumption}}}{\leq} \,  \delta'_{N,\frac{\epsilon}{8C}} \,+\, \tilde{\delta}_{n,\frac{\epsilon}{C}} \,+\,\tilde{\delta}_{N,\frac{\epsilon}{4C}} \,+\,\delta_{n,\frac{\epsilon}{2C}} 
\end{align}
for any $n \leq N$. Recall that $N = n+ m_n$, and note that there exists a function $N'_0 : (0,1) \to \mathbb{N} $ such that \smash{$N'_0(\epsilon) \leq C \big( \frac{1}{\epsilon}\big)^{q}$} with $C,q>0$, and such that for all $n\geq N'_0(\epsilon)$, the functions $\delta'_{N,\frac{\epsilon}{8C}}, \, \tilde{\delta}_{n,\frac{\epsilon}{C}}, \,\tilde{\delta}_{N,\frac{\epsilon}{4C}}, \,\delta_{n,\frac{\epsilon}{2C}} $ are subexponentially decreasing in $n$ due to~\cref{def_subexp-equivalence,lem_sums-of-subexp,ex_subexp-function-eps}. Hence, also the sum \smash{$h(n):=  \delta'_{N,\frac{\epsilon}{8C}} \,+\, \tilde{\delta}_{n,\frac{\epsilon}{C}} \,+\,\tilde{\delta}_{N,\frac{\epsilon}{4C}} \,+\,\delta_{n,\frac{\epsilon}{2C}} $} is a subexponentially decreasing function for all $n\geq N'_0(\epsilon)$ by~\cref{ex_sums-of-subexp-eps}. Furthermore, note that both $ \delta'_{N,\frac{\epsilon}{8C}} $ and $\tilde{\delta}_{N,\frac{\epsilon}{4C}}$ can be bounded by subexponentially decreasing functions that depend only on the numerical value $n$ instead of $N$ (while $\delta'_{N,\epsilon}$ may still depend on the original sequence $\pmb{\rho}$, which is allowed). Consequently, there exists a poly-log.~sample complexity \smash{$\tilde{N}_0(\epsilon,\delta)$} that is independent of the chosen sequence \smash{$(m_n)_{n\in \N}$}, and such that 
\begin{align}
 \forall\,(\epsilon,\delta) \in (0,1)\times (0,\tfrac{1}{2}] \ \ \forall\, n\geq \tilde{N}_0(\epsilon,\delta)\,:\  \Pr_{\bar{\rho}^{(n)} }\Big[\tfrac{1}{2}\big\| \hat{\rho}(\vec{x}_n) - \rho \big\|_1 > \epsilon \Big]\, \leq \, \delta\,, 
 \end{align}
because of~\cref{lem_sample-subexp}. This implies $\bar{\pmb{\rho}} \in \OP^{( \tilde{N}_0)}(\texttt{tom($\cdot$)};\rho)$, which proves~\cref{def_tom34-3}.

\subsection{Proof of~\cref{def_tom34-4}} \label{app_item4}
Let $(r_n)_{n\in\N}$ be any sequence of non-negative integers $r_n \in \N_0 $ such that $r_n \leq n$ and $r_n = o(n^\beta)$ for some $\beta \in (0,1)$.  Let $\rho \in \St(\cH)$, $\Omega^{(r_n)} \in \St(\cH_C^{\otimes r_n})$, and let $\bar{\pi}\in S_n$ be any fixed permutation. We define the state sequence  $\tilde{\pmb{\rho}}$ by 
\begin{align}
     \tilde{\rho}^{(n)}:= U_{\bar{\pi}}\big(\rho^{\otimes n-r_n} \otimes  \Omega^{(r_n)} \big) U_{\bar{\pi}}^\dagger\,,
\end{align}
and we want to show that $\tilde{\pmb{\rho}}  \in \OP^{\textrm{(eff.)}}(\texttt{tom($\cdot$)};\rho)$. Conceptually, the following proof corresponds to a rough adaptation of parts of the proof of~\cite[Theorem 4.5.2.]{renner_phd}. In the following, we write $r = r_n$ to simplify notation.

Let $P_X(x) := \tr[M_x \rho]$ denote the classical distribution of the outcomes of (the single-copy POVM) $\cM$ applied to $\rho$. First, for any $n\in\N$ and $\epsilon \in (0,1)$, we bound
\begin{align}
 \Pr_{\tilde{\rho}^{(n)} }\Big[\tfrac{1}{2}\big\| \hat{\rho}(\vec{x}_n) - \rho \big\|_1 > \epsilon \Big]\, \overset{\textnormal{\Cshref{lem_lipschitz-continuity}}}&{\leq}\, \Pr_{\tilde{\rho}^{(n)} }\Big[\tfrac{1}{2}\big\| \mu_{\hat{\rho}(\vec{x}_n)} - P_X \big\|_1 > \tfrac{\epsilon}{C} \Big] \, \label{eq_it4-steps-1}\\
  \overset{{\mathrm{triangle}\, \mathrm{ineq.}}}&{\leq}\, \Pr_{\tilde{\rho}^{(n)} }\Big[\tfrac{1}{2}\big\| \mu_{\hat{\rho}(\vec{x}_n)} - \lambda_{\vec{x}_n} \big\|_1 + \tfrac{1}{2}\big\| \lambda_{\vec{x}_n} - P_X \big\|_1 > \tfrac{\epsilon}{C} \Big]  \\
 \, \overset{{\mathrm{union}\, \mathrm{bound}}}&{\leq} \, \Pr_{\tilde{\rho}^{(n)} }\Big[\big\| \mu_{\hat{\rho}(\vec{x}_n)} - \lambda_{\vec{x}_n} \big\|_1 > \tfrac{\epsilon}{C} \Big] + \Pr_{\tilde{\rho}^{(n)} }\Big[\big\| \lambda_{\vec{x}_n} - P_X \big\|_1 > \tfrac{\epsilon}{C} \Big]  \label{eq_it4-intermediate-1}
 \end{align}
 As before, we use a post-selection de Finetti technique~\cite{Christandl_2009} to bound the first term in~\cref{eq_it4-intermediate-1} as
 \begin{align}
 \Pr_{\tilde{\rho}^{(n)} }\Big[\big\| \mu_{\hat{\rho}(\vec{x}_n)} - \lambda_{\vec{x}_n} \big\|_1 > \tfrac{\epsilon}{C} \Big] \,\overset{\textnormal{\Cshref{eq_it3-PI}}}&{=}  \Pr_{\cS_n(\tilde{\rho}^{(n)}) }\Big[\big\| \mu_{\hat{\rho}(\vec{x}_n)} - \lambda_{\vec{x}_n}  \big\|_1 > \tfrac{\epsilon}{C} \Big]  \\
 \overset{\textnormal{\Cshref{eq_it3-postselection-dF-p}}}&{\leq}\, (n+1)^{d^2-1} \sup_{\sigma \in \St(\cH)} \Pr_{\sigma^{\otimes n}} \Big[\big\| \mu_{\hat{\rho}(\vec{x}_n)} - \lambda_{\vec{x}_n} \big\|_1 > \tfrac{\epsilon}{C}  \Big]  \\
 \, \overset{\textnormal{\Cshref{eq_it3-postselection-delta}}}&{\leq} \,\tilde{\delta}_{n,\frac{\epsilon}{C}}\,,
 \end{align}
where $\cS_n$ denotes the symmetrization map. Note that the quantity $\tilde{\delta}_{n,\frac{\epsilon}{C}}$ is completely independent of $\tilde{\pmb{\rho}}$. The second term in~\cref{eq_it4-intermediate-1} can be bounded by exploiting the structure of $\tilde{\pmb{\rho}}$ and~\cref{def_tom12-2}. More precisely, due to the permutation invariance of the POVM elements \smash{$T^{(n)}_{\vec{x}_n}$}, we have
\begin{align}
\Pr_{\tilde{\rho}^{(n)} }\big[\vec{x}_n \big] \,&=\, \tr\big[ T^n_{(x_1,\dots,x_n)}\,\tilde{\rho}^{(n)}\big]  \,\overset{\textnormal{PI}}{=}\, \tr\big[ T^n_{(x_1,\dots,x_n)}\,\big(\rho^{\otimes n-r} \otimes  \Omega^{(r)}\big)\big] \\\,&=\,\tfrac{1}{n!}\sum_{\pi \in S_n} \tr\big[U_\pi ( M_{x_1} \otimes \cdots \otimes M_{x_n}) U^\dagger_\pi \big(\rho^{\otimes n-r} \otimes  \Omega^{(r)}\big)\big] \\
\,&=\, \frac{1}{\binom{n}{r}} \sum_{I \,:\, |I| = n-r} \tr\big[M_{\vec{x}_I}\rho^{\otimes n-r}  \big]\,\tr\big[\cS_r(M_{\vec{x}_{I^c}})  \Omega^{(r)} \big]\,, \label{eq_it4-intermediate-2}
\end{align}
where in the last step, we reexpressed the summation over all permutations as first summing over all choices of index subsets $I \subset \{ 1,\dots,n\}$ of length $|I|=n-r$ and then separately permuting in the respective intervals. Consequently, since
\begin{align} \label{eq_it4-triangle}
\big\|\lambda_{\vec{x}_n} - P_X \Big\|_1 \,\leq\, \big\|\lambda_{\vec{x}_n} - \lambda_{\vec{x}_I} \Big\|_1 + \big\|\lambda_{\vec{x}_I} - P_X \Big\|_1 \,\leq\, \frac{2r}{n} +  \big\|\lambda_{\vec{x}_I} - P_X \Big\|_1 
\end{align}
due to the triangle inequality, where $\vec{x}_I := (x_{i_1},\dots,x_{i_{n-r}})$ for any index subset $I = \{i_1,\dots,i_{n-r}\}$, we obtain
\begin{align}
\Pr_{\tilde{\rho}^{(n)} }\Big[\big\| \lambda_{\vec{x}_n}& - P_X \big\|_1 > \tfrac{\epsilon}{C} \Big] \,=\,\sum_{x_1,\dots,x_n} \id\Big[\big\|\lambda_{\vec{x}_n} - P_X \Big\|_1 > \tfrac{\epsilon}{C} \big] \, \Pr_{\tilde{\rho}^{(n)} }\big[\vec{x}_n \big] \\
\overset{\textnormal{\Cshref{eq_it4-intermediate-2}}}&{=}\,  \frac{1}{\binom{n}{r}} \sum_{I \,:\, |I| = n-r}\sum_{\vec{x}_n}  \id\Big[\big\|\lambda_{\vec{x}_n} - P_X \big\|_1 > \tfrac{\epsilon}{C} \Big] \, \tr\big[M_{\vec{x}_I}\rho^{\otimes n-r}  \big]\,\tr\big[\cS_r(M_{\vec{x}_{I^c}})  \Omega^{(r)} \big] \\
\overset{\textnormal{\Cshref{eq_it4-triangle}}}&{\leq}\,  \frac{1}{\binom{n}{r}} \sum_{I \,:\, |I| = n-r}\,\sum_{\vec{x}_n}  \id\Big[\big\|\lambda_{\vec{x}_{I}} - P_X \big\|_1 > \tfrac{\epsilon}{C} - \tfrac{2r}{n} \Big] \, \tr\big[M_{\vec{x}_I}\rho^{\otimes n-r}  \big]\,\tr\big[\cS_r(M_{\vec{x}_{I^c}})  \Omega^{(r)} \big] \\
\,&=\, \frac{1}{\binom{n}{r}} \sum_{I \,:\, |I| = n-r} \,\sum_{\vec{x}_I}  \id\Big[\big\|\lambda_{\vec{x}_{I}} - P_X \big\|_1 > \tfrac{\epsilon}{C} - \tfrac{2r}{n} \Big] \, \tr\big[M_{\vec{x}_I}\rho^{\otimes n-r}  \big] \sum_{\vec{x}_{I^c}}\tr\big[\cS_r(M_{\vec{x}_{I^c}})  \Omega^{(r)} \big] \\
\overset{\textnormal{POVM}}&{=}\, \frac{1}{\binom{n}{r}} \sum_{I \,:\, |I| = n-r} \,\sum_{\vec{x}_I}  \id\Big[\big\|\lambda_{\vec{x}_{I}} - P_X \big\|_1 > \tfrac{\epsilon}{C} - \tfrac{2r}{n} \Big] \, \tr\big[M_{\vec{x}_I}\rho^{\otimes n-r}  \big]  \\
\,&=\,\sum_{x_1,\dots,x_{n-r}}  \id\Big[\big\|\lambda_{\vec{x}_{n-r}} - P_X \big\|_1 > \tfrac{\epsilon}{C} - \tfrac{2r}{n} \Big] \, \tr\big[( M_{x_1} \otimes \cdots \otimes M_{x_{n-r}})\rho^{\otimes n-r}  \big] \\
\,&=\,\sum_{x_1,\dots,x_{n-r}}  \id\Big[\big\|\lambda_{\vec{x}_{n-r}} - P_X \big\|_1 > \tfrac{\epsilon}{C} - \tfrac{2r}{n} \Big] \, \tr\big[\cS_{n-r}( M_{x_1} \otimes \cdots \otimes M_{x_{n-r}})\rho^{\otimes n-r}  \big] \\
\,&=\, \Pr_{{\rho}^{\otimes n-r} }\Big[\big\|\lambda_{\vec{x}_{n-r}} - P_X \big\|_1 > \tfrac{\epsilon}{C} - \tfrac{2r}{n} \Big]\,.
\end{align}
%In the following, we consider $n$ large enough such that \smash{$\frac{\epsilon}{C} - \frac{2r}{n} >0$} which is possible since $r=o(n)$ by assumption. \GM{ATTENTION: Relationsship between $n$ and $\epsilon$. Still polynomial, but to be made precise!} 
Next, we relate the above probability to the corresponding event probability involving the guessed state,
\begin{align}
 \Pr_{{\rho}^{\otimes n-r} }\Big[\big\|\lambda_{\vec{x}_{n-r}} - &P_X \big\|_1 > \tfrac{\epsilon}{C} - \tfrac{2r}{n} \Big] \,\overset{\textnormal{triangle ineq.}}{\leq} \,  \Pr_{{\rho}^{\otimes n-r} }\Big[\big\|\lambda_{\vec{x}_{n-r}} - \mu_{\hat{\rho}} \big\|_1 + \big\| \mu_{\hat{\rho}}- P_X \big\|_1 > \tfrac{\epsilon}{C} - \tfrac{2r}{n} \Big]  \\
 \,\overset{\textnormal{union bound}}&{\leq} \,  \Pr_{{\rho}^{\otimes n-r} }\Big[\big\|\lambda_{\vec{x}_{n-r}} - \mu_{\hat{\rho}} \big\|_1  > \tfrac{\epsilon}{2C} - \tfrac{r}{n} \Big] +  \Pr_{{\rho}^{\otimes n-r} }\Big[  \big\| \mu_{\hat{\rho}}- P_X \big\|_1 > \tfrac{\epsilon}{2C} - \tfrac{r}{n} \Big]  \\
 \, \overset{\textnormal{\Cshref{eq_inverted-born-rule--,lem_sample-subexp,lem_lipschitz-continuity}}}&{\leq} \, K\,2^{-L (\frac{\epsilon}{2C} -\frac{r}{n})^{{q}} (n-r)} +  \Pr_{{\rho}^{\otimes n-r} }\Big[  \tfrac{1}{2}\big\| {\hat{\rho}}(\vec{x}_{n-r})- \rho \big\|_1 > \tfrac{\epsilon}{4C} - \tfrac{r}{2n} \Big]
\end{align}
for some constants $K,L,q>0$. Because of~\cref{def_tom12-2} and~\cref{lem_sample-subexp}, we know that
\begin{align} \label{eq_it4-steps--1}
 \Pr_{{\rho}^{\otimes n-r} }\Big[  \tfrac{1}{2}\big\| {\hat{\rho}}(\vec{x}_{n-r})- \rho \big\|_1 > \tfrac{\epsilon}{4C} - \tfrac{r}{2n} \Big] \,\leq\, \bar{K}\,2^{-\bar{L} (\frac{\epsilon}{4C} -\frac{r}{2n})^{\bar{q}} (n-r)}
\end{align}
for some constants $\bar{K},\bar{L},\bar{q}>0$. By combining all the steps~\eqref{eq_it4-steps-1}-\eqref{eq_it4-steps--1}, we finally obtain
\begin{align}
 \Pr_{\tilde{\rho}^{(n)} }\Big[\tfrac{1}{2}\big\| \hat{\rho}(\vec{x}_n) - \rho \big\|_1 > \epsilon \Big] \,\leq\,\tilde{\delta}_{n,\frac{\epsilon}{C}}+ K 2^{-{L} (\frac{\epsilon}{2C} -\frac{r}{n})^{{q}} (n-r)}  + \bar{K}\,2^{-\bar{L} (\frac{\epsilon}{4C} -\frac{r}{2n})^{\bar{q}} (n-r)} \,=:\, h(n)\,.
\end{align}
Note that while the quantity $h(n)$ depends on $r$, it is independent of $\Omega^{(r_n)}$ and $\bar{\pi}$. Since $r=o(n^\beta)$ with $\beta \in (0,1)$, we have $r \leq Ln^\beta$ for some $L>0$, and we can use~\cref{lem_sample-subexp-function-eps} to bound the exponential terms in $h(n)$. More precisely, we first choose $N_0(\epsilon) \in \N$ large enough such that
\begin{align}
\frac{\epsilon}{2C} -\frac{r}{n} \,\overset{!}{\geq}\, \frac{\epsilon}{4C} \quad \iff \quad  \frac{\epsilon}{4C} \,\overset{!}{\geq}\,  \frac{r}{n}\,, \qquad n-r  \,\overset{!}{\geq}\, \frac{n}{2}
\end{align}
for all $n \geq N_0(\epsilon)$, yielding a polynomial scaling of the form \smash{$N_0(\epsilon) \leq \tilde{C} \big(\frac{1}{\epsilon}\big)^{\frac{1}{1-\beta}}$} for some constant $\tilde{C}>0$ (and we can assume \smash{${q}\geq1,\,\bar{q}\geq 1$} without loss of generality). Combining this with~\cref{def_subexp-equivalence,lem_sums-of-subexp,ex_subexp-function-eps,ex_sums-of-subexp-eps} again shows that $h(n)$ is an exponentially decreasing function for all $n \geq N_0(\epsilon)$.
 %it follows by~\cref{lem_sample-subexp-function-eps} that the sublinear parts in the exponentials of $h(n)$ can be absorbed into the dominant exponential part. 
Hence, by~\cref{lem_sample-subexp} there exists an efficient sample complexity $\bar{N}_0(\epsilon,\delta)$ such that
\begin{align}
 \forall\,(\epsilon,\delta) \in (0,1)\times (0,\tfrac{1}{2}] \ \ \forall\, n\geq \bar{N}_0(\epsilon,\delta)\,:\   \Pr_{\tilde{\rho}^{(n)} }\Big[\tfrac{1}{2}\big\| \hat{\rho}(\vec{x}_n) - \rho \big\|_1 > \epsilon \Big] \, \leq \, \delta\,, 
 \end{align}
implying $\tilde{\pmb{\rho}}  \in \OP^{(\bar{N}_0)}(\texttt{tom($\cdot$)};\rho)$. Since the sample complexity $\bar{N}_0$ is further independent of $\Omega^{(r_n)}$ and $\bar{\pi}$, this proves~\cref{def_tom34-4}.

%%Consequently,

%\cite{cover}

%In a tomography protocol, the assignment $\vec{x} \mapsto \hat{\rho}(\vec{x})$ should of course yield a state $ \hat{\rho}$ that is consistent with the observed data $\vec{x}$. This means that for any (multi-copy) POVM $\{ M_{\vec{x}} \}_{\vec{x}\in \cX^n}$, it is required that for large enough number of samples $n$ it holds
%\begin{align}
% \big| \tr\big( M_{\vec{x}}  \hat{\rho}^{\otimes n} \big) -  \tr\big( M_{\vec{x}}  \rho^{\otimes n} \big) \big|
%\end{align}

% , in loose terms, correspond to inverting Born's rule. This means that for any POVM, .... This is exactly what~\cref{eq_tom-def1} is capturing. The hard thing to check is the scaling of the sample complexity. 

%%Now, we aim to show that~\cref{def_tom34-3,def_tom34-4} are satisfied for tomography protocols \texttt{tom($\cdot$)} that meet the requirements of~\cref{eq_product-povm-2,eq_data-processing-scheme,eq_lipschitz-continuity}.
%%
%%The simplest protocol to have in mind is the standard Pauli measurements protocol.
%%%%%%%%%%%%%%%%%%%%%%%%%%%%%%%%%%%%%%%%%%%%%%%%%%%%%%%%%%%%%%%%%%%%%%%
%%%%%%%%%%%%%%%%%%%%%%%%%%%%%%%%%%%%%%%%%%%%%%%%%%%%%%%%%%%%%%%%%%%%%%%

%------------------------------%------------------------------
%------------------------------%------------------------------ RATES
\newpage
\section{Operational state transformation rates} \label{sec:OP-rates}
%------------------------- RATES definition
Let $\rho \in \St(\cH_C)$ and $\sigma\in \St(\cH_D)$ be density operators. In this section, we consider the task of converting $n$ copies of $\rho$ into $m$ copies of $\sigma$ under a restricted set of operations $\F$ (e.g., the set of $\LOCC$ operations), 
\begin{align}
    \rho^{(n)} \overset{\F}{\longrightarrow} \sigma^{(m)}\,,
\end{align}
where the goal is to determine the optimal conversion rate, defined as $R := m/n$. Such conversion rates play a central role in the general framework of quantum resource theories, which encompasses, among others, entanglement theory. Applications of quantum resource theories range from quantum communication and computation to quantum foundations; see~\cite{Gour_2025} for a modern review. In the following, we adopt the framework introduced in~\cite{Gour_2025} to define the set of free operations $\F$.
\begin{definition-2}[Quantum Resource Theories~{\cite[Chapter 9.1]{Gour_2025}}] \label{def_QRTs}
Let $\F$ be a mapping that assigns to any two input/output physical systems $C$ and $D$, with corresponding Hilbert spaces $\cH_C$ and $\cH_D$, a unique set of CPTP operations denoted by $\F\big(C \to D \big) \subseteq \textrm{CPTP}\big(C \to D \big)$. The mapping $\F$ is called a \emph{quantum resource theory} (QRT) if it satisfies the following conditions:
\begin{enumerate}
\item \emph{Doing nothing is free.} For any physical system $C$, the identity channel $\mathrm{id}_C\in \F\big(C \to C \big)$.
\item \emph{Concatenation is free.} For any three physical systems, $A$, $B$, and $C$, if $\cE \in  \F\big(A \to B \big)$ and $\cN \in \F\big(B \to C \big)$ then $\cN \circ \cE \in  \F\big(A \to C \big)$.
\item \emph{Discarding a system is free.} For any system $C$, it holds $ \F\big(C \to 1 \big) = \textrm{CPTP} \big(C \to 1 \big) = \big\{\tr\big\}$.
\end{enumerate}
In addition, $\F$ is said to \emph{admit a tensor-product structure} if it fulfils these additional criteria:
\begin{enumerate}
\setcounter{enumi}{3}
\item \emph{Completely free operations:} For any three systems $A$, $B$, and $C$, and a channel $\cE \in \F\big(A \to B \big)$, it holds that $\cE^{A \to B} \otimes \mathrm{id}_C \in  \F\big(AC \to BC \big)$.
\item \emph{Freedom of relabeling:} For any integer $n$, a free channel $\cN \in  \F\big(C^n \to D^n \big)$, and permutation channel $\cP_\pi$ corresponding to a permutation $\pi$ on $n$ elements, the composition $\cP_\pi^{D^n} \circ \cN \circ \cP_{\pi^{-1}}^{C^n} \in \F\big(C^n \to D^n \big)$.
\end{enumerate}
The above requirements are standard for many QRTs~\cite{Gour_2025}. For our work, we require two (less standard but broadly used) additional properties:
\begin{enumerate}
\item[$5^\prime$.] For any system $C$ and any integer $n$, the permutation channel $\cP_\pi$ corresponding to a permutation $\pi$ on $n$ elements is a free operation, $\cP_\pi \in  \F\big(C^n \to C^n \big)$. Note that conditions $5'$ and $2$ imply condition $5$.
\item[6.] For any two physical systems $C$ and $D$, the set $\F\big(C \to D \big)$ is convex.
\end{enumerate}
The set $ \F\big(C \to D \big)$ is called \emph{the set of free operations} from system $C$ to system $D$, and the set $\F(C) :=  \F\big(1 \to C \big)$ is identified as \emph{the set of free states}.
\end{definition-2}
The conditions listed in~\cref{def_QRTs} have a few immediate implications~\cite{Gour_2025}: \emph{(i)} If a free channel $\cE \in  \F\big(C \to D \big)$ acts on a free state $\rho \in \cF(C)$, the resulting state $\cE(\rho)$ is a free quantum state in $\F(D)$ (this property is also referred to as \emph{the Golden Rule of QRTs}). \emph{(ii)} The preparation of and conversion to a free state is a free operation. \emph{(iii)} The partial trace is a free operation. \emph{(iv)} The tensor product of two free states or operations is again a free state or operation, respectively.  \emph{(v)} Appending a free state is a free channel. \emph{(vi)} For any integer $n$, the symmetrization map $\cS_n$, defined in~\cref{eq_symmetrization-map}, is a free operation. This follows from conditions $5'$ and $6$. Finally, we note that the set of $\LOCC$ operations satisfies all conditions listed in~\cref{def_QRTs}.\footnote{In that case, we have $C = A\otimes B$ and $U_{\pi,C^n} =U_{\pi,A^n} \otimes U_{\pi,B^n} $.}

\begin{notation}[Free operations with rate $r$] \label{not_free-ops-r}
Let $r > 0$. For a sequence of free operations $\big( \cP^{(r)}_n\big)_{n\in\N}$ with $\cP^{(r)}_n \in \F \big(C^n \to D^{m_n}\big)$, where $(m_n)_{n\in \N}$ is a %monotonically increasing 
sequence of integers such that $\lim_{n\to \infty}\frac{m_n}{n} = r$, we write
\begin{align}
{\pmb{\mathcal{P}}}^{(r)} := \big( \cP^{(r)}_n\big)_{n\in\N} \in \F \,.
 \end{align}
\end{notation}
\noindent For instance, note that traditional (iid) asymptotic transformation rates (such as entanglement cost or distillation, e.g.~see~\cite{chit19,Gour_2025}) may be expressed in our notation as
\begin{align} \label{iid_rate}
R_\F(\rho \to \sigma) \,=\, \sup\Big\{\, r>0\ \big| \ \exists\ {\pmb{\mathcal{P}}}^{(r)} \in \F \,:\ {\pmb{\mathcal{P}}}^{(r)}\big( \pmb{\rho}^{\mathrm{iid}}\big) \, \simeq \,\pmb{\sigma}^{\mathrm{iid}} \Big\} \,,
\end{align}
where we have defined the iid state sequences \smash{$\pmb{\rho}^{\mathrm{iid}} :=\big(\rho^{\otimes n}\big)_{n\in \N} $ and  $\pmb{\sigma}^{\mathrm{iid}} :=\big(\sigma^{\otimes n}\big)_{n\in \N} \,$}. This is discussed in more detail in~\cref{sec:rel-traditional-rates,sec:LOCC_rates}.

Now we can use~\cref{def_OP-set,def_tom12,def_tom34,def_QRTs} to introduce an operational asymptotic transformation rate.
\begin{definition}[Operational asymptotic transformation rate]\label{def_OP-rate}
Let \texttt{tom($\cdot$)} denote a tomography protocol satisfying~\cref{def_tom12,def_tom34}. For any density operators $\rho \in \St(\cH_C)$ and $\sigma\in \St(\cH_D)$ and any fixed set of free operations $\F$ satisfying~\cref{def_QRTs}, we define \emph{the operational asymptotic transformation rate} by
\begin{align} \label{def:eq-OP_rate}
R^{ \texttt{tom}}_\F(\rho \to \sigma) \,:= \, \sup\Big\{\, r >0\ \big| \ \exists\ {\pmb{\mathcal{P}}}^{(r)} \in \F \,:\ {\pmb{\mathcal{P}}}^{(r)}\big(\OP(\texttt{tom($\cdot$)};\rho) \big) \,\subseteq \,\OP(\texttt{tom($\cdot$)};\sigma) \Big\}\,.
%    R^{ \texttt{tom}}_\F(\rho \to \sigma) \,:=& \, \sup\Big\{\, r\ \big| \ \exists\,\big( \cP_n\big)_{n\in\N}\,,\, \cP_n \in \F \big(C^n \to D^{\lceil rn\rceil}\big) \ \ \forall\, \big(\rho^{(n)}\big)_{n\in \N} \in \OP(\texttt{tom($\cdot$)};\rho)\,:\ \\
%    & \phantom{, \sup\Big\{\, r\ \big| \ \exists\,\big( \cP_n\big)_{n\in\N}\,,\, \cP_n \in \F}  \big(\cP_n\big(\rho^{(n)}\big)\big)_{n\in\N} \equiv \big(\sigma^{(\lceil rn\rceil)}\big)_{n\in\N} \in \OP(\texttt{tom($\cdot$)};\sigma) \Big \} \\
%    =& \, \sup\Big\{\, r\ \big| \ \exists\ {\pmb{\mathcal{P}}}^{(r)} \in \F \,:\ {\pmb{\mathcal{P}}}^{(r)}\big(\OP(\texttt{tom($\cdot$)};\rho) \big) \,\subseteq \,\OP(\texttt{tom($\cdot$)};\sigma) \Big\}\,.
\end{align}
If the set in~\cref{def:eq-OP_rate} is empty, we set $ R^{ \texttt{tom}}_\F(\rho \to \sigma)= 0$.
\end{definition}
\begin{remark} \label{rmk_op-rate-convention}
Operational asymptotic transformation rates, introduced in~\cref{def_OP-rate}, are independent of the convention adopted in~\cref{not_channel-seq} for selecting the output state sequence.
\end{remark}
\begin{proof}
Let  $\rho \in \St(\cH_C)$, $\sigma\in \St(\cH_D)$, and let $r>0$ be achievable for $R^{ \texttt{tom}}_\F(\rho \to \sigma)$, i.e., there exists ${\pmb{\mathcal{P}}}^{(r)} \in \F$ such that, for any $ \pmb{\rho} \in \OP(\texttt{tom($\cdot$)};\rho) $, it holds that $ {\pmb{\mathcal{P}}}^{(r)}(\pmb{\rho}) = \pmb{\sigma} \in  \OP(\texttt{tom($\cdot$)};\sigma)$. Moreover, for any $n\in \N$ we have $\cP^{(r)}_n \in \F(C^n \to D^{m_n})$ for some %monotonically increasing 
sequence of integers $(m_n)_{n\in\N}$ such that $\lim_{n\to\infty}\frac{m_n}{n} = r$. According to~\cref{not_channel-seq}, the output sequence $\pmb{\sigma}$ is defined via $\sigma^{(m)} := \tr_{m_{j(m)}- m}\big[\cP^{(r)}_{j(m)}(\rho^{(j(m))})\big]$, where the (minimal) input index representative $j(m)$ is specified in~\cref{eq_channel-selector}.
% $\bar{n}_m = \min\{n\in \N \,:\, m_n = m \}$ for the image values $m\in \{m_n \,:\, n\in \N \}$. 

Now, let $k(m) \in \{ n\in \N \,:\,m_n = m\}$ be an arbitrary representative for each $m \in \{m_n \,:\, n\in \N \}$. For the gaps $m\notin \{m_n \,:\, n\in \N \}$, let the convention still be fixed by $k(m):=\min\{n\in N: m_n\geq m\}$. A potential issue is that, whenever $k(m) \neq j(m)$, the output \smash{$ \cP^{(r)}_{k(m)}(\rho^{(k(m))})$} could, a priori, be an arbitrary state with no apparent relation to an operational state sequence. Nevertheless, one can always construct a new sequence of free operations $\tilde{\pmb{\cP}}\vphantom{\pmb{\cP}}^{(r)}$ with rate $r$ whose output sequence remains operational.
%such that $\tilde{\pmb{\sigma}} := \big(\tilde{\cP}^{(r)}_{k(m)}(\rho^{(k(m))}) \big)_{m\in\N} \in \OP(\texttt{tom($\cdot$)};\sigma)$. 

\medskip
\noindent 
To do so, let $k(m)$ denote any such convention and let $\pmb{\cP}^{(r)}$ be a successful free operation with rate $r>0$ under that convention. For any $m\in \N$, define the auxiliary free channel $H_m : C^{k(m)}\to D^m$ by $H_m := tr_{m_{k(m)}-m}\circ \cP_{k(m)}$. Then, the output sequence $\pmb{\sigma}$ is defined as $\sigma^{(m)} =H_m\big(\rho^{(k(m))} \big) $. Next, define a new output size
\begin{align}
q_n \,:=\,\max\{m\in\mathbb{N}:k(m)\le n\}\,,
\end{align}
which corresponds to the largest output obtainable from the originally selected channel that requires at most $n$ inputs. For sufficiently large $n\in\N$, the above set is nonempty. Accordingly, define the modified channel sequence $\tilde{\pmb{\cP}}\vphantom{\pmb{\cP}}^{(r)}$ by
\begin{align}
\tilde{{\cP}}\vphantom{{\cP}}^{(r)}_{n} \,:=\, \cS_{q_n}\circ H_{q_n}
\circ\tr_{n-k(q_n)}\circ\cS_n\,,
\end{align}
where $\pmb{\cS}$ is the symmetrization map. For the finitely many values $n\in \N$ where $q_n$ is not defined, the channel $\tilde{{\cP}}\vphantom{{\cP}}^{(r)}_{n}$ may be defined arbitrarily. Note that $q_n$ is nondecreasing and satisfies $\lim_{n\to\infty} \frac{q_n}{n} = r$. The latter follows since $\lim_{m\to\infty}\frac{k(m)}{m} = \frac{1}{r}$ combined with the relations $k(q_n) \leq n < k(q_n+1)$ and $\lim_{n\to\infty} q_n = \infty$. Hence, $\tilde{\pmb{\cP}}\vphantom{\pmb{\cP}}^{(r)}$ defines again a free operation with rate $r$. 

We next show that this modified protocol is successful under any other representative convention. Let $J:=\{q_n:n\in\N\}$ denote the set of attained output sizes of the modified protocol, and let $h(m)\in\{n\in\N:q_n=m\}$ represent the new convention for each $m\in J$. First, note that an attained size $m\in J$ occurs first at input size $k(m)$. Indeed, if $q_n=m$, then $k(m)\leq n$ by definition of $q_n$, and moreover, $q_{k(m)} = m$. Thus, $k(m)\leq h(m)$ for all $m\in J$. In addition, $m\mapsto k(m)$ defines an injective map on $J$. Consequently, the state sequence $\pmb{\omega}$, defined by
\begin{align}
\omega^{(\ell)} &\,:=\,\tr_{s_\ell}\Big[ \cS_{\ell+s_\ell}\big(\rho^{(\ell+s_\ell)}\big)\Big], \qquad s_\ell \, := \,\begin{cases} h(m)-k(m), & \ell=k(m)\textrm{ for }m\in J,\\ 0, & \textrm{else}, \end{cases}
\end{align}
satisfies $\pmb{\omega}\in\OP(\texttt{tom($\cdot$)};\rho)$ due to~\cref{def_tom34-3} in~\cref{def_tom34}. Hence, we have $\pmb{\eta} := \pmb{\cP}(\pmb{\omega}) \in \OP(\texttt{tom($\cdot$)};\sigma) $ and, since $q_{h(m)} = m $ by definition, we find
\begin{align}
\tilde{\cP}^{(r)}_{h(m)} \big(\rho^{(h(m))}\big) =\cS_m\big(\eta^{(m)}\big)
\end{align}
for all $m\in J$. For the gap values $m\notin J$, a similar argument applies. Overall, this implies that $\tilde{\pmb{\cP}}\vphantom{\pmb{\cP}}^{(r)}(\pmb{\rho}) = \tilde{\pmb{\sigma}} \in  \OP(\texttt{tom($\cdot$)};\sigma)$. Since the convention choices were arbitrary, this proves the claim.

\end{proof}

\subsection{Properties of operational asymptotic transformation rates}
The above definition appears rather involved and seems to depend on the choice of the tomography protocol \texttt{tom($\cdot$)}. We now employ the techniques developed in~\cref{sec:characterization-OP-set}, in particular~\cref{thm_almost-iid-dF}, to simplify the structure of operational state transformation rates. In particular, we show that they are, in fact, independent of the choice of \texttt{tom($\cdot$)}. For the following, recall~\cref{not_channel-seq,not_trace-ops,not_subpoly-almost-iid,def_eq-rel,not_negligible-ball,not_free-ops-r}.
\begin{corollary} \label{thrm_rate-tom-indep}
%[$ R^{ \texttt{tom}}_\F$ is independent of  \texttt{tom}] 
Let \texttt{tom($\cdot$)} denote any tomography protocol satisfying~\cref{def_tom12,def_tom34}. For any density operators $\rho \in \St(\cH_C)$ and $\sigma\in \St(\cH_D)$ and any fixed set of free operations $\F$, the operational asymptotic transformation rate $ R^{ \texttt{tom}}_\F(\rho \to \sigma)$, defined in~\cref{def_OP-rate}, satisfies
\begin{align}
R^{ \texttt{tom}}_\F(\rho \to \sigma) \,=\,  \sup\Big\{\, r >0\ \big| \ \exists\ {\pmb{\mathcal{P}}}^{(r)} \in \F \,:\ {\pmb{\mathcal{P}}}^{(r)}\big(\mathcal{A}(\rho)\big) \,\subseteq \, B\big(\mathcal{A}(\sigma)\big) \Big\} \,=:\, \tilde{R}_\F(\rho \to \sigma)\,.
\end{align}
In particular, the operational asymptotic transformation rate is independent of the chosen tomography protocol.
\end{corollary}

\begin{proof}
Let $\rho \in \St(\cH_C)$ and $\sigma\in \St(\cH_D)$ be fixed, and let
\begin{align}
 \tilde{R}_\F(\rho \to \sigma) \,:=\, \sup\Big\{\, r >0 \ \big| \ \exists\ {\pmb{\mathcal{P}}}^{(r)} \in \F \,:\ {\pmb{\mathcal{P}}}^{(r)}\big(\mathcal{A}(\rho)\big) \,\subseteq \, B\big(\mathcal{A}(\sigma)\big) \Big\} \,,
\end{align}
where $\mathcal{A}(\cdot)$ denotes the set of power-law sublinear almost-iid state sequences as defined in~\cref{not_subpoly-almost-iid}. By convention, we set $ \tilde{R}_\F(\rho \to \sigma) =0$ if the above set is empty. We further introduce the notation $\OP_{(\mathrm{sym})}(\texttt{tom($\cdot$)};\rho)$ to denote the set of permutation invariant operational state sequences.

\medskip
\noindent  \underline{$R^{ \texttt{tom}}_\F  \leq \tilde{R}_\F$:} If $R^{ \texttt{tom}}_\F(\rho \to \sigma)=0$, the claim trivially holds. Let now $r > 0$ be achievable for $R^{ \texttt{tom}}_\F(\rho \to \sigma)$, i.e., there exists a sequence of free operations ${\pmb{\mathcal{P}}}^{(r)} \in \F$ such that 
\begin{align} \label{thm_rates1-it1}
{\pmb{\mathcal{P}}}^{(r)}\big(\OP(\texttt{tom($\cdot$)};\rho) \big) \,\subseteq \,\OP(\texttt{tom($\cdot$)};\sigma)\,.
\end{align}
Now, let $\pmb{\rho} \in \cA(\rho)$. Due to~\cref{prop_almost-iid-in-OP}, it holds that $\pmb{\rho}\in \OP(\texttt{tom($\cdot$)};\rho)$. Hence,~\cref{thm_rates1-it1} implies that ${\pmb{\mathcal{P}}}^{(r)}(\pmb{\rho} ) \in \OP(\texttt{tom($\cdot$)};\sigma)$. Next, for $\alpha \in (0,1)$ and $\pmb{\cS} = \big(\cS_n\big)_{n\in \N}$ denoting the symmetrization map, we define a new sequence of free operations by
\begin{align} 
\pmb{\cQ}^{(r)}\,:=\, \pmb{\tr}_\alpha \circ \pmb{\cS} \circ {\pmb{\mathcal{P}}}^{(r)} \,,
\end{align}
which is a free operation since it corresponds to a concatenation of free operations. More precisely, we have \smash{$\pmb{\cQ}^{(r)} = \big(\cQ^{(r)}_n\big)_{n\in\N}$} with \smash{$\cQ^{(r)}_n \in \F\big(C^n \to D^{\tilde{m}_n} \big)$} and $\tilde{m}_n := m_n - \lceil m_n^\alpha \rceil$, where $m_n$ denotes the number of output copies of \smash{$\cP^{(r)}_n$}. As $(\tilde{m}_n)_{n\in \N}$ %is monotonically increasing and 
satisfies $\lim_{n\to\infty} \frac{\tilde{m}_n}{n} = r$, the operation \smash{$ \pmb{\cQ}^{(r)} \in \F$} is a free operation with rate $r$. Applying this operation to the state sequence $\pmb{\rho}$, we find
\begin{align}
\pmb{\cQ}^{(r)}(\pmb{\rho}) \,=\, \pmb{\tr}_\alpha ( \pmb{\sigma} ), \quad \,  \pmb{\sigma}  \,\in\, \OP_{(\mathrm{sym})}(\texttt{tom($\cdot$)};\sigma) ,
%\pmb{\cQ}(\pmb{\rho}) \,\overset{\textnormal{\Cshref{thm_rates1-it1}}}{\in}\,  \pmb{\tr}_\alpha \circ \pmb{\cS} \big( \OP(\texttt{tom($\cdot$)};\rho) \big) \,\overset{\textnormal{\Cshref{..}}}{\subseteq} \, \pmb{\tr}_\alpha\big
\end{align}
where the above inclusion follows by the asymptotic permutation invariance of the tomography protocol,~\cref{def_tom-perm-inv}, combined with~\cref{thm_rates1-it1} and~\cref{rmk_op-rate-convention}. Then,~\cref{thm_almost-iid-dF} implies that there exists an almost-iid state sequence $ \pmb{\tau} \,\in\, \mathcal{A}(\sigma)$ such that 
\begin{align}  \label{thm_rates1-it2}
\pmb{\cQ}^{(r)}(\pmb{\rho}) \,=\, \pmb{\tr}_\alpha ( \pmb{\sigma} ) \simeq \pmb{\tau}\,.
\end{align} 
Since $\pmb{\rho} \in \cA(\rho)$ was chosen arbitrarily, this means that $r $ is achievable for $\tilde{R}_\F(\rho \to \sigma)$ and hence, $ r  \leq  \tilde{R}_\F(\rho \to \sigma)$. Consequently, after taking the supremum over all achievable rates $r$, we conclude that $R^{ \texttt{tom}}_\F(\rho \to \sigma)  \leq \tilde{R}_\F(\rho \to \sigma)$.

\medskip
\noindent \underline{$R^{ \texttt{tom}}_\F  \geq \tilde{R}_\F$:} If $ \tilde{R}_\F(\rho \to \sigma) = 0$, the claim trivially holds. Let now $\tilde{r} > 0$ be achievable for $ \tilde{R}_\F(\rho \to \sigma)$, i.e., there exists a sequence of free operations ${\pmb{\mathcal{P}}}^{(\tilde{r})} \in \F$ such that 
\begin{align} \label{thm_rates2-it1}
{\pmb{\mathcal{P}}}^{(\tilde{r})}\big(\mathcal{A}(\rho)\big) \,\subseteq \, B\big(\mathcal{A}(\sigma)\big)\,.
\end{align}
Now, let $\pmb{\rho}\in \OP(\texttt{tom($\cdot$)};\rho)$. For any $\alpha \in (0,1)$ and $\pmb{\cS} = \big(\cS_n\big)_{n\in \N}$ denoting the symmetrization map, we define a new sequence of free operations by
\begin{align} 
\pmb{\cQ}^{(\tilde{r})} \,:=\,{\pmb{\mathcal{P}}}^{(\tilde{r})}  \circ  \pmb{\tr}_\alpha \circ \pmb{\cS} \,,
\end{align}
where $\pmb{\cQ}^{(\tilde{r})} = \big(\cQ^{(\tilde{r})}_n\big)_{n\in\N}$ with $\cQ^{(\tilde{r})}_n \in \F\big(C^n \to D^{\tilde{m}_n} \big)$ and $\tilde{m}_n := m_{n- \lceil n^\alpha \rceil}$, where $m_n$ denotes the number of output copies of \smash{$\cP^{(\tilde{r})}_n$}. Note that as %$(\tilde{m}_n)_{n\in \N}$ is monotonically increasing and 
$\lim_{n\to\infty} \frac{\tilde{m}_n}{n} = r$, the operation $ \pmb{\cQ}^{(\tilde{r})} \in \F$ is a free operation with rate $r$. 

Due to the asymptotic permutation invariance of the tomography protocol,~\cref{def_tom-perm-inv}, and~\cref{thm_almost-iid-dF}, there exists an almost-iid state sequence $\pmb{\tau} \in \cA(\rho)$ such that
\begin{align}  \label{thm_rates2-it2}
 \pmb{\tr}_\alpha \circ \pmb{\cS}(\pmb{\rho}) \,\simeq \, \pmb{\tau}\,.
\end{align}
Moreover, recall that the relation in~\cref{thm_rates2-it2} is independent of the convention for selecting the output state sequence.
Consequently, since $\pmb{\mathcal{P}}^{(\tilde{r})}$ corresponds to a sequence of CPTP maps, it follows by the contractivity of the trace norm, and since $\tilde{m}_n = \Theta(n)$, that
\begin{align}
\pmb{\cQ}^{(\tilde{r})} (\pmb{\rho}) \,\overset{\textnormal{\Cshref{thm_rates2-it2}}}{\simeq} \, \pmb{\mathcal{P}}^{(\tilde{r})} (\pmb{\tau}) \,\overset{\textnormal{\Cshref{thm_rates2-it1}}}{\subseteq} \, B\big(\cA(\sigma)\big) \,.
\end{align}
By the definition of the negligible ball~\cref{not_negligible-ball} and~\cref{prop_eq-rel}, this means that there exists an almost-iid state sequence $\tilde{\pmb{\tau}} \in \cA(\sigma)$ such that $\pmb{\cQ}^{(\tilde{r})} (\pmb{\rho}) \simeq \tilde{\pmb{\tau}}$. Then, because of~\cref{lem_OP-set-robust,prop_almost-iid-in-OP}, it follows that
\begin{align} \label{thm_rates2-it3}
\pmb{\cQ}^{(\tilde{r})} (\pmb{\rho}) \,\in \, \OP(\texttt{tom($\cdot$)};\sigma)\,.
\end{align}
Since $\pmb{\rho}\in \OP(\texttt{tom($\cdot$)};\rho)$ was arbitrary, this implies that $\tilde{r}$ is achievable for $R^{ \texttt{tom}}_\F(\rho \to \sigma)$ and hence, $ \tilde{r} \leq  R^{ \texttt{tom}}_\F(\rho \to \sigma)$. Consequently, after taking the supremum over all achievable rates $\tilde{r}$, we conclude that $\tilde{R}_\F(\rho \to \sigma) \leq R^{ \texttt{tom}}_\F(\rho \to \sigma)$, and the claim follows.

\end{proof}

In summary,~\cref{thrm_rate-tom-indep} shows that operational state transformation rates are fully characterized in terms of (power-law sublinear) almost-iid state sequences. Moreover,~\cref{thrm_rate-tom-indep} substantially simplifies the analysis of operational state transformation rates and allows us to compare it to other asymptotic transformation rates that are widely studied in the literature~\cite{chit19}. In particular, we can show that these operational rates satisfy the necessary transitivity property for a reasonable interconversion rate, as described below.
\begin{proposition}[Composability] \label{prop_rates-transitivity}
For any density operators $\rho \in \St(\cH_C)$, $\sigma\in \St(\cH_D)$, $\tau \in \St(\cH_E)$, and any fixed set of free operations $\F$ satisfying~\cref{def_QRTs}, it holds that
\begin{align}
 \tilde{R}_\F(\rho \to \sigma) \,\geq\, \tilde{R}_\F(\rho \to \tau)\tilde{R}_\F(\tau \to \sigma) \,.
\end{align}
\end{proposition}
\begin{proof}
For $\tilde{R}_\F(\rho \to \tau) = 0$ or $\tilde{R}_\F(\tau \to \sigma)= 0$, the claim trivially follows. Let now $r_1,\,r_2 > 0$ be achievable for $\tilde{R}_\F(\rho \to \tau)$ and $\tilde{R}_\F(\tau \to \sigma)$, respectively. This means that there exist sequences of free operations $\pmb{\cP}^{(r_1)}, \, \pmb{\cQ}^{(r_2)} \in \F$ such that 
\begin{align} \label{prop_rates-transitivity-it1}
\pmb{\cP}^{(r_1)}\big(\mathcal{A}(\rho)\big) \,\subseteq \, B\big(\mathcal{A}(\tau)\big)\,, \qquad \pmb{\cQ}^{(r_2)}\big(\mathcal{A}(\tau)\big) \,\subseteq \, B\big(\mathcal{A}(\sigma)\big)\,.
\end{align}
Now consider a new sequence of free operations $\pmb{\cR}^{(r_2r_1)} \in \F$ defined by the composition
\begin{align}
\pmb{\cR}^{(r_2r_1)}\,=\,\pmb{\cQ}^{(r_2)}  \circ \pmb{\cP}^{(r_1)}\,,
\end{align}
with $\cP^{(r_1)}_n \in \F(C^n\to E^{m_n} )$ and $\cQ^{(r_2)}_n \in \F(E^n \to D^{\ell_n})$, leading to $\cR^{(r_2r_1)}_n \in \F(C^n\to D^{s_n} )$ where $s_n:= \ell_{m_n}$. Note that the sequence $(s_n)_{n\in\N}$ %is clearly monotonically increasing, and one finds 
satisfies $\lim_{n\to \infty} \frac{s_n}{n} = r_2 r_1$. To see this, note that since $\lim_{n\to \infty} \frac{m_n}{n} = r_1 >0$, we necessarily have $\lim_{n\to \infty} m_n = \infty$, and hence, we may write
\begin{align}  \label{prop_rates-transitivity-it2}
\lim_{n\to \infty} \frac{s_n}{n} \,=\,\lim_{n\to \infty} \left(\frac{\ell_{m_n}}{m_n}  \frac{m_n}{n} \right) \, = \,\left( \lim_{m\to \infty} \frac{\ell_{m}}{m} \right)\left( \lim_{n\to \infty}  \frac{m_n}{n} \right) \,=\,r_2 r_1\,.
\end{align}
%In the case where $r_1 = 0$, we can either have $\lim_{n\to \infty} m_n = \infty$ and~\cref{prop_rates-transitivity-it2} still holds, or it can happen that $\sup_{n\in\N} m_n = C $ for some constant $C \geq 0$. In that case, $m_n = C$ for large enough $n$, and we have $\lim_{n\to\N}  \frac{s_n}{n} = \lim_{n\to\N}  \frac{\ell_C}{n} = 0 = r_1r_2$. 
Hence, $\pmb{\cR}^{(r_2r_1)} \in \F$ is a free operation with rate $r_2 r_1$.
%\begin{align}
%\pmb{\cR}^{(r_2r_1)}\,=\, \big(\cR^{(r_2r_1)}_n\big)_{n\in\N} \,,\quad  \cR^{(r_2 r_1)}_n \,:=\,\tr_{k_n} \circ \cQ^{(r_2)}_{\lceil  r_1 n \rceil} \circ \cP^{(r_1)}_n
%\end{align}
%where $k_n := \lceil r_2 \lceil r_1 n \rceil\rceil - \lceil r_2r_1 n \rceil \geq 0$. 
On the other hand, as $\ell_n = \Theta(n)$, we have
\begin{align}
\pmb{\cQ}^{(r_2)} \circ \pmb{\cP}^{(r_1)}\big(\cA(\rho)\big) \,\overset{\textnormal{\Cshref{prop_rates-transitivity-it1}}}{\subseteq} \, \pmb{\cQ}^{(r_2)}\big(B\big(\cA(\tau)\big)\big) \,\overset{\textnormal{\Cshref{lem_data-processing-neg-ball}}}{\subseteq} \, B\big(\pmb{\cQ}^{(r_2)}\big(\cA(\tau)\big)\big)  \,\overset{\textnormal{\Cshref{prop_rates-transitivity-it1}}}{\subseteq} \, B\big(\cA(\sigma)\big)\,.
\end{align}
%Hence, for any $\pmb{\rho} \in \cA(\rho)$ there exists $\pmb{\sigma} \in \cA(\sigma)$ such that $\pmb{\cQ}^{(r_2)} \circ \pmb{\cP}^{(r_1)}(\pmb{\rho}) \simeq \pmb{\sigma}$. In particular, there exist constants $L\geq 0$, $\kappa \in [0,1)$ such that $\sigma^{(m_n)}$ is $\binom{m_n}{Lm_n^\kappa}$-almost-iid along $\sigma$ for $m_n := \lceil r_2 \lceil r_1 n \rceil\rceil $ and for all $n\in \N$. As
%\begin{align}
%\lceil r_2 \lceil r_1 n \rceil\rceil \,\leq\, \lceil r_2 ( r_1 n +1)\rceil \,\leq\,r_2 r_1 n + r_2 +1 \,\overset{\textnormal{\Cshref{lem_poly-1}}}{\leq} K r_2r_1 n \,\leq\, K\lceil r_2r_1 n\rceil
%\end{align}
%for some constant $K\geq 0$, we find that $\tr_{k_n} \big[\sigma^{(m_n)}\big]$ is $\binom{\lceil r_2r_1 n\rceil}{L'\lceil r_2r_1 n\rceil^\kappa}$-almost-iid along $\sigma$ for some constant $L'\geq0$. 
Therefore, $\pmb{\cR}^{(r_2r_1)} \big(\cA(\rho)\big) \subseteq B\big(\cA(\sigma)\big)$, meaning that $r_2r_1$ is achievable for $ \tilde{R}_\F(\rho \to \sigma)$ and thus, $r_2 r_1 \leq \tilde{R}_\F(\rho \to \sigma)$. After taking the supremum over $r_1$ and $r_2$, the claim follows.
%\GM{Say again sth. about index mismatch between $m_n$ and final partial trace?}
\end{proof}

\begin{proposition}[Weak additivity] \label{prop_rates-weak-additivity}
Let $\rho \in \St(\cH_C)$, $\sigma\in \St(\cH_D)$, and $\F$ be any fixed set of free operations satisfying~\cref{def_QRTs}. For any $k,\,\ell \in \N$, it holds that
\begin{align}
 \tilde{R}_\F\big(\rho^{\otimes k} \to \sigma^{\otimes \ell}\big) \,=\, \frac{k}{\ell}\,  \tilde{R}_\F(\rho \to \sigma)\,.
\end{align}
\end{proposition}
\begin{proof}
%Let $r > 0$ be achievable for $ \tilde{R}_\F(\rho \to \sigma)$, i.e., there exists a sequence of free operations ${\pmb{\mathcal{P}}}^{(\tilde{r})} \in \F$ such that 
%\begin{align}
%{\pmb{\mathcal{P}}}^{(\tilde{r})}\big(\mathcal{A}(\rho)\big) \,\subseteq \, B\big(\mathcal{A}(\sigma)\big)\,.
%\end{align}
Let $\tau \in \St(\cH_E)$ be any density matrix and $s\in \N$. Furthermore, fix a purification $\ket{\theta}_{EF}$ of $\tau$. First, we examine the rates of grouping and ungrouping systems. More precisely, we claim that 
\begin{align} \label{prop_weak-additivity-claim}
 \tilde{R}_\F\big(\tau^{\otimes s} \to \tau \big) \,\geq\, s \qquad \textrm{and} \qquad  \tilde{R}_\F\big(\tau \to \tau^{\otimes s} \big) \,\geq\, \frac{1}{s}\,.
\end{align}

To see the first relation in~\cref{prop_weak-additivity-claim}, let $\pmb{\tau} \in \cA(\tau^{\otimes s})$ be arbitrary. By definition, there exist constants $L \geq 0$, $\kappa \in [0,1)$ such that \smash{$\tau^{(n)}$} is $\binom{n}{r_n}$-almost-iid along $\tau^{\otimes s}$ with $r_n = \lceil L n^\kappa \rceil$ for all $n\in \N$. By definition of almost-iid states, there exists an extension $\tau_{E^nF^n}$ of \smash{$\tau^{(n)}$} such that
\begin{align}
\supp \big(\tau_{E^nF^n} \big) \,\subseteq \, \cV^n_{r_n}\big(\cH^{\ox s}_{EF},\,\ket{\theta}_{EF}^{\ox s} \big) \,\subseteq\,  \cV^{sn}_{sr_n}\big(\cH_{EF},\,\ket{\theta}_{EF} \big)  ,
\end{align}
where the last inclusion follows directly from the definition of the set $\cV$. This means that we may view the $n$-copy state \smash{$\tau^{(n)}$} with reference state \smash{$\tau^{\ox s}$} as an $m_n$-copy state with reference state $\tau$, where $m_n := sn$. Hence, after an additional full symmetrization, the state \smash{$\tilde{\tau}^{(m_n)}:= \cS_{sn}\big(\tau^{(n)}\big)$} is $\binom{sn}{sr_n}$-almost-iid along $\tau$. Since \smash{$sr_n = s\lceil L n^\kappa \rceil \leq \big\lceil \tilde{L} (sn)^\kappa \big\rceil$} for some constant \smash{$\tilde{L} \geq 0$}, this shows that the resulting state sequence $\tilde{\pmb{\tau}}$ (see also~\cref{not_channel-seq}) satisfies $\tilde{\pmb{\tau}}\in \cA(\tau)$. As $\frac{m_n}{n}=s$, this shows that $ \tilde{R}_\F\big(\tau^{\otimes s} \to \tau \big) \geq s$.

To see the second relation in~\cref{prop_weak-additivity-claim}, let $\pmb{\tau} \in \cA(\tau)$ be arbitrary. By definition, there exist constants $L \geq 0$, $\kappa \in [0,1)$ such that \smash{$\tau^{(n)}$} is $\binom{n}{r_n}$-almost-iid along $\tau$ with $r_n = \lceil L n^\kappa \rceil$ for all $n\in \N$. Now, define the number of blocks $m_n$ of size $s$ by $m_n:= \big\lfloor \frac{n}{s} \big\rfloor$. First, we trace away the $k_n := n- m_ns \geq 0 $ remaining systems, leading to the state $\tilde{\tau}^{(m_ns)} := \tr_{k_n}\big[\tau^{(n)} \big]$. By~\cref{rmk_properties_almost_iid_states}, $\tilde{\tau}^{(m_ns)}$ is an $\binom{m_ns}{r_n}$-almost-iid state along $\tau$. By definition of almost-iid states, there exists an extension $\tilde{\tau}_{E^{m_ns}F^{m_ns}}$ of \smash{$\tilde{\tau}^{(m_ns)}$} such that
\begin{align}
\supp \big(\tilde{\tau}_{E^{m_ns}F^{m_ns}} \big) \,\subseteq \, \cV^{m_ns}_{r_n}\big(\cH_{EF},\,\ket{\theta}_{EF} \big) \,\subseteq\,  \cV^{m_n}_{r_n}\big(\cH^{\ox s}_{EF},\,\ket{\theta}^{\ox s}_{EF} \big)  ,
\end{align}
where the last inclusion follows directly from the definition of the set $\cV$, and from the fact that $r_n$ initial defects can yield at most $r_n$ defected blocks of size $s$. This means that we may view the $m_ns$-copy state \smash{$\tilde{\tau}^{(m_ns)}$} with reference state $\tau$ as an $m_n$-copy state $\bar{\tau}^{(m_n)}$ with reference state \smash{$\tau^{\ox s}$}. In addition, note that the initial full permutation invariance of $\tilde{\tau}^{(m_ns)}$ directly implies permutation invariance w.r.t.~the block structure. Furthermore, as $m_n = \Theta(n)$, it follows that \smash{$r_n \leq \big\lceil \tilde{L} m_n^\kappa \big\rceil$} for some constant $\tilde{L} \geq 0$. Therefore, the resulting state sequence, defined by $\bar{\tau}^{(m_n)} := \tilde{\tau}^{(m_ns)}$ (see also~\cref{not_channel-seq}), satisfies $\bar{\pmb{\tau}}\in \cA(\tau^{\ox s})$. Since $\lim_{n\to\infty} \frac{m_n}{n} = \frac{1}{s}$, this shows that $ \tilde{R}_\F\big(\tau \to \tau^{\otimes s} \big) \geq \frac{1}{s}$.

\medskip
\noindent Now, we can combine~\cref{prop_weak-additivity-claim} with~\cref{prop_rates-transitivity} as follows. On the one hand, we have
\begin{align}
 \tilde{R}_\F\big(\rho^{\otimes k} \to \sigma^{\otimes \ell}\big) \,\overset{\textnormal{\Cshref{prop_rates-transitivity}}}&{\geq}\, \tilde{R}_\F\big(\rho^{\otimes k} \to \rho \big) \tilde{R}_\F\big(\rho \to \sigma \big) \tilde{R}_\F\big(\sigma \to \sigma^{\otimes \ell}\big) \\
 \overset{\textnormal{\Cshref{prop_weak-additivity-claim}}}&{\geq}\,  \frac{k}{\ell} \,  \tilde{R}_\F\big(\rho \to \sigma \big)\,.
\end{align}
On the other hand, we find
\begin{align}
 \tilde{R}_\F\big(\rho \to \sigma \big) \,\overset{\textnormal{\Cshref{prop_rates-transitivity}}}&{\geq}\, \tilde{R}_\F\big(\rho \to \rho^{\otimes k} \big) \tilde{R}_\F\big(\rho^{\otimes k} \to \sigma^{\otimes \ell}\big)  \tilde{R}_\F\big(\sigma^{\otimes \ell} \to \sigma \big) \\
 \overset{\textnormal{\Cshref{prop_weak-additivity-claim}}}&{\geq}\,  \frac{\ell}{k} \,  \tilde{R}_\F\big(\rho^{\otimes k} \to \sigma^{\otimes \ell}\big) \,,
\end{align}
which concludes the proof.
\end{proof}

\subsection{Relation to iid asymptotic rates} \label{sec:rel-traditional-rates}
Within our framework, we may define iid asymptotic state transformation rates as follows.
\begin{definition}[iid asymptotic transformation rates]  \label{def_iid_rate}
Let $\rho \in \St(\cH_C)$ and $\sigma\in \St(\cH_D)$, and $\F$ be any fixed set of free operations satisfying~\cref{def_QRTs}. The iid asymptotic state transformation rate is defined by
\begin{align} \label{def:eq_iid_rate}
R_\F(\rho \to \sigma) \,=\, \sup\Big\{\, r>0\ \big| \ \exists\ {\pmb{\mathcal{P}}}^{(r)} \in \F \,:\ {\pmb{\mathcal{P}}}^{(r)}\big( \pmb{\rho}^{\mathrm{iid}}\big) \, \simeq \,\pmb{\sigma}^{\mathrm{iid}} \Big\} \,,
\end{align}
where \smash{$\pmb{\rho}^{\mathrm{iid}} :=\big(\rho^{\otimes n}\big)_{n\in \N} $ and  $\pmb{\sigma}^{\mathrm{iid}} :=\big(\sigma^{\otimes n}\big)_{n\in \N} \,$}. If the set in~\cref{def:eq_iid_rate} is empty, we set $R_\F(\rho \to \sigma)=0$.
\end{definition}
\noindent We make an important remark concerning the definition of the iid rates introduced above. The main difference from the conventional asymptotic conversion rates considered in the literature lies in the approximation criterion: the relation ``$\simeq$" requires the approximation error to decay subexponentially in the number of input copies $n$, whereas the traditional definition only requires the error to vanish asymptotically, without imposing any particular rate of convergence. Consequently, the two definitions may, a priori, yield different conversion rates. For general classes of free operations $\F$ and general conversion tasks, establishing the equivalence of these notions does not appear to be feasible. However, in~\cref{sec:LOCC_rates}, we show that they coincide for the $\LOCC$ entanglement cost and distillation rates. We believe that requiring subexponential error decay is the more operationally meaningful notion and the one most relevant for practical applications.

In the previous section, in particular in~\cref{thrm_rate-tom-indep}, we established that operational asymptotic state transformation rates can be characterized in terms of almost-iid state sequences. This naturally raises the question of whether these rates can be further related to the corresponding rates for iid state sequences. The following theorem establishes a partial result toward proving this equality.
%Recall the definitions of the iid and almost-iid rates, repectively,
%\begin{align}
%R_\F(\rho \to \sigma) \,&=\, \sup\Big\{\, r\ \big| \ \exists\ {\pmb{\mathcal{P}}}^{(r)} \in \F \,:\ {\pmb{\mathcal{P}}}^{(r)}\big( \pmb{\rho}^{\mathrm{iid}}\big) \, \simeq \,\pmb{\sigma}^{\mathrm{iid}} \Big\} \,, \\
%\tilde{R}_\F(\rho \to \sigma)\,&=\,  \sup\Big\{\, r\ \big| \ \exists\ {\pmb{\mathcal{P}}}^{(r)} \in \F \,:\ {\pmb{\mathcal{P}}}^{(r)}\big(\mathcal{A}(\rho)\big) \,\subseteq \, B\big(\mathcal{A}(\sigma)\big) \Big\}  \,.
%\end{align}
%%%%%%%%%%%%%%%%%%%%%%%%%%%%%%%%%% OLD THEOREM MOVING-DEFECTS
%\begin{theorem}[Moving the input defects to the output] \label{thm_moving-input-to-output}
%Let $\rho \in \St(\cH_C)$ and $\sigma\in \St(\cH_D)$, and $\F$ be any fixed set of free operations satisfying~\cref{def_QRTs}. Then, the almost-iid asymptotic transformation rate $ \tilde{R}_\F(\rho \to \sigma)$, introduced in~\cref{thrm_rate-tom-indep}, satisfies
%\begin{align}
%\tilde{R}_\F(\rho \to \sigma) \,=\,  \sup\Big\{\, r\ \big| \ \exists\ {\pmb{\mathcal{P}}}^{(r)} \in \F \,:\ {\pmb{\mathcal{P}}}^{(r)}\big( \pmb{\rho}^{\mathrm{iid}}\big) \, \in \, B\big(\mathcal{A}(\sigma)\big) \Big\}  \,.
%\end{align}
%\end{theorem}
%%%%%%%%%%%%%%%%%%%%%%%%%%%%%%%%%%
 \begin{theorem} \label{thm_almost-iid-iid-rates-relation}
 Let $\rho \in \St(\cH_C)$ and $\sigma\in \St(\cH_D)$, and $\F$ be any fixed set of free operations satisfying~\cref{def_QRTs}. Then, the (operational) almost-iid asymptotic transformation rate $ \tilde{R}_\F(\rho \to \sigma)$, introduced in~\cref{thrm_rate-tom-indep}, is lower bounded by the iid asymptotic rate, i.e., 
 \begin{align}
R_\F(\rho \to \sigma) \,\leq\,  \tilde{R}_\F(\rho \to \sigma)\,.
 \end{align}
 \end{theorem}

\begin{proof}
For $R_\F(\rho \to \sigma) =0$, the claim clearly holds. Now let $r>0$ be achievable for $R_\F(\rho \to \sigma)$, i.e., there exists a sequence of free operations ${\pmb{\mathcal{P}}}^{(r)} \in \F$, with $\cP^{(r)}_n \in \F(C^n \to D^{m_n})$ for each $n\in \N$, such that $\lim_{n\to\infty}\frac{m_n}{n} = r$ and
\begin{align} \label{thm_moving-input-it1}
%\pmb{\mathcal{P}}^{(r)}\big( \pmb{\rho}^{\mathrm{iid}}\big)  \in  B\big( \pmb{\sigma}^{\mathrm{iid}}\big)\,.
\tfrac{1}{2}\big\|\cP^{(r)}_n(\rho^{\otimes n})-\sigma^{\otimes m_n}\big\|_1 \,\leq \, \eps_n
\end{align}
for a supexponentially decreasing function $\eps_n$. By~\cref{def_subexp}, there exists $\alpha \in (0,1)$ and $n_0 \in \N$ such that $\eps_n \leq 2^{-n^\alpha}$ for all $n\geq n_0$.
In words, we are given an operation $\pmb{\mathcal{P}}^{(r)}$ that acts on iid inputs, and our goal is to construct an operation that also works for almost-iid inputs. The key idea is to apply $\pmb{\mathcal{P}}^{(r)}$ blockwise to suitable subsystems. On perfect blocks, that is, blocks containing no defects, the original operation produces the desired output state. The remaining, defective blocks can then be treated separately, and we show that their contribution becomes asymptotically negligible.

\noindent \textbf{Step 1:} \textit{(Block structure)}
For any $N\in \N$, let
\begin{align}  \label{thm_moving-input-it1-0}
\ell_N \,:=\, \left\lceil\exp\big(\sqrt{\ln N}\big)\right\rceil, \qquad s_N\,:=\, \left\lfloor\frac{N}{\ell_N}\right\rfloor,
\end{align}
denote the block length and the number of blocks, respectively. Note that $\ell_N \leq N^{\frac{1}{\sqrt{\ln N}}} +1 = N^{o(1)}$ and hence, $\ell_N = o(N^\lambda)$ for any $\lambda > 0$. We may then write $N = \ell_N s_N + k_N$ with $k_N:= N - \ell_N s_N \geq 0$. Note that since $\lim_{N\to\infty} \frac{\ell_N }{N} = 0$, we have
\begin{align} \label{thm_moving-input-it1-1}
\lim_{N\to\infty} \frac{\ell_N s_N}{N} \,=\,1\,,
\end{align}
and hence, $k_N = o(N)$ is sublinear in $N$. Accordingly, we define a new sequence of free operations
\begin{align} \label{thm_moving-input-it1-2}
\pmb{\cQ}^{(r)}\,:=\big(\cQ^{(r)}_N \big)_{N\in \N} \,, \quad \, \cQ^{(r)}_N \,:= \, (\cP^{(r)}_{\ell_N})^{\otimes s_N} \circ \tr_{k_N} \,\in\, \F(C^N \to D^{\tilde{m}_N})\,,
\end{align}
where $ \tr_{k_N}$ denotes the partial trace of the last $k_N$ subsystems, and $\tilde{m}_N := m_{\ell_N} s_N$. Moreover, as %$(\tilde{m}_N)_{N\in \N}$ is monotonically increasing and
\begin{align}  \label{thm_moving-input-it1-3}
 \lim_{N\to \infty} \frac{\tilde{m}_N}{N} \,=\,  \lim_{N\to \infty}\frac{ \ell_N s_N}{N}\frac{m_{\ell_N} }{\ell_N} \,\overset{\textnormal{\Cshref{thm_moving-input-it1-1}}}{=}\, \lim_{N\to \infty} \frac{m_{\ell_N}}{\ell_N} \,=\, r\,,
\end{align}
the map $\pmb{\cQ}^{(r)} \in \F$ is indeed a free operation with rate $r$. The final protocol will then include an additional random permutation,
\begin{align} \label{thm_moving-input-it1-4}
\tilde{\pmb{\cQ}}{\vphantom{\pmb{\cQ}}}^{(r)} \,:=\, \pmb{\cS} \circ \pmb{\cQ}^{(r)} \,.
\end{align}
In the following, we will write $\ell_N \equiv \ell$ and $s_N \equiv s$ for brevity.

%%%%%%%%%%%%%%%%%%%%%%%%%%%%%%%%%%
%%%%%%%%%%%%%%%%%%%%%%%%%%%%%%%%%%
\medskip
\noindent \textbf{Step 2:} \textit{(Purified picture)}
Now let $N\in \N$ and $\pmb{\rho}\in \cA(\rho)$ be arbitrary. It then follows by~\cref{rmk_properties_almost_iid_states} that there exists constants $L\geq0$, $\kappa \in [0,1)$ such that for any $N\in \N$, the state $\bar{\rho}^{(\ell s)} := \tr_{k_N}( \rho^{(N)})$ is $\binom{\ell s}{r_N}$-almost-iid along $\rho$ with $r_N := \lceil LN^\kappa \rceil$. Let $\ket{\theta}_{CE}$ be a purification of $\rho$ with $E \simeq C$. By definition of almost-iid states, there exists an orthonormal family of states $\{\ket{\Psi_i} \}_{i\in\cT_N}$ from \smash{$\cV^{\ell s}_{r_N}(\cH_{CE},\,\ket{\theta}_{CE} )$}, and an extension $\bar{\rho}_{C^{\ell s}E^{\ell s}}$ of $\bar{\rho}^{(\ell s)}$ such that
\begin{align}
   \bar{\rho}_{C^{\ell s}E^{\ell s}} \,=\, \sum_{i,j \in\cT_N} c_{i,j} \ket{\Psi_i}\!\bra{\Psi_j}
\end{align}
with coefficients $c_{i,j} \in \C$ such that $c_{j,j} \geq 0 $ and $\sum_{j \in\cT_N} c_{j,j} =1$. In addition, we employ~\cref{lem_purif_almost_iid} to construct a (normalized) purification $\ket{\bar{\rho}}_{C^{\ell s}E^{\ell s} H}$ of $  \bar{\rho}_{C^{\ell s}E^{\ell s}}$ specified by
\begin{align} \label{thm_moving-input-it2-9}
\ket{\bar{\rho}}_{C^{\ell s}E^{\ell s} \bar{H}} \,=\, \sum_{i \in \cT_N} \alpha_i \ket{\Psi_i}_{C^{\ell s}E^{\ell s}} \ket{h_i}_{H}
\end{align}
where $\alpha_i \in \C$ and $ \{\ket{h_i}_{H}\}_{i\in \cT_N}$ are normalized (but not necessarily orthogonal) states. 

\medskip
As a next step, we will construct a Stinespring dilation of the map $ (\cP^{(r)}_\ell)^{\otimes s}$ that will act on $\ket{\bar{\rho}}_{C^{\ell s}E^{\ell s} H}$ in order to analyse the action of $\pmb{\cQ}^{(r)}$ on $\pmb{\rho}$. First, we use the Stinespring dilation to construct an isometry $\cU : C^\ell \to D^{m_\ell}F^{(\ell)}$ with $F^{(\ell)}\simeq C^\ell D^{m_\ell}$ such that 
\begin{align}  \label{thm_moving-input-it2-1}
 \forall\, \xi_{C^\ell} \in \St(\cH^{\otimes \ell}_C)\,: \quad \tr_{F^{(\ell)}} \big( \cU \xi_{C^\ell} \cU^\dagger\big) \,=\, \cP^{(r)}_\ell (\xi_{C^\ell})\,.
\end{align} 
Then, recall that by~\cref{thm_moving-input-it1}, it holds that $ \pmb{\cP}^{(r)} (\pmb{\rho}^{\mathrm{iid}}) \simeq \pmb{\sigma}^{\mathrm{iid}} $. Accordingly, by Uhlmann's theorem combined with the Fuchs-van de Graaf inequality, there exists a purification \smash{$\ket{\Sigma}_{D^{m_\ell} E^\ell F^{(\ell)}}$} of $\sigma^{\otimes m_\ell}$ such that
\begin{align}  \label{thm_moving-input-it2-2}
\tfrac{1}{2}\norm{(\cU \otimes \id_{E^\ell})\proj{\theta}_{CE}^{\otimes \ell}\, (\cU^\dagger \otimes \id_{E^\ell}) - \proj{\Sigma}_{D^{m_\ell} E^\ell F^{(\ell)}}}_1 \,\leq\, \sqrt{2\eps_\ell} \,.
\end{align}

\medskip
On the other hand, let $\ket{\varphi}_{DG}$ denote a purification of $\sigma$ with a Hilbert space $G$ such that $\dim(G^{m_\ell})\geq \dim(E^\ell F^{(\ell)})$. Since $\ket{\Sigma}_{D^{m_\ell} E^\ell F^{(\ell)}}$ and $\ket{\varphi}^{\otimes m_\ell}_{DG}$ are both purifications of $\sigma^{\otimes m_\ell}$, there exists an isometry $V : E^\ell F^{(\ell )} \to  G^{m_\ell}$ such that
\begin{align}  \label{thm_moving-input-it2-3}
(\id_{D^{m_\ell}} \otimes V)\ket{\Sigma}_{D^{m_\ell} E^\ell F^{(\ell)}}\,=\, \ket{\varphi}^{\otimes m_\ell}_{DG}\,.
\end{align}
Then, we define a new isometry $\cW$ by combining the maps in~\cref{thm_moving-input-it2-2,thm_moving-input-it2-3} as
\begin{align}
\cW\,:\,C^\ell E^\ell \to D^{m_\ell} G^{m_\ell} \,,\quad \cW \,:=\, (\id_{D^{m_\ell}} \otimes V) \circ (\cU \otimes \id_{E^\ell})\,.
\end{align}
By construction, this isometry satisfies
\begin{align}  \label{thm_moving-input-it2-8}
 \forall\, \xi_{C^\ell E^\ell} \in \St(\cH^{\otimes \ell}_{CE})\,: \quad \tr_{G^{m_\ell}}\big( \cW  \xi_{C^\ell E^\ell} \cW^\dagger \big) \,=\, \cP^{(r)}_\ell (\xi_{C^\ell})\,,
\end{align}
where $ \xi_{C^\ell} = \tr_{E^\ell}( \xi_{C^\ell E^\ell})$, and 
\begin{align}  \label{thm_moving-input-it2-4}
\tfrac{1}{2} \norm{ \cW\,\proj{\theta}_{CE}^{\otimes \ell}\,\cW^\dagger - \proj{\varphi}^{\otimes m_\ell}_{DG}}_1 \,\leq\, \sqrt{2\eps_\ell}\,.
\end{align}
Moreover, since the tensor product of Stinespring dilations of any channels is again a Stinespring dilation of the tensor product of these channels, it follows that $\cW^{\otimes s}$ is a Stinespring dilation of $(\cP^{(r)}_\ell)^{\otimes s}$. Finally, we use the above to decompose the output of the purified protocol as
\begin{align} \label{thm_moving-input-it2-6}
\ket{\Phi_\ell}\,:=\, \cW\ket{\theta}^{\otimes \ell}_{CE} \,=\, \sqrt{1-p_\ell} \,\ket{\varphi}^{\otimes m_\ell}_{DG} + \sqrt{p_\ell} \,\ket{\zeta^\perp}_{D^{m_\ell}G^{m_\ell}}\,,\qquad p_\ell \overset{\textnormal{\Cshref{thm_moving-input-it2-4}}}{\leq} 2\eps_\ell\,,
\end{align}
where the vector $ \ket{\zeta^\perp}_{D^{m_\ell}G^{m_\ell}}$ satisfies $ \braket{\zeta^\perp | \varphi}^{\otimes m_\ell} = 0$.

%Note that w.l.o.g., we can assume $\braket{\tilde{\sigma}| \cW | \theta}^{\otimes n} \geq 0$ (otherwise, we may rescale \smash{$\ket{\Sigma}_{D^{m_n} E^nF^{(n)}}$} by a corresponding phase factor). Then,~\cref{lem_basis-unitary} implies that there exists a unitary $\tilde{V}$ on the space ${D^{m_n}G^{m_n}H}$ such that
% \begin{align} \label{thm_moving-input-it2-6}
%    \tilde{V}\,\cW\ket{\theta}_{CE}^{\otimes n}=\ket{\tilde{\sigma}}_{D^{m_n}G^{m_n}H} 
% \end{align} 
% and
% \begin{align} \label{thm_moving-input-it2-7}
% \norm{\id_{D^{m_n}G^{m_n}H} - \tilde{V}}_{\infty} \,=\, \Big\|\cW\ket{\theta}_{CE}^{\otimes n}-\ket{\tilde{\sigma}}_{D^{m_n}G^{m_n}H} \Big\|_2 \,\overset{\textnormal{\Cshref{lem_norm-identities} \eqref{lem_norm-identities-it4}}}{\leq}\, \sqrt{2}\varepsilon_n \,,
% \end{align}
% where $\norm{\cdot}_2$ denotes the Euclidean norm and $\norm{\cdot}_{\infty}$ is the corresponding operator norm. Correspondingly, we may also define the isometry
% \begin{align} \label{thm_moving-input-it2-10}
%\tilde{\cW} \,:\,C^n E^n \to D^{m_n} G^{m_n}H\,,\quad \tilde{\cW} := \tilde{V}\circ \cW \,.
%\end{align}
%\medskip
%\noindent \textbf{Step 3:} \textit{(Evaluating $\pmb{\cQ}^{(r)}(\pmb{\rho})$)} 
\medskip
With the above preparation, especially~\cref{thm_moving-input-it1-2,thm_moving-input-it2-8,thm_moving-input-it2-9,thm_moving-input-it2-6}, we can now rewrite the action of the free operation $\pmb{\cQ}^{(r)}$ on the input almost-iid state sequence $\pmb{\rho}$ for each $N\in \N$ (and recalling that $N = \ell_N s_N + k_N$) as
\begin{align} \label{thm_moving-input-it2-12}
\cQ^{(r)}_N\big(\rho^{(N)}\big) \,&=\, \tr_{H} \, \tr_{G^{m_\ell s}}\Big[\big(\cW^{\otimes s} \otimes \mathrm{id}_{H} \big) \proj{\bar{\rho}}_{C^{\ell s}E^{\ell s} H}  \big(\cW^{\otimes s} \otimes \mathrm{id}_{H} \big)^\dagger \Big] \,,
%\,\overset{(\star)}&{\simeq}\, \tr_{\bar{H}}\,\tr_{H^s} \, \tr_{G^{m_ns}}\Big[\big(\tilde{\cW}^{\otimes s} \otimes \id_{\bar{H}} \big) \proj{\bar{\rho}}_{C^{ns}E^{ns} \bar{H}}  \big(\tilde{\cW}^{\otimes s} \otimes \id_{\bar{H}} \big)^\dagger \Big] \,.  \label{thm_moving-input-it3-0}
\end{align}
and hence, it suffices to analyze the vector
\begin{align}  \label{thm_moving-input-it2-7}
\big(\cW^{\otimes s} \otimes \mathrm{id}_{H} \big) \ket{\bar{\rho}}_{C^{\ell s}E^{\ell s} H} \,\overset{\textnormal{\Cshref{thm_moving-input-it2-9}}}&{=}\, \sum_{i \in \cT_N} \alpha_i \ket{h_i}_{H} \otimes \cW^{\otimes s} \ket{\Psi_i}_{C^{\ell s}E^{\ell s}} \,.
\end{align}
Because $\ket{\Psi_i}_{C^{\ell s}E^{\ell s}}  \in \cV^{\ell s}_{r_N}\big(\cH_{CE},\,\ket{\theta}_{CE} \big)$ for any $i\in \cT_N$, we may write it as
\begin{align}
\ket{\Psi_i}_{C^{\ell s}E^{\ell s}}  \,=\, \pi_i \big(\ket{\theta}^{\otimes \ell \tilde{s}}_{CE} \otimes \ket{\Omega_i}_{C^{\ell(s-\tilde{s})}E^{\ell(s-\tilde{s})}} \big)\,,\qquad \tilde{s} := \max\{ s - r_N, 0\}
\end{align}
where $\pi_i \in S_s$ denotes a (simultaneous) permutation of the $s$ blocks, $\tilde{s}$ is the minimal number of unfaulty (iid) blocks of the form $\ket{\theta}^{\otimes \ell}_{CE}$, and $\ket{\Omega_i}$ represents any `trash'-vector on the remaining faulty blocks. Thus, for any $i\in \cT_N$, it holds
\begin{align} \label{thm_moving-input-it2-10}
\cW^{\otimes s} \ket{\Psi_i}_{C^{\ell s}E^{\ell s}} \,&=\,  \pi_i \Big(\cW^{\otimes \tilde{s}}  \ket{\theta}^{\otimes \ell \tilde{s}}_{CE} \otimes \cW^{\otimes s-\tilde{s}}  \ket{\Omega_i}_{C^{\ell(s-\tilde{s})}E^{\ell(s-\tilde{s})}} \Big)  \\
\overset{\textnormal{\Cshref{thm_moving-input-it2-6}}}&{=} \, \pi_i \Big( \ket{\Phi_\ell}^{\otimes \tilde{s}} \otimes   \ket{\tilde{\Omega}_i}_{D^{m_\ell(s-\tilde{s})}G^{m_\ell(s-\tilde{s})} } \Big)
\end{align}
with $ \ket{\tilde{\Omega}_i} :=  \cW^{\otimes s-\tilde{s}}  \ket{\Omega_i} $. Furthermore, notice that since $s_N = N^{1-o(1)}$, the number clean blocks scales as
\begin{align} \label{thm_moving-input-it2-11}
\lim_{N\to\infty}\frac{\tilde{s}_N}{s_N} \,=\,1 \quad \implies \quad \tilde{s}_N = N^{1-o(1)}\,.
\end{align}

%Because of~\cref{thm_moving-input-it2-11} and due to~\cref{lem_purif_almost_iid}, the vector $\ket{\tilde{\sigma}}^{\otimes \tilde{s}}_{D^{m_n}G^{m_n}H}$ satisfies
%\begin{align}
%\supp \big( \tr_{ H^{\tilde{s}} } \big[ \proj{\tilde{\sigma}}^{\otimes \tilde{s}}_{D^{m_n}G^{m_n}H} \big] \big) \,\subseteq\, \mathrm{span}\, \cV\big(\cH^{\ox \tilde{s} m_n}_{DG},\,\ket{\varphi}_{DG}^{\ox \tilde{s} m_n -\tilde{s} \lceil L'm_n^{\kappa'} \rceil} \big) \,,
%\end{align}
%and hence,
%\begin{align}
%\supp \big( \tr_{H^s} \big[ \tilde{\cW}^{\otimes s} \proj{\Psi_i}_{C^{ns}E^{ns}} \tilde{\cW}^{\otimes s\dagger} \big] \big) \,\subseteq\, \mathrm{span}\, \cV\big(\cH^{\ox s m_n}_{DG},\,\ket{\varphi}_{DG}^{\ox s m_n -\tilde{s} \lceil L'm_n^{\kappa'} \rceil - m_n(s-\tilde{s})} \big) \,,
%\end{align}
%for all $i\in \cT_N$, leading to
%\begin{align} \label{thm_moving-input-it3-3}
%\supp \Big(\tr_{\bar{H}}\, \tr_{H^s} \Big[ \big(\tilde{\cW}^{\otimes s} \otimes \id_{\bar{H}} \big) \proj{\bar{\rho}}_{C^{ns}E^{ns} \bar{H}} \big(\tilde{\cW}^{\otimes s} \otimes \id_{\bar{H}} \big)^\dagger \Big] \Big) \,\subseteq\, \mathrm{span}\, \cV\big(\cH^{\ox s m_n}_{DG},\, \ket{\varphi}_{DG}^{\ox s m_n - e_n} \big) \,,
%\end{align}
%where $e_n := \tilde{s} \lceil L'm_n^{\kappa'} \rceil + m_n(s-\tilde{s})$ denotes the number of defects.
\medskip
\noindent \textbf{Step 3:} \textit{(Coherent binomial truncation)} Next, we aim to show that the output state $\cQ^{(r)}_N\big(\rho^{(N)}\big)$ is close to a (generalized) power-law sublinear almost-iid state along $\sigma$. To do so, we focus on the states $\cW^{\otimes s} \ket{\Psi_i}$ for each $i\in \cT_N$ and analyse the action on the clean $\ket{\theta}^{\otimes \ell}$ and defective $\ket{\Omega_i}$ input blocks separately. For the clean blocks, the main conceptual idea is to turn the approximation error, which is represented by the $p_\ell$ factor in~\cref{thm_moving-input-it2-6}, into additional defects in the output, similarly as was done in the proof of~\cref{prop_close-almost-iid}. The defective input blocks, on the other hand, will simply remain defected blocks in the output state.

Formally, we have $\tilde{s}$ clean input blocks that are transformed via $\cW$ into the state $\ket{\Phi_\ell}^{\otimes \tilde{s}}$, which can be expanded as
\begin{align}
\ket{\Phi_\ell}^{\otimes \tilde{s}} \,=\, \sum_{J\subseteq [\tilde{s}]} (1-p_\ell)^{\frac{\tilde{s}-|J|}{2}} p_\ell^{\frac{|J|}{2}}\ket{\Phi_{\ell,J}}\,,\qquad  \ket{\Phi_{\ell,J}} \,:=\, \bigotimes_{t\in [\tilde{s}]\setminus J} \ket{\varphi}^{\otimes m_\ell}_{DG,t}\, \bigotimes_{t\in J} \ket{\zeta^\perp}_{D^{m_\ell}G^{m_\ell},t}
\end{align}
where the sum ranges over all subsets $J \subseteq \{1,\dots,\tilde{s} \}$ denoting the places of occurrence of an orthogonal block state $\ket{\zeta^\perp}_{D^{m_\ell}G^{m_\ell}}$ in the binomial expansion. Clearly, the state $\ket{\Phi_\ell}^{\otimes \tilde{s}}$ is itself not almost-iid along $\ket{\varphi}_{DG}$ with a small number of defects since it always contains a component with $|J| = \tilde{s}$ defects. However, it is possible to find such an almost-iid state that is close to $\ket{\Phi_\ell}^{\otimes \tilde{s}}$. To find that state, the idea is to truncate the sum in the expansion above in a way that the number of defects remain power-law sublinear in $N$ while ensuring the weight of the discarded component to be subexponentially small in $N$. More precisely, we fix any parameter $\lambda \in (\kappa, 1)$ (recalling that $r_N = \lceil L N^\kappa \rceil$ is the number of defects of the input state $\rho^{(N)}$) and define
\begin{align}
a_N \,:=\, \lceil N^\lambda \rceil\,, \qquad \ket{\Phi'_\ell} \,:=\,  \sum_{\substack{J\subseteq [\tilde{s}] \\ |J| < a_N }} (1-p_\ell)^{\frac{\tilde{s}-|J|}{2}} p_\ell^{\frac{|J|}{2}}\ket{\Phi_{\ell,J}}\,,
%  \qquad \ket{\tilde{\Phi}_\ell}_{D^{m_\ell \tilde{s}}G^{m_\ell \tilde{s}}} := \frac{\ket{\Phi'_\ell}}{\norm{\Phi'_\ell}_2}\,,
\end{align}
and the discarded (unnormalized) component is
\begin{align}
 \ket{R_\ell} \,:=\,  \sum_{\substack{J\subseteq [\tilde{s}] \\ |J| \geq a_N }} (1-p_\ell)^{\frac{\tilde{s}-|J|}{2}} p_\ell^{\frac{|J|}{2}}\ket{\Phi_{\ell,J}}\,.
\end{align}
Note that~\cref{thm_moving-input-it2-11} ensures that $\tilde{s} > a_N$ for sufficiently large values of $N\in \N$. Since the vectors $\ket{\Phi_{\ell,J}}$ are normalized and orthogonal for different subsets $J \neq J'$, the norm of the discarded component can be calculated as
\begin{align} \label{thm_moving-input-it3-1}
\big\| \ket{R_\ell} \big\|^2_2 \,=\,  \sum_{\substack{J\subseteq [\tilde{s}] \\ |J| \geq a_N }} (1-p_\ell)^{\tilde{s}-|J|} p_\ell^{|J|} \,=\, \sum_{j=a_N}^{\tilde{s}} \binom{\tilde{s}}{j}  (1-p_\ell)^{\tilde{s}-j} p_\ell^{j}\,. 
\end{align}
The expression in~\cref{thm_moving-input-it3-1} can be related to a random variable $Z \sim \mathrm{Bin}(\tilde{s}, p_\ell)$ that follows a binomial distribution, namely
\begin{align} \label{thm_moving-input-it3-2}
\Pr\big[Z\geq a_N \big] \,=\,   \sum_{j=a_N}^{\tilde{s}} \binom{\tilde{s}}{j}  (1-p_\ell)^{\tilde{s}-j} p_\ell^{j} \,\leq\,  \left( \frac{e \mu_N}{a_N} \right)^{a_N}\,,
\end{align}
where the last bound corresponds to Chernoff's inequality, see~\cite[Theorem 2.3.1]{Vershynin_2018}, and we further used that $\mu_N := \E[Z] = p_\ell \tilde{s} $ and $a_N \geq \mu_N$ for sufficiently large values of $N\in \N$ since $\lim_{N\to\infty} a_N = \infty$ and
\begin{align} \label{thm_moving-input-it3-3}
\mu_N \overset{\textnormal{\Cshref{thm_moving-input-it2-6}}}{\leq}  2\eps_{\ell_N}\tilde{s}_N \overset{\textnormal{\Cshref{thm_moving-input-it2-11}}}{\leq} 2^{-\ell_N^\alpha + \log(N)+1 } \,\overset{\textnormal{\Cshref{thm_moving-input-it1-0}}}{=}\, 2^{- \ell_N^\alpha(1 + o(1)) } \ \xrightarrow{N\to \infty} \ 0 \,.
\end{align}
Furthermore,~\cref{thm_moving-input-it3-3} can be used to bound~\cref{thm_moving-input-it3-2} by
\begin{align}
 \left( \frac{e \mu_N}{a_N} \right)^{a_N} =\, 2^{a_N(\log\mu_N -\log a_N + \log e)} \overset{\textnormal{\Cshref{thm_moving-input-it3-3}}}{\leq} 2^{-a_N \ell_N^\alpha(1 + o(1)) }\,,
\end{align}
where we used again that $a_N = \lceil N^\lambda \rceil$ and $\lim_{N\to\infty} \frac{\log N}{\ell_N^\alpha} = 0$. To sum up, we found that
\begin{align}  \label{thm_moving-input-it3-4}
\big\| \ket{R_\ell} \big\|_2 \,=\, \sqrt{\Pr\big[Z\geq a_N \big]} \,\leq\,2^{-\frac{1}{2}a_N \ell_N^\alpha(1 + o(1)) }
\end{align}
for sufficiently large $N\in \N$.

\medskip
Hence, we may express the output state in~\cref{thm_moving-input-it2-10} as
\begin{align}   \label{thm_moving-input-it3-5}
\cW^{\otimes s} \ket{\Psi_i}_{C^{\ell s}E^{\ell s}} \,=\,  \underbrace{\pi_i \Big(\ket{\Phi'_\ell} \otimes   \ket{\tilde{\Omega}_i}_{D^{m_\ell(s-\tilde{s})}G^{m_\ell(s-\tilde{s})} } \Big)}_{=: \,\ket{G_i}_{D^{m_\ell s}G^{m_\ell s}}} +  \underbrace{\pi_i \Big(\ket{R_\ell} \otimes   \ket{\tilde{\Omega}_i}_{D^{m_\ell(s-\tilde{s})}G^{m_\ell(s-\tilde{s})} } \Big) }_{=:\, \ket{R_{\ell,i}}_{D^{m_\ell s}G^{m_\ell s}}}
\end{align}
and consequently,~\cref{thm_moving-input-it2-7} can be rewritten as
\begin{align}   \label{thm_moving-input-it3-6}
\big(\cW^{\otimes s} \otimes \mathrm{id}_{H} \big) \ket{\bar{\rho}}_{C^{\ell s}E^{\ell s} H} \,=\,\sum_{i \in \cT_N} \alpha_i \ket{h_i}_{H} \otimes \ket{G_i}_{D^{m_\ell s}G^{m_\ell s}} + \underbrace{\sum_{i \in \cT_N} \alpha_i \ket{h_i}_{H} \otimes \ket{R_{\ell,i}}_{D^{m_\ell s}G^{m_\ell s}}}_{=:\, \ket{R_N}} \,.
\end{align}
The norm of the second component in~\cref{thm_moving-input-it3-6}, that will be discarded, can by bounded as
\begin{align}   \label{thm_moving-input-it3-7}
\big\| \ket{R_N} \big\|_2  \overset{\textnormal{triangle}}{\leq} \sum_{i \in \cT_N}| \alpha_i |\, \big\| \ket{h_i}_{H}  \ket{R_{\ell,i}}_{D^{m_\ell s}G^{m_\ell s}} \big\|_2  \,=\, \sum_{i \in \cT_N}| \alpha_i |\, \big\| \ket{R_{\ell}} \big\|_2 \overset{\textnormal{\Cshref{thm_moving-input-it3-4}}}{\leq} |\cT_N| \,2^{-\frac{1}{2}a_N \ell_N^\alpha(1 + o(1)) }\,.
\end{align}
Recalling from~\cref{rmk_properties_almost_iid_states} that $|\cT_N| \leq 2^{N h(r_N/N) +r_N \log(d_{CE})}$ with $r_N = \lceil L N^\kappa \rceil$, we can bound
\begin{align}
\log |\cT_N|\,\leq\, CN^\kappa \log N
\end{align}
for some constant $C> 0$. Furthermore, for $a_N = \lceil N^\lambda \rceil$ with $\lambda \in (\kappa ,1)$ and $\ell_N $ defined in~\cref{thm_moving-input-it1-0}, we have
\begin{align}
\lim_{N\to \infty} \frac{\log |\cT_N|}{\tfrac{1}{2}a_N \ell_N^\alpha} \,\leq\, \lim_{N\to \infty} 2C \, \frac{N^\kappa \log N}{N^\lambda \ell_N^\alpha} \,=\, 0\,,
\end{align}
and consequently,
\begin{align}  \label{thm_moving-input-it3-8}
\big\| \ket{R_N} \big\|_2 \,\leq\, 2^{-\frac{1}{2}a_N \ell_N^\alpha(1 + o(1)) }\,.
\end{align}
This implies that for any parameter $\nu\in(0,\lambda)$, there exists $N_0 \in \N$ such that $\| \ket{R_N} \|_2 \leq 2^{-N^\nu}$ for all $N \geq N_0$, see also~\cref{def_subexp-equivalence,ex_subexp-function}. Thus, we have shown that the weight of the discarded component $\| \ket{R_N} \|_2$ is subexponentially decreasing in $N$.

\medskip
\noindent \textbf{Step 4:} \textit{(The truncated state is almost-iid)}
Proceeding from~\cref{thm_moving-input-it3-6}, we define the truncated state for sufficiently large $N \in \N$ by
\begin{align}   \label{thm_moving-input-it4-1}
\ket{\hat{G}_N}_{D^{m_\ell {s}}G^{m_\ell {s}} H} := \frac{\ket{G_N}}{\norm{\ket{G_N}}_2}\,,\qquad \ket{G_N} \,:=\, \sum_{i \in \cT_N} \alpha_i \ket{h_i}_{H} \otimes \ket{G_i}_{D^{m_\ell s}G^{m_\ell s}} \,.
\end{align}
By construction, it holds that
\begin{align}  \label{thm_moving-input-it4-2}
\ket{G_i}_{D^{m_\ell s}G^{m_\ell s}} \,\in \, \mathrm{span}\, \cV^{m_\ell s}_{\tilde{r}_N}\big(\cH_{DG},\,\ket{\varphi}_{DG} \big)\,,\qquad \tilde{r}_N \,:=\, m_\ell(r_N +a_N)
\end{align}
for all $i\in \cT_N$. Intuitively, the number $\tilde{r}_N$ of output defects arises from the $r_N$ faulty input blocks and from the approximation error, which is truncated to retain at most $a_N$ faulty blocks. Hence, it follows from~\cref{lem_purif_almost_iid} that
\begin{align}  \label{thm_moving-input-it4-3}
\supp\big(\tr_H\big[\proj{\hat{G}_N}_{D^{m_\ell {s}}G^{m_\ell {s}} H}\big]\big) \,\subseteq\, \mathrm{span}\, \cV^{m_\ell s}_{\tilde{r}_N}\big(\cH_{DG},\,\ket{\varphi}_{DG} \big)\,.
\end{align}
On the other hand, because $m_\ell = \Theta(\ell_N)$, $\ell_N = N^{o(1)}$ and $\lambda > \kappa$, the number of output defects scales as $\tilde{r}_N = N^{\lambda + o(1)}$. Since $\tilde{m}_N = \Theta(N)$ due to~\cref{thm_moving-input-it1-3}, this implies that there exist constants $\tilde{L} \geq 0$ and $\tilde{\kappa} \in (\lambda,1)$ such that
\begin{align}  \label{thm_moving-input-it4-4}
\tilde{r}_N \,\leq\, \big\lceil  \tilde{L}\tilde{m}_N^{\tilde{\kappa}} \big\rceil \,.
\end{align}
for all $N \in \N$. Hence, after a random permutation, we find that the following state
\begin{align}  \label{thm_moving-input-it4-5}
\hat{\omega}^{(\tilde{m}_N)} \,:=\,  \cS_{\tilde{m}_N} \circ  \tr_{G^{m_\ell {s}}}\circ \tr_H\big[\proj{\hat{G}_N}_{D^{m_\ell {s}}G^{m_\ell {s}} H}\big]
\end{align}
is a $\binom{\tilde{m}_N}{\tilde{r}_N}$-almost-iid state along $\sigma$ due to~\cref{thm_moving-input-it4-3}. 
%Finally, we define a state sequence by $\hat{\pmb{\omega}} := \big(\hat{\omega}^{(\tilde{m}_N)} \big)_{N\in\N}$

\medskip
\noindent We emphasise that while the construction of the states $\hat{\omega}^{(\tilde{m}_N)}$ depended on the number of defects $r_N$ via the parameter $\kappa$ and hence, on the input sequence $\pmb{\rho} \in \cA(\rho)$, the free operation $\pmb{\cQ}^{(r)}$ remained unchanged and is hence universally applicable to \emph{any} input sequence. In other words, the coherent binomial truncation serves as a proof technique and does not represent any physical operation.

\medskip
\noindent \textbf{Step 5:} \textit{(Conclusion)}
Finally, recall~\cref{thm_moving-input-it1-4,thm_moving-input-it2-12,thm_moving-input-it4-5}. For all sufficiently large values of $N\in \N$, we can then estimate
\begin{align}
\tfrac{1}{2} \Big\| \tilde{\cQ}^{(r)}_N ( \rho^{(N)})& -  \hat{\omega}^{(\tilde{m}_N)}\Big\|_1 \nonumber \\
\overset{\textnormal{DPI}}&{\leq} \tfrac{1}{2}\norm{\big(\cW^{\otimes s} \otimes \mathrm{id}_{H} \big) \proj{\bar{\rho}}_{C^{\ell s}E^{\ell s} H}  \big(\cW^{\otimes s} \otimes \mathrm{id}_{H} \big)^\dagger -  \proj{\hat{G}_N}_{D^{m_\ell {s}}G^{m_\ell {s}} H}}_1 \\
\overset{\textnormal{\Cshref{lem_norm-identities} \eqref{lem_norm-identities-it3}}}&{\leq} \, \norm{\big(\cW^{\otimes s} \otimes \mathrm{id}_{H} \big) \ket{\bar{\rho}}_{C^{\ell s}E^{\ell s} H} -  \ket{\hat{G}_N}_{D^{m_\ell {s}}G^{m_\ell {s}} H}}_2 \\
\overset{\textnormal{triangle}}&{\leq}  \Big\|\big(\cW^{\otimes s} \otimes \mathrm{id}_{H} \big) \ket{\bar{\rho}}_{C^{\ell s}E^{\ell s} H} - \ket{G_N}\Big\|_2 + \Big\| \ket{G_N} -  \ket{\hat{G}_N}_{D^{m_\ell {s}}G^{m_\ell {s}} H}\Big\|_2 \\
\overset{\textnormal{\Cshref{thm_moving-input-it3-6}}}&{=} \, \big\| \ket{R_N}\big\|_2 + \Big| \big\|\ket{G_N}\big\|_2 - 1\Big| \\
\overset{\textnormal{triangle, \Cshref{thm_moving-input-it3-6}}}&{\leq} \, 2\,\big\| \ket{R_N}\big\|_2 \quad \xrightarrow{\textnormal{\Cshref{thm_moving-input-it3-8}}} \quad \textrm{subexponentially decreasing in $N$}\,. \label{thm_moving-input-it5-0}
\end{align}
For the finitely many small values of $N\in \N$, we may define the states $ \hat{\omega}^{(\tilde{m}_N)}$ arbitrarily. Then, the states $\{ \hat{\omega}^{(\tilde{m}_N)}\}_{N\in \N}$ define a state sequence $ \hat{\pmb{\omega}} $ by means of~\cref{not_channel-seq} which satisfies
% $ \hat{\pmb{\omega}} := \big( \bar{\omega}^{(M)}\big)_{M\in \N}$, with $\bar{\omega}^{(M)} :=  \hat{\omega}^{(\tilde{m}_{\bar{N}(M)})}$ and $\bar{N}(M) := \min\{ N\in \N \,:\, \tilde{m}_N = M \}$, defines a state sequence that satisfies 
$ \hat{\pmb{\omega}} \,\in \, \cA(\sigma)$ due to~\cref{thm_moving-input-it4-4,thm_moving-input-it4-5}. Therefore,~\cref{thm_moving-input-it5-0} shows that
\begin{align}
\tilde{\pmb{\cQ}}{\vphantom{\pmb{\cQ}}}^{(r)}(\pmb{\rho})\, \simeq \, \hat{\pmb{\omega}} \, \in\, \cA(\sigma)
\end{align}
as desired. Hence, since $\pmb{\rho} \in \cA(\rho)$ was arbitrary, the rate $r$ is achievable for $\tilde{R}_\F(\rho \to \sigma)$, meaning that $r\leq  \tilde{R}_\F(\rho \to \sigma)$. After taking the supremum over all achievable rates $r$, we find $ R_\F(\rho \to \sigma) \leq \tilde{R}_\F(\rho \to \sigma) $, which proves the claim.
 \end{proof}
 
% \begin{corollary} \label{cor_almost-iid-iid-rates-relation}
% Let $\rho \in \St(\cH_C)$ and $\sigma\in \St(\cH_D)$, and $\F$ be any fixed set of free operations satisfying~\cref{def_QRTs}. Then, the iid asymptotic transformation rate $R_\F$ is upper bounded by the almost-iid (or operational) asymptotic transformation rate $\tilde{R}_\F$, i.e., 
% \begin{align}
%R_\F(\rho \to \sigma) \,\leq\,  \tilde{R}_\F(\rho \to \sigma)\,.
% \end{align}
% \end{corollary}
% \begin{proof}
%Let $r\geq0$ be achievable for $R_\F(\rho \to \sigma)$, i.e., there exists a sequence of free operations ${\pmb{\mathcal{P}}}^{(r)} \in \F$ such that ${\pmb{\mathcal{P}}}^{(r)}\big( \pmb{\rho}^{\mathrm{iid}}\big)  \simeq \pmb{\sigma}^{\mathrm{iid}}$. Since $ \pmb{\sigma}^{\mathrm{iid}} \in \mathcal{A}(\sigma)$, we have that ${\pmb{\mathcal{P}}}^{(r)}(\pmb{\rho}^{\mathrm{iid}}) \in  B\big(\mathcal{A}(\sigma)\big) $. Thus, it follows immediately by~\cref{thm_moving-input-to-output} that $r$ is also achievable for $\tilde{R}_\F(\rho \to \sigma)$. After taking the supremum over all achievable rates $r$, the claim follows.
%\end{proof}

\medskip
\noindent
\textbf{Relative entropy of resource:} Let $\F$  be any fixed set of free operations satisfying~\cref{def_QRTs}, and for any physical system $\cH_C$, let $\F(C)$ denote the corresponding set of free operations. In the following, we define the family of sets $(\F_n)_{n\in \N}$ with $\F_n := \F(C^n)$. Assume the following properties:
\begin{enumerate} \label{list_BP-axioms}
\item The set $\F_n$ is closed and convex (as a subset of $\St(C^{n})$).
\item The set $\F_n$ is closed under permutations of the subsystems of $C^n$.
\item The family $(\F_n)_{n\in \N}$ is closed under tensor products: If $\sigma \in \F_n$ and $\omega \in \F_m$ then $\sigma \otimes \omega \in \F_{mn}$.
\item The family $(\F_n)_{n\in \N}$ is closed under the partial trace: $\rho_{n+1} \in \F_{n+1}$ implies $\rho_n = \tr_1[\rho_{n+1}] \in \F_n$ (where $\tr_1$ denotes the partial trace over the last subsystem).
\item Full-rank state: There exists $\gamma \in \F_1$ with $\gamma > 0$.
\end{enumerate}
Note that the requirements in~\cref{def_QRTs} actually imply all of the above properties except for the closedness of $\F_n$ and the existence of a full-rank state. Nevertheless, when considering asymptotic state transformation rates, one may replace the set of free operations $\F$ by its closure $\overline{\F}$ without altering the rates. In that case, closedness of $\F_n$ is also implied. Another example concerns the set of $\LOCC$ operations which is convex but not closed~\cite{LOCC14}. Nevertheless, the corresponding set of free operations is given by the set of separable states, $\F_n =  \mathrm{SEP}(A^n:B^n)$, which is closed~\cite{chit19}. On the other hand, the existence of a full-rank state is known to be granted in many quantum resource theories of interest, such as in the case of $\LOCC$ operations. The above properties are also referred to as Brand{\~a}o-Plenio axioms~\cite{brandao_Stein_10}. We further refer to~\cite{Kuroiwa-Yamasaki2020} for a general treatment of quantum resource theories with more minimal assumptions.

Then, for any $\rho \in \St(\cH)$ and set $\cM \subseteq  \St(\cH)$ of states, the \emph{relative entropy of resource}~\cite{VPRK97} is defined as
\begin{align} \label{eq_def_ER}
D(\rho\| \cM) :=  \inf_{\sigma \in \cM} D(\rho \| \sigma) \, .
\end{align}
For a family of free states $(\F_n)_{n\in \N}$ as specified above, its regularized version is given by
\begin{align} \label{eq_def_regul_2}
D^\reg(\rho \| \F) 
:=\lim_{n \to \infty} \frac{1}{n} D(\rho^{\otimes n} \| \F_n)
=\inf_{n \in \N} \frac{1}{n} D(\rho^{\otimes n} \| \F_n)\, ,
\end{align}
which is well-defined as ensured by Fekete's lemma and the requirements listed above.
\begin{proposition} \label{prop_rates-upper-bound}
Let $\rho \in \St(\cH_C)$, $\sigma\in \St(\cH_D)$, and let $\F$ be any fixed set of free operations satisfying~\cref{def_QRTs} and containing a full-rank free state. Then,
\begin{align}
 \tilde{R}_\F(\rho \to \sigma) \,\leq\, \frac{D^\reg(\rho \| \F)}{D^\reg(\sigma \| \F)}\,.
\end{align}
\end{proposition}
\begin{proof}
Here, we will follow the same logic as in~\cite[Theorem 11.6]{Gour_2025} (see also~\cite{Horodecki2002,chit19}) combined with the insight from~\cite{MSR2026-2} that the relative entropy of resource is robust under almost-iid states. For $ \tilde{R}_\F(\rho \to \sigma) = 0$, the claim immediately follows due to the positivity of the relative entropy. Now let $\tilde{r} > 0$ be achievable for $ \tilde{R}_\F(\rho \to \sigma)$, i.e., there exists a sequence of free operations ${\pmb{\mathcal{P}}}^{(\tilde{r})} \in \F$ such that 
\begin{align} 
\pmb{\mathcal{P}}^{(\tilde{r})}\big(\mathcal{A}(\rho)\big)  \in  B\big(\mathcal{A}(\sigma)\big)\,.
\end{align}
That is, for $\pmb{\rho}^{\mathrm{iid}} \in \cA(\rho)$ there exists an almost-iid sequence $\pmb{\tau} \in \cA(\sigma)$ such that for any $n\in \N$, we have
\begin{align}
\frac{1}{2} \norm{\mathcal{P}^{(\tilde{r})}_n\big( \rho^{\ox n}\big) - \tau^{(m_n)} }_1 \,\leq\, \eps_n\,,
\end{align}
where $\eps_n$ is subexponentially decreasing in $n$ and $(m_n)_{n\in \N}$ %is a monotonically increasing sequence with 
satisfies $\lim_{n\to\infty}\frac{m_n}{n} = \tilde{r}>0$. It is shown in~\cite{win16} that the relative entropy of resource is asymptotically continuous. By further recalling that $\mathcal{P}^{(\tilde{r})}_n( \F_n) \subseteq \F_{m_n}$, we evaluate for all $n \in \N$
\begin{align}
D\big(\rho^{\otimes n}\| \F_n\big) \,\overset{\textnormal{DPI}}&{\geq}\, D\big(\mathcal{P}^{(\tilde{r})}_n( \rho^{\otimes n}) \| \F_{m_n}\big) \\
 \,\overset{\textnormal{\cite[Lem. 7]{win16}}}&{\geq}\, D\big(  \tau^{(m_n)}\| \F_{m_n}\big) - \eps_n \kappa_{m_n} -  (1+\eps_n)h\Big(\frac{\eps_n}{1+\eps_n}\Big) \\
 \overset{\textnormal{\cite[Cor. 3.3]{MSR2026-2}}}&{\geq}\,  D\big(  \sigma^{\otimes m_n}\| \F_{m_n}\big) - \eps_n \kappa_{m_n} -  (1+\eps_n)h\Big(\frac{\eps_n}{1+\eps_n}\Big)\nonumber \\
 &\hphantom{{\geq}\,  D\big(  \sigma^{\otimes m_n}\| \F\big) - \eps_n \kappa_{m_n}aa} - 3m_n h\Big(2 \sqrt{\frac{r_{m_n}}{m_n}}\Big) - 12 \sqrt{m_nr_{m_n}}  \log \frac{d_D}{\lambda_{\min}(\gamma)} \,,
\end{align}
where $\kappa_n := \sup_{\omega^{(n)}} D(\omega^{(n)} \| \F_n) \leq n\log \big( \frac{1}{\lambda_{\min}(\gamma)}\big) $ (e.g., see~\cite[Lemma A.8]{MSR2026-2}) with $\gamma > 0$ denoting the full-rank state that exists by assumption, and $r_n=o(n)$ denotes the number of defects. Hence, dividing both sides of the inequality above by $n$ and taking the limit $n\to \infty$ yields
\begin{align}
D^\reg(\rho \| \F) \,\geq\, \lim_{n\to \infty} \frac{m_n}{n} \frac{1}{m_n}  D\big(  \sigma^{\otimes m_n}\| \F_{m_n}\big) \,=\, \tilde{r} D^\reg(\sigma \| \F) \,.
\end{align}
After taking the supremum over all achievable rates $\tilde{r}$, this completes the proof. Finally, we remark that the statement of~\cref{prop_rates-upper-bound} may hold for more general resource monotones that are robust under almost-iid states. For instance, see~\cite{Ferrari2023} in the iid case. 
\end{proof}

\begin{corollary} \label{cor_basic-rates-relations}
Let $\rho \in \St(\cH_C)$, $\sigma\in \St(\cH_D)$, and $\F$ be as above. For $\sigma \notin \F(D)$, we find
\begin{align}
 \tilde{R}_\F(\rho \to \rho) \,=\,\begin{cases}\infty\,, & \rho \in \F(C)\\ 1\,, &\mathrm{else} \end{cases} \qquad \mathrm{and} \qquad  \tilde{R}_\F(\rho \to \sigma) \,\leq\,  \frac{1}{\tilde{R}_\F(\sigma \to \rho)} \,.
\end{align}
\end{corollary}
\begin{proof}
First, note that $ \tilde{R}_\F(\rho \to \rho) \geq 1$ for any $\rho \in \St(\cH_C)$ since the identity is a free operation, $\mathrm{id}_C \in \F(C\to C)$. If $ \rho \in \F(C)$, then $ \tilde{R}_\F(\rho \to \rho) = \infty$ since tracing and appending free states is a free operation, and this process can be repeated arbitrary many times. For $ \rho \notin \F(C)$, we have $0< D^\reg(\rho \| \F) < \infty$ and thus,~\cref{prop_rates-upper-bound} implies that $ \tilde{R}_\F(\rho \to \rho) \leq 1$ and hence, $ \tilde{R}_\F(\rho \to \rho) = 1$.
For the second relation, we can now use~\cref{prop_rates-transitivity} to evaluate
\begin{align}
1 \,=\,\tilde{R}_\F(\sigma \to \sigma) \,\overset{\textnormal{\Cshref{prop_rates-transitivity}}}{\geq}\, \tilde{R}_\F(\sigma \to \rho)\tilde{R}_\F(\rho \to \sigma)\,,
\end{align}
which proves the claim.
\end{proof}

\subsection{Special case: entanglement theory ($\F=\LOCC$)} \label{sec:LOCC_rates}
%In the following, we focus on the special case where the set of free operations is given by $\LOCC$ (local operations and classical communication). In particular, we study the interconversion from or to multiple copies of entangled Bell pairs.
As an application, we employ the concept operational asymptotic state transformation rates in the context of (bipartite) entanglement measures. More precisely, given a bipartite density operator $\rho_{AB}$ between two distant parties $A$ and $B$, one may operationally characterize the entanglement that is present in $\rho_{AB}$ by inferring the amount of entanglement initially shared between $A$ and $B$ (quantified by a number of shared Bell pairs) that is needed to construct the state $\rho_{AB}$ using only local operations and classical comcharacterizmunication ($\LOCC$) between $A$ and $B$. On the other hand, a different characterization may be given by the amount of entanglement that can be gained from $\rho_{AB}$ by $\LOCC$ operations. These two conceptual ideas lead to the notion of \emph{entanglement cost} and \emph{entanglement distillation}, respectively.

In the formulation of asymptotic entanglement cost and distillation, the conversion tasks are considered not for a single bipartite state $\rho_{AB}$, but for \emph{many} copies of it. Traditionally, these copies are assumed to form the iid state $\rho_{AB}^{\otimes n}$. However, there is no fundamental or operational justification for assuming such an iid structure. This suggests that a more natural operational formulation of entanglement conversion is obtained by replacing the iid assumption with the framework proposed here, namely by defining entanglement cost and distillation in terms of operational asymptotic transformation rates. In this section, we show that the traditional and operational notions are closely related.

Let 
\begin{align}
\Phi_2\, = \, \proj{\Phi_2}\,,\qquad \ket{\Phi_2} \,:= \, \frac{1}{\sqrt{2}}(\ket{00}+\ket{11})
\end{align} 
denote the $2\times 2$ dimensional maximally entangled state (i.e., a Bell state). For $\rho \in \St(\cH_{AB})$, the (iid) entanglement distillation and entanglement cost are defined as
\begin{align} \label{eq_E_D_and_E_C}
E_D(A:B)_{\rho} := R(\rho \to \Phi_2) \qquad \textnormal{and} \qquad E_C(A:B)_{\rho} := R(\Phi_2 \to \rho)^{-1} \, ,
\end{align}
where $ R(\rho \to \sigma) \equiv  R_{\LOCC}(\rho \to \sigma)$ denotes the iid asymptotic transformation rate introduced in~\cref{def_iid_rate}. Note that the iid asymptotic transformation rate is composable in the sense that for any $\rho \in \St(\cH_C)$, $\sigma\in \St(\cH_D)$, $\tau \in \St(\cH_E)$, it holds
\begin{align} \label{eq_triangle_rate}
 R(\rho \to \sigma )\, \geq\, R(\rho \to \tau) R(\tau \to \sigma) \, ,
\end{align}
which follows by analogous arguments as in~\cref{prop_rates-transitivity} (see also~\cite{Horo03,chit19}). This property is natural as it means that transforming $\rho \to \tau \to \sigma$ cannot be more efficient (in terms of rate) than transforming $\rho \to \sigma$ directly. Furthermore, note that since
\begin{align} \label{eq_rho_2_rho}
 R(\Phi_2 \to \Phi_2) = 1\, ,
\end{align}
we have
\begin{align}
1 
\overset{\textnormal{\Cshref{eq_rho_2_rho}}}{=} R(\Phi_2 \to \Phi_2) 
\overset{\textnormal{\Cshref{eq_triangle_rate}}}{\geq}R(\Phi_2 \to \rho) R(\rho \to \Phi_2) 
\overset{\textnormal{\Cshref{eq_E_D_and_E_C}}}{=}E_C(A:B)_{\rho}^{-1} E_D(A:B)_{\rho} \, ,
\end{align}
which reveals the well-known relation $E_D(A:B)_{\rho} \leq E_C(A:B)_{\rho}$.

Analogously, we introduce the notion of \emph{operational entanglement distillation} and \emph{cost} through operational asymptotic state transformation rates.
\begin{definition} \label{def_operational-distillation-cost}
Let $\rho \in \St(\cH_{AB})$ and $\F=\LOCC$. The operational entanglement distillation and entanglement cost are defined as
\begin{align}
\tilde{E}_D(A:B)_{\rho} \,:=\, \tilde{R}_{\LOCC}(\rho \to \Phi_2) \qquad \textnormal{and} \qquad \tilde{E}_C(A:B)_{\rho}\, :=\,  \tilde{R}_{\LOCC}(\Phi_2 \to \rho)^{-1} \, ,
\end{align}
respectively, where $\tilde{R}_{\LOCC}$ denotes the (operational) almost-iid asymptotic transformation rate introduced in~\cref{thrm_rate-tom-indep}.
\end{definition}
\begin{corollary} \label{cor:separation_Dreg}
Let $\rho \in \St(\cH_{AB})$. By~\cref{cor_basic-rates-relations,prop_rates-upper-bound}, and recalling that $D^\reg(\Phi_2\| \mathrm{SEP}) = 1$ where $\mathrm{SEP}$ denotes the set of separable states, we further have the relation
\begin{align}
\tilde{E}_D(A:B)_{\rho} \,\leq \,D^\reg(\rho \| \mathrm{SEP}) \,\leq \,  \tilde{E}_C(A:B)_{\rho}\,.
\end{align}
\end{corollary}
\begin{corollary}
Let $\rho \in \St(\cH_{AB})$ and $k\in \N$. By~\cref{prop_rates-weak-additivity}, weak additivity holds for both operational entanglement distillation and cost, i.e.,
\begin{align}
\tilde{E}_D(A:B)_{\rho^{\otimes k }} \,= \, k  \tilde{E}_D(A:B)_{\rho} \qquad \textrm{and} \qquad \tilde{E}_C(A:B)_{\rho^{\otimes k }} \,=\, k \tilde{E}_C(A:B)_{\rho} \,.
\end{align}
\end{corollary}
In the following, we will relate the notions of operational entanglement distillation and cost to the traditional iid counter parts.

\subsubsection{Subexponentially decaying approximation error}
We recall two equivalent formulations for the definition of (standard) asymptotic state transformation rates that are widely found in the literature (e.g., see~\cite{BBPSSW96,bennett_mixed-state_1996,HHT01,BrandaoGour2015,Ferrari2023,Gour_2025}). In the following, let $C:=AB$ and $D:=\bar{A}\bar{B} $ denote the input and output bipartite systems, respectively. For any $\rho \in \St(\cH_C)$ and $\sigma\in \St(\cH_D)$, the (standard) asymptotic state transformation rate can be expressed as
\begin{align} \label{eq:EC_def-traditional-forall}
{R}(\rho\to\sigma) \,&=\, \lim_{\eps \to 0^+} \sup_{m,n\in\N} \Big\{\, \frac{m}{n}\ \big| \ \inf_{\cP\in\LOCC(C^n \to D^m)} \tfrac{1}{2} \big\|\cP(\rho^{\otimes n}) - \sigma^{\otimes m} \big\|_1 \leq \eps \,\Big\} \\
&=\, \sup \Big\{ \,r >0 \ \big| \ \lim_{n\to \infty}\inf_{\cP^{(r)}_n\in\LOCC(C^n \to D^{m_n})} \tfrac{1}{2} \big\|\cP^{(r)}_n(\rho^{\otimes n}) - \sigma^{\otimes m_n} \big\|_1 = 0\,\Big\}\,,
\end{align}
where in the last line, $(m_n)_{n\in \N}$ denotes any %monotonically increasing 
sequence of positive integers such that $\lim_{n\to\infty} \frac{m_n}{n} = r$. On the other hand, the iid asymptotic transformation rate introduced in~\cref{def_iid_rate} can be rewritten as
\begin{align}
R^{\mathrm{(exp)}}(\rho\to\sigma) \,=\,  \sup \Big\{ \,r>0 \ \big| \ \inf_{\cP^{(r)}_n\in\LOCC(C^n \to D^{m_n})} \tfrac{1}{2} \big\|\cP^{(r)}_n(\rho^{\otimes n}) - \sigma^{\otimes m_n} \big\|_1\, \textrm{ is subexp.~decreas.~in } n \,\Big\} .
\end{align}
In words,~\cref{def_iid_rate} requires the approximation error $\epsilon_n := \inf_{\cP_n\in\LOCC} \tfrac{1}{2} \big\|\cP_n(\rho^{\otimes n}) - \sigma^{\otimes m_n} \big\|_1$ to decay subexponentially in $n$, whereas the traditional definition imposes no specific rate of convergence beyond asymptotic vanishing. Hence,~\cref{def_iid_rate} may a priori differ from the traditional notions found in the literature. While the inequality $R^{\mathrm{(exp)}}(\rho\to\sigma) \leq R(\rho\to\sigma) $ is immediate, establishing equality between these two notions does not appear feasible in general. Nevertheless, we prove below that they indeed coincide for the entanglement distillation and cost rates.

Within this section, we will denote the respective entanglement cost and distillation as
\begin{align} \label{eq_bar_E_D_and_E_C}
E_D(A:B)_{\rho} := R(\rho \to \Phi_2) \qquad \textnormal{and} \qquad E_C(A:B)_{\rho} := R(\Phi_2 \to \rho)^{-1} \, , \\
E^{\mathrm{(exp)}}_D(A:B)_{\rho} := R^{\mathrm{(exp)}}(\rho \to \Phi_2) \qquad \textnormal{and} \qquad E^{\mathrm{(exp)}}_C(A:B)_{\rho} := R^{\mathrm{(exp)}}(\Phi_2 \to \rho)^{-1} \,.
\end{align}

\begin{proposition} \label{prop:subexp_entanglement-cost-distillation}
Let $\rho \in \St(\cH_{AB})$. It holds that
\begin{align}
E^{\mathrm{(exp)}}_D(A:B)_{\rho} \,= \, E_D(A:B)_{\rho}  \qquad \textnormal{and} \qquad  E^{\mathrm{(exp)}}_C(A:B)_{\rho}  \,=\,  E_C(A:B)_{\rho} \,,
\end{align}
i.e., the approximation error in the definitions of (iid) entanglement cost and distillation can be assumed to be subexponentially decreasing in the number of input copies.
\end{proposition}
\begin{proof}
\noindent {\textbf{Entanglement distillation:}} The direction $E^{\mathrm{(exp)}}_D(A:B)_{\rho} \leq E_D(A:B)_{\rho} $ is immediate. For the other direction, first note that if $ E_D(A:B)_{\rho}=0$, the claim trivially follows. Now, let $r> 0 $ be achievable for $ E_D(A:B)_{\rho}$, i.e., there exists a sequence of channels $\big( \cP^{(r)}_n\big)_{n\in\N}$ with $\cP^{(r)}_n\in\LOCC(C^n \to D^{m_n})$, where $(m_n)_{n\in \N}$ %is a monotonically increasing sequence \
satisfies $\lim_{n\to\infty} \frac{m_n}{n} = r$, such that
\begin{align} \label{prop_subexp-distillation-1}
% \eps_n \,:=\, \tfrac{1}{2} \big\|\cP^{(r)}_n(\rho^{\otimes n}) - \Phi_2^{\otimes m_n} \big\|_1 \qquad \textrm{satisfies} \qquad  \lim_{n\to \infty} \eps_n \,=\,0\,.
\lim_{n\to \infty} \tfrac{1}{2} \big\|\cP^{(r)}_n(\rho^{\otimes n}) - \Phi_2^{\otimes m_n} \big\|_1 \,=\, 0 \,.
\end{align}
Now, the idea is to append a decoupling $\LOCC$ protocol $\cD$, as properly described in~\cref{sec:decoupling} below, to the distillation protocol $\cP^{(r)}$. To analyse the whole protocol, let $\ell \in \N$ be fixed, and set
\begin{align} \label{prop_subexp-distillation-2}
 \eps_{\ell} \,:=\, \tfrac{1}{2} \big\|\cP^{(r)}_{\ell}(\rho^{\otimes \ell}) - \Phi_2^{\otimes m_{\ell}} \big\|_1 \,.
\end{align}
Without loss of generality, we assume $\dim{\bar{A}} = \dim{\bar{B}} =2$, and we further denote 
\begin{align}
\sigma_{m_{\ell}}\, :=\, \cP^{(r)}_{\ell}(\rho^{\otimes \ell}) \,\in\, \St((\bar{A}\bar{B})^{\otimes m_{\ell}})\,.
\end{align}
It is convenient to apply the distillation protocol $\cP^{(r)}$ block-wise: Any $n \in \N$ may be rewritten as $n = N\ell + s$ where $N,s \in \N$ with $N = \lfloor \frac{n}{\ell} \rfloor$ and $s \leq \ell -1$. The following $\LOCC$ channel then yields
\begin{align}
\big(\cP^{(r)}_{\ell}\big)^{\otimes N} \circ \tr_s \, (\rho^{\otimes n}) \,=\,\sigma_{m_{\ell}}^{\otimes N}\,.
\end{align}
As a next step, we choose $\nu_n := 2^{-n^\alpha}$ for any $\alpha \in (0,1)$, which is a subexponentially decreasing function in $n$ with $\nu_n \in (0,1)$ for all $n\in \N$. From~\cref{sec:decoupling}, we know that there exists a ($\LOCC$) decoupling protocol $\cD$ such that
\begin{align} \label{prop_subexp-distillation-4}
%\tilde{\sigma}_{m_n} \,:=\, \cD (\sigma_{m_n}) \qquad \textrm{such that} \qquad 
\tfrac{1}{2} \norm{ \cD (\sigma^{\otimes N}_{m_\ell}) -  \Phi_2^{\otimes k_n}}_1 \,\leq\, 2\sqrt{26\nu_n} \,,
\end{align}
and according to~\cref{eq_condition_decoupling_2}, the maximal number $k_n$ of output Bell pairs is given by
\begin{align}
k_n  \,=\, \left\lfloor - H^{\nu_n}_{\max}(\bar{A}_1^{m_\ell N}|\bar{B}_1^{m_\ell N})_{\sigma^{\otimes N}} - 2 \log\left( \frac{1}{\nu_n}\right) \right\rfloor \,.
\end{align}
Then, we use the (strong) quantum asymptotic equipartition property (QAEP)~\cite{tomamichel09} and the continuity of conditional entropy~\cite{AF03,win16} to determine the rate of the overall constructed $\LOCC$ channel
\begin{align}
\tilde{\cP}_n \,:=\, \cD \circ \big(\cP^{(r)}_{\ell}\big)^{\otimes N} \circ \tr_s \,\in \, \LOCC(C^n \to D^{k_n}) \,.
\end{align}
Formally, recalling that $\ell \in \N$ is fixed and $s\leq \ell-1$, the rate is given by
\begin{align}
\lim_{n\to \infty} \frac{k_n}{n} \,&\geq\,\lim_{n\to \infty} \bigg[- \frac{H^{\nu_n}_{\max}(\bar{A}_1^{m_\ell N}|\bar{B}_1^{m_\ell N})_{\sigma^{\otimes N}}}{N\ell}\left(\frac{1}{1+\tfrac{s}{N\ell} }\right) -  \frac{2\log\left( \frac{1}{\nu_n}\right)-1}{n} \bigg]\\
\overset{\textnormal{QAEP}}&{=}\, -\frac{1}{\ell} H(\bar{A}_1^{m_\ell }|\bar{B}_1^{m_\ell })_{\sigma} \\
\overset{\textnormal{continuity, \Cshref{prop_subexp-distillation-2}}}&{\geq}\,  -\frac{1}{\ell} H(\bar{A}_1^{m_\ell }|\bar{B}_1^{m_\ell })_{\Phi_2^{\otimes m_{\ell}}} - \frac{2m_\ell}{\ell} \eps_\ell - \frac{(1+\eps_\ell)}{\ell}h\Big(\frac{\eps_\ell}{1+\eps_\ell}\Big) \\
&=\, \frac{m_\ell}{\ell}  -\bigg( \underbrace{\frac{2m_\ell}{\ell} \eps_\ell + \frac{(1+\eps_\ell)}{\ell}h\Big(\frac{\eps_\ell}{1+\eps_\ell}\Big)}_{=:\, \delta_\ell}\bigg)\,. \label{prop_subexp-distillation-3}
\end{align}
Note that in the QAEP step, one needs to be careful with the $n$-dependence of the smoothing error $\nu_n$. However, by keeping track of the bounds in the proof of the QAEP and inserting $\nu_n = 2^{-n^\alpha}$, the statement indeed continues to hold (for instance, see also the proof of~\cite[Proposition 4.2]{MSR2026}). 

Finally, recall that by assumption, we have $\lim_{\ell\to\infty} \frac{m_\ell}{\ell} = r$ and hence, $\lim_{\ell\to\infty} \delta_\ell = 0$. Furthermore, the approximation error in~\cref{prop_subexp-distillation-4} is subexponentially decreasing in $n$. On the other hand, since the bound in~\cref{prop_subexp-distillation-3} holds for any $\ell \in \N$, we may choose for any $\delta \in (0,1)$ a large enough value $\ell$ such that $ \frac{m_\ell}{\ell} \geq r - \frac{\delta}{2}$ and $\delta_\ell \leq  \frac{\delta}{2}$. Then,~\cref{prop_subexp-distillation-3} implies that the rate $ r - \delta$ is achievable for $E^{\mathrm{(exp)}}_D(A:B)_{\rho} $ and thus, $r-\delta \leq E^{\mathrm{(exp)}}_D(A:B)_{\rho} $.  After taking the supremum over all $\delta \in (0,1)$ and achievable rates $r$, we find $E_D(A:B)_{\rho}  \leq E^{\mathrm{(exp)}}_D(A:B)_{\rho}$, which proves the desired equality.

\medskip
\noindent {\textbf{Entanglement cost:}} The direction $E_C(A:B)_{\rho} \leq E^{\mathrm{(exp)}}_C(A:B)_{\rho}$ is immediate since imposing a rate of decay on the approximation error is a stronger requirement than merely asking that the error vanish. For the converse inequality, we aim to prove that $E^{\mathrm{(exp)}}_C(A:B)_{\rho} \leq E^\infty_F(A:B)_{\rho}$, where $E^\infty_F$ denotes the regularised entanglement of formation~\cite{bennett_mixed-state_1996}. The proof idea is to follow the same reasoning given in~\cite{HHT01} while keeping the error bounds explicit. Conceptually, we proceed in three steps:
\begin{enumerate}
 \item As in~\cite[Lemma 1]{HHT01}, we show that the modified cost $E^{\mathrm{(exp)}}_C$ is weakly additive under tensor powers.
 \item We adapt~\cite[Lemma 2]{HHT01} and prove the one-state estimate $E^{\mathrm{(exp)}}_C(A:B)_{\omega} \leq E_F(A:B)_{\omega}$ for any $\omega \in \St(\cH)$. Here, strong typicality is the key tool which will further guarantee a supexponentially decreasing approximation error.
 \item The one-state estimate is then applied to $\omega = \rho^{\otimes \ell}$. By weak additivity, this produces the regularized quantity $E^\infty_F(A:B)_{\rho}$.
\end{enumerate}

\noindent Recall that
\begin{align}  \label{eq:Ef-definition}
 E_F(A:B)_\rho\, := \,\min_{\{p_i,\ket{\psi_i}\}} \sum_i p_i E(\psi_i)\,, \qquad E(\psi_i) \,:=\, H\bigl(\tr_B\proj{\psi_i}\bigr)\,,
\end{align}
where $\{p_i,\ket{\psi_i}\} $ denotes any pure state ensemble such that $\rho = \sum_i p_i\proj{\psi_i}$ and $H(\cdot)$ is the von Neumann entropy, and the regularization is given by
\begin{align} \label{eq:regularised-Ef}
 E^\infty_F(A:B)_\rho \, :=\, \lim_{\ell\to\infty}\frac{1}{\ell}E_f(\rho^{\otimes \ell})  \,=\, \inf_{\ell\in\N}\frac{1}{\ell}E_f(\rho^{\otimes \ell}) \,.
\end{align}

\medskip
First, we prove weak additivity of the modified entanglement cost. For every fixed $\ell \in\N$,
\begin{align}  \label{eq:weak-additivity}
E_C^{\mathrm{(exp)}}(A^\ell:B^\ell)_{\rho^{\otimes \ell}} =\ell E_C^{\mathrm{(exp)}}(A:B)_\rho.
\end{align}
This can be seen directly from the operational definition, suppose first that a protocol converts $n$ Bell pairs into $m_n$ copies of $\rho^{\otimes \ell}$. The same output is simply $\ell m_n$ copies of $\rho$.  Conversely, a protocol producing $m_n$ copies of $\rho$ also produces $\lfloor \frac{m_n}{\ell}\rfloor$ copies of $\rho^{\otimes \ell}$: The output is grouped into blocks of size $\ell$, and the fewer than $\ell$ left-over systems are discarded. This discarding is local and cannot increase the trace distance. The number $n$ of input Bell pairs is unchanged in both constructions, so any subexponential approximation error bound remains of the same form. Therefore,
\begin{align}
 R^{\mathrm{(exp)}}(\Phi_2\to\rho^{\otimes \ell}) =\frac{1}{\ell}R^{\mathrm{(exp)}}(\Phi_2\to\rho),
\end{align}
which is equivalent to~\cref{eq:weak-additivity}.

\medskip
Now, we focus on the one-state estimate $E^{\mathrm{(exp)}}_C(A:B)_{\omega} \leq E_F(A:B)_{\omega}$. Let $\omega_{AB}$ be any finite-dimensional bipartite state and fix any number
\begin{align} \label{eq:c-choice}
 c>E_F(A:B)_\omega \,.
\end{align}
We will construct, for every sufficiently large $n$, an $\LOCC$ protocol that uses at most $\lceil nc\rceil$ Bell pairs and prepares an approximation to $\omega^{\otimes n}$ with error $2^{-\Omega(\sqrt n)}$.

Since the Hilbert spaces are finite-dimensional, the minimum in~\cref{eq:Ef-definition} is attained by a finite ensemble. Fix an optimal ensemble $\cE=\{p_i,\ket{\psi_i}\}_{i=1}^k$ such that $p_i > 0$ for all $i\in [k]$ and
\begin{align} \label{eq:optimal-ensemble}
 \omega=\sum_{i=1}^k p_i\proj{\psi_i}, \qquad E_F(A:B)_\omega =\sum_{i=1}^k p_iE(\psi_i)\,.
\end{align}

\smallskip
\noindent
\textit{First use of typicality:}
Let $s=(i_1,\ldots,i_n)\in\{1,\ldots,k\}^n$ and define the unnormalised vector
\begin{align} \label{eq:HHT-sequence}
 \ket{\psi_s}:=\bigotimes_{t=1}^n\sqrt{p_{i_t}}\ket{\psi_{i_t}}\,, \qquad p_s:=\braket{\psi_s|\psi_s}=\prod_{t=1}^n p_{i_t}\,. 
\end{align}
Furthermore, we choose
\begin{align}
 \delta_1=\delta_1(n):=n^{-1/4}
\end{align}
and, writing $N_i(s)$ for the number of occurrences of the symbol $i$ in the string $s$, we define the strongly typical set
\begin{align}  \label{eq:classical-typical-set}
 T_{\delta_1}^{(n)}  :=\left\{s\ :\ \left|\frac{N_i(s)}{n}-p_i\right|\leq\delta_1   \text{ for all }i\in [k]\right\}.
\end{align}
As in~\cite{HHT01}, we want to approximate $\omega^{\otimes n}$ by the state
\begin{align} \label{eq:HHT-typical-state}
 \omega_{T_{\delta_1}^{(n)}} :=\frac{1}{p_{T_{\delta_1}^{(n)}}} \sum_{s\in T_{\delta_1}^{(n)}}\proj{\psi_s}\,, 
\end{align}
where $p_{T_{\delta_1}^{(n)}} = \sum_{s\in T_{\delta_1}^{(n)}} p_s$. Hoeffding's inequality~\cite{Hoeffding1963} and the union bound imply
\begin{align} \label{eq:first-typicality-error}
 \eps_{1,n} \,:=\, 1-p_{T_{\delta_1}^{(n)}} \,=\, \sum_{s\notin T_{\delta_1}^{(n)}}p_s \,\leq\, 2k\exp(-2n\delta_1^2) \, =\, 2k\exp(-2\sqrt n)\,.
\end{align}
Then, we have $\omega^{\otimes n}=p_{T_{\delta_1}^{(n)}} \omega_{T_{\delta_1}^{(n)}}+(1-p_{T_{\delta_1}^{(n)}})\tau_n$ for some state $\tau_n$, and for all $n\in \N$,
\begin{align} \label{eq:first-truncation-distance}
 \tfrac{1}{2}\Big\|\omega^{\otimes n} -\omega_{T_{\delta_1}^{(n)}}\Big\|_1\, \leq\, \eps_{1,n}.
\end{align}
Moreover, note that conditioned on a typical word $s$, the normalised target vector is
\begin{align} \label{eq:conditional-product}
 \ket{\bar\psi_s} \,:=\, p_s^{-1/2}\ket{\psi_s} \, \simeq \, \bigotimes_{i=1}^k\ket{\psi_i}^{\otimes N_i(s)}\,,
\end{align}
where $\simeq$ denotes equality up to a common local permutation of the tensor factors. Thus the first typicality step controls how often each pure state $\ket{\psi_i}$ occurs. It does not yet specify how many Bell pairs are needed to prepare those pure states; that is the role of the second typicality step.

\smallskip
\noindent
\textit{Second use of typicality:}
We next repeat the argument of~\cite{bennett96} to find the number of Bell pairs needed to prepare the pure states $\ket{\psi_i}^{\otimes N_i(s)}$ and to bound the corresponding error. For each $i \in [k]$, we use the Schmidt decomposition to write
\begin{align} \label{eq:Schmidt-decomposition}
 \ket{\psi_i} \, =\,\sum_{j=1}^{d_i}\sqrt{q_j^{(i)}} \ket{e_j^{(i)}}_A\ket{f_j^{(i)}}_B\,, \qquad E(\psi_i)\,=\,H(q^{(i)})\,,
\end{align}
where $\{ \ket{e_j^{(i)}}_A \}_j$ and $\{\ket{f_j^{(i)}}_B \}_j$ denote orthonormal sets of vectors, and we assume $q_j^{(i)} >0$ for all $j \in [d_i]$. Consider $N$ copies of $\ket{\psi_i}$.  A Schmidt string
$j^N=(j_1,\ldots,j_N)$ occurs with squared Schmidt coefficient
\begin{align}
 q_{j^N}^{(i)}:=\prod_{t=1}^Nq_{j_t}^{(i)}.
\end{align}
Next, we choose
\begin{align}
 \delta_2=\delta_2(n):=n^{-1/4}
\end{align}
and define another strongly typical set
\begin{align} \label{eq:Schmidt-typicality}
T_{\delta_2}^{(N)}(q^{(i)}) :=\left\{j^N\ :\  \left|\frac{N_j(j^N)}{N}-q_j^{(i)}\right|\leq\delta_2  \text{ for all }j \in [d_i]\right\}\,.
\end{align}
Again, by Hoeffding's inequality, we find that
\begin{align}  \label{eq:second-typicality-error}
 \beta_{i,N} \, :=\,\sum_{j^N\notin T_{\delta_2}^{(N)}(q^{(i)})}q_{j^N}^{(i)} \, \leq \, 2d_i\exp(-2N\delta_2^2)\,.
\end{align}
Then, we want to approximate the state $\ket{\psi_i}^{\otimes N}$ by a coherent truncation to typical Schmidt strings,
\begin{align} \label{eq:coherent-truncation}
 \ket{\psi_{i,N}'} \,:=\, \frac{1}{\sqrt{1-\beta_{i,N}}}  \sum_{j^N\in T_{\delta_2}^{(N)}(q^{(i)})} \sqrt{q_{j^N}^{(i)}}  \ket{e_{j^N}^{(i)}}_A\ket{f_{j^N}^{(i)}}_B\,,
\end{align}
and its squared overlap with the original state is given by
\begin{align} \label{eq:truncation-overlap}
 \left|\braket{\psi_i^{\otimes N}|\psi_{i,N}'}\right|^2 \, =\, 1-\beta_{i,N}\,.
\end{align}
We next bound the Schmidt rank of $\ket{\psi_{i,N}'}$.  For any \smash{$j^N\in T_{\delta_2}^{(N)}(q^{(i)})$} and \smash{$ a_i:=\sum_{j=1}^{d_i}\log_2\frac{1}{q_j^{(i)}}<\infty $}, we have
\begin{align} \label{eq:typical-string-probability}
 -\frac1N\log_2q_{j^N}^{(i)} \,=\, \sum_{j=1}^{d_i}\frac{N_j(j^N)}{N}  \log_2\frac1{q_j^{(i)}} \, \leq \, H(q^{(i)})+a_i\delta_2 \,=\,E(\psi_i)+a_i\delta_2\,.
\end{align}
Therefore, every string \smash{$j^N\in T_{\delta_2}^{(N)}(q^{(i)})$} has probability at least $2^{-N(E(\psi_i)+a_i\delta_2)}$, and since the probabilities of all strings sum to one, we find
\begin{align} \label{eq:typical-Schmidt-rank}
 \left|T_{\delta_2}^{(N)}(q^{(i)})\right| \,\leq \, 2^{N(E(\psi_i)+a_i\delta_2)}\,.
\end{align}
Consequently, the Schmidt rank of $\ket{\psi_{i,N}'}$ is bounded by~\cref{eq:typical-Schmidt-rank}. Since a maximally entangled state of Schmidt rank at least $D$ can be converted by deterministic $\LOCC$ into any pure state of Schmidt rank at most $D$, we thus find that $\ket{\psi_{i,N}'}$ can be prepared \emph{exactly} from $ \left\lceil N\bigl(E(\psi_i)+a_i\delta_2\bigr)\right\rceil $ Bell pairs.

\smallskip
\noindent
Now, fix $s\in T_{\delta_1}^{(n)}$.  By applying the above construction to each factor in~\cref{eq:conditional-product}, we obtain
\begin{align}   \label{eq:cost-psi-bar}
 \ket{\bar\psi_s'} \, :=\, \bigotimes_{i=1}^k\ket{\psi_{i,N_i(s)}'}\,.
\end{align}
The number $m(s)$ of Bell pairs needed to prepare this state satisfies
\begin{align}  \label{eq:uniform-cost}
 m(s)\, \leq\, \sum_{i=1}^k  \left\lceil N_i(s)\bigl(E(\psi_i)+a_i\delta_2\bigr)\right\rceil \,\overset{\textnormal{\Cshref{eq:classical-typical-set}}}&{\leq} \, nE_F(A:B)_\omega +n\delta_1\sum_{i=1}^kE(\psi_i) +n\delta_2\max_i a_i+k \\
 &\leq\, nE_F(A:B)_\omega+Cn^{3/4}+k,
\end{align}
where $C<\infty$ depends only on the fixed optimal ensemble $\{p_i,\ket{\psi_i}\} $. Because $c>E_F(A:B)_\omega$, the last expression is at most $\lceil nc\rceil$ for all sufficiently large $n$. Thus the same deterministic supply of $\lceil nc\rceil$ Bell pairs works for every typical word $s \in T_{\delta_1}^{(n)}$.

\smallskip
\noindent
\textit{Full $\LOCC$ protocol:} Consider the state
\begin{align} \label{eq:HHT-prepared-state}
 \omega_{T_{\delta_1}^{(n)}}'\, :=\, \frac{1}{p_{T_{\delta_1}^{(n)}}}  \sum_{s\in T_{\delta_1}^{(n)}} p_s \proj{\bar\psi_s'}
\end{align}
where $\ket{\bar\psi_s'}$ is defined in~\cref{eq:cost-psi-bar}. This state can be prepared by one-way $\LOCC$: Alice samples $s$ with probability \smash{$p_s/p_{T_{\delta_1}^{(n)}}$}, communicates $s$ to Bob, and the parties prepare $\ket{\bar\psi_s'}$ using the construction above.

\smallskip
\noindent
Finally, we bound the error of this preparation. For every typical $s\in T_{\delta_1}^{(n)}$ and all sufficiently large $n$, we have $N_i(s)\geq \frac{np_i}{2}$.~\cref{eq:second-typicality-error,eq:truncation-overlap} therefore imply for all typical $s$ that
\begin{align} \label{eq:conditional-error}
 \frac12\norm{\proj{\bar\psi_s}-\proj{\bar\psi_s'}}_1 \,=\, \sqrt{1-\prod_{i=1}^k(1-\beta_{i,N_i(s)})} \,\leq\, \sqrt{\sum_{i=1}^k\beta_{i,N_i(s)}}\, \leq \, K_2\exp(-L_2\sqrt n)
\end{align}
for constants $K_2,L_2>0$. By the triangle inequality, we thus find
\begin{align}
 \frac12\Big\| \omega_{T_{\delta_1}^{(n)}} -\omega_{T_{\delta_1}^{(n)}}' \Big\|_1  \leq K_2\exp(-L_2\sqrt n)
\end{align}
for sufficiently large $n$. Combining this with~\cref{def_subexp-equivalence,eq:first-typicality-error,eq:first-truncation-distance}, we hence obtain an $\LOCC$ channel
$\mathcal D_n$ such that
\begin{align}
 \frac12\norm{
 \mathcal D_n(\Phi_2^{\otimes\lceil nc\rceil})
 -\omega^{\otimes n}}_1
 \leq K2^{-L\sqrt n}
 \label{eq:output-indexed-error}
\end{align}
for suitable constants $K,L>0$ and for all $n\in \N$.

Note that at this point, $n$ is the number of output copies, while the definition of $R^{\mathrm{(exp)}}(\Phi_2\to\omega)$ instead indexes protocols by the number of input Bell pairs. To translate between the two conventions, let $N$ be the number of available Bell pairs and set $n_N:= \lfloor\frac{N}{c}\rfloor$. Since $\lceil n_N c\rceil\leq N$, the parties discard the unused Bell pairs and apply $\mathcal D_{n_N}$.  Denote this overall $N$-input channel by $\cP_N$. Then
\begin{align}  \label{eq:input-indexed-error}
 \lim_{N\to\infty}\frac{n_N}{N}\,=\,\frac1c\,, \qquad \frac12\norm{\cP_N(\Phi_2^{\otimes N}) -\omega^{\otimes n_N}}_1 \,\leq\, \tilde{K}2^{-\tilde{L}\sqrt N}
\end{align}
for suitable $\tilde{K},\tilde{L}>0$.  Hence $1/c$ is an achievable formation rate and $ E^{\mathrm{(exp)}}_C(A:B)_{\omega} \leq c$. Since the choice of $c>  E_F(A:B)_{\omega}$ was arbitrary,
\begin{align}  \label{eq:one-state-bound}
 E^{\mathrm{(exp)}}_C(A:B)_{\omega} \,\leq\, E_F(A:B)_{\omega}\,.
\end{align}

\medskip
\noindent
\textit{Regularisation step:}
Apply~\cref{eq:one-state-bound} to the state $\omega=\rho^{\otimes \ell}$ for any fixed $\ell\in\N$. Weak additivity gives
\begin{align}
 \ell  E^{\mathrm{(exp)}}_C(A:B)_{\rho} \,\overset{\textnormal{\Cshref{eq:weak-additivity}}}{=}\, E^{\mathrm{(exp)}}_C(A^\ell:B^\ell)_{\rho^{\otimes \ell}} \,\overset{\textnormal{\Cshref{eq:one-state-bound}}}{\leq}\, E_F(A^\ell:B^\ell)_{\rho^{\otimes \ell}}\,.
\end{align}
Therefore
\begin{align}  \label{eq:modified-cost-upper-bound}
 E^{\mathrm{(exp)}}_C(A:B)_{\rho} \, \leq\, \inf_{\ell \in\N}\frac{1}{\ell}E_F(A^\ell:B^\ell)_{\rho^{\otimes \ell}} \,=\, E_F^\infty(A:B)_{\rho}\,.
\end{align}
We remark that the constants in the subexponential estimate may depend on the fixed integer $\ell$, which is however harmless because $\ell$ is chosen before the asymptotic limit is taken.

\medskip
\noindent
Finally, recall that~\cite{HHT01} proved that for finite-dimensional bipartite states, the traditional entanglement cost is
\begin{align}  \label{eq:HHT-theorem}
  E_C(A:B)_{\rho} =E_F^\infty(A:B)_{\rho}\,.
\end{align}
Hence, we find
\begin{align}
  E_C(A:B)_{\rho} \, \leq \,E^{\mathrm{(exp)}}_C(A:B)_{\rho}  \,\overset{\textnormal{\Cshref{eq:modified-cost-upper-bound}}}{\leq}\, E_F^\infty(A:B)_{\rho} \, \overset{\textnormal{\Cshref{eq:HHT-theorem}}}{=}\,  E_C(A:B)_{\rho}\,,
\end{align}
which proves the claim.
\end{proof}

We remark that finite dimensionality guarantees both a finite optimal ensemble $\cE=\{p_i,\ket{\psi_i}\}_{i=1}^k$ and finite Schmidt alphabets $\{1,\ldots,d_i\}$. Consequently, both applications of Hoeffding's inequality involve finitely many bounded random variables, and the constants $a_i$ in~\cref{eq:typical-string-probability} are finite. These are the specific ingredients that need not be available in infinite dimensions. In fact, while~\cref{eq:HHT-theorem} for the traditional entanglement cost has been established for infinite dimensional systems in~\cite{Yamasaki2025}, the analogous equality with the subexponential entanglement cost as in~\cref{prop:subexp_entanglement-cost-distillation} need not hold. % A corresponding counterexample can be found in~\cref{lem_entanglement-cost-counterex}. 

\subsubsection{Entanglement distillation is operational}
In this section, we combine most of our results, in particular~\cref{thm_almost-iid-iid-rates-relation,prop:subexp_entanglement-cost-distillation}, to prove that the operational (almost-iid) and traditional (iid) notions of entanglement distillation coincide. This is achieved by an additional use of decoupling, as detailed below.
\begin{corollary} \label{cor_ED-equal}
Let $\rho \in \St(\cH_{AB})$ and $\F = \LOCC$. Then, the almost-iid asymptotic entanglement distillation satisfies
\begin{align}
E_D(A:B)_{\rho} \,=\,  R_\LOCC(\rho \to \Phi_2)\,=\, \tilde{R}_\LOCC(\rho \to \Phi_2) \,=\,\tilde{E}_D(A:B)_{\rho}\,.
\end{align}
\end{corollary}
\begin{proof}
Consider the following rate definitions,
\begin{align}
 R(\rho \to \Phi_2)  \,\overset{\textnormal{\Cshref{def_iid_rate}}}&{=}\,  \sup\Big\{\, r >0\ \big| \ \exists\ {\pmb{\mathcal{P}}}^{(r)} \in \LOCC \,:\ {\pmb{\mathcal{P}}}^{(r)}\big( \pmb{\rho}^{\mathrm{iid}}\big) \, \in \, B\big(\pmb{\Phi}_2^{\mathrm{iid}}\big) \Big\}  \\
T(\rho \to \Phi_2) \, &:=\,  \sup\Big\{\, r > 0\ \big| \ \exists\ {\pmb{\mathcal{P}}}^{(r)} \in \LOCC \,:\ {\pmb{\mathcal{P}}}^{(r)}\big( \pmb{\rho}^{\mathrm{iid}}\big) \, \in \, B\big(\mathcal{A}(\Phi_2)\big) \Big\} \,.
\end{align}
We have the relations $ R(\rho \to \Phi_2)  \leq \tilde{R}(\rho \to \Phi_2) $, which is the content of~\cref{thm_almost-iid-iid-rates-relation}, and $  \tilde{R}(\rho \to \Phi_2) \leq  T(\rho \to \Phi_2)$ which follows immediately since $\pmb{\Phi}_2^{\mathrm{iid}} \in \mathcal{A}(\Phi_2) $. Next, we claim that $T(\rho \to \Phi_2) \leq R(\rho \to \Phi_2)$. To see this, first note that for $T(\rho \to \Phi_2)=0$, the claim trivially follows. Now, let $r > 0$ be achievable for $ T(\rho \to \Phi_2)$, i.e., there exists a sequence of free operations ${\pmb{\mathcal{P}}}^{(r)} \in \LOCC$, with $\cP^{(r)}_n \in \LOCC(C^n \to D^{m_n})$ for each $n\in \N$, such that 
\begin{align} \label{eq:ED-equal-it1}
\pmb{\mathcal{P}}^{(r)}\big( \pmb{\rho}^{\mathrm{iid}}\big)  \in  B\big(\mathcal{A}(\Phi_2)\big)\,,
\end{align}
where the sequence $(m_n)_{n\in \N}$ %is a monotonically increasing sequence 
satisfies $\lim_{n\to\infty}\frac{m_n}{n} = r$. 
%Since $r=0$ is always achievable via trace operations, we can assume $r>0$ without loss of generality. 
The idea is to append a decoupling $\LOCC$ protocol $\cD$, as properly described in~\cref{sec:decoupling}, to the distillation protocol $\cP^{(r)}$. More precisely, we assume $D = \bar{A}\bar{B}$ with $\dim{\bar{A}} = \dim{\bar{B}} =2$ without loss of generality, and due to~\cref{eq:ED-equal-it1}, there exists an almost-iid sequence $\pmb{\tau} \in \cA(\Phi_2)$ such that for any $n \in \N$,
\begin{align}
\frac{1}{2}\norm{ \cP^{(r)}_{n}(\rho^{\otimes n}) - \tau^{(m_n)}}_1 \,\leq \, \eps_n
\end{align}
with $\eps_n$ being a subexponentially decreasing function. Then, we choose $\nu_n := 2^{-n^\alpha}$ for any $\alpha \in (0,1)$, which is also a subexponentially decreasing function in $n$ with $\nu_n \in (0,1)$. From~\cref{sec:decoupling}, we know that there exists a ($\LOCC$) decoupling protocol $\pmb{\cD}$ (depending on $\pmb{\tau}$) such that
\begin{align} \label{eq:ED-equal-it2}
\frac{1}{2} \norm{ \cD_m \big( \tau^{(m_n)}\big) -  \Phi_2^{\otimes k_n}}_1 \,\leq\, 2\sqrt{26\nu_n} \,,
\end{align}
and according to~\cref{eq_condition_decoupling_2}, the maximal number $k_n$ of output Bell pairs is given by
\begin{align}
k_n  \,=\, \left\lfloor - H^{\nu_n}_{\max}(\bar{A}_1^{m_n}|\bar{B}_1^{m_n})_{ \tau^{(m_n)}} - 2 \log\left( \frac{1}{\nu_n}\right) \right\rfloor \,.
\end{align}
By exploiting the almost-iid structure of $\pmb{\tau}$, and recalling that %$(m_n)_{n\in \N}$ is monotonically increasing with 
$\lim_{n\to\infty}\frac{m_n}{n} =r >0$, we evaluate
\begin{align}
\lim_{n\to \infty} \frac{k_n}{n} \,&=\,\lim_{n\to \infty} \bigg[- \frac{H^{\nu_n}_{\max}(\bar{A}_1^{m_n}|\bar{B}_1^{m_n})_{ \tau^{(m_n)}}}{m_n}\Big(\frac{m_n}{n}\Big) -  \frac{2\log\left( \frac{1}{\nu_n}\right)}{n} \bigg]\\
\overset{\textnormal{\cite[Prop.~4.2]{MSR2026}}}&{=}\, - H(\bar{A} |\bar{B})_{\Phi_2}r \\
 \,&=\, r\,. \label{eq:ED-equal-it3}
\end{align}
Note that in the second (QAEP) step, one needs to be careful with the $n$-dependence of the smoothing parameter $\nu_n$. However, by keeping track of the bounds in the proof of~\cite[Proposition 4.2]{MSR2026} and inserting $\nu_n = 2^{-n^\alpha}$, the statement indeed continues to hold. Therefore, it holds
\begin{align}
\frac{1}{2}\norm{ \cD_m \circ \cP^{(r)}_{n}(\rho^{\otimes n}) -  \Phi_2^{\otimes k_n}}_1 \,&\leq \, \frac{1}{2}\norm{ \cD_m \circ \cP^{(r)}_{n}(\rho^{\otimes n}) - \cD_m( \tau^{(m_n)})}_1 + \frac{1}{2}\norm{\cD_m( \tau^{(m_n)}) -  \Phi_2^{\otimes k_n}}_1 \\
&\leq\, \eps_n + 2\sqrt{26\nu_n}\,,
\end{align}
and hence, the rate $r$ is achievable for $R(\rho\to \Phi_2)$, meaning that $r \leq R(\rho\to \Phi_2)$ After taking the supremum over all achievable rates $r >0 $, we find $ T(\rho\to \Phi_2) \leq  R(\rho\to \Phi_2)$ and hence,
\begin{align}
 R(\rho\to \Phi_2) \,\leq\, \tilde{R}(\rho\to \Phi_2) \,\leq\,  T(\rho\to \Phi_2) \,\leq\,   R(\rho\to \Phi_2)\,,
\end{align}
showing that all of these quantities are equal. This proves the desired claim.
\end{proof}

As a consequence of the above, we can show that for pure states, the distillable entanglement and entanglement cost coincide and the iid or more operational almost-iid notions are equivalent.
\begin{corollary}[Pure states] \label{cor_pure-states}
Let $\ket{\psi}_{AB} \in \cH_{AB}$. Then,
\begin{align}
 \tilde{E}_D(A:B)_{\psi} =  E_D(A:B)_{\psi} =  \tilde{E}_C(A:B)_{\psi} =  E_C(A:B)_{\psi} = H(A)_{\psi} \, .
\end{align}
\end{corollary}
\begin{proof}
Note that because $\ket{\psi}_{AB}$ is pure, we have
\begin{align}
H(A)_{\psi} \overset{\textnormal{\cite{bennett96}}}{=} E_D(A:B)_{\psi} \overset{\textnormal{\Cshref{cor_ED-equal}}}{=} \tilde{E}_D(A:B)_{\psi} \overset{\textnormal{\Cshref{cor_basic-rates-relations}}}{\leq} \tilde{E}_C(A:B)_{\psi} \overset{\textnormal{\Cshref{thm_almost-iid-iid-rates-relation}}}{\leq}  E_C(A:B)_{\psi}  \overset{\textnormal{\cite{bennett96}}}{=} H(A)_{\psi} \, .
\end{align}
\end{proof}

\begin{corollary}[Operational transformations with pure endpoints] \label{cor:moving-pure-endpoints}
Let $\rho \in\St(\cH_{AB})$, and let $\ket{\psi}_{AB},\,\ket{\phi}_{AB} \in \cH_{AB}$ be entangled pure states. Then,
\begin{align} \label{eq:moving-pure-endpoints}
\tilde{R}_{\LOCC}(\rho\to\phi) \,=\, \frac{E_D(A:B)_\rho}{H(A)_\phi}\,, \qquad \tilde{R}_{\LOCC}(\psi\to \rho) \,=\, \frac{H(A)_\psi}{\tilde{E}_C(A:B)_\rho}\,.
\end{align}
In particular, it holds
\begin{align}
\tilde{R}_{\LOCC}(\psi\to\phi) \,=\, \frac{H(A)_\psi}{H(A)_\phi} \,=\, R_{\LOCC}(\psi\to\phi) \,.
\end{align}
\end{corollary}

\begin{proof}
By~\cref{prop_rates-transitivity}, we find
\begin{align}
\tilde{R}_{\LOCC}(\rho\to\phi) \,\geq\, \tilde{R}_{\LOCC}(\rho\to\Phi_2)\tilde{R}_{\LOCC}(\Phi_2 \to\phi) \,\overset{\textnormal{\Cshref{cor_ED-equal}}}{=}\, \frac{E_D(A:B)_{\rho}}{\tilde{E}_C(A:B)_{\phi}} \,\overset{\textnormal{\Cshref{cor_pure-states}}}{=}\,  \frac{E_D(A:B)_{\rho}}{H(A)_{\phi}} \,.
\end{align}
For the other direction, we similarly evaluate
\begin{align}
\tilde{R}_{\LOCC}(\rho\to\phi) \,\overset{\textnormal{\Cshref{prop_rates-transitivity}}}{\leq}\,  \frac{\tilde{R}_{\LOCC}(\rho\to\Phi_2)}{\tilde{R}_{\LOCC}(\phi \to\Phi_2) } \,\overset{\textnormal{\Cshref{cor_ED-equal,cor_pure-states}}}{=}\,  \frac{E_D(A:B)_{\rho}}{H(A)_{\phi}} \,,
\end{align}
which thus shows equality. Analogously, we have
\begin{align}
\tilde{R}_{\LOCC}(\psi\to\rho)  \,\overset{\textnormal{\Cshref{prop_rates-transitivity}}}&{\geq}\, \tilde{R}_{\LOCC}(\psi\to\Phi_2)\tilde{R}_{\LOCC}(\Phi_2\to\rho) \,\overset{\textnormal{\Cshref{cor_pure-states}}}{=}\, \frac{H(A)_{\psi}}{\tilde{E}_C(A:B)_{\rho}} \,, \\
\tilde{R}_{\LOCC}(\psi\to\rho) \,\overset{\textnormal{\Cshref{prop_rates-transitivity}}}&{\leq}\,  \frac{\tilde{R}_{\LOCC}(\Phi_2 \to \rho)}{\tilde{R}_{\LOCC}(\Phi_2 \to \psi) } \,\overset{\textnormal{\Cshref{cor_pure-states}}}{=}\, \frac{H(A)_{\psi}}{\tilde{E}_C(A:B)_{\rho}} \,,
\end{align}
which proves the second equality. Finally, the relation $\tilde{R}_{\LOCC}(\psi\to\phi) = \frac{H(A)_\psi}{H(A)_\phi}$ from the above due to~\cref{cor_pure-states}.
\end{proof}

Finally, we remark that there exist bipartite entangled mixed states for which operational entanglement distillation and cost differ.
\begin{corollary}
There exists a state $\rho \in \St(\cH_{AB})$ such that
\begin{align}
 \tilde{E}_D(A:B)_{\rho} \,<\,  \tilde{E}_C(A:B)_{\rho}\,.
\end{align}
\end{corollary}
\begin{proof}
The existence of bound entanglement shows that the traditional entanglement distillation is not faithful, i.e., there exists an entangled state $\rho \in \St(\cH_{AB})$ such that $E_D(A:B)_{\rho} = 0$~\cite{HHH98}. By~\cref{cor_ED-equal}, the operational distillable entanglement is likewise not faithful, $\tilde{E}_D(A:B)_{\rho} = 0$. Combining this with~\cref{cor:separation_Dreg} and the fact that the regularized relative entropy of entanglement \emph{is} faithful~\cite{brandao_Stein_10}, the claim follows.
\end{proof}

%%%%%%%%%%%%%%%%%%%%%%%%%%%%%%%%%%%%%%%%%%%%%%%%%%%%%%%%%%%%%%%%%%%%%%%%%%%%%%%%%%%%%%%%%%%%%%%%%%%%%%%%%%%%%%%%%%%%%%%
\subsubsection{Decoupling protocol: achieving iid Bell pairs} \label{sec:decoupling}
In this section, we present a $\LOCC$ protocol that converts a bipartite state $\rho_{A_1^n B_1^n} \in \St(\cH_{AB}^{\otimes n})$ into iid Bell pairs $\ket{\Phi_{d_A}}^{\otimes k}$, where the number $k$ of Bell pairs is determined by the (entanglement) structure of $\rho_{A_1^n B_1^n} $.
% This shows that in the definitions of entanglement cost and distillable entanglement, it makes no difference if we assume iid or almost-iid Bell pairs. \RR{Double check -- as discussed}
The $\LOCC$ protocol achieving this task is depicted in~\cref{fig_Bell_protocol}. It works in three steps:
\begin{enumerate}[(1)]
    \item We start with a bipartite state $\rho_{A_1^n B_1^n} \in \St(\cH_{AB}^{\otimes n})$. Let $\rho_{A_1^n B_1^n E} = \proj{\Psi}_{A_1^n B_1^n E}$ be a purification of $\rho_{A_1^n B_1^n}$.
    \item Then, we apply a (Haar) random unitary $U_{A_1^n}$ on the $A_1^n$ system. The output is split into the first $k$ subsystems which we now denote as the $C$ system, and the remaining $n-k$ subsystems which we measure in the computational (product) basis. More precisely, the measurement symbol in~\cref{fig_Bell_protocol} performs the operation
    \begin{align}
      M_{A_{k+1}^n}: \ Z \mapsto \sum_x \bra{x} Z \ket{x}\, \proj{x}_{A_{k+1}^n} \, .
    \end{align}
    We implement the measurement as an isometry $V_{XS|A_{k+1}^n }$ (via the Stinespring dilation)\footnote{Formally, the isometry implementing the measurement is given by $V_{XS|A_{k+1}^n} = \sum_x \ket{x}_X \otimes \ket{x}_S \otimes \bra{x}_{A_{k+1}^n}$.} by means of an additional auxiliary system $S$. The output state is denoted by $\omega_{C X B_1^n S  E}$.
     \item The classical measurement outcome $X$ is communicated from the $A_1^n$ system to the $B_1^n$ system, and a ``retrieval" map $\cR$ is then jointly applied on the $B_1^n$ and $X$ systems. The retrieval map $\cR$ is chosen such that it produces $k$ (maximally entangled) Bell pairs between the output systems $C$ and $D$. The output state is denoted by $\chi_{C D S E}$. 
\end{enumerate}
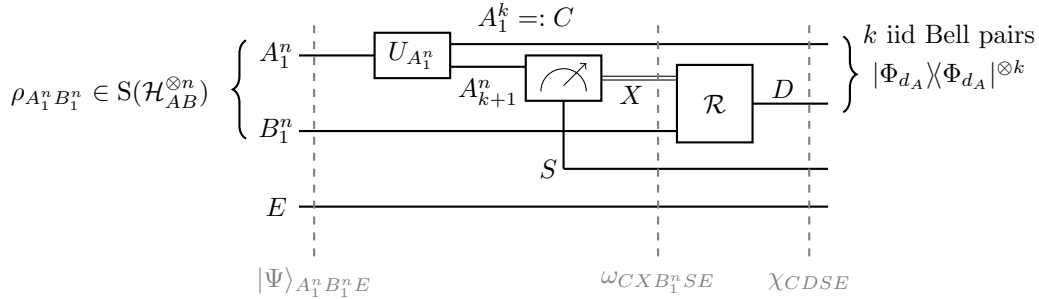
\begin{figure}[htb!]
    \centering
    \begin{tikzpicture}[thick,scale=1]
    \def\y{1}
    \draw (0,0) -- (1,0);
    \draw [draw=black] (1,0.3) rectangle (2,-0.3);
    \node at (1.5,0) {$U_{A_1^n}$};
    \draw (2,-0.15) -- (3,-0.15);
    
    \draw [draw=black] (3,0.0) rectangle (4,-0.6);
    \draw[thin] (4,-0.325) -- (5,-0.325);
    \draw[thin] (4,-0.275) -- (5,-0.275);
    \draw[thin] (3.8,-0.475) arc (0:180:0.3);
    \draw[thin,->] (3.5,-0.4) -- (3.8,-0.1);
 
    \draw (2,0.15) -- (7,0.15);

    \draw (0,-\y) -- (5,-\y);
    \draw [draw=black] (5,-0.275+0.15) rectangle (6,-\y-0.15);
    \draw (6,-0.275/2-0.5*\y) -- (7,-0.275/2-0.5*\y);
    \node at (5.5,-0.275/2-0.5*\y) {$\cR$};
    
    \draw (3.5,-0.6) -- (3.5,-1.5*\y);
    \draw (3.5,-1.5*\y) -- (7,-1.5*\y);    
     
    \draw (0,-2*\y) -- (7,-2*\y);     

    % notation
    \node at (-0.3,0) {$A_1^n$};
    \node at (-0.3,-\y) {$B_1^n$}; 
    \node at (-0.3,-2*\y) {$E$};  
    \node at (3,0.5) {$A_1^k =:C$};
    \node at (2.5,-0.5) {$A_{k+1}^n$}; 
    \node at (3.3,-1.5*\y) {$S$};
    \node at (6.4,-0.275/2-0.5*\y+0.2) {$D$};  
    \node at (4.4,-0.52) {$X$};      

    \draw[dashed,gray] (0.2,0.4)--(0.2,-2*\y-0.7);
    \draw[dashed,gray] (4.75,0.4)--(4.75,-2*\y-0.7);  
    \draw[dashed,gray] (6.75,0.4)--(6.75,-2*\y-0.7); 

    \node[gray] at (0.2,-2*\y-0.7-0.3) {$\ket{\Psi}_{A_1^n B_1^n E}$}; 
    \node[gray] at (4.75,-2*\y-0.7-0.3) {$\omega_{C X B_1^n S  E}$};   
    \node[gray] at (6.75,-2*\y-0.7-0.3) {$\chi_{C D S E}$};

    \draw [rotate=180,thick, decorate,decoration={brace,amplitude=5pt}] (0.7,+1.1) -- (0.7,-1.3*\y+1.1);
    \node at (-2.5,-0.5*\y) {$\rho_{A_1^n B_1^n} \in \St(\cH_{AB}^{\otimes n})$};
    % \node at (-3,-0.5*\y+0.5) {$n$ almost-iid Bell pairs};
    \draw [thick, decorate,decoration={brace,amplitude=5pt}] (7.2,0+0.25) -- (7.2,-\y+0.25);
    \node at (8.6,-0.5*\y+0.25) {$\proj{\Phi_{d_A}}^{\otimes k}$};
     \node at (8.6,-0.5*\y+0.25+0.5) {$k$ iid Bell pairs};
     
    \end{tikzpicture}
    \caption{An $\LOCC$ protocol that converts an $n$-fold bipartite state into $k$ iid Bell pairs.} %We show in the text that $k$ is asymptotically equal to $n$, i.e.~$\lim_{n \to \infty} \frac{k}{n}=1$. Therefore, the protocol is asymptotically optimal.}
    \label{fig_Bell_protocol}
\end{figure}

\noindent In order to analyse the above protocol, we make use of the decoupling theorem~\cite[Theorem~3.1]{DBWR14}. Let $\eps \in (0,1)$. The theorem implies that if
\begin{align} \label{eq_condition_decoupling_1}
    H^{\eps}_{\min}(A_1^n|E)_{\rho} + H^{\eps}_{\min}(A_1^n|C X)_{\tau} \geq 2\log\left(\frac{1}{\eps}\right),
\end{align}
where $\tau_{A_1^n C X}$ denotes the Choi–Jamio{\l}kowski representation of the map $\cT_{A_1^n \to C X } := \id_{A_1^k} \otimes M_{A_{k+1}^n}$, then there exists a unitary $U_{A_1^n}$ such that
\begin{align} \label{eq_decoupling_thrm}
    \norm{\cT_{A_1^n \to C X }\left(U_{A_1^n} \rho_{A_1^n E} U^\dagger_{A_1^n} \right) - \tau_{C X}\otimes \rho_E}_1 = \norm{\omega_{CXE} - \frac{\id_{CX}}{d_C d_X}\otimes \omega_E}_1 \leq 13\eps\,.
\end{align}
Note that in the first equality above, we used that $\rho_E = \omega_E$, and the explicit form of the representation
\begin{align} \label{eq_CJ-rep}
    \tau_{A_1^n C X} = \proj{\psi}_{A_1^k, C} \otimes \frac{1}{d_X} \sum_x \proj{x}_{A_{k+1}^n}\otimes \proj{x}_X \, ,
\end{align}
where \smash{$\proj{\psi}_{A_1^k, C} \equiv \proj{\Phi_{d_A}}^{\otimes k}$} denotes a %$d_A \times d_A$ 
maximally entangled state between $k$ pairs of the input \smash{$A_1^k$} and output $C$ systems, implying that $\tau_{C X} = \frac{\id_{CX}}{d_C d_X}$. 
Because the trace distance is contractive under trace-preserving completely positive maps~\cite[Theorem~8.16]{wolf_notes},~\cref{eq_decoupling_thrm} further implies that
\begin{align} \label{eq_decoupling_DPI}
    \norm{\omega_{XE}- \frac{\id_{X}}{d_X}\otimes \omega_{E}}_1 \leq \norm{\omega_{CXE} - \frac{\id_{CX}}{d_C d_X}\otimes \omega_E}_1 \overset{\textnormal{\Cshref{eq_decoupling_thrm}}}{\leq} 13\eps \, ,
\end{align}
and as a consequence,
\begin{align}
    \norm{\omega_{CXE} - \frac{\id_{C}}{d_C}\otimes \omega_{XE}}_1 &\leq  \norm{\omega_{CXE} - \frac{\id_{CX}}{d_C d_X}\otimes \omega_E}_1 + \norm{\frac{\id_{CX}}{d_C d_X}\otimes \omega_E - \frac{\id_{C}}{d_C}\otimes \omega_{XE}}_1 \\
    &=\norm{\omega_{CXE} - \frac{\id_{CX}}{d_C d_X}\otimes \omega_E}_1 + \norm{\frac{\id_{X}}{d_X}\otimes \omega_E -  \omega_{XE}}_1 \\
    \overset{\textnormal{\Cshref{eq_decoupling_thrm,eq_decoupling_DPI}}}&{\leq} 26\eps \label{eq_decoupling_thrm-2}.
\end{align}
The above decoupling statement now allows us to construct the retrieval map $\cR$ as following. Since the protocol consists of a unitary channel followed by an isometric channel, the state $\omega_{C X B_1^n S  E}$ is pure. On the other hand, we define \smash{$\ket{G}_{CDXE\bar{E}} := (\ket{\Phi_{d_A}}^{\otimes k})_{CD}\otimes \ket{\xi}_{XE\bar{E}}$} as a purification of the state $ \frac{\id_{C}}{d_C}\otimes \omega_{XE}$ with additional purifying systems $D$ and $\bar{E}$, and where $\ket{\xi}_{XE\bar{E}}$ is any purification of $ \omega_{XE}$.\footnote{Without loss of generality, we can choose the dimension of the $\bar{E}$ system large enough such that the $D\bar{E}$ system is at least as large as the $B_1^n S$ system.} Now, due to Uhlmann's theorem and the Fuchs-van de Graaf inequalities~\cite{fuchs99}, there exists a purification $\tilde{\omega}_{CDXE\bar{E}}$ of $\omega_{CXE}$ such that
\begin{align} \label{eq_uhlmann-fdv-purifications}
    \Big\lVert\tilde{\omega}_{CDXE\bar{E}} -\proj{G}_{CDXE\bar{E}} \Big\rVert_1 \leq 2\sqrt{ \norm{\omega_{CXE} - \frac{\id_{C}}{d_C}\otimes \omega_{XE}}_1} \overset{\textnormal{\Cshref{eq_decoupling_thrm-2}}}{\leq} 2\sqrt{26\eps} \, .
\end{align}
In addition, since $\omega_{C X B_1^n S  E}$ is also a purification of $\omega_{C X E}$, there exists an isometry $W_{D\bar{E}|B_1^n S}$ acting on the purifying systems that relates these two purifications,
\begin{align}
    W_{D\bar{E}|B_1^n S}\, \omega_{C X B_1^n S  E}\, W^\dagger_{D\bar{E}|B_1^n S} = \tilde{\omega}_{CDXE\bar{E}} \overset{\textnormal{\Cshref{eq_uhlmann-fdv-purifications}}}{\approx} \proj{G}_{CDXE\bar{E}} = (\proj{\Phi_{d_A}}^{\otimes k})_{CD}\otimes \proj{\xi}_{XE\bar{E}}\, . \label{eq_GM_s1}
\end{align}
Finally, by exploiting the explicit form of the state $\omega_{C X B_1^n S  E} $ after applying the measurement isometry $V_{XS|A_{k+1}^n}$, a direct calculation shows that if the isometry $W$ is applied on the $B_1^n X$ system instead of the $B_1^n S$ system, the output state remains the same (since the systems $X$ and $S$ are essentially copies of each other), i.e.,
\begin{align}
     W_{D\bar{E}|B_1^n X}\, \omega_{C X B_1^n S  E}\, W^\dagger_{D\bar{E}|B_1^n X} \overset{\textnormal{\Cshref{eq_GM_s1}}}{\approx} (\proj{\Phi_{d_A}}^{\otimes k})_{CD}\otimes \proj{\xi}_{SE\bar{E}} =: \chi_{CDSE\bar{E}}\, .
\end{align}
Thus, the retrieval function $\cR$ is just defined by the isometry $W_{D\bar{E}|B_1^n X}$ and yields the maximally entangled output state $(\proj{\Phi_{d_A}}^{\otimes k})_{CD}$ up to an error that scales as $O(\sqrt{\eps})$. This concludes the protocol.

It remains to be understood how large one can choose $k$ for the protocol to yield $\chi_{CD} \approx \proj{\Phi_{d_A}}^{\otimes k}$.
This is quantified by~\cref{eq_condition_decoupling_1} which ensures that the decoupling theorem holds. Since $\rho_{A_1^n B_1^n E} = \proj{\Psi}_{A_1^n B_1^n E}$, the duality of the smooth min- and max-entropy~\cite[Proposition~6.14]{marco_book} ensures that
\begin{align} \label{eq_duality_decoupling}
    H^{\eps}_{\min}(A_1^n|E)_{\rho} \overset{\textnormal{duality}}{=} - H^{\eps}_{\max}(A_1^n|B_1^n)_{\rho}\, .
\end{align}
Furthermore, the simple structure of $\tau_{A_1^n C X}$ implies
\begin{align} \label{eq_decoupling2_ass}
    H^{\eps}_{\min}(A_1^n|C X)_{\tau} \geq H_{\min}(A_1^n|C X)_{\tau} 
    \overset{\textnormal{\Cshref{eq_CJ-rep}}}{=}  H_{\min}(A_1^k|C)_{\tau} + \underbrace{H_{\min}(A_{k+1}^n|X)_{\tau}}_{=0} 
    = -k \, .
\end{align}
With the help of~\cref{eq_duality_decoupling,eq_decoupling2_ass}, we can strengthen the decoupling condition given in~\cref{eq_condition_decoupling_1} to
\begin{align} \label{eq_condition_decoupling_2}
k   \,\overset{!}{\leq}\, - H^{\eps}_{\max}(A_1^n|B_1^n)_{\rho} - 2 \log\left( \frac{1}{\eps}\right) 
%\\  \overset{\textnormal{\Cshref{prop_strong_AEP}}}&{=} - n H(A|B)_{\Phi} - 2 \log \left(\frac{1}{\eps}\right) +o(n) \\
% &= n - 2 \log \left(\frac{1}{\eps}\right) +o(n)   \, .
\end{align}
This shows that we can choose the maximal number of iid Bell pairs to be $ k= \lfloor - H^{\eps}_{\max}(A_1^n|B_1^n)_{\rho} - 2 \log\left( \frac{1}{\eps}\right) \rfloor $. Moreover, note that it is clear from \Cref{fig_Bell_protocol} that the protocol is $\LOCC$ between the systems $A_1^n$ and $B_1^n$. We emphasise that both the unitary $U_{A_1^n}$ and the retrieval map $\cR$, defined by the isometry $W_{D\bar{E}|B_1^n X}$, are constructed depending on the initial state $\rho_{A_1^n B_1^n}$ and therefore, the overall $\LOCC$ protocol is dependent on $\rho_{A_1^n B_1^n}$.

Finally, we remark that due to Nielsen's theorem~\cite{nielsenChuang_book}, one can further convert $k$ iid copies of a maximally entangled state $\ket{\Phi_{d_A}}$ into $M\geq k$ iid copies of Bell ebits $\ket{\Phi_2}$ via $\LOCC$ operations, i.e.,
\begin{align}
\ket{\Phi_{d_A}}^{\otimes k} \quad \overset{\LOCC}{\longrightarrow} \quad \ket{\Phi_2}^{\otimes M}\,,\qquad M \,=\, \lfloor k\,\log_2(d_A) \rfloor \,.
\end{align}

%--------------------------------------------------------------------------
\begin{lemma}[Universal Bell decoupling protocol] \label{lem_universal-decoupling}
Let $n,\,w\in \N$ with $w\leq n$ and $\delta \in (0,1)$. Let $\St^n_w(\cH_{AB}, \ket{\Phi_2})$ denote the set of $\binom{n}{w}$-almost-iid states along $\Phi_2$ with $\dim(\cH_A)= \dim(\cH_B)=2$.
% and let $\cT_{n,w} := \dim \St^n_w(\cH_{AB}, \Phi_2)$. 
Then, there exists $k\in \N$ and a (one-way) $\LOCC$ channel $\cD_n^{w,\delta}$ satisfying
\begin{align}
\sup_{\omega^{(n)} \in \St^n_w(\cH_{AB}, \ket{\Phi_2})}  \Big\|\cD_n^{w,\delta} \big(\omega^{(n)}\big) - \Phi_2^{\otimes k} \Big\|_1 \,\leq\, \delta
\end{align}
and
\begin{align}
k \,\geq \, n-5\log D_n^w -4\log \left(\frac{1}{\delta}\right) - c_0  \,,
\end{align}
where $ D_n^w := \sum_{j=0}^w \binom{w}{j}3^j \leq 2^{nh(w/n) +2w}$ and $c_0 :=  2\log(104) + 1$.
\end{lemma}

\begin{proof}
For $n \in \N$ and $w\in\{0,\ldots,n\}$, we define the following (Hamming) space
\begin{align} \label{eq:bounded-defect-subspace}
\cK_n^w (\Phi_2) \,:=\, \mathrm{span} \left\{ (I_{A^n}\ox P_{B^n})\ket{\Phi_2}^{\ox n} \,\middle|\, P\in\{I,X,Y,Z\}^{\ox n},\ \mathrm{wt}(P)\leq w \right\},
\end{align}
and its dimension is given by
\begin{align} \label{eq:bounded-defect-dimension}
D_n^w \,:= \dim\cK_n^w (\Phi_2) \,=\, \sum_{j=0}^{w}\binom{n}{j}3^j\,.
\end{align}
Note that for a Pauli string $P = P_1\ox\cdots\ox P_n  \in  \{I,X,Y,Z\}^{\ox n} $ with $I,X,Y,Z$ denoting the single-qubit Pauli operators, we write $\mathrm{wt}(P)$ for the number of factors $P_i$ that are different from the identity.  Thus, $\mathrm{wt}(P)$ simply counts the pairs on which $P$ acts non-trivially. Furthermore, the vectors in~\cref{eq:bounded-defect-subspace} are maximally entangled and form an orthonormal basis of $\cK_n^w (\Phi_2)$ since, for any two Pauli strings $P$ and $Q$, it holds
\begin{align} \label{eq:Pauli-Bell-orthogonality}
\bra{\Phi_2}^{\ox n} (I\ox P^\dagger Q) \ket{\Phi_2}^{\ox n} \,=\, \frac{1}{2^n}\tr(P^\dagger Q) \,=\, \delta_{P,Q}\,.
\end{align}
Next, let $\Pi_n^w$ denote the (orthogonal) projector onto $\cK_n^w (\Phi_2)$ and we introduce the universal test state
\begin{align}
\Omega_{A^nB^n} \,:= \,\frac{\Pi_n^w}{D_n^w}\,.
\end{align}
The proof idea is to design a decoupling protocol depending on the state $\Omega_{A^nB^n}$. Then we will show that, due to the structure of $\Omega_{A^nB^n}$, the constructed decoupling protocol will also work for any state that is supported in $\cK_n^w (\Phi_2)$. In the following, we will write $D \equiv D_n^w$ for brevity.

\medskip
\noindent \textbf{Step 1:} \textit{(Constructing the decoupling channel)}
Let $\{\ket{\Phi_\ell}\}_{\ell=1}^{D}$ denote the orthonormal basis of $\cK_n^w (\Phi_2)$ consisting of maximally entangled vectors $\ket{\Phi_\ell}$, and define a purification of $\Omega_{A^nB^n}$ as
\begin{align} \label{eq:moving-note-universal-test-purification}
\ket{\Gamma}_{A^nB^n R} \,:=\, \frac1{\sqrt{D}} \sum_{\ell=1}^{D} \ket{\Phi_\ell}_{A^nB^n}  \ket{\ell}_R\,.
\end{align}
We now apply the decoupling protocol described above to the state $\Omega_{A^nB^n}$, using $\ket{\Gamma}_{A^nB^n R}$ as its purification and treating $R$ as the purifying environment. Then, for any $\eps \in (0,1)$ and $k\in \N_0$ with
\begin{align} \label{eq:universal-sufficient-condition}
k   \,\leq \, - H^{\eps}_{\max}(A_1^n|B_1^n)_{\Omega} - 2 \log\left( \frac{1}{\eps}\right),
\end{align}
we obtain a (one-way) $\LOCC$ channel $\cD^{w,\eps}_n$ obeying
\begin{align} \label{eq:moving-note-universal-Choi-error}
\Big\| (\mathrm{id}_R\otimes\cD^{w,\eps}_n)\big(\proj{\Gamma}_{A^nB^n R}\big) -\tfrac{\id_R}{D}\otimes\Phi_2^{\otimes k} \Big\|_1 \,\leq \, 2\sqrt{26\eps}\,.
\end{align}
% Notice that the map $\cD^{w,\eps}_n$ depends on the parameters $n$, $w$, and $\eps$.
As a next step, we bound the conditional smooth max-entropy of $\Omega_{A^nB^n}$. Since every basis vector $\ket{\Phi_\ell}$ is maximally entangled and  $\Pi_n^w=\sum_{\ell=1}^{D}\proj{\Phi_\ell}$, it holds for any state $\xi_{B^n} \in \St(\cH_B^{\otimes n})$ that
\begin{align}\label{eq:projector-overlap}
\tr\big[\Pi_n^w(\id_{A^n}\ox\xi_{B^n})\big] \,=\, \sum_{\ell=1}^{D} \tr_{B^n}\big[\tr_{A^n}(\proj{\Phi_\ell})\xi_{B^n}\big] \,= \, \sum_{\ell=1}^{D} \frac{1}{2^n}\tr(\xi_{B^n}) \,= \, \frac{D}{2^n}\,.
\end{align}
By further using that $\| X \|_1 \leq \sqrt{\mathrm{rank}(X)} \|X\|_2$ for any hermitian operator $X$, we can evaluate following the fidelity expression (using the convention $F(\rho,\sigma) := \|\sqrt{\rho}\sqrt{\sigma}\|_1^2$) as
\begin{align}\label{eq:fidelity-first-line}
F(\Omega_{A^nB^n},\id_{A^n}\ox \xi_{B^n}) \,&=\, \frac{1}{D} \norm{\Pi_n^w\big(\id_{A^n}\ox\sqrt{\xi_{B^n}}\big)}^2_1 \\
\,&\leq\,  \norm{\Pi_n^w\big(\id_{A^n}\ox\sqrt{\xi_{B^n}}\big)}^2_2 \\
\,&=\, \tr\big[\Pi_n^w(\id_{A^n}\ox\xi_{B^n})\big] \\
\overset{\textnormal{\Cshref{eq:projector-overlap}}}&{=} \frac{D}{2^n}\,, \label{eq:fidelity-last-line}
\end{align}
where in the second step, we used that $\mathrm{rank}\big(\Pi_n^w(\id_{A^n}\ox\sqrt{\xi_{B^n}})\big) \leq \mathrm{rank}(\Pi_n^w) = D$. By the definition of the conditional max-entropy, we therefore find
\begin{align} \label{eq:universal-conditional-max-entropy}
H_{\max}(A^n|B^n)_\Omega \,=\, \sup_{\xi_{B^n}\in\St(\cH^{\ox n}_{B})} \log  F(\Omega_{A^nB^n},\id_{A^n}\ox\xi_{B^n}) \,\overset{\textnormal{\Cshref{eq:fidelity-last-line}}}{\leq}\,  \log D - n\,.
\end{align}
In fact, equality holds in~\cref{eq:universal-conditional-max-entropy} as can be seen by choosing $\xi_{B^n} = \frac{\id_{B^n}}{2^n}$. Hence, since smoothing can only decrease the max-entropy, we find
\begin{align} \label{eq:smooth-max-entropy-bound}
H_{\max}^{\eps}(A^n|B^n)_\Omega \,\leq\, H_{\max}(A^n|B^n)_\Omega \,=\, \log D - n\,.
\end{align}
%Intuitively, this bound shows that $n$ perfect Bell pairs contribute $-n$ to the conditional max-entropy, whereas the uncertainty about a $D$-dimensional family of possible defect vectors costs at most $\log D$. 
Thus, the sufficient condition~\cref{eq:universal-sufficient-condition} for the number of output Bell pairs $k$ is satisfied whenever
\begin{align} \label{eq:moving-note-Hamming-decoupling-condition}
k \,\leq\, n -\log D_n^w - 2 \log\left( \frac{1}{\eps}\right)\,.
\end{align}

\medskip
\noindent \textbf{Step 2:} \textit{(From the universal test state to an arbitrary almost-iid input)}
First, we argue that if a state $\omega^{(n)} \in \St(\cH_{AB}^{\ox n})$ is $\binom{n}{w}$-almost-iid along $\ket{\Phi_2}$, then its support is contained in the space $\cK_n^w(\Phi_2)$ defined in~\cref{eq:bounded-defect-subspace}. To see this, recall that any $\omega^{(n)} \in   \St^n_w(\cH_{AB}, \ket{\Phi_2})$ satisfies 
\begin{align}
\supp(\omega^{(n)}) \subseteq  \mathrm{span}\,\cV^n_w(\cH_{AB},\ket{\Phi_2}_{AB})
\end{align}
by definition. The set $\cV^n_w(\cH_{AB},\ket{\Phi_2}_{AB})$ consists of vectors of the form $\ket{\psi_j}_{A^n B^n } =\pi_j \big(\ket{\Phi_2}_{AB}^{\otimes n-w} \otimes \ket{\xi_j}_{A^w B^w}\big)$ for some permutation $\pi_j$ and state $\ket{\xi_j}_{A^w B^w} \in \cH_{AB}^{\ox w}$. As any pure state $\ket{\xi_j}_{A^w B^w} $ may be expressed as $\ket{\xi_j}_{A^w B^w}  = (\id_{A^w}\otimes X_{B^w})\ket{\Phi_2}_{AB}^{\otimes w} $ for some operator $X_{B^w}$ (e.g., see~\cite{wolf_notes}), and the operator $X_{B^w}$ may be expanded in the (multi-qubit) Pauli basis, it follows that $\ket{\psi_j}_{A^n B^n } \in \cK_n^w(\Phi_2)$ and thus, 
\begin{align} \label{eq:almost-iid-support-Hamming}
\supp(\omega^{(n)}) \,\subseteq \,  \cK_n^w(\Phi_2)\,.
\end{align}
%\begin{align}
%X_{B^w} = \sum_{P \in\{I,X,Y,Z\}^{\ox w}} c_{P} P\,,
%\end{align}

\medskip
\noindent Next, let $\bar{\cD}_n^{\omega,\epsilon}$ denote the restriction of $\cD_n^{\omega,\epsilon}$ to input states that are supported in $\cK_n^w(\Phi_2)$, and further define a replacer channel $\cR_\Phi(Z)$ on this subspace by
\begin{align} \label{eq:replacer-channel}
\cR_\Phi(Z) \,:=\, \tr(Z)\Phi_2^{\otimes k}\,.
\end{align}
Then, we have
\begin{align}
(\mathrm{id}_R \ox\cR_{\Phi})\big(\proj{\Gamma}_{A^nB^n R}\big) \,=\, \tfrac{\id_R}{D}\ox\Phi_2^{\ox k}\,.
\end{align}
Moreover, since $\ket{\Gamma}_{A^nB^n R}$ is a maximally entangled vector between the $D$-dimensional system $\cH_R$ and the $D$-dimensional subspace $\cK_M^w(\Phi_2)$, the left hand side of~\cref{eq:moving-note-universal-Choi-error} corresponds precisely to the trace norm of the (normalized) Choi operator $J(\Delta)$ with
\begin{align}
\Delta \,:=\, \bar{\cD}_n^{w,\epsilon} -\cR_{\Phi}\,,\qquad J(\Delta)\,:=\, (\Delta \ox \mathrm{id}_R)\big(\proj{\Gamma}_{A^nB^n R}\big)\,.
\end{align}
Now, the idea is to use the trace norm of the Choi operator $J(\Delta)$ as a bound for the diamond norm of the channel difference operator $\Delta$, i.e.,
\begin{align} \label{eq:Choi-to-diamond}
\norm{\Delta}_\diamond \,\leq\, D\norm{J(\Delta)}_1\,.
\end{align}
One way to justify~\cref{eq:Choi-to-diamond} is to recall that any pure state $\ket{\psi} \in \cK_n^w(\Phi_2) \otimes \cH_R$ may be written as $\ket{\psi} = ( \id_{\cK_n^w} \otimes X_R)\ket{\Gamma}$ for some operator $X$ satisfying $\tr(X^\dagger X)=D$ (e.g., see~\cite{wolf_notes}), and consequently,
\begin{align} \label{eq:Choi-to-diamond-derivation}
\norm{(\Delta \ox \mathrm{id}_R)(\proj{\psi})}_1 = \norm{ ( \id \otimes X_R)J(\Delta)( \id \otimes X^\dagger_R)}_1 \leq \norm{X}_\infty^2 \norm{J(\Delta)}_1 = D\norm{J(\Delta)}_1\,.
\end{align}
By convexity, the same bound holds for mixed inputs, and~\cref{eq:Choi-to-diamond} follows (see also~\cite[Exercise 3.6]{watrous_2018}). Putting everything together, we find for any almost-iid input state $\omega^{(n)} \in  \St^n_w(\cH_{AB}, \ket{\Phi_2})$ that
\begin{align}
 \Big\|\cD_n^{w,\eps} \big(\omega^{(n)}\big) - \Phi_2^{\otimes k} \Big\|_1 \,\overset{\textnormal{\Cshref{eq:replacer-channel}}}&{=}\,  \Big\|{\cD}_n^{w,\eps} \big(\omega^{(n)}\big) - \cR_\Phi\big(\omega^{(n)} \big) \Big\|_1 \\
  \,\overset{\textnormal{\Cshref{eq:almost-iid-support-Hamming}}}&{\leq}\, \norm{ \bar{\cD}_n^{w,\epsilon} -\cR_{\Phi}}_\diamond \\
   \,\overset{\textnormal{\Cshref{eq:Choi-to-diamond}}}&{\leq} \, D\norm{ (\Delta \ox \mathrm{id}_R)\big(\proj{\Gamma}_{A^nB^n R}\big)}_1 \\
    \,\overset{\textnormal{\Cshref{eq:moving-note-universal-Choi-error}}}&{\leq}\,2D\sqrt{26\eps}\,.
\end{align}
Since the bound above is independent of the state $\omega^{(n)}$, we obtain the same bound after taking the supremum over all states $\omega^{(n)} \in  \St^n_w(\cH_{AB}, \ket{\Phi_2})$. Finally, for $\delta \in (0,1)$ we set
\begin{align}
\eps \,:=\, \frac{\delta^2}{104D^2}\qquad \implies \qquad 2D\sqrt{26\eps} \,=\, \delta\,.
\end{align}
Inserting this choice into~\cref{eq:moving-note-Hamming-decoupling-condition} yields
\begin{align}
k \,\leq \, n -\log D_n^w -2\log \left(\frac{104D^2}{\delta^2}\right)  \,=\, n-5\log D_n^w -4\log \left(\frac{1}{\delta}\right) -2\log(104) \,.
\end{align}
Setting $k:= \big\lfloor n-5\log D_n^w -4\log \left(\frac{1}{\delta}\right) -2\log(104) \big\rfloor$ concludes the proof.
\end{proof}

This finite block lemma becomes particularly useful when the number of defects $w$ is sublinear in the block length $n$. The following proposition shows how to exploit this to find a universal sequence of decoupling protocols that functions for any almost-iid input $\pmb{\omega} \in \cA(\Phi_2)$.

\begin{proposition}[Universal Bell decoupling with unit rate] \label{prop:unit-rate-universal-sanitizer}
There exists a universal sequence of (one-way) $\LOCC$ channels $\pmb{\cD} = (\cD_n)_{n\in\N}$ such that every almost-iid Bell state sequence $\pmb{\omega} \in \cA(\Phi_2)$ is converted into a sequence of perfect iid Bell pairs,
\begin{align}
\pmb{\cD} \big(  \cA(\Phi_2)\big) \subseteq B\big(\pmb{\Phi}_2^{\mathrm{iid}}\big)\,,
\end{align}
with a subexponentially decreasing error and asymptotic rate one.
\end{proposition}

\begin{proof}
For any $n\in \N$, set
\begin{align}
w_n\, :=\,
\min\bigg\{n, \left\lceil\frac{n}{\log(n+1)}\right\rceil \bigg\}\,, \qquad \delta_n\,:=\,2^{-\sqrt n}\,,
\end{align}
and invoke~\cref{lem_universal-decoupling} to define a sequence of (one-way) $\LOCC$ operations \smash{$\pmb{\cD} := \big(\cD^{w,\delta}_n\big)_{n\in\N}$}. Now, let $\pmb{\omega} \in \cA(\Phi_2)$ be arbitrary. By definition, there exist constants $L\geq0$ and $\kappa \in [0,1)$ such that $\omega^{(n)}$ is \smash{$\binom{n}{r_n}$}-almost-iid along $\ket{\Phi_2}$ with $r_n = \lceil Ln^\kappa \rceil$ for each $n\in \N$. Furthermore, it is straightforward to see that there exists $N_0 \in \N$ such that $r_n \leq w_n$ for all $n\geq N_0$. Hence, in this case we have $\omega^{(n)} \in \St^n_{w_n}(\cH_{AB}, \ket{\Phi_2})$ and~\cref{lem_universal-decoupling} therefore implies that
\begin{align}
 \Big\|\cD_n^{w,\delta} \big(\omega^{(n)}\big) - \Phi_2^{\otimes k} \Big\|_1 \,\leq\, \delta_n\,, \qquad \textrm{with} \quad k \,\geq \, n-5\log D_n^w -4\log \left(\frac{1}{\delta_n}\right) - c_0  \,,
\end{align}
where $ D_n^w \leq 2^{nh(w/n) +2w}$ and $c_0$ is a constant. The approximation error $\delta_n$ is subexponentially decreasing in $n$ by construction. Moreover, we have
\begin{align}
\lim_{n\to\infty} \frac{k}{n} \,\geq\, \lim_{n\to\infty} \bigg[1 - 5 h\Big(\frac{w_n}{n}\Big) + 10 \frac{w_n}{n} - 4\frac{1}{\sqrt{n}} - \frac{c_0}{n}\bigg] \,=\, 1
\end{align}
since $w_n = o(n)$ is sublinear. Thus, $\pmb{\cD}$ is a universal $\LOCC$ operation with rate one.
\end{proof}

%------------------------------------------------------------------
\subsubsection{Single-shot properties}
So far, we have discussed transformation rates in the asymptotic regime. We now connect these asymptotic quantities to the single-shot definitions of entanglement cost and distillation, following the standard approach. We expect the same single-shot formulas to underlie both the operational and conventional frameworks, with the distinction arising only in how the corresponding asymptotic quantities are obtained through regularization. We follow the notation introduced in~\cite{Gour_2025}.

For any bipartite finite-dimensional system $\cH_{AB}$ and for any $M\in \N$ with $M\leq \min\{d_A,d_B\}$, let
\begin{equation}
  |\Phi_M\rangle\, :=\, \frac{1}{\sqrt M}\sum_{j=1}^{M}\ket{j}_A\ket{j}_B\,, \qquad \Phi_M \,:=\,\proj{\Phi_M}\,.
\end{equation}
denote a maximally entangled state of rank $M$. The family of maximally entangled states $\{\Phi_M\}_{M\in \N}$ satisfies the following two properties,
\begin{align}
\Phi_N \otimes \Phi_M \, \sim \, \Phi_{NM}\,,\qquad N\geq M \ \implies \ \Phi_N\, \xrightarrow{\LOCC} \, \Phi_M\,,
\end{align} 
where the relation $\rho \sim \sigma$ indicates that the two states $\rho$ and $\sigma$ can be converted exactly into each other by $\LOCC$ operations. Note that in more general quantum resource theories, any family of states satisfying the above properties is sometimes referred to as \emph{golden unit}, see~\cite[Definition 11.1]{Gour_2025}.
\begin{definition-2}[Single-shot transformation rates~{\cite[Definition 11.2]{Gour_2025}}] \label{def_single-shot-rates}
Let $\rho \in \St(\cH_{AB})$ and $\eps\in (0,1)$. The \emph{$\eps$-single-shot entanglement cost} and \emph{distillation of} $\rho$ are defined as
\begin{align}
 \operatorname{Cost}^{\eps}(A:B)_\rho \,&:=\, \min\Big\{ \log M \,:\, \inf_{\cP \in \LOCC}\tfrac{1}{2} \norm{\cP(\Phi_M) - \rho }_1 \leq \eps\,, \ M\in\N   \Big\},  \label{eq:ss-cost} \\
 \operatorname{Distill}^{\eps}(A:B)_\rho \,&:=\,  \max\Big\{  \log M \,:\,   \inf_{\cP \in \LOCC}\tfrac{1}{2} \norm{\cP(\rho) - \Phi_M }_1 \leq\eps\,, \ M\in\N \Big\}, \label{eq:ss-distill} 
\end{align}
respectively.
\end{definition-2}
Let us further recall a few known relations. For any $\rho \in \St(\cH_{AB})$, the traditional (iid) asymptotic rates are related to the single-shot quantities via
\begin{align}
E_C(A:B)_{\rho} \overset{\textnormal{\cite[Lem.~11.10]{Gour_2025}}}&{=} \lim_{\eps \to 0^+} \liminf_{n\to \infty} \frac{1}{n}  \operatorname{Cost}^{\eps}(A^n:B^n)_{\rho^{\otimes n}} \overset{\textnormal{\cite{HHT01}}}{=}  \lim_{\eps \to 0^+} \limsup_{n\to \infty} \frac{1}{n}  \operatorname{Cost}^{\eps}(A^n:B^n)_{\rho^{\otimes n}}, \label{eq:EC-asymptotic-single-shot} \\
E_D(A:B)_{\rho} \overset{\textnormal{\cite[Lem.~11.10]{Gour_2025}}}&{=} \lim_{\eps \to 0^+} \limsup_{n\to \infty} \frac{1}{n}  \operatorname{Distill}^{\eps}(A^n:B^n)_{\rho^{\otimes n}} \overset{\textnormal{\cite{Rains99}}}{=}  \lim_{\eps \to 0^+} \liminf_{n\to \infty} \frac{1}{n}  \operatorname{Distill}^{\eps}(A^n:B^n)_{\rho^{\otimes n}}. \label{eq:ED-asymptotic-single-shot}
\end{align}
%Notice the distinction between the $\liminf$ and $\limsup$ expressions, which arises from different, though ultimately equivalent, asymptotic definitions of entanglement cost and distillation. To illustrate this point, let us focus on entanglement cost. In~\cite[Definition 11.3]{Gour_2025}, cost is defined by requiring that \emph{there exists} numbers $n,\,m\in \N$ such that a conversion $\Phi_2^{\otimes n} \xrightarrow{\LOCC} \rho^{\otimes m}$ is approximately possible. In addition, one can further argue that the values $n,\,m$ may be chosen arbitrarily large. This existence gives rise to the $\liminf_{n\to \infty} $ in~\cref{eq:EC-asymptotic-single-shot}. On the other hand, a definition of entanglement cost as given in~\cref{eq:EC_def-traditional-forall} or~\cref{def_iid_rate}, implicitly requires that an approximate conversion $\Phi_2^{\otimes n} \xrightarrow{\LOCC} \rho^{\otimes m}$ must eventually hold \emph{for all} values of $n$, giving rise to the $\limsup_{n\to \infty} $ in~\cref{eq:EC-asymptotic-single-shot}. These two definitions are in fact equivalent as is further discussed in~\cite{HHT01}. Within our framework, the definition of operational asymptotic rates are analogous to~\cref{def_iid_rate}.
Notice the distinction between the $\liminf$ and $\limsup$ expressions, which reflects different, though ultimately equivalent, asymptotic definitions of entanglement cost and distillation. To illustrate this point, let us focus on entanglement cost. In~\cite[Definition 11.3]{Gour_2025}, the cost is defined by requiring the \emph{existence} of integers $k,n\in\mathbb N$ for which the approximate conversion \smash{$ \Phi_2^{\otimes k}\xrightarrow{\LOCC}\rho^{\otimes n}$} is possible at a given accuracy and rate. Moreover, it is easily shown that the target blocklength $n$ may be chosen arbitrarily large. Thus, the definition only requires arbitrarily large blocklengths at which the normalized one-shot cost is small; it does not require such a bound to hold at every sufficiently large blocklength. This quantifier structure is captured by the $\liminf_{n\to\infty}$ in~\cref{eq:EC-asymptotic-single-shot}. By contrast, the definitions in~\cref{eq:EC_def-traditional-forall,def_iid_rate} require such approximate conversions to exist \emph{for every} sufficiently large target blocklength $n$. This uniform requirement leads to the corresponding $\limsup_{n\to\infty}$. In finite-dimensional bipartite entanglement theory under $\LOCC$, these two definitions of entanglement cost are equivalent, as shown in~\cite{HHT01}, while for more general quantum resource theories, such an equivalence does not trivially follow. The definition of operational asymptotic rates adopted in our framework is analogous to that in~\cref{def_iid_rate}. Finally, note that an analogous distinction arises for entanglement distillation, with the roles of $\liminf$ and $\limsup$ reversed.

It is further shown in~\cite[Theorem 13.20]{Gour_2025} that for any $\eps\in (0,\tfrac{1}{2})$, the $\eps$-single-shot entanglement cost can be related to a smoothed conditional max-entropy, namely
\begin{align} \label{eq:single-shot-cost-hmax}
 \inf_{\rho_{ABX}} {H}^{2\eps}_{0}(A|X)_\rho \,\leq\, \operatorname{Cost}^{\eps}(A:B)_\rho \,\leq\, \inf_{\rho_{ABX}} {H}^{\eps^2}_{0}(A|X)_\rho\,,
\end{align}
where the infimum ranges over all classical systems $X$ and all extensions $\rho_{ABX}$ of $\rho_{AB}$.\footnote{Note that in~\cite[Theorem 13.20]{Gour_2025}, the $\eps$-single-shot entanglement cost is defined with respect to the squared purified distance $P^2$ instead of the trace distance $T$. Since these are related by means of the Fuchs-van de Graaf inequalities, $T(\rho, \sigma)^2 \leq P^2(\rho,\sigma) \leq 2T(\rho, \sigma)$ for any states $\rho$ and $\sigma$,~\cite[Theorem 13.20]{Gour_2025} then implies~\cref{eq:single-shot-cost-hmax}.} Here (following the notation of~\cite{Gour_2025}), the conditional max-entropy is defined via
\begin{align}
{H}^{\eps}_{0}(A|X)_\rho := \min_{\omega_{XA} \in B_\eps(\rho_{XA})} {H}_{0}(A|X)_\omega\,,\quad {H}_{0}(A|X)_\omega := \max_{\tau_X \in \St(\cH_X)} \log \tr\big[ \Pi^{XA}_\omega(\tau_X \otimes \id_A)\big]
\end{align}
with %$B_\eps(\rho) = \big\{\omega : \frac{1}{2}\norm{\omega-\rho}_1\leq \eps \big\}$ and 
$\Pi^{XA}_\omega$ denoting the projector onto the support of $\omega_{XA}$. Next, we show how the operational asymptotic rates can be retrieved from the single-shot quantities. Due to~\cref{cor_ED-equal}, it follows directly that~\cref{eq:ED-asymptotic-single-shot} also holds for the operational entanglement distillation $\tilde{E}_D(A:B)_\rho$. For the operational entanglement cost, we obtain the following relation.

\begin{lemma} \label{lem_single-shot-cost}
Let $\rho \in \St(\cH_{AB})$. Then,
\begin{align}
\tilde{R}_{\LOCC}(\Phi_2 \to \rho) \,=\, \sup\Big\{\, r >0\ \big| \ \exists\ {\pmb{\mathcal{P}}}^{(r)} \in \LOCC \,:\ {\pmb{\mathcal{P}}}^{(r)}\big( \pmb{\Phi}_2^{\mathrm{iid}}\big) \, \in \, B\big(\mathcal{A}(\rho)\big) \Big\} \,.
\end{align}
\end{lemma}
\begin{proof}
Let
\begin{align} \label{eq:single-shot-cost-0}
T( \Phi_2\to \rho ) \, :=&\,  \sup\Big\{\, r>0\ \big| \ \exists\ {\pmb{\mathcal{P}}}^{(r)} \in \LOCC \,:\ {\pmb{\mathcal{P}}}^{(r)}\big( \pmb{\Phi}_2^{\mathrm{iid}}\big) \, \in \, B\big(\mathcal{A}(\rho)\big) \Big\} \,. 
\end{align}
It immediately follows that $\tilde{R}_{\LOCC}(\Phi_2 \to \rho) \leq T( \Phi_2\to \rho )$ since $\pmb{\Phi}_2^{\mathrm{iid}} \in \cA(\Phi_2)$. For the other direction, let $r>0$ be achievable for $T( \Phi_2\to \rho )$, i.e., there exists a sequence of free operations ${\pmb{\mathcal{P}}}^{(r)} \in \LOCC$ such that 
\begin{align} \label{eq:single-shot-cost-1}
\pmb{\mathcal{P}}^{(r)}\big( \pmb{\Phi}_2^{\mathrm{iid}}\big)  \in  B\big(\mathcal{A}(\rho)\big)\,.
\end{align}
By~\cref{prop:unit-rate-universal-sanitizer}, there exists a universal sequence of Bell decoupling protocols $\pmb{\cD}$ satisfying
\begin{align}  \label{eq:single-shot-cost-2}
\pmb{\cD} \big(  \cA(\Phi_2)\big) \subseteq B\big(\pmb{\Phi}_2^{\mathrm{iid}}\big)
\end{align}
with rate one. Hence, the composition $\tilde{\pmb{\cP}}\vphantom{\pmb{\cP}}^{(r)}:={\pmb{\mathcal{P}}}^{(r)}  \circ\pmb{\cD} $ is again a sequence of $\LOCC$ operations with rate $r$, and it satisfies
\begin{align}  \label{eq:single-shot-cost-3}
\tilde{\pmb{\cP}}\vphantom{\pmb{\cP}}^{(r)} \big(  \cA(\Phi_2) \big) \,\overset{\textnormal{\Cshref{eq:single-shot-cost-2}}}{\subseteq} \, {\pmb{\mathcal{P}}}^{(r)}\big(B\big(\pmb{\Phi}_2^{\mathrm{iid}}\big) \big)  \,\overset{\textnormal{\Cshref{lem_data-processing-neg-ball}}}{\subseteq} \, B\big( {\pmb{\mathcal{P}}}^{(r)}\big(\pmb{\Phi}_2^{\mathrm{iid}}\big) \big) \,\overset{\textnormal{\Cshref{eq:single-shot-cost-0}}}{\subseteq} \, B\big(\mathcal{A}(\rho)\big)\,.
\end{align}
Therefore, $r$ is achievable for $\tilde{R}_{\LOCC}(\Phi_2 \to \rho)$, implying that $r\leq \tilde{R}_{\LOCC}(\Phi_2 \to \rho)$. After taking the supremum over all achievable rates $r>0$, we find $T( \Phi_2\to \rho ) \leq \tilde{R}_{\LOCC}(\Phi_2 \to \rho)$, and the claim follows.
\end{proof}

\begin{proposition}
Let $\rho \in \St(\cH_{AB})$. Then,
\begin{gather}
%\tilde{E}_D(A:B)_\rho \,=\, \inf_{f \in \mathrm{SD}} \limsup_{n\to \infty} \frac{1}{n}  \operatorname{Distill}^{\eps}(A^n:B^n)_{\rho^{\otimes n}} \,=\, \lim_{\eps \to 0^+} \limsup_{n\to \infty} \frac{1}{n}  \operatorname{Distill}^{\eps}(A^n:B^n)_{\rho^{\otimes n}}\,, \\
\tilde{E}_C(A:B)_\rho \,=\, \inf_{\pmb{\omega}\in \cA(\rho)} \inf_{f \in \mathsf{SD}} \limsup_{n\to \infty} \frac{1}{n}  \operatorname{Cost}^{f(n)}(A^n:B^n)_{\omega^{(n)}} \,,
\end{gather}
where $\mathsf{SD} := \big\{f : \N \to (0,1) \ |\  f \textrm{ is subexponentially decreasing}\, \big\}$.
\end{proposition}
\begin{proof}
For the proof, we introduce the notation
\begin{align}
C_*(\rho) \,:=&\, \inf_{\pmb{\omega}\in \cA(\rho)} \inf_{f \in \mathsf{SD}} \limsup_{n\to \infty} \frac{1}{n}  \operatorname{Cost}^{f(n)}(A^n:B^n)_{\omega^{(n)}}\,.
\end{align}
Our goal is to show that $\tilde{E}_C(A:B)_\rho  = C_*(\rho)$.

\medskip
\noindent
\underline{$ \tilde{E}_C(A:B)_\rho  \leq C_*(\rho)$:} Fix any $c>C_*(\rho)$. By definition, there exist $\pmb{\omega}\in \cA(\rho)$ and $f\in  \mathsf{SD}$ such that
\begin{equation}  \label{eq:below-c}
  \limsup_{n\to\infty} \frac1n\operatorname{Cost}^{f(n)}(A^n:B^n)_{\omega^{(n)}}\,<\, c \,.
 \end{equation}
Choose any $r \in (0,\frac{1}{c}) $ and set $m_n:=\lceil rn\rceil$ for any $n\in \N$. Then $\lim_{n\to \infty}\frac{m_n}{n} = r$, and by~\cref{eq:below-c}, we have
\begin{equation}  \label{eq:cost-below-n}
 \log M_n \,:=\, \operatorname{Cost}^{f(m_n)}(A^{m_n}:B^{m_n})_{\omega^{(m_n)}} \,<\, c m_n\, <\, n 
\end{equation}
for all sufficiently large $n$, where the last inequality further used that $cr<1$. Furthermore, choose any subexponentially decreasing function $\delta \in \mathsf{SD}$. Then, according to~\cref{def_single-shot-rates}, there exists an $\LOCC$ map $\cP$ such that
\begin{align}
\tfrac{1}{2} \norm{\cP(\Phi_{M_n}) -\omega_{m_n}}_1 \,\leq  \, f(m_n) + \delta(n)\,.
\end{align}
On the other hand, since $M_n < 2^n$ due to~\cref{eq:cost-below-n}, there exists an exact $\LOCC$ conversion mapping $\Phi_2^{\otimes n} \to \Phi_{2^n} \to \Phi_{M_n}$. Concatenating these two mappings yields an $\LOCC$ protocol $\tilde{\cP}_n$ such that
\begin{align}
\tfrac{1}{2} \norm{\tilde{\cP}_n(\Phi_2^{\otimes n}) -\omega_{m_n}}_1 \,\leq  \, f(m_n) + \delta(n)\,.
\end{align}
Since $m_n = \lceil rn \rceil$ scales linearly in $n$, the function $ f(m_n)$ and hence the sum $ f(m_n) + \delta(n)$ is subexponentially decreasing in $n$. With~\cref{lem_single-shot-cost}, this therefore implies that $r$ is achievable for $\tilde{R}_{\LOCC}(\Phi_2 \to \rho)$ and thus, $r\leq \tilde{R}_{\LOCC}(\Phi_2 \to \rho)$. By taking the supremum over all $r\in (0,\frac{1}{c})$ and $c >C_*(\rho)$, we find $\frac{1}{C_*(\rho)}\leq \tilde{R}_{\LOCC}(\Phi_2 \to \rho)$, which proves the first direction.

\medskip
\noindent
\underline{$ \tilde{E}_C(A:B)_\rho  \geq C_*(\rho)$:} 
Let $r> 0$ be achievable for $\tilde{R}_{\LOCC}(\Phi_2 \to \rho)$, i.e., there exists a sequence of operations ${\pmb{\mathcal{P}}}^{(r)} \in \LOCC$ such that 
\begin{align} \label{eq:asymptotic-single-shot-1}
\pmb{\mathcal{P}}^{(r)}\big( \pmb{\Phi}_2^{\mathrm{iid}}\big)\, \simeq\, \pmb{\tau} \,\in\, \mathcal{A}(\rho)\,.
\end{align}
Recalling that $ \tilde{E}_C(A:B)_\rho \leq  E_C(A:B)_\rho = E^\infty_F(A:B)_\rho < \infty$, we can assume $\tilde{R}_{\LOCC}(\Phi_2 \to \rho)>0$. The free operations satisfy $\cP^{(r)}_N \in \LOCC(C^N \to D^{m_N})$ for all $N\in \N$ with $C :=\bar{A}\bar{B}$, $D:= AB$, and 
\begin{align}
\lim_{N\to\infty} \frac{m_N}{N} = r\,.
\end{align} 
Next, fix any number $s \in (0,r)$ and, for any $n\in \N$ representing the target blocklength, define the input copy number $N_n := \lceil \frac{n}{s}\rceil$. Note that for sufficiently large values of $n$, we find 
\begin{align}
m_{N_n}\, =\, \frac{m_{N_n}}{N_n}N_n \,>\, sN_n\, \geq \, n\,.
\end{align} 
Hence, in this case one may trace out the remaining $m_{N_n} - n$ systems after applying the protocol $\cP_{N_n}$, which yields
\begin{align}  \label{eq:asymptotic-single-shot-2}
\frac{1}{2}\norm{\tr_{m_{N_n} - n}\circ \cP_{N_n}(\Phi_2^{\otimes N_n}) - \omega^{(n)}}_1 \,\overset{\textnormal{DPI}}{\leq}\,\frac{1}{2}\norm{ \cP_{N_n}(\Phi_2^{\otimes N_n}) - \tau^{(m_{N_n})}}_1 \,\overset{\textnormal{\Cshref{eq:asymptotic-single-shot-1}}}{\leq}\, f(m_{N_n})\,,
\end{align}
where we defined $ \omega^{(n)} :=\tr_{m_{N_n}-n}\big[\tau^{(m_{N_n})}\big] $, and $f\in \mathsf{SD}$ is a subexponentially decreasing function. Since
\begin{align}
\frac{m_{N_n}-n}{m_{N_n}} \,= \,  1-\frac{n}{N_n}\frac{N_n}{m_{N_n}}   \,\overset{n\to \infty}{\longrightarrow}\,  1 - \frac{s}{r} \,<\,1\,,
\end{align}
it follows by~\cref{lem_subpoly-almost-iid} that $\pmb{\omega} = \big( \omega^{(n)}\big)_{n\in\N} \in \cA(\rho)$. In addition, the error function $h(n) :=f(m_{N_n}) $ is also subexponentially decreasing since $m_{N_n} = \Theta(n)$, hence $h\in \mathsf{SD}$. Consequently,
\begin{align}
 \operatorname{Cost}^{h(n)}(A^n:B^n)_{\omega^{(n)}} \,\overset{\textnormal{\Cshref{eq:asymptotic-single-shot-2}}}{\leq}\, N_n
\end{align}
for all sufficiently large $n \in \N$, and therefore
\begin{align}
C_*(\rho)\, \leq\,  \limsup_{n\to\infty}\frac{N_n}{n}\,  =\, \frac{1}{s}\,.
\end{align}
After taking the supremum over all achievable $r>0$ and $s\in (0,1)$ we obtain $C_*(\rho) \leq \frac{1}{ \tilde{R}_{\LOCC}(\Phi_2 \to \rho)}$, which proves the second direction.

\end{proof}

\newpage
\appendix
\section*{Appendix}
%\newpage
\section{Technical statements}

%------------------------------ Typical sequences
\begin{proposition}[Typical sequences] \label{prop_typical-sequences} \phantom{....} 
\begin{enumerate}[(i)]

\item  \label{lem_typical_itm1} Let $P_X$ be a probability distribution on $\cX$ and let $\vec{x} = (x_1, \dots ,x_n)$ be an $n$-tuple chosen according to the i.i.d.~product distribution $(P_X)^n$. Then, for any $\epsilon \geq 0$,
\begin{align}
\Pr_{\vec{x}} \Big[ D\big( \lambda_{\vec{x}} \| P_X \big) >  \epsilon \Big] \, \leq \, 2^{-n\big(\epsilon - |\cX|\frac{\log{(n+1)}}{n} \big)}\,,
\end{align}
where $\lambda_{\vec{x}}$ denotes the empirical frequency distribution with respect to $\vec{x}$, and where $D$ denotes the relative entropy.

\item \label{lem_typical_itm2} Let $P_X$ be a probability distribution on $\cX$ and let $\vec{x} = (x_1, \dots ,x_n)$ be an $n$-tuple chosen according to the i.i.d.~product distribution $(P_X)^n$. Then, for any $\epsilon \geq 0$,
\begin{align}
\Pr_{\vec{x}} \Big[ \big\|  \lambda_{\vec{x}} - P_X \big\|_1 >  \epsilon \Big] \, \leq \, 2^{-n\big(\frac{\epsilon^2}{2 \ln{2}} - |\cX|\frac{\log{(n+1)}}{n} \big)}\,,
\end{align}
where $\lambda_{\vec{x}}$ denotes the empirical frequency distribution with respect to $\vec{x}$.

\item \label{lem_typical_itm3} Let $\vec{y} = (y_1, \dots ,y_N)$ be a fixed $N$-tuple with elements in $\cX$, and let $\vec{x} =  (x_1, \dots ,x_n) $ be sampled at random from $\vec{y}$ with or without replacement, where $n\leq N$. Then, for any $\epsilon \geq 0$,
\begin{align}
\Pr_{\vec{x}} \Big[ \big\|  \lambda_{\vec{x}} - \lambda_{\vec{y}} \big\|_1 >  \epsilon \Big] \, \leq \, 2^{-n\big(\frac{\epsilon^2}{2 \ln{2}} - |\cX|\frac{\log{(n+1)}}{n} \big)}\,,
\end{align}
where $\lambda_{\vec{y}}$ and  $\lambda_{\vec{x}}$ denote the empirical frequency distributions with respect to $\vec{y}$ and  $\vec{x}$, respectively.

\end{enumerate}
\end{proposition}

\begin{proof}
\begin{itemize}
\item[\eqref{lem_typical_itm1}] See~\cite[Theorem 11.2.1]{cover}.

\item[\eqref{lem_typical_itm2}] The claim follows directly by combining~\eqref{lem_typical_itm1} and~\cref{lem_Pinsker}.

\item[\eqref{lem_typical_itm3}] First, consider the case where $\vec{x} =  (x_1, \dots ,x_n) $ is sampled at random from $\vec{y}$ \emph{with} replacement. This scenario is equivalent to choosing the $n$-tuple $\vec{x} = (x_1, \dots ,x_n)$ according to the i.i.d.~product distribution $( \lambda_{\vec{y}})^n$. Hence,~\eqref{lem_typical_itm2} directly applies with $P_X$ replaced by $\lambda_{\vec{y}}$, and the assertion follows.

For the case where $\vec{x} =  (x_1, \dots ,x_n) $ is sampled at random from $\vec{y}$ \emph{without} replacement, the proof idea is to relate this case to the analogue scenario with replacement in the spirit of Hoeffding's result~\cite[Theorem 4]{Hoeffding1963}. Accordingly, let $\vec{z} =  (z_1, \dots ,z_n) $ be sampled at random from $\vec{y}$ with replacement. Let $D_{\vec{x}} = D\big( \lambda_{\vec{x}} \|  \lambda_{\vec{y}} \big)$ and $D_{\vec{z}} = D\big( \lambda_{\vec{z}} \|  \lambda_{\vec{y}}  \big)$. By applying a Chernoff bound, we find for any $\epsilon \geq 0$ and $t>0$ that
\begin{align}
\Pr_{\vec{x}} \big[ D_{\vec{x}} > \epsilon \big] \,=\, \Pr_{\vec{x}} \big[2^{tD_{\vec{x}}} > 2^{t\epsilon}\big] \,\leq\, 2^{-t\epsilon}\, \E_{\vec{x}}\big[2^{tD_{\vec{x}}} \big]\,,
\end{align}
where the last step is Markov's inequality. As a next step, we apply~\cite[Corollary A.2.d.]{Marshall2011} to the function $\xi(\vec{x}) = 2^{t D_{\vec{x}}}$, leading to
\begin{align}
\E_{\vec{x}}\big[2^{tD_{\vec{x}}} \big] \,\leq\, \E_{\vec{z}}\big[2^{tD_{\vec{z}}} \big]\,.
\end{align}
This corollary is applicable since \textit{(a)} the relative entropy is jointly convex~\cite{cover} and in particular, the map $ \nu \mapsto D(\nu \|  \lambda_{\vec{y}} )$ is convex in the distribution $\nu$, \textit{(b)} the exponential map $u \mapsto 2^{tu}$ is convex and increasing for $u \in \R$ and any $t >0$, and \textit{(c)} hence the composite map $f( \nu ) = 2^{t D(\nu \|  \lambda_{\vec{y}} )}$ is convex. Furthermore, it is straightforward to check that any map of the form $\xi(\vec{x}) = f(\lambda_{\vec{x}})$ with $f$ convex satisfies the requirements listed in~\cite[Section 12.A., Eqs.~(3) \& (4)]{Marshall2011}. Choosing $t=n$, we further obtain
\begin{align}
 \E_{\vec{z}}\big[2^{nD_{\vec{z}}} \big] \,=\, \sum_{\nu \in \cP_n} \Pr_{\vec{z}} \big[\lambda_{\vec{z}} = \nu \big]\, 2^{n D(\nu \|  \lambda_{\vec{y}} ) } \leq |\cP_n|\,,
\end{align}
where $\cP_n$ denotes the set of types with denominator n, and the last step follows due to the bound \smash{$ \Pr_{\vec{z}} \big[\lambda_{\vec{z}} = \nu \big] \leq 2^{-n D(\nu \|  \lambda_{\vec{y}} )}$} for any $\nu \in \cP_n$, see~\cite[Theorem 11.1.4]{cover}. Furthermore, we know that \smash{$|\cP_n| \leq (n+1)^{|\cX|}$}, see~\cite[Theorem 11.1.1]{cover}. By combining all of the above, we then find
\begin{align}
\Pr_{\vec{x}} \big[ D\big( \lambda_{\vec{x}} \|  \lambda_{\vec{y}} \big) > \epsilon \big]\,\leq\, 2^{-n\epsilon}\, \E_{\vec{x}}\big[2^{nD_{\vec{x}}} \big] \, \leq \, 2^{-n\epsilon}\, \E_{\vec{z}}\big[2^{nD_{\vec{z}}} \big] \, \leq \, (n+1)^{|\cX|} \, 2^{-n\epsilon} \,.
\end{align}
 Finally, combining the above relation with~\cref{lem_Pinsker} yields
 \begin{align}
 \Pr_{\vec{x}} \Big[ \big\|  \lambda_{\vec{x}} - \lambda_{\vec{y}} \big\|_1 >  \epsilon \Big] \, \leq \, (n+1)^{|\cX|} \, 2^{-n\frac{\epsilon^2}{2 \ln{2}}} \,=\,2^{-n\big(\frac{\epsilon^2}{2 \ln{2}} - |\cX|\frac{\log{(n+1)}}{n} \big)}\,.
 \end{align}

\end{itemize}
\end{proof}

%------------------------------ Pinsker
\begin{lemma}[Pinsker's inequality] \label{lem_Pinsker}
Let $P$ and $Q$ be probability distributions. Then
\begin{align}
\big\| P-Q \big\|_1 \,\leq\, \sqrt{(2\ln{2})\,D\big(P\|Q\big)}\,,
\end{align}
where $D$ denotes the relative entropy.
\end{lemma}
\begin{proof}
See~\cite[Lemma 11.6.1]{cover}.
\end{proof}

%------------------------------ Gentle measurement lemma & Lemma about trace and euclidean norm
\begin{lemma}[Gentle measurement lemma~{\cite[Corollary 3.15]{watrous_2018}}] \label{lem_gentle-measurement}
Let $\rho\in \St(\cH_C)$ be a density operator, and let $P\in \Pos(\cH_C)$ be a positive semidefinite operator satisfying $P\leq \id_C$ and $\tr(P\rho) >0$. Then it holds that
\begin{align}
    F\left(\rho\,,\frac{\sqrt{P}\,\rho\,\sqrt{P}}{\tr(P\rho)} \right) \,\geq\, \tr(P\rho) \,,
\end{align}
where $F(\rho,\sigma) := \|\sqrt{\rho}\sqrt{\sigma}\|_1^2$ denotes the fidelity. 
\end{lemma}

\begin{lemma} \label{lem_norm-identities}
    Let $\ket{\Phi}$, $\ket{\Psi}\,\in\cH$ be any two (nonzero) vectors in a Hilbert space, and let $\| \cdot\|_1$ and $\| \cdot\|_2$ denote the trace norm and the Euclidean norm, respectively. Then, the following holds:
    \begin{enumerate}[(i)]
    \item \(
        \big\| \ket{\Psi} \bra{\Phi}\big\|_1 = \big\|\ket{\Psi} \big\|_2 \,\big\|\ket{\Phi} \big\|_2\,.
    \) \label{lem_norm-identities-it1}
    \item If both $\ket{\Phi}$ and $\ket{\Psi}$ are normalized, then
    \(
        \big\| \ket{\Psi} -\ket{\Phi}\big\|_2 \,=\, \sqrt{2\big(1-\mathrm{Re}(\braket{\Psi|\Phi})\big)} \,.
        %\,\leq\,\sqrt{2\big(1-|\braket{\Psi|\Phi}|\big)}
    \) \label{lem_norm-identities-it2}
    \item If both $\ket{\Phi}$ and $\ket{\Psi}$ are normalized, then
    \begin{align*}
        \tfrac{1}{2}\big\| \proj{\Psi} -\proj{\Phi}\big\|_1 \,\leq\, \big\| \ket{\Psi} -\ket{\Phi}\big\|_2 \sqrt{1-\frac{ \big\| \ket{\Psi} -\ket{\Phi}\big\|^2_2}{4}}  \,\leq\,  \big\| \ket{\Psi} -\ket{\Phi}\big\|_2 \,.
    \end{align*} \label{lem_norm-identities-it3}
     \item If both $\ket{\Phi}$ and $\ket{\Psi}$ are normalized and $\braket{\Psi|\Phi}\geq 0$, then
    \begin{align*}
       \big\| \ket{\Psi} -\ket{\Phi}\big\|_2  \,\leq\, \sqrt{2}\,  \tfrac{1}{2}\big\| \proj{\Psi} -\proj{\Phi}\big\|_1   \,.
    \end{align*} \label{lem_norm-identities-it4}
\end{enumerate}
\end{lemma}
\begin{proof} Recall that $\|A\|_1 = \tr\big[\sqrt{A^\dagger A}\,\big]$ for any linear operator $A$ on $\cH$, and $\big\|\ket\Psi \big\|_2 = \sqrt{\braket{\Psi|\Psi}}$ for any vector $\ket\Psi \in \cH$.

\noindent \eqref{lem_norm-identities-it1} In particular, for $A:=\ket{\Psi} \bra{\Phi} $, we have
\begin{align}
    A^\dagger A \,=\,\braket{\Psi|\Psi} \proj{\Phi}\,=\, \big\|\ket{\Psi}  \big\|_2^2\, \proj{\Phi} \quad \implies \quad \sqrt{A^\dagger A}\,=\,\frac{\big\|\ket{\Psi}  \big\|_2}{\big\|\ket{\Phi}  \big\|_2}  \,\proj{\Phi}\,.
\end{align}
Hence,
\begin{align}
    \big\| \ket{\Psi} \bra{\Phi}\big\|_1\,=\,\big\|A\big\|_1 \,=\, \tr\left[\frac{\|\ket{\Psi}  \|_2}{\|\ket{\Phi}  \|_2}  \,\proj{\Phi} \right] \,=\,\big\|\ket{\Psi}  \big\|_2\,\big\|\ket{\Phi}  \big\|_2\,.
\end{align}

\noindent \eqref{lem_norm-identities-it2} Since $\braket{\Psi|\Psi} = \braket{\Phi|\Phi} =1$, we find
\begin{align}
     \big\| \ket{\Psi} -\ket{\Phi}\big\|_2^2 \,=\,\big(\bra{\Psi} -\bra{\Phi}\big)\big(\ket{\Psi} -\ket{\Phi}\big)\,=\,2- \big(\braket{\Psi|\Phi}+\braket{\Phi|\Psi}\big)\,=\,2\big(1-\mathrm{Re}(\braket{\Psi|\Phi})\big)\,.
\end{align}

\noindent \eqref{lem_norm-identities-it3} First, note that
\begin{align}
|\braket{\Psi|\Phi}| \,\geq\, \big|\mathrm{Re}(\braket{\Psi|\Phi})\big| \,\overset{\eqref{lem_norm-identities-it2}}{=}\, \bigg| 1 - \frac{  \big\| \ket{\Psi} -\ket{\Phi}\big\|^2_2}{2} \bigg|\,.
\end{align}
Then, we evaluate
\begin{align}
 \tfrac{1}{2}\big\| \proj{\Psi} -\proj{\Phi}\big\|_1 \,\overset{\textnormal{F.v.d.G.}}{=}\, \sqrt{1-|\braket{\Psi|\Phi}|^2} \,\leq\,  \big\| \ket{\Psi} -\ket{\Phi}\big\|_2 \sqrt{1-\frac{ \big\| \ket{\Psi} -\ket{\Phi}\big\|^2_2}{4}}  \,\leq\,  \big\| \ket{\Psi} -\ket{\Phi}\big\|_2 \,.
\end{align}

\noindent \eqref{lem_norm-identities-it4} Because of the Fuchs-van de Graaf inequality and $\braket{\Psi|\Phi}\geq0$, it holds that
\begin{align}
\epsilon \,:=\,  \tfrac{1}{2}\big\| \proj{\Psi} -\proj{\Phi}\big\|_1 \,=\, \sqrt{1-|\braket{\Psi|\Phi}|^2} \quad \implies \quad \braket{\Psi|\Phi} \,=\, \sqrt{1-\epsilon^2} \,.
\end{align}
Furthermore, note that
\begin{align}
1- \sqrt{1-\epsilon^2} \,=\, \frac{(1- \sqrt{1-\epsilon^2})(1+ \sqrt{1-\epsilon^2})}{1+ \sqrt{1-\epsilon^2}} \,=\, \frac{\epsilon^2}{1+ \sqrt{1-\epsilon^2}}\,.
\end{align}
With that, we evaluate
\begin{align}
  \big\| \ket{\Psi} -\ket{\Phi}\big\|_2 \,\overset{\eqref{lem_norm-identities-it2}}{=}\,  \sqrt{2\big(1-\braket{\Psi|\Phi} \big)} \,=\,\sqrt{2\big(1- \sqrt{1-\epsilon^2}\big)} \,=\, \epsilon\,\sqrt{\frac{2}{1+ \sqrt{1-\epsilon^2}}} \,\leq \,\epsilon\,\sqrt{2}\,,
\end{align}
which proves the claim.
\end{proof}
\newpage
 \begin{lemma} \label{lem_number-theory}
 For any $j \in (0,1)$, consider the set $S_{j} := \big\{n \in \N \,:\,  n = \ell + \lceil \ell^{j} \rceil \,, \  \ell \in \N  \big\}$. Then, for any $s\in \N$ with $s\geq 2$, it holds that
 \begin{align}
 S_{\frac{1}{s}} \,\cup\,  S_{\frac{1}{2s}} \,=\,\N \setminus \{1,3\}\,.
 \end{align}
 \end{lemma}
 
 \begin{proof}
First we note that for any $j \in (0,1)$ and $n(\ell):=  \ell + \lceil \ell^{j}\rceil $, the gap between two values, $\Delta n = n(\ell+1) - n(\ell)$, satisfies $1\leq \Delta n \leq 2$ for all $\ell\in \N$. Furthermore, we argue that the integers $\{1,3\}$ always correspond to skipped integers for any choice of $j$. To see that, we simply evaluate the first few terms: $n(0) = 0, \,n(1) = 2,\, n(2) = 4$, which is true independently of $j$.

Now, for rationals $j = \frac{1}{s}$ with $s\in \N$ and $s\geq 2$, the gaps occur precisely at positions $\ell^s$, and indeed we evaluate
\begin{align}
n(\ell^s) \,=\,  \ell^s + \ell\,,\quad n(\ell^s+1) \,=\,  \ell^s+1 + \lceil (\ell^s+1)^{\frac{1}{s}}\rceil \,=\,  \ell^s+1 + \ell + 1 = n(\ell^s) + 2   \,. 
\end{align} 
This means that the set of skipped integers is given by $\bar{S}_{\frac{1}{s}} = \big\{\ell^s +\ell +1 \, |\, \ell \in \N \big\}$. Similarly, one finds $\bar{S}_{\frac{1}{2s}} = \big\{\ell^{2s} +\ell +1 \, |\, \ell \in \N \big\}$. Next, we aim to show that $\bar{S}_{\frac{1}{s}} \cap \bar{S}_{\frac{1}{2s}} = \{3\}$. In fact, any element $\bar{n} \in \bar{S}_{\frac{1}{s}} \cap \bar{S}_{\frac{1}{2s}}$ would need to satisfy
\begin{align} \label{lem_number-theory-it}
\bar{n}\,=\,x^s +x +1 \,=\,y^{2s} + y + 1 \quad \iff \quad x^s +x  \,=\,y^{2s} + y 
\end{align}
for some $x,\,y \in \N$. If $x = y = 1$, we recover $\bar{n} = 3$, which is always skipped as observed before. For $y=1$ and $x \geq 2$,~\cref{lem_number-theory-it} is clearly never satisfied. Now consider the remaining general case with $x\in \N$ and $y \geq 2$.
%assume (by contradiction) that~\cref{lem_number-theory-it} holds with $x,\,y \geq 2$.
Define the function $f(\ell) := \ell^s + \ell $ and evaluate
\begin{align}
f(y^2) \,&=\,  y^{2s} + y^2  \, > \, y^{2s} + y  \,, \\
f(y^2-1) \,&=\,  (y^2-1)^s + y^2  - 1 \,\overset{(*)}{<}\,y^{2s} + y\,,
\end{align}
where the step $(*)$ follows from the identity $a^s - b^s =(a-b) \sum_{i=0}^{s-1} a^{s-i-1} b^i$ by setting $a= y^2$ and $b= y^2-1$, implying that
\begin{align}
y^{2s} -  (y^2-1)^s \,&=\,  \sum_{i=0}^{s-1}  y^{2(s-i-1)} (y^2-1)^i \,=\, y^{2s-2} +  \sum_{i=1}^{s-1}  y^{2(s-i-1)} (y^2-1)^i  \\
\,&> \,  y^{2s-2} \\
 \,&>\,y^2 -y -1 \quad \forall\, s\geq 2 \,.
\end{align}
Hence, we have shown that 
\begin{align}
 f(y^2-1) < y^{2s} + y < f(y^2)\,.
\end{align}
But since there is no integer value lying between $y^2-1$ and $y^2$, there cannot exist any $x \in \N$ such that $f(x) = y^{2s} + y$ and thus,~\cref{lem_number-theory-it} cannot hold in this case. This proves that $\bar{S}_{\frac{1}{s}} \cap \bar{S}_{\frac{1}{2s}} = \{3\}$, which immediately implies that $ S_{\frac{1}{s}} \,\cup\,  S_{\frac{1}{2s}} \,=\,\N \setminus \{1,3\}$ as desired.
 \end{proof}

\newpage
\bibliographystyle{arxiv_no_month}
\bibliography{bibliofile}

\end{document}